\documentclass[
  aps,              %
  pra,              %
  10pt,             %
  onecolumn,        %
  amsmath,amssymb,  %
  groupedaddress,   %
  nofootinbib,      %
  longbibliography, %
  floatfix          %
]{revtex4-2}

\usepackage[T1]{fontenc}
\usepackage{lmodern}

\usepackage{amsthm}
\usepackage{bm}

\usepackage{graphicx}
\usepackage{pgfplots}
\pgfplotsset{compat=1.18}
\usepgfplotslibrary{groupplots,fillbetween}
\usetikzlibrary{backgrounds,fadings,arrows.meta,decorations.pathreplacing}
\pgfplotsset{table/search path={figs}}
\pgfplotsset{
  paperaxis/.style={
    tick align=outside, tick pos=left, line width=0.5pt,
    grid style={line width=0.3pt, draw=black!12},
    legend style={draw=black!25, fill=white, fill opacity=0.9, text opacity=1, rounded corners=1pt},
  }
}
\graphicspath{{figs/}}

\makeatletter
\newcommand{\figladder}{%
  \renewcommand{\tiny}{\@setfontsize\tiny{6}{7}}%
  \renewcommand{\scriptsize}{\@setfontsize\scriptsize{8}{9.5}}%
  \renewcommand{\footnotesize}{\@setfontsize\footnotesize{9}{11}}%
  \renewcommand{\small}{\@setfontsize\small{10}{12}}%
  \renewcommand{\normalsize}{\@setfontsize\normalsize{10.95}{13.6}}%
  \renewcommand{\large}{\@setfontsize\large{12}{14}}%
}
\makeatother
\usepackage{etoolbox}
\AtBeginEnvironment{tikzpicture}{\figladder\normalsize}

\usepackage{booktabs}

\usepackage{microtype}
\usepackage[dvipsnames]{xcolor}
\newtheorem{theorem}{Theorem}
\newtheorem{proposition}{Proposition}
\newcommand{\braket}[1]{\langle #1 \rangle}
\newcommand{\ket}[1]{\lvert #1 \rangle}

\usepackage[colorlinks=true,allcolors=blue,breaklinks=true]{hyperref}
\makeatletter
\g@addto@macro\UrlBreaks{\UrlOrds}
\makeatother
\hypersetup{
  pdftitle={Captured weight and boundary leakage bound the error of sample-based spectral functions},
  pdfauthor={Nicolas Bonilla Vargas},
  pdfkeywords={sample-based quantum diagonalization, QSCI, spectral functions, a posteriori error bound, Hubbard model, quantum simulation}
}

\begin{document}

\title{Captured weight and boundary leakage bound the error of
sample-based spectral functions}

\author{Nicol\'as Bonilla Vargas}
\email{ngbonillav@unal.edu.co}
\affiliation{Universidad Nacional de Colombia, Bogot\'a, Colombia}
\affiliation{SRH University, Munich, Germany}
\affiliation{Daita AI, Munich, Germany}

\begin{abstract}
Rayleigh--Ritz on a subspace from sample-based quantum diagonalization yields the spectral function
of a probe state; we ask what bounds its error. The probe's captured Born weight $w$ does not do so
alone: an analytic counterexample forces any bound indexed on it to its trivial value. On Hubbard
rings at broadening $\eta=0.18\,t$, the probe's whole support ($w=1$) still gives a relative error
of $0.35$--$0.43$: Rayleigh--Ritz displaces the poles a subspace keeps. The error is bounded instead
by the missed weight and the Hamiltonian coupling across the subspace boundary. Given the exact
ground state, we prove for any orthogonal projector that the relative $L_1$ error lies between $1-w$
and $\min\{1+w,(1+\sqrt{w})(\sqrt{1-w}+\widehat\Lambda(\eta)/\eta)\}$, where $\widehat\Lambda$,
that coupling at resolution $\eta$, is computed from the Ritz pairs and the retained probe. We test
it on $L$-site rings with classical Born draws of an exactly time-evolved probe, or their
infinite-shot ranking; our device runs are execution records only. On the ranked subspaces the
bound is vacuous at every operating fraction of the resource scan, at all five sizes up to a sector
of $10\,306\,296$ determinants: its leakage branch exceeds the trivial bound $1+w$ by factors of
$1.66$ to $8.61$. Only well above those fractions is it informative, and there it is calibrated: at
$L=6$--$12$ it first beats the trivial bound at a true error of $(1.2$--$7.3)\times10^{-3}$, though
an unproven fit in $1-w$ is tighter on 11 of 14 splits. For the Born-ranked $A(k,\omega)$ at $85\%$
of its sector, the bound is $350$--$900$ times the true error. In the determinant basis, zero
leakage at full weight needs every symmetry-allowed determinant: the probe's measured Krylov
support at $L=6$--$14$. Plane-wave orbitals reduce a momentum probe's support to one block, $1/L$ of
an exponentially growing sector. At full weight a second-order certificate follows; whether
one exists below full weight, where every operating fraction lies, is open.
\end{abstract}

\maketitle

\section{The operational question}
\label{sec:intro}

A sample-based quantum calculation returns a list of bitstrings. Diagonalizing the Hamiltonian in
the span of the configurations they name yields an energy and, as this work and several others show,
a frequency-resolved response as well. What it does not yield is a statement of how wrong that
response is. This paper addresses the question a user faces once the shots are spent:
\emph{given the subspace those shots populated, how wrong is the spectral function built on it?}

The natural diagnostic is the captured Born weight, the fraction of the probe's squared norm that
lies in the subspace; the probe is the ground state acted on by the operator whose response is
sought. It does not
answer the question. Rayleigh--Ritz does not delete the poles a subspace misses; it displaces the
poles it keeps. A subspace that holds the probe's entire Born support, where the captured weight is
exactly one, therefore still carries an error of order one. On half-filled Hubbard rings of $L$
sites, at Lorentzian broadening $\eta=0.18\,t$, the relative $L_1$ error is $0.43$ at $L=6$ and $0.35$ at $L=8$
($0.28$ on the open chain of Sec.~\ref{sec:moments}, a different lattice). No choice of constant rescues the captured
weight: an analytic counterexample forces the constant of any bound indexed on it to its trivial
value (Sec.~\ref{sec:retraction}).

What enters a bound, together with the missed weight, is the Hamiltonian coupling across the
subspace boundary. Theorem~\ref{thm:leakage} makes this precise for Rayleigh--Ritz on any
orthogonal projector: the $L_1$ error of the broadened spectral function lies between the missed
Born weight and the smaller of the trivial bound and a leakage branch, built from that weight and
the coupling that leaks across the boundary at resolution $\eta$. The leakage term is a posteriori:
it has a closed form in the Ritz pairs the reconstruction produces and in the retained
probe. Both terms take the probe and the energy reference from the exact ground state, so the bound
is a certificate for benchmarking a sample-based pipeline where exact diagonalization is available. The identity behind it also shows that the error surviving at full
weight is generated by the boundary coupling alone (Sec.~\ref{sec:retraction}), whatever
value the bound itself takes there.

What the subspace structure does guarantee is moment exactness, the classical moment-matching
property of Krylov projections (Sec.~\ref{sec:moments}). A subspace containing the order-$K$ Krylov
space reproduces the spectral moments through order $2K+1$ and no further; we state that reach
exactly and measure it on coordinate subspaces (Theorem~\ref{thm:moments}). A non-circular count
then gives the number of shots needed to capture every high-probability configuration
(Proposition~\ref{prop:shots}).

We then measure where the bound is informative. The operating fractions of the resource scan are
its acceptance fractions: the fractions of the symmetry sector that the infinite-shot ranking of
the scan of Sec.~\ref{sec:prices} needs to meet this paper's $5\%$ accuracy criterion, relative
$L_1<0.05$ at $\eta=0.15\,t$. The leakage branch of the bound is worse than the
trivial bound at every operating fraction of the resource scan, at five sizes up to a sector of
dimension $10\,306\,296$, by factors of $1.66$ to $8.61$. At $L=12$ that verdict is converged in
Krylov depth. The bound becomes informative only well above those fractions, and to certify the
$5\%$ criterion at $\eta=0.18\,t$ it needs at least $0.98$ of the sector at $L=8$. An independent implementation recovers
the acceptance fractions of the scan to under $2\%$ and confirms the $5\%$ criterion on the
published subspaces: the reconstructions are reproduced, and what they lack is a certificate.

Where the bound is informative, it is calibrated. Over sixteen $(L,\eta)$ cells at $L=6$--$12$, it
first crosses below the trivial bound at a true relative error between $1.2\times10^{-3}$ and
$7.3\times10^{-3}$: a spread of $5.9$, over true errors that span five decades. Once it is below
that bound, the true error is already at that level. A two-constant rule in the missed weight,
fitted on some sizes and tested on the others, localizes this crossing more tightly on eleven of
fourteen splits. That rule is a regularity of these systems
with no guarantee outside its fitted range; the certificate is a proven inequality, given the exact
probe. The Born-ranked (infinite-shot) $A(k,\omega)$ of Sec.~\ref{sec:frontier}, at $85\%$ of its
sector, is non-vacuous in all sixteen momentum channels. Its relative bound there, $0.77$--$1.75$
against the trivial $2$, is $350$--$900$ times its true error.

The leakage depends on the basis in which the device measures. In the determinant basis, the Krylov
space of a probe created at one site, the seed site, has coordinate support on every determinant
that the reflection about the seed site allows, measured at $L=6$ to $14$. At full captured weight,
therefore, only a subspace holding all of them has zero leakage. If the device samples in
plane-wave orbitals instead, a momentum-resolved probe has its Krylov support in one momentum block,
$1/L$ of the sector. The rotation to that basis costs about $L(L-1)$ further two-qubit gates, whose
effect on the sampled distribution is unmeasured, and the block, like the sector, grows
exponentially with $L$
(Sec.~\ref{sec:frontier:obstruction}).

No orbital-rotation invariant can bound the boundary leakage usefully on every state: on fibers of a
pairing model, families of states on which the invariants are held fixed, every such bound is
vacuous. None can therefore replace the certificate as an a priori
predictor computed from the state before any circuit runs (Sec.~\ref{sec:nogo}). The natural
candidates are the efficiently measurable monotones of fermionic non-Gaussianity. On
number-conserving states the one-body member collapses to the single invariant
$\mathcal F_1=4\,\mathrm{tr}[\gamma(1-\gamma)]$, with $\gamma$ the one-particle reduced density
matrix [Eq.~\eqref{eq:collapse}]. An invariant of the occupation spectrum cannot see a boundary.
Along a geminal orbit every antiflatness $\mathcal F_k$, like every other orbital-rotation
invariant, is fixed, while
the orbit contains a point of classical cost $O(K)$, for $K$ geminals, in the fixed computational basis
(Theorem~\ref{thm:witness}). That basis is the pairing basis, in which the point has $2^{K}$
determinants but bond dimension $2$. The orbit also contains points at which Born sampling is
plausibly hard (Sec.~\ref{sm:nogo}), a hardness we do not establish. On those fibers the captured weight
is frozen too, yet the subspace-aware certificate still separates a non-vacuous value from the
trivial bound (Proposition~\ref{prop:leakinv}).

Within each of the six $(L,N)$ strata of the Hubbard interaction sweep, the rank correlation of
$\mathcal F_1$ with the determinant support is exactly $-1$. Its exact permutation
$p=3.35\times10^{-13}$ assumes independent orderings, a null that six strata sharing one mechanism
do not satisfy. The correlation is a deterministic monotone relation whose sign reverses between
families, as it must when an orbital invariant is compared with a basis-dependent count
(Sec.~\ref{sec:nogo}).

Sections~\ref{sec:prices} and~\ref{sec:physics} then price the method in support, resolution and
shots, and check the full-sector references against the opposite finite-size scaling of the charge
and spin gaps of the Hubbard chain.

Four bodies of work bound what is new here. \emph{(i)}~Patel, Rangi and Tam~\cite{patel2026}
construct Green's functions from sampled Krylov subspaces for dynamical mean-field
theory~\cite{dmft1996}, independently and concurrently, and cite this work as ``a closely related
sample-based approach to dynamical spectral functions, constructed directly from bitstring-sampled
subspaces''. \emph{(ii)}~Dynamical DMRG computed all three momentum-resolved channels reported here
for the half-filled extended Hubbard chain at $L=60$--$120$ in 2007~\cite{benthien2007}.
Correction-vector DMRG treats the Hubbard chain at $L=48$, the Heisenberg chain at $L=64$ (the
$k=\pi$ cut to $L=128$) and the $t$--$J$ ladder at $48\times2$ sites~\cite{nocera2016}, and the
classical toolbox for such spectra is broader
still~\cite{kuehner1999,whitefeiguin2004,holzner2011,kpm2006,paeckel2019,tdasci2024}.
\emph{(iii)}~Frequency-resolved responses of lattice models, all but one of them correlated, have
been measured on devices far larger than the $12$-qubit spectral run of this work. They include a
retarded spin Green's function of the Fermi--Hubbard chain at up to $30$ qubits~\cite{vilchez2025},
an angle-resolved-photoemission-like $A(k,\omega)$ of a free-fermion $27$-site chain at $54$
qubits~\cite{granet2026akw}, $S(q,\omega)$ on $50$ qubits against neutron scattering in
$\mathrm{KCuF_3}$~\cite{lee2026sqw}, and quench spectroscopy of $L=101$ spins~\cite{millar2026}. At a
scale comparable to ours, $S(q,\omega)$ of a $20$-site Heisenberg chain was computed on a trapped-ion
device and compared with neutron data for copper sulfate~\cite{quantinuum_dsf2026}. The two device
runs of this work are execution records only: the spectral one retained its entire sector, so its
agreement with exact diagonalization is exact by coverage rather than a test of the method
(Sec.~\ref{sec:scope:hw}). \emph{(iv)}~Computable Gaussian-invariant quantifiers of fermionic
non-Gaussianity include the antiflatness of Ref.~\cite{sierant2026} used here. Tarabunga
\emph{et al.}~\cite{tarabunga2026} separate covariance-matrix occupation entropies from full-state
natural-orbital participation entropies. Two members of the first family are strong monotones that
lower-bound the non-Gaussian gate count; the second family upper-bounds the classical simulation cost
in an orthonormal Gaussian basis. The quadratic member, which coincides with the $k=1$ antiflatness,
is one of the two strong monotones. Leone and Bittel~\cite{leonebittel2026} independently prove it an
efficiently measurable strong pure-state Gaussian monotone. Tarabunga \emph{et al.} and Leone and
Bittel posted their works both in July 2026; the priority for those measures is theirs. What this
work adds is the contrast between these measures and the basis-dependent determinant support, a
constructive instance of that contrast, and its extension from the determinant support to the
leakage. Every system studied here is classically tractable: Hubbard rings at $L\le14$ and nineteen
active spaces with a largest minimal bond dimension of $42$. No quantum advantage is claimed. The one
quantitative statement this work makes about its own cost---that the basis-dependent bond dimension
tracks the determinant support---argues, wherever that bond dimension is small, for a tensor network
rather than for a sampler (Sec.~\ref{sec:honest}).

The paper is organized as follows. Section~\ref{sec:method} fixes the sampling primitive and separates
shot-budgeted reconstructions from full-sector references. Section~\ref{sec:leakage}
proves the bound and shows that no bound on the captured weight alone can replace it;
Sec.~\ref{sec:moments} states what the subspace structure certifies; Sec.~\ref{sec:frontier} measures
where the certificate is informative; Sec.~\ref{sec:nogo} addresses a priori predictors; Secs.~\ref{sec:prices}
and~\ref{sec:physics} price the method and validate the references; and Sec.~\ref{sec:honest} reports
the device runs, the molecular suite and the relation to classical methods. The systems are the half-filled
one-dimensional Hubbard model at $U/t=8$, rings and open chains at $L=4$ to $14$, and nineteen molecular
active spaces. Three quantities share the symbol $\mathcal S$: the static
ground-state support of Sec.~\ref{sec:nogo}, the dynamical support of Sec.~\ref{sec:prices}, and the set
of configurations the sampler returned, the $\mathcal S$ of Theorem~\ref{thm:leakage}. The proof of
Theorem~\ref{thm:leakage} is in Appendix~\ref{app:proof}; the other proofs, full tables and extended
discussions are in the Supplemental Material~\cite{supplemental}, cited as Sec.~S$n$.

\section{One sampling primitive, four channels: which are reconstructed and which are references}
\label{sec:method}

The quantum processor is asked for one thing: computational-basis bitstrings of a time-evolved
response seed. The subspace those bitstrings name, its diagonalization and every spectral function built
on it are computed classically, and Table~\ref{tab:provenance} states, channel by channel, which objects
here were reconstructed from measurements and which are full-sector references.

\subsection{Construction}
For a Hamiltonian $\hat H$ with $N$-electron ground state $\lvert 0\rangle$ of energy $E_0$, the
single-particle spectral function has an electron-addition and an electron-removal branch,
\begin{equation}
\begin{aligned}
A_{p}(\omega)=&\sum_{n}\bigl|\langle n^{N+1}\lvert \hat c^\dagger_p\rvert 0\rangle\bigr|^{2}\,
\delta\!\bigl(\omega-(E^{N+1}_n-E_0)\bigr)\\
+&\sum_{m}\bigl|\langle m^{N-1}\lvert \hat c_p\rvert 0\rangle\bigr|^{2}\,
\delta\!\bigl(\omega-(E_0-E^{N-1}_m)\bigr),
\end{aligned}
\label{eq:lehmann}
\end{equation}
energies measured from $E_0$; the chemical potential $\mu=\tfrac12(E^{N+1}_0-E^{N-1}_0)$ places the
Fermi level mid-gap for the particle--hole-symmetric Hubbard model at half filling. The subspace is
built in three steps. \emph{(1)}~Prepare the response seeds
$\lvert\phi^{\pm}\rangle=\hat c^{\dagger}_p\lvert 0\rangle,\ \hat c_p\lvert 0\rangle$ in the $(N{\pm}1)$ sectors (every simulation here takes $\lvert 0\rangle$, the seeds and $e^{-i\hat Ht_k}$ exact; how the non-unitary seed is prepared on a device is not specified by this work). \emph{(2)}~Time-evolve each, $\lvert\phi^{\pm}(t_k)\rangle=e^{-i\hat H
t_k}\lvert\phi^{\pm}\rangle$, and measure computational-basis bitstrings at each of the $K{+}1$
snapshots; the union of the \emph{distinct} bitstrings returned is the determinant subspace
$\mathcal S^{\pm}$. \emph{(3)}~Form $\hat H_{\mathcal S^{\pm}}=P_{\mathcal S^{\pm}}\hat H P_{\mathcal
S^{\pm}}$, diagonalize it as $\sum_m \tilde E^{\pm}_m\lvert\tilde m^{\pm}\rangle\langle\tilde
m^{\pm}\rvert$, and evaluate the Lehmann sum inside $\mathcal S^{\pm}$. The retarded Green's function is
then assembled entirely on the classical processor,
\begin{widetext}
\begin{equation}
G_{pq}(\omega)=\sum_{m}\frac{\langle 0\lvert \hat c_p\rvert\tilde m^{+}\rangle\langle\tilde m^{+}\lvert \hat
c^\dagger_q\rvert 0\rangle}{\omega-(\tilde E^{+}_m-E_0)+i\eta}
+\sum_{m}\frac{\langle 0\lvert \hat c^{\dagger}_q\rvert\tilde m^{-}\rangle\langle\tilde m^{-}\lvert \hat
c_p\rvert 0\rangle}{\omega-(E_0-\tilde E^{-}_m)+i\eta},
\label{eq:green}
\end{equation}
\end{widetext}
whose imaginary part is the Lorentzian-broadened $A_{pq}(\omega)$ at resolution $\eta$. The seed selects
the channel: charged seeds $\hat c^\dagger_q\lvert 0\rangle$ in the $(N{\pm}1)$ sectors give $A(\omega)$
and, by Fourier transform of $G_{pq}$, $A(k,\omega)$; number-conserving density and spin seeds
$\hat n_q\lvert 0\rangle,\ \hat S^z_q\lvert 0\rangle$ in the neutral $N$ sector give $S(q,\omega)$ and
$S^{zz}(q,\omega)$. Figure~\ref{fig:circ} schematizes the
primitive---what one circuit and one readout can address; Table~\ref{tab:provenance} states separately
what was actually run for each channel here.

\begin{table*}[tp]\centering\small
\caption{\textbf{Provenance of every dynamical object in this paper: the subspace it was built on, and
the measurements that built it.} ``Full sector'' means no subspace selection of any kind---the
continued fraction runs on the whole sector at the stated recursion depth $n_\ell$, so the row is a
reference against which reconstructions are scored, not a reconstruction. ``Observed bitstrings'' means
the subspace is the union of the distinct bitstrings actually returned by $T$ draws split evenly over
the $K{+}1$ snapshots, which is the protocol of Sec.~\ref{sec:method}; one \emph{stream} is one
independent pseudorandom realization of that draw. ``Born ranking'' means the subspace is the top
fraction of the sector under the accumulated Born weight of the evolved state, which is the
$T\to\infty$ limit of the same protocol and carries no shot budget. Sector dimensions are exact; budgets
are per channel and per stream.}
\label{tab:provenance}
\setlength{\tabcolsep}{4pt}
\setlength{\parskip}{0pt}
\renewcommand{\arraystretch}{1.15}
\begin{tabular}{@{}p{0.225\linewidth}p{0.115\linewidth}p{0.315\linewidth}p{0.225\linewidth}@{}}
\toprule
\raggedright \textbf{Channel and system} & \raggedright \textbf{Shown in} & \raggedright \textbf{Determinant subspace} & \raggedright \textbf{Shots $T$}\tabularnewline
\midrule
\raggedright $A(\omega)$; open Hubbard chain $L{=}6$, $U/t{=}8$, $\eta{=}0.15\,t$; $\dim=300$
 & \raggedright Fig.~\ref{fig:method}
 & \raggedright observed bitstrings, plus the largest-amplitude determinant of the probe, inserted and not
   drawn; $|\mathcal S|=245$ of $300$ in panel~(a), $242.5$ on average in panel~(b)
 & \raggedright $5.2\times10^{4}$ per stream ($4\times10^{3}$ at each of $13$ snapshots, $K{=}12$), $8$ streams\tabularnewline
\raggedright $A(\omega)$, orbital projected; $19$ molecules at dissociation; $\dim=2$--$300$
 & \raggedright Table~\ref{tab:molecules}
 & \raggedright observed bitstrings, plus the largest-amplitude determinant of the probe, inserted and not drawn; the
   fraction $|\mathcal S|/D$ is the last column of that table
 & \raggedright $2.16\times10^{5}$ per molecule ($8\times10^{3}$ at each of $27$ snapshots, $K{=}26$), one stream\tabularnewline
\raggedright $A(\omega)$; $\mathrm{N_2}$ along dissociation, CAS$(6,6)$
 & \raggedright Fig.~\ref{fig:hero}(a)
 & \raggedright full CAS sector
 & \raggedright none\tabularnewline
\raggedright $A(\omega)$; hardware, \texttt{ibm\_fez}, $L{=}6$, $U/t{=}4$, $\eta{=}0.15\,t$; $\dim=300$
 & \raggedright Fig.~\ref{fig:heron}(a)
 & \raggedright post-selected from device bitstrings read in reversed bit order, so every retained shot is a
   device error (Sec.~\ref{sec:scope:hw}); $|\mathcal S|=300$ of $300$, the entire sector, so the reconstruction
   equals the exact $A(\omega)$ by coverage
 & \raggedright $3.5\times10^{5}$: $5\times10^{4}$ on each of $7$ circuits, $k{=}0$ (the $t{=}0$
   reference) to $k{=}6$\tabularnewline
\midrule
\raggedright $A(k,\omega)$; $L{=}12$, $\eta{=}0.18\,t$; $\dim=731\,808$
 & \raggedright Fig.~\ref{fig:lattice}
 & \raggedright full sector, $n_\ell=150$ per momentum
 & \raggedright none\tabularnewline
\raggedright $A(k,\omega)$; $L{=}8$, $\eta{=}0.18\,t$; $\dim=3920$
 & \raggedright Fig.~\ref{fig:akwsampled}(a,b)
 & \raggedright Born ranking, $3332$ of $3920$ per channel, over $17$ snapshots ($K{=}16$)
 & \raggedright none: (a)--(b) are the $T\to\infty$ limit; at finite shots $5.2\times10^{6}$ per channel
   suffice and $2.6\times10^{6}$ do not (see text)\tabularnewline
\raggedright $A(\omega)$; $L{=}8$, $\eta{=}0.15\,t$; $\dim=3920$
 & \raggedright Fig.~\ref{fig:akwsampled}(c)
 & \raggedright observed bitstrings; sector fraction $0.835\pm0.0025$ at the largest budget
 & \raggedright up to $2.60\times10^{6}$ over $19$ snapshots ($K{=}18$), $8$ streams\tabularnewline
\midrule
\raggedright $S(q,\omega)$; $L{=}12$, $\eta{=}0.18\,t$; $\dim=853\,776$
 & \raggedright Fig.~\ref{fig:sqw}
 & \raggedright full $N$-particle sector, $n_\ell=200$
 & \raggedright none\tabularnewline
\raggedright $S^{zz}(q,\omega)$; $L{=}12$, $\eta{=}0.10\,t$; $\dim=853\,776$
 & \raggedright Fig.~\ref{fig:spin}
 & \raggedright full $N$-particle sector, $n_\ell=200$
 & \raggedright none\tabularnewline
\midrule
\raggedright Required sector fraction versus $L$; $L{=}4$--$14$, $\eta{=}0.15\,t$
 & \raggedright Fig.~\ref{fig:scaling}(b,c)
 & \raggedright both selections, plotted together: Born ranking (open squares) and observed bitstrings (filled
   circles)
 & \raggedright measured: $2.4\times10^{4}$, $2.6\times10^{5}$, $3.1\times10^{6}$, $1.9\times10^{7}$ at
   $L=6,8,10,12$ (Table~\ref{tab:shots}, $4$--$6$ seeds per size)\tabularnewline
\bottomrule
\end{tabular}
\end{table*}

\begin{figure*}[tp]\centering\includegraphics[width=\linewidth]{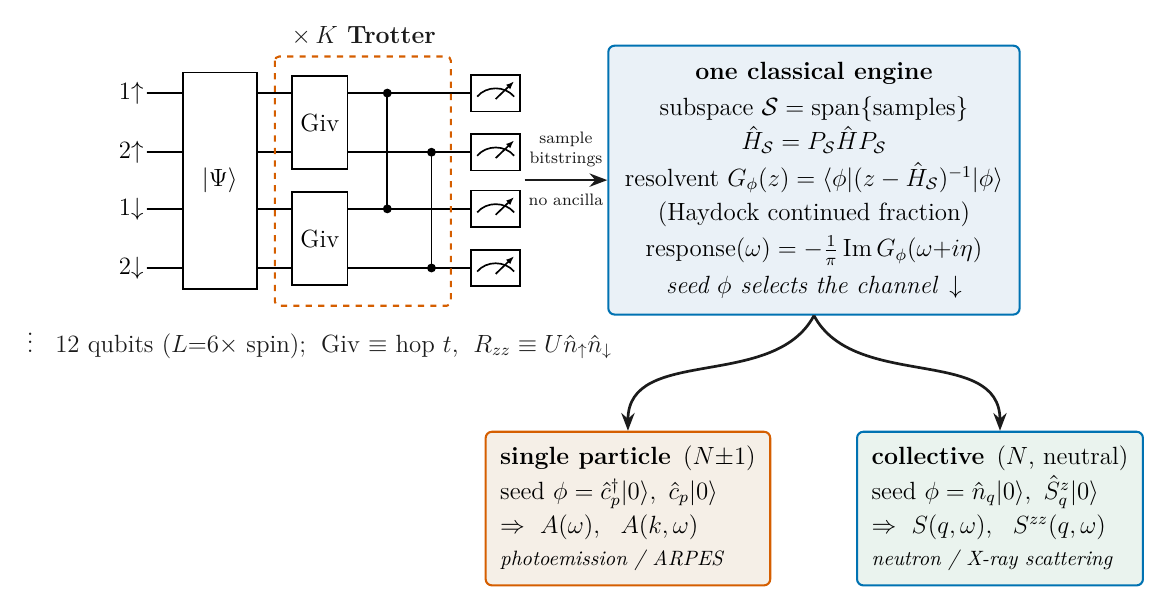}
\caption{\textbf{One sampling primitive, four dynamical responses.} An ancilla-free circuit---state preparation and a Trotter layer repeated $K$ times (Givens rotations for the hopping
$t$, native $R_{zz}$ for the on-site $U\hat n_\uparrow\hat n_\downarrow$)---is measured in the
computational basis; no Hadamard test or controlled unitary is run on the device. The sampled
bitstrings define a sector-specific subspace $\mathcal S$ and feed a single classical engine---a Haydock
continued-fraction resolvent~\cite{haydock1972} $G_\phi(z)=\langle\phi|(z-\hat H_{\mathcal S})^{-1}|\phi\rangle$, whose
single-particle case ($\phi=\hat c^\dagger_p|0\rangle$) is the Green's function of
Eq.~\eqref{eq:green}. The seed $\phi$ then selects the channel: particle addition/removal seeds
$\hat c^{\dagger}_p|0\rangle,\hat c_p|0\rangle$ in the $(N{\pm}1)$ sectors give the single-particle
spectral functions $A(\omega)$ and $A(k,\omega)$ (photoemission/ARPES), while number-conserving density
and spin seeds $\hat n_q|0\rangle,\hat S^z_q|0\rangle$ in the $N$ sector give the dynamical structure
factors $S(q,\omega)$ and $S^{zz}(q,\omega)$ (neutron~\cite{squires2012} and resonant inelastic
X-ray~\cite{ament2011} scattering)---all from \emph{one} sampling
primitive, the seed selecting the sector and hence the subspace. Here $A(\omega)$ is reconstructed from sampled bitstrings
(Fig.~\ref{fig:method}) and $A(k,\omega)$ on a Born-ranked subspace, the infinite-shot limit of the
protocol (Fig.~\ref{fig:akwsampled}); the $A(k,\omega)$ of Fig.~\ref{fig:lattice}, $S(q,\omega)$ and
$S^{zz}(q,\omega)$ are full-sector references, not reconstructed here (Table~\ref{tab:provenance}).}\label{fig:circ}\end{figure*}

\subsection{Relation to quantum subspace and Krylov methods}\label{sec:method:krylov}
For equally spaced snapshots the span of the $K{+}1$ time-evolved seeds is the order-$K$ Krylov
space of the propagator $e^{-i\hat H\Delta t}$~\cite{shen2023}, a real-time counterpart of the
polynomial Krylov space $\operatorname{span}\{\phi,\hat H\phi,\dots,\hat H^{K}\phi\}$, which block
encoding builds exactly without time evolution~\cite{kirby2023}; diagonalizing $\hat H$ inside a
subspace containing such a space is the route
of quantum subspace expansion~\cite{mcclean2017qse}, filter diagonalization~\cite{parrish2019} and
real-time quantum Krylov algorithms~\cite{cortesgray2022}, whose error theory is now
characterized~\cite{epperly2022}. What distinguishes the present method is that these states are used
not to \emph{estimate matrix elements}~\cite{bauer2016,endo2020,rungger2020,greenediniz2024,bespalova2025,onqse2024}---the step requiring Hadamard tests, compute--uncompute overlaps, Hamiltonian-moment or quantum equation-of-motion measurements on the device---but
as a \emph{sampling distribution} selecting a determinant
subspace~\cite{robledomoreno2025,kanno2023,skqd2025,teqsci2024,extsqd2025,adaptqsci,shirakawa2025},
extended since to embedded fragments and to auxiliary-field
post-processing~\cite{shajan2024,danilov2025}, and here lifted to the $(N{\pm}1)$ manifold and to a
frequency-resolved observable. Apart from the concurrent Green's-function extension of Ref.~\cite{patel2026}
(Sec.~\ref{sec:intro}), that lineage stops at
energies, static correlators and, at most, discrete levels: the closest published results are a QSCI quasiparticle band
structure returning discrete poles~\cite{ohgoe2025}, a molecular spectrum built from a neutral-sector
autocorrelation with classically projected long-time dynamics~\cite{santini2025}, and, from the
measured-overlap side rather than the sampled one, a dynamical structure factor of a Kitaev spin
liquid obtained by quantum subspace expansion~\cite{umeano2025}, or a hybrid quantum--classical
framework for many-fermion response and structure~\cite{du2026}. Closest on lattice models, Nogaki
\emph{et al.}~\cite{nogaki2025} apply symmetry-adapted SQD to ground- and excited-state \emph{energies} and a
\emph{static} superconducting pair correlation of the two-leg Hubbard ladder, and Wray
\emph{et al.}~\cite{wray2025} measure how SQD converges on a variable-length cuprate chain. The readout
is that of SQD, the circuits are the Trotterized Krylov evolutions of Krylov quantum
diagonalization~\cite{yoshioka2025}, sampled as in sample-based Krylov
diagonalization~\cite{skqd2025,piccinelli2025}, and all dynamical information is assembled
classically through Eq.~\eqref{eq:green} by one continued-fraction
resolvent~\cite{haydock1972,golub2010}.

\subsection{What the subspace reconstruction is: Rayleigh--Ritz, not restriction}\label{sec:method:rr}
Step~\emph{(3)} is the definition the rest of this paper turns on.
With $P_{\mathcal S}$ the orthogonal projector onto the span of the determinants in $\mathcal S$ and
$\phi$ the seed,
\begin{widetext}
\begin{equation}
A_{\mathcal S}(\omega)=\frac{\eta}{\pi}\,
\bigl\lVert\,\bigl(\omega+E_0+i\eta-H_{\mathcal S}\bigr)^{-1}P_{\mathcal S}\phi\,\bigr\rVert^{2},
\qquad H_{\mathcal S}=P_{\mathcal S}\hat HP_{\mathcal S},
\label{eq:ritz}
\end{equation}
\end{widetext}
the spectral function of the \emph{Rayleigh--Ritz operator} $P_{\mathcal S}\hat HP_{\mathcal S}$, whose
poles are the Ritz values $\tilde E_m$ of Eq.~\eqref{eq:green}---not eigenvalues of $\hat H$---and
not the restriction of $A$ to a subset of its own eigenpairs. The reproduction repository computes
exactly this, at the small sizes by forming and diagonalizing the submatrix $\hat H[\mathcal S,\mathcal
S]$ and at the large ones by a matrix-free restricted matrix--vector product inside the same continued
fraction; the consequence is the mechanism of this paper: \emph{spectral weight the subspace retains is
not kept in place, it is moved in $\omega$}. Truncation does not delete Lehmann poles, it displaces every surviving one, so the
weight a subspace captures cannot by itself bound how far the reconstruction has moved---which is why
no weight-only bound survives contact with Eq.~\eqref{eq:ritz} (Sec.~\ref{sec:retraction}), and why
Theorem~\ref{thm:leakage} carries a second term indexed on the coupling
$P_{\mathcal S}\hat HQ_{\mathcal S}$ across the boundary of $\mathcal S$, with $Q_{\mathcal S}$ the
complementary projector---exactly what Eq.~\eqref{eq:ritz} discards. Two consequences used later, that
determinants outside $\operatorname{supp}(\phi)$ still move the Ritz spectrum and that the error does not
vanish on $\operatorname{supp}(\phi)$ itself, are stated in Sec.~\ref{sm:method} and measured in
Sec.~\ref{sec:retraction}.

\subsection{Provenance: what was reconstructed from shots, channel by channel}\label{sec:method:provenance}
Table~\ref{tab:provenance} gives, for every dynamical object in this paper, the subspace it was built on
and the measurements that built it. The neutral channels $S(q,\omega)$ and $S^{zz}(q,\omega)$, the
$A(k,\omega)$ of Fig.~\ref{fig:lattice} and the $\mathrm{N_2}$ series of Fig.~\ref{fig:hero} are
full-sector references computed by continued fraction, not reconstructions. The Born ranking of
Fig.~\ref{fig:akwsampled}(a,b) and of the published curve of Fig.~\ref{fig:scaling} is the infinite-shot
limit of the declared protocol, and its finite-shot replacement is measured with its budget: at $L=8$,
$2.60\times10^{6}$ shots over $19$ snapshots reach sector fraction $0.835\pm0.0025$ at relative $L_1$
error $(4.81\pm0.27)\times10^{-3}$ (Fig.~\ref{fig:akwsampled}(c)), and the literal configuration of
Fig.~\ref{fig:akwsampled}(a,b) reaches the accuracy of panel~(b) at $5.2\times10^{6}$ shots per channel
and not at half that budget. The full declarations, including the branch sum rules and what is guaranteed
in advance, are in Sec.~\ref{sm:method}.

\begin{figure*}[tp]\centering\definecolor{inkPrim}{HTML}{1B1B1B}\definecolor{inkMute}{HTML}{6E6E6E}\definecolor{clExact}{HTML}{1F3A5F}
\definecolor{coral}{HTML}{F0653F}\definecolor{gridClr}{HTML}{ECEAE5}\definecolor{thm}{HTML}{2E7D5B}
\begin{tikzpicture}[font=\rmfamily]
\begin{groupplot}[
  group style={group size=2 by 1, horizontal sep=1.9cm},
  width=7.7cm, height=5.9cm, paperaxis,
  label style={font=\small}, tick label style={font=\footnotesize},
  axis line style={inkPrim, line width=0.6pt}, tick style={inkPrim, line width=0.6pt},
  title style={yshift=1pt},
]
\nextgroupplot[
  xlabel={frequency \; $\omega/t$}, ylabel={$A(\omega)$},
  title={(a)\quad $A(\omega)$ from $5.2\times10^{4}$ simulated shots},
  xmin=3, xmax=14, ymin=0, ymax=0.50, ymajorgrids, y grid style={gridClr},
  legend style={draw=none, fill=white, fill opacity=0.75, text opacity=1, at={(0.985,0.985)}, anchor=north east}, legend cell align=left,
]
  \addplot[draw=none, fill=clExact, fill opacity=0.16, forget plot] table[x=w,y=Aex]{aw_method.dat} \closedcycle;
  \addplot[clExact, line width=1.3pt] table[x=w,y=Aex]{aw_method.dat}; \addlegendentry{exact}
  \addplot[coral, line width=1.1pt, dash pattern=on 3pt off 2.2pt] table[x=w,y=Asm]{aw_method.dat}; \addlegendentry{bitstring-sampled}
  \node[anchor=north east, align=right, inkPrim] at (rel axis cs:0.985,0.56)
     {rel-$L_1=0.017$\\[-1pt] $\int_W\! A_{\mathcal S}=0.4865$};
\nextgroupplot[
  xlabel={Krylov order \; $K$ \; (time-evolution steps)}, ylabel={rel-$L_1$ error},
  title={(b)\quad error versus Krylov order},
  xmin=-0.4, xmax=12.4, ymode=log, ymin=0.013, ymax=2.6, ymajorgrids, y grid style={gridClr},
  ytick={0.02,0.05,0.1,0.2,0.5}, yticklabels={$2\%$,$5\%$,$10\%$,$20\%$,$50\%$},
  legend style={draw=none, fill=white, fill opacity=0.85, text opacity=1, at={(0.985,0.985)}, anchor=north east}, legend cell align=left,
]
  \fill[inkMute, fill opacity=0.06] (axis cs:-0.4,0.013) rectangle (axis cs:4,2.6);
  \node[inkPrim, anchor=center, align=center] at (axis cs:1.75,0.19) {discovery\\ regime};
  \addplot[name path=lo, draw=none, forget plot] coordinates {(0,0.5495) (1,0.5169) (2,0.5125) (3,0.5072) (4,0.3953) (5,0.1346) (6,0.0822) (7,0.0481) (8,0.0297) (9,0.0227) (10,0.0191) (11,0.0172) (12,0.0132)};
  \addplot[name path=hi, draw=none, forget plot] coordinates {(0,0.6261) (1,0.6267) (2,0.5807) (3,0.5948) (4,0.4793) (5,0.1820) (6,0.1081) (7,0.0713) (8,0.0402) (9,0.0307) (10,0.0275) (11,0.0259) (12,0.0239)};
  \addplot[coral, fill opacity=0.16, forget plot] fill between[of=lo and hi];
  \addplot[coral, line width=1.3pt, mark=*, mark size=2.2pt,
           mark options={fill=coral, draw=white, line width=0.4pt}] coordinates {(0,0.5878) (1,0.5718) (2,0.5466) (3,0.5510) (4,0.4373) (5,0.1583) (6,0.0951) (7,0.0597) (8,0.0350) (9,0.0267) (10,0.0233) (11,0.0216) (12,0.0185)}; \addlegendentry{sampled (8 seeds)}
\end{groupplot}
\end{tikzpicture}
\caption{\textbf{The method validated at $L=6$, and its error against Krylov order.} \textbf{(a)}~The
single-particle addition branch $A^{+}(\omega)$, reconstructed from classically simulated shots
(independent draws from the exact Born distribution of the evolved probe, with the probe's
largest-amplitude determinant inserted into each subspace rather than drawn), reproduces the exact one
(open Hubbard chain, $L=6$, $U/t=8$; broadening $\eta=0.15\,t$) to relative-$L_1$ $0.017$---the exact (filled) and reconstructed
(dashed) curves are indistinguishable. The analytic addition-branch sum rule $\int A^{+}\,d\omega=1-\langle n_p\rangle=0.5$ is an identity of
the ground state and not a check on the reconstruction; integrated over the plotted window the exact
and reconstructed curves give $0.486588$ and $0.486471$, agreeing with each other to
$2.4\times10^{-4}$ relative and both falling $2.7\%$ short of $0.5$ because Lorentzian tails leave the
window (Sec.~\ref{sm:physics}). \textbf{(b)}~Convergence of the relative-$L_1$ error with
the Krylov order $K$ (number of time-evolution steps sampled), seed-averaged over eight independent
runs (band $=$ standard deviation), $\eta=0.15\,t$, so the budget is $T=4000(K{+}1)$ and reaches
$5.2\times10^{4}$ at $K=12$. The error falls by a factor $32$ across the thirteen points, from $0.59$
at $K=0$ to $0.0185$ at $K=12$, with one non-monotone step at $K=3$ inside the discovery phase, while
the mean subspace returned grows from $75$ to $242$ determinants of the $300$-dimensional sector. No
geometric rate is fitted to the window $4\le K\le9$ here (Sec.~\ref{sec:moments:notbuy}): one is
not implied by moment exactness, and inside that window the decay in $K$ is confounded with the growth
of $|\mathcal S|$.}\label{fig:method}\end{figure*}
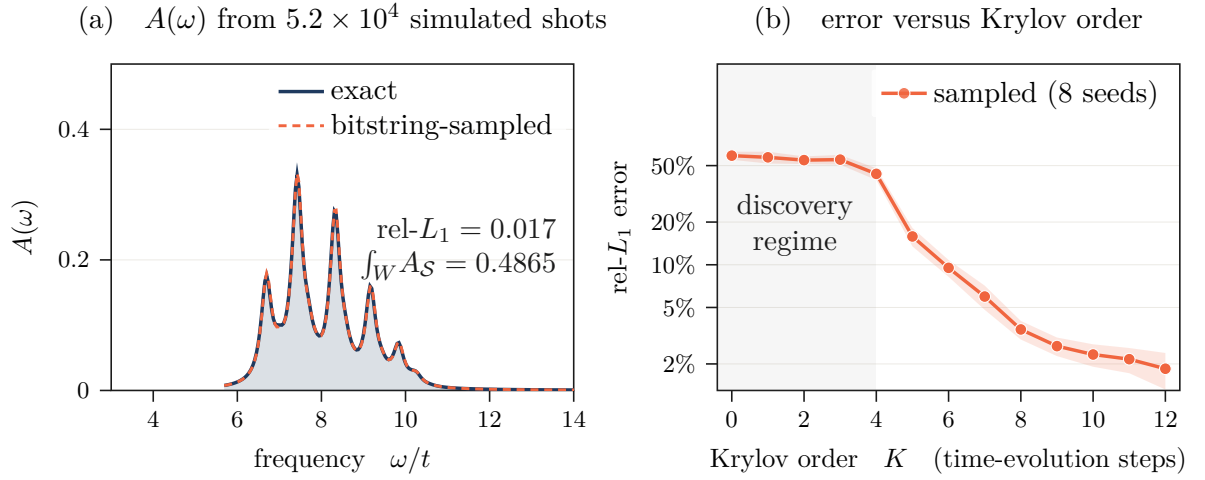

\begin{figure*}[tp]\centering
\scalebox{0.92}{\input{figs/fig_akw_sampled_honest.tex}}
\input{figs/fig5_caption.tex}
\end{figure*}

\begin{figure*}[tp]\centering\includegraphics[width=\linewidth]{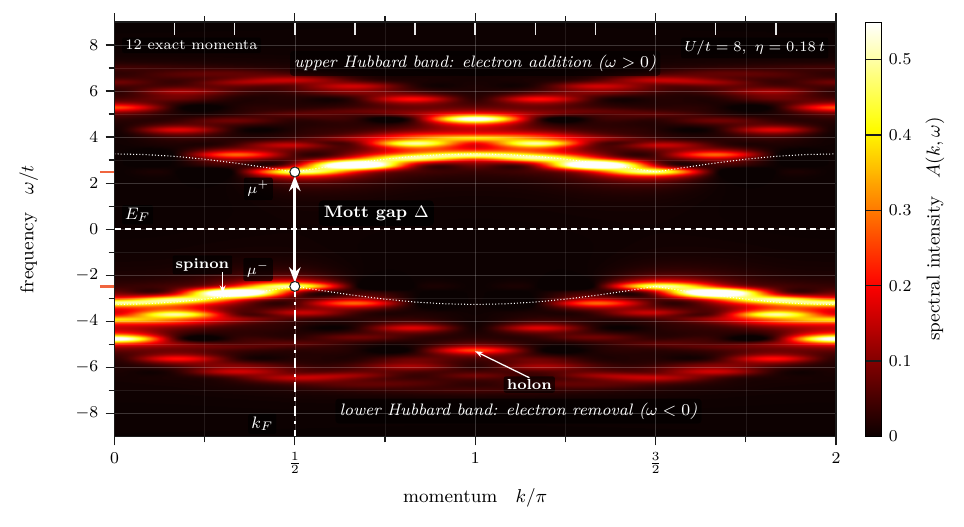}
\caption{\textbf{Momentum-resolved spectral function: the dynamical target.} The
single-particle spectral function $A(k,\omega)$ of the half-filled periodic Hubbard ring ($U/t=8$, $L=12$,
Lorentzian broadening $\eta=0.18\,t$), computed from the resolvent $\langle\phi|(z-\hat H)^{-1}|\phi\rangle$ by a
Haydock continued fraction~\cite{haydock1972} of depth $n_\ell=150$ over the full $(N{\pm}1)$ sectors---a
finite-depth continued fraction, that is a $150$-node Gauss quadrature of the spectral
measure~\cite{golub2010} and not the Lehmann sum
(Sec.~\ref{sm:physics})---the frequency- and momentum-resolved
observable that sample-based diagonalization targets but its energy-only output cannot itself deliver.
The two dispersing Hubbard bands---electron removal (lower, $\omega<0$; photoemission) and electron
addition (upper, $\omega>0$; inverse photoemission)---are separated by the correlation-driven
\emph{Mott gap} $\Delta=\mu^{+}-\mu^{-}=4.97\,t$ between the highest occupied state ($\mu^{-}$, lower
edge) and the lowest unoccupied state ($\mu^{+}$, upper edge), obtained from the exact $(N{\pm}1)$
ground-state energies (ticks at $\pm2.484\,t$); the broadened spectral peaks of the plotted $A(k_F,\omega)$
sit slightly wider, at $\pm2.49\,t$ (a $4.99\,t$ splitting), the $\sim\!0.02\,t$ excess being the
$\eta=0.18\,t$ Lorentzian broadening; the thermodynamic Lieb--Wu value~\cite{liebwu1968} is $4.68\,t$, the
$L=12$ excess being finite-size. The gap is \emph{direct} at the
Fermi momentum $k_F=\pi/2$. Particle--hole symmetry places the Fermi level at mid-gap ($\omega=0$; here
$\omega$ is measured from the chemical potential $\mu=U/2$, whereas Figs.~\ref{fig:method},
\ref{fig:akwsampled} and~\ref{fig:heron} measure it from $E_0$),
verified to double-precision round-off, and each momentum obeys the analytic spectral sum rule
$\int A(k,\omega)\,d\omega=1$ (the plotted broadened grid recovers it to $\sim\!2\%$ within the $\pm9t$ window).
In the thermodynamic limit the lower Hubbard band is spin--charge separated; at $L=12$ the plot is compatible with a narrow spinon-like branch (bandwidth $\pi J/2=0.79\,t$~\cite{descloizeaux1962}, $J=4t^2/U$; dotted, drawn with its mirror image at $\omega>0$,
validated against the exact peak dispersion to ${<}0.1\,t$) that rides a broad holon (charge, ${\sim}4t$)
band, as expected from the Bethe-ansatz solution~\cite{liebwu1968}, and is not offered as evidence of it. At $L=12$ that
structure is three to seven poles per momentum (Table~\ref{tab:poles}), and the in-panel labels
``spinon'' and ``holon'' name those branches, not thermodynamic-limit excitations.
The momentum axis is a periodic cubic-spline interpolation of the $12$ physical momenta (ticks, top;
the in-panel ``12 exact momenta'' counts these allowed momenta, not an exact spectrum).
Faithfully \emph{sampling} this full momentum-resolved observable---as opposed to the ground-state
energy---requires a large fraction of the $(N{+}1)$ sector (Fig.~\ref{fig:scaling}), the cost
quantified in Sec.~\ref{sec:prices}.}\label{fig:lattice}\end{figure*}

\section{Two budgets: the missed weight and the boundary}
\label{sec:leakage}

The $L_1$ error of a spectral function reconstructed on a sampled subspace lies between the Born weight
the subspace missed and that weight plus a boundary-leakage term, $\Lambda_P(\eta)/\eta$, computed from
the Ritz pairs the reconstruction has already produced and one matrix--vector product per Ritz vector
(Theorem~\ref{thm:leakage}). The leakage term does not need the exact $A$; it does need the exact ground
state (H4). The bound is two-sided, it holds for \emph{any} orthogonal projector and therefore for any set
of configurations a sampler might return, and its upper branch has two terms: the first is fixed by the
captured weight and vanishes once the probe's support is held; the second is not controlled by the weight,
although it falls as the subspace grows. Section~\ref{sec:frontier} measures where the resulting
certificate is informative, Sec.~\ref{sec:nogo} shows that no invariant of the state can replace it, and
Sec.~\ref{sec:prices} prices both terms in shots.

\subsection{The objects}

Let $\phi$ be the probe---$\hat c^{\dagger}_p\ket{\Psi_0}$ or
$\hat c_p\ket{\Psi_0}$ for the charged channels, $\hat n_q\ket{\Psi_0}$ or $\hat S^{z}_q\ket{\Psi_0}$
for the neutral ones---$H$ the Hamiltonian restricted to the sector of $\phi$, $(N{\pm}1)$ or $N$, and
$E_0$ the ground-state energy of the $N$-electron sector. Write $z=\omega+E_0+i\eta$ and $R(z)=(z-H)^{-1}$. The $\eta$-broadened spectral
function of Eq.~\eqref{eq:green} is
\begin{widetext}
\begin{equation}
A(\omega)\;=\;\frac{\eta}{\pi}\,\bigl\lVert R(z)\phi\bigr\rVert^{2},
\qquad\text{so that}\qquad
\int_{\mathbb R}A(\omega)\,d\omega\;=\;\lVert\phi\rVert^{2}
\label{eq:lehmannA}
\end{equation}
\end{widetext}
for every $\eta>0$, the sum rule holding because the Poisson kernel integrates to one on each Lehmann
pole. Let $P$ be an orthogonal projector on the sector, $Q=I-P$, and
\begin{widetext}
\begin{equation}
\phi_P=P\phi,\quad \phi_Q=Q\phi,\quad
w_P=\frac{\lVert\phi_P\rVert^{2}}{\lVert\phi\rVert^{2}},\quad
H_P=PHP,\quad R_P(z)=(z-H_P)^{-1}\big|_{\operatorname{ran}P},
\label{eq:ritz3}
\end{equation}
\end{widetext}
with Ritz pairs $H_P\ket{\tilde m}=\tilde E_m\ket{\tilde m}$ and $c_m=\braket{\tilde m|\phi}$. The
reconstruction of Eq.~\eqref{eq:ritz} is $A_P(\omega)=(\eta/\pi)\lVert R_P(z)\phi_P\rVert^{2}$, and it
obeys the same sum rule with $\lVert\phi_P\rVert^{2}$ on the right. When $P=P_{\mathcal S}$ is diagonal
in the determinant basis---the only case a device produces---we write $S$ for the configuration set
and $w_S$, $\Lambda_S$ for the quantities below.

The magnitude the theorem is indexed on is the Hamiltonian coupling that crosses the boundary of the
subspace, averaged over the resolution window rather than over the spectrum. With
$g_m=QH\ket{\tilde m}$,
\begin{widetext}
\begin{equation}
\Lambda_P(\eta)^{2}\;:=\;\frac{\eta}{\pi}\int_{\mathbb R}
\bigl\lVert QH\,R_P(z)\,\phi_P\bigr\rVert^{2}d\omega
\;=\;\sum_{m,m'}c_m\bar c_{m'}\,\braket{g_{m'},g_m}\,
\frac{2i\eta}{(\tilde E_{m'}-\tilde E_m)+2i\eta},
\label{eq:lambda}
\end{equation}
\end{widetext}
the second equality being exact and proved in Appendix~\ref{app:proof}; the sum is real, the $(m,m')$ and $(m',m)$ terms
being complex conjugates. Two features of Eq.~\eqref{eq:lambda} carry the practical content of this
paper. It is assembled from objects the reconstruction has already produced---the Ritz vectors and
their coefficients---plus one matrix--vector product $H\ket{\tilde m}$ per Ritz vector, so evaluating
it costs a small multiple of the diagonalization that produced $A_P$ and \emph{does not require} $A$.
And it vanishes exactly when $\operatorname{span}(S)$ contains $\mathcal K(H,\phi_P)$, the Krylov space of the retained part of the probe: $QHR_P(z)\phi_P$ vanishes for every $z$ exactly when $\mathcal K(H_P,\phi_P)$ is invariant under $H$, and it then coincides with $\mathcal K(H,\phi_P)$. For a subspace that holds the whole probe, $w_P=1$, the condition is that $\operatorname{span}(S)$ contain $\mathcal K(H,\phi)$ itself, which for the local probes here, in the site basis, was measured at $L=6$--$14$ to mean every reflection-allowed determinant of the sector (Sec.~\ref{sec:frontier:obstruction}).

\subsection{The bound}

\begin{theorem}[a posteriori leakage certificate]
\label{thm:leakage}
Let $H=H^{\dagger}$ act on the sector of $\phi$, let $\eta>0$, let $P$ be \emph{any} orthogonal
projector on that sector, and let $A$, $A_P$, $w_P$ and $\Lambda_P(\eta)$ be as in
Eqs.~\eqref{eq:lehmannA}--\eqref{eq:lambda}. Then
\begin{widetext}
\begin{equation}
\lVert\phi\rVert^{2}\bigl(1-w_P\bigr)
\;\le\;\lVert A-A_P\rVert_{1}\;\le\;
\min\Bigl\{\,\lVert\phi\rVert^{2}\bigl(1+w_P\bigr),\;
\bigl(\lVert\phi\rVert+\lVert\phi_P\rVert\bigr)
\bigl(\,\underbrace{\lVert\phi_Q\rVert}_{\text{missed weight}}
+\underbrace{\Lambda_P(\eta)/\eta}_{\text{boundary leakage}}\,\bigr)\Bigr\},
\label{eq:twobudget}
\end{equation}
\noindent equivalently, dividing by $\lVert\phi\rVert^{2}$ and writing
$\widehat\Lambda_P=\Lambda_P/\lVert\phi\rVert$,
\begin{equation*}
1-w_P\;\le\;\frac{\lVert A-A_P\rVert_{1}}{\lVert\phi\rVert^{2}}\;\le\;
\min\Bigl\{\,1+w_P,\;\bigl(1+\sqrt{w_P}\bigr)
\bigl(\sqrt{1-w_P}+\widehat\Lambda_P(\eta)/\eta\bigr)\Bigr\}.
\end{equation*}
\end{widetext}
\end{theorem}

\noindent Four hypotheses carry it, each a property of the \emph{pipeline} that produced $A_P$ rather
than of the inequality: \emph{(H1)} $A$ and $A_P$ are referred to the same $E_0$; \emph{(H2)} the $L_1$
norm is taken over the whole line and normalized by $\lVert\phi\rVert^{2}$; \emph{(H3)} $A$ is the exact
spectral function of the full sector; \emph{(H4)} the probe $\phi$ is built from the exact $N$-electron
ground state. Each is load-bearing and each is checked against this paper's own pipeline in
Sec.~\ref{app:hyp}. A shot-based pipeline breaks (H1), and in general (H4) with it
(Sec.~\ref{sec:leakage:branches}). If only $E_0$ is estimated on a sampled subspace, the probe being
kept exact, and $A_P$ is referred to $E_0+\delta$, the bound is guaranteed only once an additive
$\lVert\phi_P\rVert^{2}(4/\pi)\arctan(\lvert\delta\rvert/2\eta)$ is added to it, a term the certificate
does not control; at $L=8$, with $E_0$ from $60\%$ of the $N$ sector, that term moves the bound by
$+0.20\%$ while the measured error grows by a factor $16$ (Sec.~\ref{app:hyp}). The proof, from the
Feshbach--L\"owdin split of the resolvent (Eq.~\eqref{eq:split}) and Cauchy--Schwarz in $\omega$, is in
Appendix~\ref{app:proof}.

\subsection{Reading the two branches}
\label{sec:leakage:branches}

\paragraph{What the weight does not control.} The first summand of the leakage branch, $\lVert\phi_Q\rVert=\lVert\phi\rVert\sqrt{1-w_P}$, is the Born weight of the probe the subspace failed to hold; Table~\ref{tab:provenance} prices it channel by channel. The second, $\Lambda_P(\eta)/\eta$, reads the off-diagonal block $PHQ$, and the captured weight does not control it: along the sampler's ranking it falls as $S$ grows (Sec.~\ref{sec:frontier}), but a subspace can hold the whole probe and still leak, and Sec.~\ref{sec:retraction} measures how much.

\paragraph{A certificate, not an error estimate.} The gap between the bound and the measured error is
large, and it is not the accident of a constant. At \emph{fixed} $P$ the looseness is stable---within
$\pm17\%$ ($L=6$) and $\pm25\%$ ($L=8$) over a decade in $\eta$---but across subspaces it varies by a
factor $51$, and at the threshold where the leakage branch first improves on the trivial one it is
$3\times10^{2}$ to $2.4\times10^{3}$ (Sec.~\ref{sm:frontier}). Equation~\eqref{eq:twobudget} is therefore an exclusion, not an error
estimate; Sec.~\ref{sec:frontier} reports what it excludes, and at what fraction of the sector.

\paragraph{Neither term is free of the ground state.} Given the probe, $\Lambda_P(\eta)$ is computable
from $H$, $\operatorname{ran}P$, $\phi_P$ and the Ritz vectors alone, without $A$; for $P=P_{\mathcal S}$,
$\phi_P$ consists of the amplitudes of $\Psi_0$ on the determinants that the probe operator maps into
$S$, which the reconstruction $A_P$ uses as well. The weight term, $\lVert\phi_Q\rVert=\lVert\phi\rVert\sqrt{1-w_P}$,
needs the full norm of the probe as well as its retained part ($\lVert\phi\rVert^{2}=1-\langle\hat n_p\rangle$
for electron addition). Under (H4) both terms are therefore evaluable without $A$, but neither without
$\Psi_0$. If $\Psi_0$ is itself a sampled-subspace estimate, (H4) fails: the probe is another vector, and
the theorem bounds how far the reconstruction lies from that vector's own spectral function, not from $A$;
$w_P$ estimated against it is biased in the optimistic direction, and $\Lambda_P$ is evaluated for the
wrong probe, with an error of unknown sign. Nothing the pipeline computes then bounds the distance from
its reconstruction to $A$, and no certificate in this work has been evaluated with such a probe. The
complete bound is therefore \emph{not} evaluable by a user who does not already have the ground state.
Its regime is the one in which Sec.~\ref{sec:frontier} evaluates it: exact-diagonalization sizes, where
the exact ground state is available and the bound is applied, subspace by subspace, to the configuration
sets a sampled pipeline returns; there Sec.~\ref{sec:frontier:calib} gives the conservative substitution
for $w_P$ that makes the decision itself a posteriori, given the exact probe.

\paragraph{What is new, and the prior art.} What this section adds is the conjunction of three things: a closed form
for $\Lambda_P(\eta)$ from the Ritz pairs on a subspace of \emph{coordinates}, with no Krylov structure,
invariance or gap assumption; a bound two-sided on the same object; and a measured calibration of where
it crosses the trivial bound (Sec.~\ref{sec:frontier:calib}). Each step of the proof is standard. The split is L\"owdin
partitioning~\cite{lowdin1962}, the algebra of the Feshbach projection~\cite{feshbach1958}; Rayleigh--Ritz
and Krylov error analysis, for resolvents and matrix functions~\cite{saad2011,golub2010} and for
eigenpairs~\cite{temple1928,weinstein1934,kato1949,epperly2022}, with Paige's loss-of-orthogonality
analysis and the restart convergence theory of Beattie, Embree and Sorensen~\cite{paige1976,beattie2005},
contains every rung of its steps, and the integration step is the reduced-basis bound on a Galerkin
projection with inf--sup constant $\eta$~\cite{rozza2008,hesthaven2016}. A posteriori estimates for Krylov
approximations to matrix functions are due to Frommer and Simoncini, Jia and Lv, and Frommer and
Schweitzer~\cite{frommersimoncini2008,jialv2015,frommerschweitzer2016}, with Lanczos and
stochastic-Lanczos-quadrature versions~\cite{chen2022lanczos,xuchen2025block,chen2021slq}; on quantum
Krylov, Oliveira, Ziems and Glaser assess reliability without knowledge of the true
spectrum~\cite{oliveira2026}; and a response-targeted selection criterion for classical selected
configuration interaction is Reinholdt, Kjellgren and Kongsted's~\cite{reinholdt2026lr}. The priority
for the a posteriori framing is theirs. At the fractions sampling reaches here the certificate is vacuous
(Sec.~\ref{sec:frontier}, which also measures where it is informative); the full comparison with that
literature is in Sec.~\ref{sm:leakage}.

\subsection{The captured weight alone does not bound the error}
\label{sec:retraction}

The certificate one would want for a sampled reconstruction indexes the error on the captured
weight alone,
\begin{equation}
\lVert A-A_S\rVert_{1}\;\le\;2\lVert\phi\rVert^{2}(1-w_S)\;+\;C\,\rho^{-K_S},
\label{eq:weightonly}
\end{equation}
with $w_S=\lVert P_S\phi\rVert^{2}/\lVert\phi\rVert^{2}$ and $K_S$ the largest Krylov order
contained in $\operatorname{span}(S)$ (we set $K_S=-1$ when $\phi\notin\operatorname{span}(S)$).
It is false, for every constant $C$ and every $\rho$, and it fails hardest in exactly the regime
a sampling experiment aims for. 

\paragraph{A two-level counterexample forces the trivial constant.}
Take $H=\big(\begin{smallmatrix}0&g\\ g&0\end{smallmatrix}\big)$, $\phi=e_0$, $S=\{0\}$, $E_0=0$. Then
$P_S\phi=\phi$, so $w_S=1$ exactly and the weight term vanishes identically; and
$\mathcal K_0=\operatorname{span}\{\phi\}\subseteq\operatorname{span}(S)$ while
$\mathcal K_1\not\subseteq\operatorname{span}(S)$, so $K_S=0$ exactly and the second term is the bare
constant $C$. The exact $A$ is a pair of Lorentzians of weight $\tfrac12$ at $\pm g$; the
reconstruction, built on $H_S=P_SHP_S=0$, is one Lorentzian of weight $1$ at the origin, so
$\lVert A-A_S\rVert_{1}\to2\lVert\phi\rVert^{2}$ as $g/\eta\to\infty$. Since $g$ is free,
Eq.~\eqref{eq:weightonly} requires $C\ge2\lVert\phi\rVert^{2}$: false if $C$ is a constant, vacuous if it may depend on the instance. The same holds at every Krylov order: on a chain of $K+2$ sites with $\phi=e_0$ and $S$ its first $K+1$ sites, $w_S=1$ and $K_S=K$, and the exact and reconstructed spectra have strictly interlacing poles, so the same limit requires $C\rho^{-K}\ge2\lVert\phi\rVert^{2}$ for every $K$, which no fixed $C$ satisfies when $\rho>1$. With $w_S<1$ strictly the weight term alone fails without bound: in a three-level family its violation diverges as $\varepsilon^{-2}$ in the amplitude $\varepsilon$ of the discarded component (Sec.~\ref{app:counterex}). By contrast
$\lVert A-A_S\rVert_{1}\le\lVert\phi\rVert^{2}(1+w_S)$ holds for every $S$ under no hypothesis at all,
$A$ and $A_S$ being nonnegative with known integrals: that trivial bound is the only surviving member of
the family of Eq.~\eqref{eq:weightonly}, and Theorem~\ref{thm:leakage} keeps it explicitly, as the first branch
of its minimum.
\paragraph{The violation is systematic, and worst where the reconstruction is best.} In a pooled sweep
of $606$ subspaces from five Hamiltonian families the weight term of Eq.~\eqref{eq:weightonly}
alone---its $C=0$ member, $\lVert A-A_S\rVert_{1}\le2\lVert\phi\rVert^{2}(1-w_S)$, the constant having
been disposed of analytically above---fails on $469$, and on all $363$
with $w_S\ge0.99$, the mildest by a factor $2.00$ (Fig.~\ref{fig:thm1iii-violation},
Sec.~\ref{app:counterex}). In the
configuration of Fig.~\ref{fig:akwsampled} its median violation factor rises from $18.5$ to $29.9$
between sector fractions $0.50$ and $0.95$ while the measured error falls by three orders of magnitude.
The sweep and its declared limits are in Sec.~\ref{sm:leakage}; the conclusion does not depend on it,
the counterexample being analytic.

\paragraph{Why it fails.} The reconstruction is Rayleigh--Ritz on $H_S=P_SHP_S$, not a restriction of
$A$ to a subset of its eigenpairs: weight the subspace keeps is \emph{displaced} in $\omega$, not
discarded, and the two-level example is the extreme case. What the integrals do give is the lower bound
$\lVert\phi\rVert^{2}(1-w_S)\le\lVert A-A_S\rVert_{1}$, with $0$ violations across the $606$ subspaces,
the left-hand half of Theorem~\ref{thm:leakage}.

\paragraph{The two terms are mutually exclusive, and one is off at every operating point here.}
$S$ is a set of determinants, so $P_S$ is diagonal in the working basis and
$P_S\phi=\phi\iff\operatorname{supp}(\phi)\subseteq S\iff w_S=1$, while $K_S\ge0$ means precisely
$\phi=H^{0}\phi\in\operatorname{span}(S)$. Hence the exclusion lemma
\begin{equation}
K_S\ \ge\ 0\qquad\Longleftrightarrow\qquad w_S\ =\ 1 ,
\label{eq:exclusion}
\end{equation}
verified with $0$ discrepancies across the $378$ deposited subspaces, of which $357$ have $K_S=-1$ (Sec.~\ref{app:counterex}). The two terms of
Eq.~\eqref{eq:weightonly} are therefore never both informative: either $w_S=1$ and the whole bound is the
second term---the two-level counterexample---or $K_S=-1$ and $C\rho^{-K_S}=C\rho$, an undefined
constant with no decay in it. Counting settles which branch applies here. The seed is
$\phi=\hat c^{\dagger}_{0\uparrow}\ket{\Psi_0}$ at half filling, so $\operatorname{supp}(\phi)$ lies in the set of $(N{+}1)$ determinants with orbital $0\!\uparrow$ occupied, a fraction $n_\uparrow/L=\tfrac12+\tfrac1L$ of the sector, and fills it except where symmetry zeroes an amplitude of $\Psi_0$. By Eq.~\eqref{eq:exclusion}, $K_S\ge0$ demands $|S|\ge\lvert\operatorname{supp}(\phi)\rvert$, a floor of $0.571$ at $L=8$ ($2237$ of $3920$), $0.600$ at $L=10$, and above $0.54$ at $L=12$ and $14$ at any amplitude threshold from $10^{-12}$ to $10^{-10}$, against the sector fractions $0.5561$, $0.3625$, $0.1700$, $0.0800$
at which the resource scan of Sec.~\ref{sec:prices} operates (Fig.~\ref{fig:scaling}; the coarser
replication grid of Table~\ref{tab:fracgrid} reads $0.180$ at $L=12$ for the same operating point).
From $L=8$ on, every operating point of that scan lies below its own floor, so $K_S=-1$ by counting
alone, with no
numerics, no ranking and no protocol, and any point below the same line inherits the verdict
(Sec.~\ref{app:counterex}). 
Theorem~\ref{thm:leakage} carries no such
bookkeeping: its two terms are defined and simultaneously active for every $S$, neither is free of
$\Psi_0$ (Sec.~\ref{sec:leakage:branches}), and $\Lambda_S(\eta)$ is evaluated, subspace by subspace, in
Sec.~\ref{sec:frontier}.

\paragraph{The magnitude, not the constant.}
Two measurements at the same, exactly unit, captured weight show that Eq.~\eqref{eq:weightonly} was indexed
not on a bad constant but on the wrong quantity. First, let $S^\star=\operatorname{supp}(\phi)$, the
smallest determinant subspace of unit captured weight---fixed by the seed, with no ranking, no
tie-breaking and no zero-amplitude padding---of size $(\tfrac12+\tfrac1L)D$ at $L=6$ and $0.571\,D$ at $L=8$, where $426$ of the $4900$ ground-state amplitudes vanish by symmetry. Its weight term
is identically zero, and its relative $L_1$ error is $0.43$ at $L=6$ and $0.35$ at $L=8$ on the periodic rings at $\eta=0.18t$, rising to $1.15$ and $0.73$ at $\eta=0.05t$, against a hard ceiling of $1+w_S=2$.
Second, $S^\star$ is the $K=0$ rung of the Krylov ladder
$S_K=\bigcup_{a\le K}\operatorname{supp}(H^{a}\phi)$, every rung of which contains $\phi$ and so carries
$w_S=1$ identically, while its moment reach grows: $S_K$ reproduces the exact spectral moments through
order $2K{+}1$---to the $10^{-15}$ level, and in general failing at $2K{+}2$, as on the open chain of
Table~\ref{tab:moments} at both rungs that do not fill the sector---whereas coordinate
subspaces of the same size drawn at random miss already $\mu_0$, in all of $300$ draws
(Sec.~\ref{sec:moments}). Along a family of constant, exactly unit captured weight the error therefore
runs from order unity to moment-exact through order $2K{+}1$. The captured weight, identical on every
rung, does not order them.

What does is elementary. Writing $\mu_j,\mu^{S}_j$ for the moments of the unbroadened measures of
Sec.~\ref{sec:moments}, every determinant subspace with $w_S=1$ reproduces $\mu_0$ and $\mu_1$ exactly,
and the first moment that can fail is $\mu_2$, with deficit
\begin{equation}
\mu_2-\mu^{S}_2\;=\;\lVert Q_S H P_S\phi\rVert^{2},\qquad Q_S=I-P_S ,
\label{eq:mu2}
\end{equation}
the squared coupling of the Hamiltonian across the boundary of $S$ ($E_0$ drops out because
$Q_S\phi=0$), zero exactly when $S$ also holds $\operatorname{supp}(H\phi)$---the $K_S\ge1$ condition.
That boundary coupling is the magnitude the replacement is indexed on: it controls the leakage term
through $\Lambda_S(\eta)\le\lVert P_SHQ_S\rVert\,\lVert\phi_S\rVert$, and Theorem~\ref{thm:leakage}
averages it over the resolution window rather than over the spectrum. The captured weight says what the
subspace \emph{contains}; the error is set by what the Hamiltonian does to the seed at the subspace's
\emph{edge}. Both halves of this paper follow from that displacement: Sec.~\ref{sec:frontier} measures
where the boundary term is small enough for the certificate to say something, and Sec.~\ref{sec:nogo}
shows why no orbital-rotation invariant can predict it usefully on every state in advance.

\makeatletter
\@ifundefined{c@conjecture}{\newtheorem{conjecture}{Conjecture}}{}
\makeatother

\section{Moment exactness on a Krylov-adequate subspace}
\label{sec:moments}

A subspace that contains the order-$K$ Krylov space of the probe reproduces its spectral moments exactly
through order $2K+1$ (Theorem~\ref{thm:moments}); the statement is classical, and what is measured here
is its exact reach on the coordinate subspaces this work builds. That reach is sharp: on the construction
of Table~\ref{tab:moments} the moments part company at order $2K+2$, while coordinate subspaces of the
same size drawn at random fail already at the zeroth moment. The guarantee is structural: it depends on
which determinants the subspace holds, not on how many, and it is decidable from the subspace and the probe,
without the exact spectral function. It is not an error bound, and it does not imply convergence in $K$
(Sec.~\ref{sec:moments:notbuy}).

\subsection{Statement, proof and measured reach}
\label{sec:moments:exact}

Write $\mathcal K_K=\operatorname{span}\{\phi,H\phi,\dots,H^{K}\phi\}$ for the order-$K$ Krylov space
of the seed $\ket{\phi}=\hat c^{\dagger}_{p}\ket{0}$, and let $d\mu$ and $d\mu_{\mathcal S}$ be the
unbroadened Stieltjes measures whose Lorentzian convolutions are the exact $A$ and the reconstruction
$A_{\mathcal S}$, the Lehmann sum evaluated inside $\operatorname{span}(\mathcal S)$ as in
Sec.~\ref{sec:method} (Sec.~\ref{app:moments}).

\begin{theorem}[Moment exactness]
\label{thm:moments}
Let $P$ be the orthogonal projector onto $\operatorname{span}(\mathcal S)$ and $H_{\mathcal S}=PHP$.
If $\operatorname{span}(\mathcal S)\supseteq\mathcal K_K$, then
\begin{widetext}
\begin{equation}
\int \omega^{j}\,d\mu_{\mathcal S}(\omega)\;=\;\int \omega^{j}\,d\mu(\omega)
\;=\;\braket{\phi|(H-E_0)^{j}|\phi},\qquad 0\le j\le 2K+1 .
\label{eq:momex}
\end{equation}
\end{widetext}
\end{theorem}

\begin{proof}
Because $\mathcal K_K\subseteq\operatorname{span}(\mathcal S)$ we have $PH^{a}\phi=H^{a}\phi$ for
$0\le a\le K$, whence by induction $H_{\mathcal S}^{\,b}\phi=P H^{b}\phi$ for $0\le b\le K+1$. Then for
any $a\le K$ and $b\le K+1$, using only $P=P^{\dagger}=P^{2}$ and $H=H^{\dagger}$,
\begin{widetext}
\begin{equation}
\braket{\phi|H_{\mathcal S}^{\,a+b}|\phi}=\braket{H_{\mathcal S}^{\,a}\phi|H_{\mathcal S}^{\,b}\phi}
=\braket{H^{a}\phi|P H^{b}\phi}=\braket{P H^{a}\phi|H^{b}\phi}=\braket{\phi|H^{a+b}|\phi},
\label{eq:rrsplit}
\end{equation}
\end{widetext}
and every $j\le 2K+1$ is such an $a+b$. Both measures are shifted by the same $E_0$, so the identity
survives term by term in the binomial expansion of $(H-E_0)^{j}$.
\end{proof}

\noindent The proof is three lines of Rayleigh--Ritz, and the statement is the moment-matching property
of Krylov projections that underlies Lanczos--Gauss quadrature~\cite{golub2010}. The proof does not identify $A_{\mathcal S}$ with a
Gaussian-quadrature rule, which is not the object built here: every reconstruction diagonalizes
$P_{\mathcal S}HP_{\mathcal S}$ on a coordinate subspace, so $d\mu_{\mathcal S}$ carries
$|\mathcal S|$ nodes, not $K+1$ (Sec.~\ref{sm:moments}).

Table~\ref{tab:moments} measures Eq.~\eqref{eq:momex} on the construction this work uses---the union of
the determinant supports of $\{\phi,H\phi,\dots,H^{K}\phi\}$---for the open Hubbard chain at $L=6$,
$U/t=8$, in the $(N{+}1,S_z{=}{+}\tfrac12)$ sector of dimension $D=300$. Through order $2K+1$ the moments
agree to a relative $10^{-14}$; at order $2K{+}2$ they part company by eleven orders of magnitude at
$K=0$ and seven at $K=1$. The order $2K+1$ is the exact reach of the construction, not merely a
sufficient claim.

\begin{table*}[tp]\centering\small
\caption{\textbf{Moment exactness and its exact reach.} Open Hubbard chain, $L=6$, $U/t=8$; seed
$\ket{\phi}=\hat c^{\dagger}_{0\uparrow}\ket{0}$; $(N{+}1,S_z{=}{+}\tfrac12)$ sector of dimension $D=300$;
$\mathcal S=\bigcup_{a\le K}\operatorname{supp}(H^{a}\phi)$. The largest Krylov order inside
$\operatorname{span}(\mathcal S)$, $K_{\mathcal S}$, and the captured Born weight $w_{\mathcal S}$ are
measured, not assumed. $\epsilon_j$ is the \emph{relative} discrepancy of the $j$-th moment, quoted as
the largest value over six equivalent implementations of the same identity (the deposited generator,
\texttt{src/moments\_table.py}, regenerates the roundoff residues to the same order and every other entry
to the precision printed); the absolute discrepancy at order $2K{+}2$ follows in parentheses. Only relative residues are quoted below that order, because
$\mu_j$ spans five decades between $j=0$ and $j=6$ and the absolute roundoff residue is
route-dependent (Sec.~\ref{app:moments:sharp}). At $K=2$ the ladder has filled the sector, so
$\operatorname{span}(\mathcal S)$ carries every eigenvector of nonzero Lehmann weight,
$A_{\mathcal S}=A$ identically, and there is no order-$2K{+}2$ failure to exhibit; that row is kept to
mark where the ladder terminates, and its single entry is the roundoff level of the moments.}
\label{tab:moments}
\begin{tabular}{cccccc}\toprule
$K$ & $|\mathcal S|$ & $|\mathcal S|/D$ & $w_{\mathcal S}$ & $K_{\mathcal S}$ &
$\max_{j\le 2K+1}\epsilon_j\ \longrightarrow\ \epsilon_{2K+2}$ (absolute) \\ \midrule
$0$ & $200$ & $0.667$ & $1$ & $0$ &
$1.8\times10^{-15}\ \longrightarrow\ 9.5\times10^{-4}$\ \ $(3.10\times10^{-2})$ \\
$1$ & $280$ & $0.933$ & $1$ & $1$ &
$3.2\times10^{-15}\ \longrightarrow\ 6.0\times10^{-8}$\ \ $(1.34\times10^{-4})$ \\
$2$ & $300$ & $1.000$ & $1$ & full sector &
$7.5\times10^{-15}\ \longrightarrow\ $ no failure: $A_{\mathcal S}=A$ \\
\bottomrule\end{tabular}
\end{table*}

A negative control fixes what the table is about. Over $300$ coordinate subspaces of the \emph{same}
size $|\mathcal S|=200$---two thirds of the sector---drawn uniformly from it, every single draw fails
already at the zeroth moment, against $4\times10^{-16}$ for the Krylov-support subspace of identical
size (Sec.~\ref{app:moments:sharp}). What
Theorem~\ref{thm:moments} certifies is \emph{which} determinants the subspace holds, not \emph{how
many}.

The ladder terminates once the subspace contains the cyclic subspace generated by $\phi$, and there the
reconstruction is exact: at $L=6$ that is $K=2$ on the open chain, and already $K=1$ on the periodic
rings of Sec.~\ref{sec:prices}, where it stalls at $296$ of $300$ determinants with $A_{\mathcal S}=A$
to $10^{-14}$ (Sec.~\ref{app:moments:sharp}). These three rows are the whole validation, and they say
nothing about the large-$K$ behavior of $\lVert A-A_{\mathcal K_K}\rVert_1$, which is the open question
of Conjecture~\ref{conj:geom}.

\subsection{What moment exactness does not buy}
\label{sec:moments:notbuy}
Moment exactness is a list of $2K+2$ linear identities, and the $L_1$ distance between two
Lorentzian-broadened measures is not a finite linear functional of them, so geometric convergence in the
Krylov order does not follow from it.

\begin{conjecture}[Geometric convergence in the Krylov order]
\label{conj:geom}
Fix $\eta>0$ and let $A_{\mathcal K_K}$ be the Rayleigh--Ritz reconstruction on $\mathcal K_K$. Are
there $C(\eta)<\infty$ and $\rho(\eta)>1$, uniform over a class of Hamiltonians containing the models
studied here, such that $\lVert A-A_{\mathcal K_K}\rVert_{1}\le C(\eta)\,\rho(\eta)^{-K}$ for every
$K$?
\end{conjecture}

\noindent If the class contains the two-level family of Sec.~\ref{sec:retraction}, that family already forces
$C(\eta)\ge2\lVert\phi\rVert^{2}$ at $K=0$, so any content the conjecture has must be carried by
$\rho(\eta)^{-K}$ alone; and the chain family of the same subsection, on which
$\operatorname{span}(S)=\mathcal K_K$, forces $C(\eta)\rho(\eta)^{-K}\ge2\lVert\phi\rVert^{2}$ at every
$K$ as $g/\eta\to\infty$, so the class must bound the couplings relative to $\eta$. Nor is
Eq.~\eqref{eq:momex} an error bound: for a coordinate subspace, containing $\mathcal K_K$ forces
$w_{\mathcal S}=1$, and at unit captured weight the error is $O(1)$---$0.28$ on the open chain of
Table~\ref{tab:moments} at $\eta=0.18\,t$, and up to $0.43$ on the periodic rings
(Sec.~\ref{sec:retraction}). The panel of Fig.~\ref{fig:method} read against the conjecture, and the full
discussion, are in Sec.~\ref{sm:moments}.

\subsection{The shot count}
\label{sec:moments:shots}

Proposition~\ref{prop:shots} counts the shots that populate such a subspace.

\begin{proposition}[Sampling capture]
\label{prop:shots}
Let $p(x)=\tfrac{1}{K+1}\sum_{k=0}^{K}p_k(x)$ with
$p_k(x)=\lvert\braket{x|e^{-iHt_k}\phi}\rvert^{2}/\lVert\phi\rVert^{2}$, draw $T$ bitstrings
independently from $p$, or, for $T$ a multiple of $K{+}1$, $T/(K{+}1)$ independently from each $p_k$ as in the protocol of
Sec.~\ref{sec:method}, and let $\mathcal S$ be the set of distinct outcomes. For
$\varepsilon,\delta\in(0,1)$, if
\begin{equation}
T\;\ge\;\varepsilon^{-1}\log\!\bigl(1/(\varepsilon\delta)\bigr)
\label{eq:shots}
\end{equation}
then with probability at least $1-\delta$ the set $\mathcal S$ contains every configuration $x$ with
$p(x)\ge\varepsilon$.
\end{proposition}

The count is written in $\varepsilon$ and not in $|\mathcal S|$ because the form
$T=O(\varepsilon^{-1}\log(|\mathcal S|/\delta))$ is circular: $|\mathcal S|$ is the random output of
the very $T$ draws the bound is meant to fix. With $\mathcal N_\varepsilon=\{x:p(x)\ge\varepsilon\}$,
each member contributing at least $\varepsilon$ to a probability distribution,
$|\mathcal N_\varepsilon|\le1/\varepsilon$ \emph{deterministically}; a union bound over
$\mathcal N_\varepsilon$ of the event ``$x$ never drawn in $T$ trials'' gives
$\Pr[\mathcal N_\varepsilon\not\subseteq\mathcal S]\le\varepsilon^{-1}e^{-\varepsilon T}$, hence
Eq.~\eqref{eq:shots}. For the stratified draws, $\Pr[x\notin\mathcal S]=\prod_k(1-p_k(x))^{T/(K+1)}\le
e^{-Tp(x)}$, so the same count holds verbatim; and wherever the sampler returns more than
$\varepsilon^{-1}$ distinct configurations this logarithm is the smaller of the two as well
(Sec.~\ref{app:moments:shots}).

The parameter $\varepsilon$ thresholds a Born probability, not a spectral accuracy: capturing the
high-probability configurations is a statement about the sampler, not about the spectrum
(Sec.~\ref{sec:moments:notbuy}); what the resulting subspace costs is priced in
Sec.~\ref{sec:prices} (Sec.~\ref{sm:moments} gives the two compositions of the bound this work does not
make).

\section{Where the leakage certificate is informative: its non-vacuity threshold, measured}
\label{sec:frontier}

Whether the certificate of Theorem~\ref{thm:leakage} says anything on a given subspace is decided by
one comparison, computed from the Ritz pairs and the retained probe, given the exact probe (H4). Its
\emph{non-vacuity threshold} $\mathrm{FR}^{\ast}$ is the sector fraction above which the leakage branch
of the bound first improves on the trivial bound $\lVert\phi\rVert^{2}(1+w_S)$, a property of the
inequality chain that proves the theorem rather than of the reconstruction. Swept at $L=6$--$12$, the
certificate first improves on the trivial bound at a true relative error of $1.2\times10^{-3}$ to
$7.3\times10^{-3}$ in every one of sixteen $(L,\eta)$ cells (Sec.~\ref{sec:frontier:calib}); above the
threshold it gives converged bounds below the trivial one (Sec.~\ref{sec:frontier:certified});
evaluated on the published subspaces at five sizes up to $L=14$, at the operating points of the
resource scan, which lie well below the threshold, it is vacuous (Sec.~\ref{sec:frontier:threshold}); and the
displacement of the threshold is located rung by rung in the proof, with the basis structure behind it
(Sec.~\ref{sec:frontier:obstruction}).

Non-vacuity is a condition on one number. With $\widehat\Lambda_S=\Lambda_S/\lVert\phi\rVert$, the
leakage branch of Eq.~\eqref{eq:twobudget} improves on the trivial branch exactly when
\begin{widetext}
\begin{equation}
\frac{\widehat\Lambda_S(\eta)}{\eta}\;<\;\tau(w_S)\;:=\;\frac{1+w_S}{1+\sqrt{w_S}}-\sqrt{1-w_S},
\qquad \tau(1)=1 ,
\label{eq:tau}
\end{equation}
\end{widetext}
so a subspace holding the entire Born weight of the probe certifies something as soon as its leakage,
in units of the resolution it was asked for, drops below unity.

\subsection{What is evaluated}
\label{sec:frontier:hyp}

We evaluate Theorem~\ref{thm:leakage} on the subspaces of Fig.~\ref{fig:scaling} at $L=6$, $8$, $10$ and
$12$, with the ranking protocol of Sec.~\ref{sec:prices}, over fractions
$0.20$--$0.95$ and $\eta=0.10$--$0.25\,t$: $204$ evaluations, of which $187$ are converged in Krylov
depth by the depth test of the sweep---the internal indicator $\Lambda_{\rm in}/\Lambda_{\rm out}$, the
truncation part of $\Lambda$ over its boundary part, below $0.05$, and $\Lambda_{\rm out}$ constant to
$10^{-3}$ relative over the last three depths of the ladder---and no row flagged not converged is
quoted. To that grid we add twenty evaluations at the
published subspaces themselves, at the integer counts $246$, $2180$, $19\,183$, $124\,407$ and
$824\,504$ at $L=6$--$14$. The sixteen up to $L=12$ pass the depth test and the $L=14$ row its
indicator half only (Table~\ref{tab:published}). The subspaces are the $T\to\infty$ limit of the
protocol, not finite-shot draws. What is certified is the object the code forms, the reorthogonalized
Krylov reconstruction $A_{\mathcal K}$ at a declared depth, against a reference whose drift is declared,
and no measured error within a factor $20$ of that drift is quoted. The protocol, the controls and the
costs are in Secs.~\ref{sm:frontier} and~\ref{app:hyp}.

\subsection{The certificate at the operating fractions of the resource scan}
\label{sec:frontier:threshold}

Evaluated on the published subspaces themselves, the
leakage branch of the bound exceeds the trivial bound $\lVert\phi\rVert^{2}(1+w_S)$ at all five sizes
and all four resolutions, by factors running from $1.66$ ($L=6$, $\eta=0.25\,t$) to $8.61$ ($L=14$,
$\eta=0.10\,t$), and by $2.40$ to $4.61$ at the $\eta=0.18\,t$ of the momentum-resolved channel;
Table~\ref{tab:published} carries the twenty values. They are monotone in both
arguments---rising with $L$ at every resolution, falling with $\eta$ at every size---so the certificate
is furthest from non-vacuity at the largest sector and the finest resolution, the regime the method
targets.

The largest classically converged case carries the evidence. At $L=12$ and the published
subspace---$\lvert S\rvert=124\,407$ of $731\,808$ determinants---$w_S=0.998622$, so
Eq.~\eqref{eq:tau} asks for $\widehat\Lambda_S/\eta<\tau(w_S)=0.9625$; the measured value is $4.236$
at $\eta=0.18\,t$, converged in depth (Sec.~\ref{sm:frontier}). Over the $606$ live subspaces of the
pooled sweep the leakage branch is the smaller of the two on only $61$. The reconstructions at
the operating fractions of the resource scan---Fig.~\ref{fig:scaling} and the operating points of
Secs.~\ref{sec:prices} and~\ref{sec:physics}---are therefore uncertified; the Born-ranked
reconstruction of Fig.~\ref{fig:akwsampled}, at $85\%$ of its sector, lies above the threshold
(Sec.~\ref{sec:frontier:certified}).

\begin{table*}[tp]\centering\small
\caption{\textbf{The certificate on the subspaces this paper publishes, including the size the sweep
never reached.} Each row is Theorem~\ref{thm:leakage} evaluated at the integer determinant count
$\lvert\mathcal S\rvert=\mathrm{round}(\mathrm{frac}\times\dim)$ that Fig.~\ref{fig:scaling} reports,
not at the nearest point of the grid of Table~\ref{tab:thresholds}. Columns~$5$--$8$ are the vacuity
factor---the leakage branch of Eq.~\eqref{eq:twobudget} over the trivial branch
$\lVert\phi\rVert^{2}(1+w_S)$---so a value above $1$ means the certificate says nothing the trivial
bound does not. All twenty evaluations are the $T\to\infty$ limit of the ranking protocol; the sixteen up to
$L=12$ pass the depth test of Sec.~\ref{sec:frontier:hyp}, and the $L=14$ row its indicator half only
(see below); controls and costs are in Sec.~\ref{app:hyp}. The last column is the measured
relative $L_1$ error at the resolution and against the threshold that \emph{define} the published
fraction ($\eta=0.15\,t$, rel-$L_1<0.05$), from code sharing neither reference, normalization nor
recursion depth with the scan that produced the fractions. At $L=14$ the ground state is the stored checkpoint, the count is one determinant above the deposited $824\,503$, and the $K=18$ ranking, re-run at that size, reproduces the stored selection as a set (Sec.~\ref{sm:frontier}). The $L=14$ row is taken at the deepest recursion run, depth $250$, where $\Lambda_{\rm out}$ was not recorded, so only the indicator half of the depth test applies: the depth indicator is $0.045$ and the factor at $\eta=0.10\,t$ still falls by $0.17\%$ between depths $200$ and $250$ ($8.627\to8.613$); the vacuity is not at risk.}\label{tab:published}
\setlength{\tabcolsep}{6pt}
\begin{tabular}{lcccccccc}
\toprule
 & & \multicolumn{2}{c}{published subspace} & \multicolumn{4}{c}{vacuity factor at $\eta/t=$} & rel-$L_1$\\
\cmidrule(lr){3-4}\cmidrule(lr){5-8}
$L$ & $\dim(N{+}1)$ & fraction & $\lvert\mathcal S\rvert$ & $0.10$ & $0.15$ & $0.18$ & $0.25$ & at $\eta=0.15\,t$\\
\midrule
$6$  & $300$          & $0.8200$ & $246$      & $4.54$ & $2.94$ & $2.40$ & $1.66$ & $0.0477$\\
$8$  & $3\,920$       & $0.5561$ & $2\,180$   & $5.65$ & $3.69$ & $3.05$ & $2.17$ & $0.0478$\\
$10$ & $52\,920$      & $0.3625$ & $19\,183$  & $6.50$ & $4.28$ & $3.54$ & $2.52$ & $0.0484$\\
$12$ & $731\,808$     & $0.1700$ & $124\,407$ & $8.04$ & $5.20$ & $4.27$ & $3.00$ & $0.0485$\\
$14$ & $10\,306\,296$ & $0.0800$ & $824\,504$ & $8.61$ & $5.59$ & $4.61$ & $3.26$ & $0.0463$\\
\bottomrule
\end{tabular}
\end{table*}

Two controls come with the table. Applying this paper's own acceptance criterion to the same sweep with
that independent code returns the reported fractions to under $2\%$ (Table~\ref{tab:thresholds}), and
the last column of Table~\ref{tab:published} evaluates the defining criterion at the published
subspaces themselves, $0.0463$--$0.0485$ against the threshold $0.05$. The reconstructions and the
curve are thus reproduced independently; what those fractions lack is a certificate.

\subsection{What the certificate certifies, and at what shot cost}
\label{sec:frontier:certified}

What the sweep does certify, at each size the strongest statement whose bracket is converged, is
collected in Table~\ref{tab:certified}: at $L=12$, $60\%$ of a $731\,808$-dimensional sector gives
relative $L_1\le0.717$, a factor $2.8$ below the trivial cap, tightening to $\le0.218$ at $70\%$ and
$\eta=0.25\,t$. One further point of the map lies above the threshold: for the Born-ranked
(infinite-shot) reconstruction of Fig.~\ref{fig:akwsampled}, at $85\%$ of its sector, the certificate
of the depth-$500$ Krylov reconstruction is non-vacuous in all sixteen channels, at a relative bound of
$0.77$--$1.75$ against the trivial $2$, and loose, $350$--$900$ times the true error of the channel
integrated over the whole frequency line (\texttt{cert\_akw.json}). At the acceptance criterion used
throughout this paper, relative $L_1<0.05$, which the reconstruction at $L=8$, $\eta=0.18\,t$ reaches at
$\mathrm{FR}\approx0.52$, the certificate requires at least $0.98$ of the sector at $L=8$, the crossing
lying in $[0.977,0.983]$, and on our grid it is not reached at all at $L=6$, $10$ or $12$: at that
accuracy Theorem~\ref{thm:leakage} does not certify the operating points of Fig.~\ref{fig:scaling}, and
the regime where it would costs more shots than the method saves, since reaching the non-vacuous band
multiplies the shots per snapshot needed for the marginal determinant by $7.5$--$10$ at $L=8$ and by
$20$--$35$ at $L=12$ (Table~\ref{tab:shotfrontier}).

\subsection{The crossing of the trivial bound is calibrated}
\label{sec:frontier:calib}

Across all sixteen $(L,\eta)$ cells the crossing of the trivial
bound occurs at a true relative $L_1$ error of $1.2\times10^{-3}$ to $7.3\times10^{-3}$, median
$2.8\times10^{-3}$: a spread of $5.9$ over true errors spanning five decades, and in every cell at least
$47$ times the drift of our own reference. The spread is a drift, not scatter, and it runs the
conservative way in $L$: the larger the system, the smaller the true error at which the certificate
first says anything. The crossing of $\widehat\Lambda_S(\eta)/\eta$ therefore locates a true error to
within that factor, and it is computed from $H$, $S$, the retained probe $\phi_S$, the probe norm
$\lVert\phi\rVert$ and $\eta$, without the spectral function it bounds, given the exact probe (H4),
with four limits (Sec.~\ref{sm:frontier}). It locates one error level per cell and is not an error
estimate. A user without $w_S$ can discharge $\tau(w_S)$ conservatively: an assumed floor
$w_S\ge0.995$ gives $\tau\ge0.9280$, moving the threshold by at most $7\%$. The calibration rests on
these sixteen cells, $L=6$--$12$, and is not extrapolated to $L=14$. A two-constant rule in $1-w_S$,
$1-w_S<c\,\eta^{p}$ with $c$ and $p$ fitted on some sizes and tested on the others, localizes the
crossing error more tightly than the certificate on eleven of the fourteen train/test splits
(\texttt{null\_ws7.py}): as a tracker of the error the certificate has no edge in sharpness over the
captured weight. Its edge is provability. It is an inequality with a proof, whose decision survives
replacing $w_S$ by a floor, whereas the fitted rule is a regularity of these systems with no guarantee
outside its fitted range---no inequality indexed on the captured weight alone bounds the error
(Sec.~\ref{sec:retraction})---and it asks for $1-w_S$ to $10^{-5}$--$10^{-4}$ and, stripped of its
second parameter, is the looser of the two on thirteen of the fourteen splits.

\subsection{Where the threshold comes from}
\label{sec:frontier:obstruction}

Measured rung by rung along the chain of inequalities that proves Theorem~\ref{thm:leakage}, at $L=8$,
the displacement of the threshold is $44\%$ the pointwise step to a Hilbert-space distance, $38\%$ the
operator-norm step $\lVert R(z)\rVert\le1/\eta$, $11\%$ the Feshbach/Minkowski split and $4.6\%$
Cauchy--Schwarz in $\omega$, whose total margin for tightening is a factor $1.05$ typically and $3.29$
at worst. The subspace adds one structural fact. For a subspace holding the whole probe, $\Lambda_S=0$
requires $\operatorname{span}(S)\supseteq\mathcal K(H,\phi)$, and the reflection $j\to-j$ about the
seed site commutes with $H$ and has $\phi$ as an eigenvector, so every vector of $\mathcal K$ vanishes on
the reflection-fixed determinants of the opposite parity---$4$, $18$, $48$, $200$ and $600$ at
$L=6$--$14$---while a reflection-projected power iteration (\texttt{z\_suppK\_sym.py}) finds every other
determinant in the support: $296$, $3902$, $52\,872$, $731\,608$ and $10\,305\,696$ at $L=6$--$14$. At
full captured weight, therefore, no coordinate subspace short of all of them has zero leakage. That is a
statement that $\Lambda_S\neq0$, not about its size; the vacuity at the operating fractions of the
resource scan, all below full captured weight, is measured in Table~\ref{tab:published}, not deduced
from it. Nor is the ranking the weak link: a Krylov-diagonal oracle ranking gives a \emph{worse} certificate than the
Born ranking at $L=6$, $\eta=0.18\,t$, at every fraction checked from $0.50$ on. At full captured weight
the leakage branch is of first order in the boundary coupling and the error of second, and the same
identity gives an a posteriori certificate of second order,
\begin{equation}
\lVert A-A_P\rVert_1\;\le\;\frac{\Lambda_P(\eta)\,\tilde\Lambda_P(\eta)}{\eta^{2}}\qquad(\phi_Q=0),
\label{eq:secondorder}
\end{equation}
where $\tilde\Lambda_P$ is Eq.~\eqref{eq:lambda} with the energy difference in its kernel reversed in
sign (equal to $\Lambda_P$ for a real Hamiltonian and probe) (Sec.~\ref{sm:frontier}); whether one
exists below full captured weight, where every published fraction lies, is open. The measurement basis
is a choice, and for the momentum-resolved probe of Fig.~\ref{fig:akwsampled} it is a forced one: the
Krylov space lies inside one momentum block of exactly $\dim(N{+}1)/L$
plane-wave determinants, on which the certificate becomes non-vacuous at $48$--$50$ of $50$ determinants
($L=6$) and $455$--$483$ of $490$ ($L=8$) (\texttt{z\_kblock.py}), a factor $5$--$7$ fewer than the site
determinant basis requires. The block exists only after the measurement is rotated to the plane-wave
orbitals, about $L(L-1)$ further two-qubit gates whose effect on the sampled distribution is unmeasured
(Sec.~\ref{sm:frontier}); in the site basis the same probes at $L=8$ have support on $3848$ to $3920$ of
the $3920$ determinants (Sec.~\ref{sm:leakage}). The gain is $1/L$ against an exponential, and whether
the support fills the block at $L\ge10$ is unmeasured.

\section{No shortcut to the boundary: an obstruction and its witness}
\label{sec:nogo}

Section~\ref{sec:leakage} leaves one question open: can the a posteriori certificate be replaced by an
a priori predictor computed from the state before any circuit runs? The natural candidates are the
monotones of fermionic non-Gaussianity, efficiently measurable from two-point functions
(Sec.~\ref{sm:nogo} reviews them). None can play that role, for either term of
Eq.~\eqref{eq:twobudget}: not for the static determinant support $|\mathcal S|_{\varepsilon_w}$ of the
ground state that the shots must populate, and not for the leakage $\Lambda_S(\eta)/\eta$ of the
coordinate subspace the sampler returned. What replaces the predictor is the measured map of
Sec.~\ref{sec:frontier}. Three results establish this. On number-conserving states the one-body
antiflatness is a single invariant of the one-particle density matrix (Eq.~\eqref{eq:collapse}); a
geminal family carries antiflatness of every order at bond dimension $2$ (Theorem~\ref{thm:witness});
and on fibers on which every invariant is frozen the subspace-aware certificate still separates, so
the leakage bound that invariants do supply is vacuous there (Proposition~\ref{prop:leakinv}).

\subsection{The collapse: one-body magic is a single 1-RDM invariant}
\label{sec:nogo:collapse}

We work with the fermionic AntiFlatness (FAF)~\cite{sierant2026},
\begin{widetext}
\begin{equation}
\mathcal F_k \;=\; 2n_{\rm o}-\tfrac12\,\mathrm{tr}\bigl[(\Gamma^{\top}\Gamma)^k\bigr] \;=\; 2n_{\rm o}-\sum_{j=1}^{2n_{\rm o}}\lambda_j^{2k},
\qquad \Gamma_{mn}=-\tfrac{i}{2}\langle[\gamma_m,\gamma_n]\rangle,
\label{eq:faf}
\end{equation}
\end{widetext}
with $\Gamma$ the $4n_{\rm o}\times4n_{\rm o}$ Majorana covariance matrix, $\{\lambda_j\}\in[0,1]$ the
moduli of its eigenvalue pairs $\pm i\lambda_j$, and $2n_{\rm o}$ the number of spin-orbital \emph{modes}
(twice the spatial-orbital count), distinct from the electron number of the $(N\pm1)$ sectors. Each
$\mathcal F_k$ vanishes if and only if the state is Gaussian and is invariant under Gaussian
(free-fermion) unitaries by construction~\cite{sierant2026}; the index $k$ is the replica (moment) order
of $\mathrm{tr}[(\Gamma^{\top}\Gamma)^k]$, and every $\mathcal F_k$ is efficiently measurable from
two-point functions~\cite{sierant2026,bittel2025,tarabunga2026}.

On the number-conserving states of electronic structure the anomalous correlators $\langle c_ic_j\rangle$
vanish, so $\Gamma$ is fixed by the one-particle reduced density matrix $\gamma$ and its singular values
are $\lambda_i=|2g_i-1|$, with $g_i\in[0,1]$ the natural spin-orbital occupations. The $k=1$ member then
reduces to a single 1-RDM invariant,
\begin{widetext}
\begin{equation}
\boxed{\;\mathcal F_1 \;=\; \sum_i\bigl[1-(2g_i-1)^2\bigr] \;=\; 4\sum_i g_i(1-g_i)
\;=\; 4\,\mathrm{tr}[\gamma(1-\gamma)]\;,}
\label{eq:collapse}
\end{equation}
\end{widetext}
exact on any number-conserving state and verified state by state to $10^{-13}$ on the $38$ molecular
ground states of this work. For spin-restricted natural orbitals it equals $2N_{\mathrm u}$, the
effectively-unpaired-electron index~\cite{takatsuka1978,staroverov2000}, and
$\mathrm{tr}[\gamma(1-\gamma)]$ is one quarter of the one-body quadratic (linear)
entropy~\cite{gigena2015}, a trace-form one-body entanglement entropy~\cite{gigena2020}. The
algebra is an immediate number-conserving specialization of the antiflatness of
Ref.~\cite{sierant2026}, noted independently for fixed-particle-number states in a study of
interacting integrable models~\cite{swietek2026}. Being a function of the occupation spectrum alone,
$\mathcal F_1$ is invariant under orbital rotations (the spin-polarized cases and the dictionary are in
Sec.~\ref{sm:nogo}).

\subsection{The geminal witness, and what an invariant cannot certify}
\label{sec:nogo:witness}

What obstructs the operative quantities---the thresholded support $|\mathcal S|_{\varepsilon_w}$ and
the leakage---is a symmetry, exhibited by a family
whose amplitudes are $\Theta(1)$ and therefore threshold-robust: an antisymmetrized product of $K$
strongly orthogonal singlet geminals (perfect-pairing / generalized-valence-bond),
\begin{equation}
\ket{\Psi_K}\;=\;\bigotimes_{i=1}^{K}\Bigl(\cos\theta\,\hat b^{\dagger}_{i\uparrow}\hat b^{\dagger}_{i\downarrow}
+\sin\theta\,\hat a^{\dagger}_{i\uparrow}\hat a^{\dagger}_{i\downarrow}\Bigr)\ket{0},
\label{eq:apsg}
\end{equation}
on $2K$ spatial orbitals ($4K$ modes, $N=2K$ electrons), each geminal occupying a disjoint
bonding/antibonding pair $(b_i,a_i)$.

\begin{theorem}[Geminal orbit witness]\label{thm:witness}
At maximal pairing $\theta=\pi/4$ the family \eqref{eq:apsg} satisfies, exactly,
\begin{widetext}
\begin{equation}
\mathcal F_1=2N_{\mathrm u}=4K,\qquad \mathcal F_2=4K,\qquad
|\mathcal S|_{\mathrm{pair}}=2^{K},\qquad \chi=2,
\end{equation}
\end{widetext}
with $\mathcal F_{1,2}$ the order-$1$ and order-$2$ antiflatness of Eq.~\eqref{eq:faf},
$|\mathcal S|_{\mathrm{pair}}$ the determinant support in the pairing basis, and $\chi$ the minimal
matrix-product-state bond dimension in the pairing ordering. In that localized (pairing) basis the state
is a bond-dimension-$2$ tensor network, representable, contractible and Born-samplable in $O(K)$.
\end{theorem}

\noindent One-body magic, higher-order magic and determinant support therefore all diverge with $K$:
magic of every order is unbounded at bond dimension $2$, where the simulation cost grows only linearly, as $O(K)$. The
proof, four block-diagonal computations, is given in the Supplemental Material
(Sec.~\ref{app:witness}). The values of $\mathcal F_1$, $\mathcal F_2$ and
$|\mathcal S|_{\mathrm{pair}}$ are reproduced by exact diagonalization for $K=1,\dots,4$, and so is
$\chi=2$, the largest Schmidt rank over all cuts of the pairing ordering (Fig.~\ref{fig:witness} and
Sec.~\ref{app:witness}). The orbit is the
argument: $\mathcal F_k$ is invariant under number-conserving orbital rotations~\cite{sierant2026}, so
the whole orbit $\{U\ket{\Psi_K}\}$ carries the same $(\mathcal F_1,\mathcal F_2)=(4K,4K)$ while the
cost in the fixed computational basis is $O(K)$ at the pairing point and, at generic rotated points, not
established here (Born sampling there is plausibly hard, Sec.~\ref{sm:nogo}). A Gaussian-invariant fermionic
magic therefore does not lower-bound the bond dimension, nor the cost beyond linear order in $K$, and it determines the cost only if the plausibly hard
rotated points are in fact easy. The same orbit reaches the second
budget of Eq.~\eqref{eq:twobudget}:

\begin{proposition}[Fibers on which every orbital-rotation invariant gives a vacuous leakage certificate]
\label{prop:leakinv}
Let $\mathcal U_U$ be the Gaussian unitary implementing $U\in\mathrm U(n_{\rm o})$, and let the same
physical problem written in the rotated orbital basis be
$(H_U,\Psi_U,\phi_U)=(\mathcal U_U^{\dagger}H\mathcal U_U,\,\mathcal U_U^{\dagger}\Psi,\,
\mathcal U_U^{\dagger}\phi)$, the determinant basis---in which the device measures, and of which $S$ is a
coordinate subspace---held fixed. Along this orbit the spectrum of $H$, $E_0$, $\lVert\phi\rVert^{2}$,
every Lehmann weight and hence $A(\omega)$ pointwise are constant, as is every orbital-rotation invariant
of the state, in particular $\mathcal F_1=4\,\mathrm{tr}[\gamma(1-\gamma)]$ and all
$\mathcal F_k$. $\Lambda_S(\eta)/\eta$ is not. An upper bound on $\Lambda_S(\eta)/\eta$ is available from
invariants alone---combining $\Lambda_S\le\lVert P_SHQ_S\rVert\,\lVert\phi_S\rVert$ with
$\lVert Q_SHP_S\rVert=\lVert Q_S(H-c)P_S\rVert\le\min_c\lVert H-c\rVert$ gives
$\Lambda_S(\eta)/\eta\le(W/2)\lVert\phi\rVert\sqrt{w_S}/\eta$, $W$ the spectral width of the $(N\pm1)$
sector---but it is vacuous by a factor of at least $W/2\eta$, which grows extensively with system size.
That no invariant bound can be useful on these fibers is seen at fixed invariants: any
$B(\text{invariants},|\mathcal S|,w_S)$ valid for $\Lambda_S(\eta)/\eta$ is constant on a fiber of fixed
arguments and must therefore exceed the fiber maximum. In the geminal family $\mathrm{APSG}(\pi/4)$ on
$2K$ orbitals with $H=\varepsilon\hat N-g\sum_iA^{\dagger}_iA_i+\delta h$ we exhibit fibers
($K=2,3,4$; $\delta=0.01$; $\eta=0.10$; $|\mathcal S|=2^{K-1}$; spectrum, Lehmann weights and all
$\mathcal F_k$ matched to $10^{-14}$, $w_S$ to $2.2\times10^{-16}$) whose endpoints carry
$\widehat\Lambda_S(\eta)/\eta=0.237$ and $2.326$, $0.323$ and $2.336$, $0.391$ and $2.346$. Since the
larger value exceeds $\tau(w_S)=0.795$ of Eq.~\eqref{eq:tau}, above which the leakage branch of
Theorem~\ref{thm:leakage} no longer improves on the trivial bound
$\lVert\phi\rVert^{2}(1+w_S)$, every valid invariant bound is vacuous on these fibers---while the
subspace-aware certificate is non-vacuous at the aligned endpoint ($0.856$ against a trivial $1.962$).
The statement is one of non-determination and non-vacuity, not of non-existence: invariants do not fail
to bound the leakage, they fail to bound it usefully.
\end{proposition}

\noindent The fibers (Table~\ref{tab:fiber}), an exhaustive search over the $276$ subspaces of that
size at $K=2$, and four limits of the exhibit---a matched pair at one angle, a margin over the trivial
bound that shrinks from $2.29$ to $1.69$ between $K=2$ and $K=4$, a small-leakage endpoint realized
only by the stabilizer of the state, and sector dimensions of at most $3920$---are in
Sec.~\ref{sm:nogo}.

The \emph{raw} determinant support---the unthresholded Slater rank---is obstructed separately, by
continuity: no continuous functional of a fixed set of reduced density matrices can upper-bound the
count more tightly than the sector dimension near a single determinant, states arbitrarily close to one
in every such functional carrying any support up to the whole sector, rank being coefficient-insensitive and
not upper-semicontinuous (though, unlike the CP tensor rank, whose sets of rank at most $r$ are not
closed~\cite{desilvalim2008}, it is lower-semicontinuous); and
at half filling the invariant $\mathcal F_1$ cannot lower-bound the count beyond $2$, the two-determinant
cat having maximal $\mathcal F_1$. That continuity argument (Sec.~\ref{app:theorem}) does \emph{not}
transfer to the operative, coefficient-\emph{sensitive} $|\mathcal S|_{\varepsilon_w}$, which collapses to $1$ along the same
path as the functionals; the symmetry above does.

\paragraph{The cost side: a bilateral obstruction whose sign is not stable.} Unlike $\mathcal F_1$, the
determinant support is basis dependent, and both directions of the decoupling are realized
(Fig.~\ref{fig:decoupling}): a free-fermion ground state has $\mathcal F_1=0$ and a site-basis support
$|\mathcal S|_{\varepsilon_w}$ (open chain, $\varepsilon_w=2.5\times10^{-3}$) of $35$, $336$, $3496$ and
$37361$ at $L=4$--$10$; along a Hubbard chain raising $U$ raises
$\mathcal F_1$ while lowering $|\mathcal S|_{\varepsilon_w}$; across the nineteen molecules the two rise
together. Within each of the six $(L,N)$ strata of the Hubbard sweep the rank correlation of
$\mathcal F_1$ with $|\mathcal S|_{\varepsilon_w}$ is exactly $-1$ (stratified Kendall $\tau=-1.000$,
exact one-sided permutation $p=3.35\times10^{-13}$ against independent orderings within each stratum,
a null that six strata sharing one mechanism do not satisfy), while it is $+1.000$ in the chain and
ladder families and $+0.86$ across the molecules: an excellent rank predictor whose sign reverses
between families, as an orbital invariant compared with a basis-dependent count must. The bond
dimension $\chi$, a rigorous floor on the raw support, tracks the thresholded one better (mean
within-stratum $\rho=0.948$), but it is itself basis dependent, measured rather than predicted, and
blind to the leakage. Computable Gaussian-invariant measures of non-Gaussianity include the
antiflatness used here~\cite{sierant2026} and the two families of Tarabunga \emph{et al.}~\cite{tarabunga2026},
whose natural-orbital participation entropies, themselves Gaussian invariant, also upper-bound the
classical simulation cost in an orthonormal Gaussian basis; their quadratic occupation entropy, which
equals $\mathcal F_1$, is a strong pure-state monotone, proved by Tarabunga \emph{et al.} and,
independently, by Leone and Bittel~\cite{leonebittel2026}. The contrast
between an orbital invariant and a basis-dependent count is our observation, and what the
construction adds is a family on which every order of the invariant is maximal at $O(K)$ cost, and
fibers on which the invariants are frozen while the certificate of Theorem~\ref{thm:leakage} still
separates. The stratified statistics, the molecular set, which does not discriminate, and the two
families where $\mathcal F_1$ wins are in Secs.~\ref{sm:nogo} and~\ref{app:stats}.

\section{Three prices: support, resolution, and shots}\label{sec:prices}

On half-filled periodic Hubbard rings at $U/t=8$ ($L=4$--$14$, $8$ to $28$ qubits) the method is
priced on three axes, each measured here: the determinant support the sampler must populate grows
as $e^{1.05\,L}$, more slowly than the sector's $e^{1.30\,L}$; the sector fraction, unlike the
support, is a property of the question put to the model rather than of the model, its exponent
moving by a factor $2.2$ over a decade of the resolution $\eta$; and the shot budget $T$ that
populates the support, measured at $L=4$--$12$, lies two orders of magnitude above it and grows as
$e^{1.12\,L}$ over $L=6$--$12$, its $L=14$ value being a floor of $1.3\times10^{8}$ shots for a
single seed and branch.

\paragraph{The dynamical support.} $|\mathcal S|_\eta$ here is the number of
determinants needed for the reconstruction of $A(\omega)$ to reach relative $L_1<\mathrm{thr}$ at
resolution $\eta$ and recursion depth $n_\ell$, so it depends on all three
(\texttt{scaling\_lanczos\_mf.py}, \texttt{stage\_frac}: $\mathrm{thr}=0.05$, $\eta=0.15\,t$,
$n_\ell=150$). It is distinct from the static ground-state support $|\mathcal S|_{\varepsilon_w}$ of
Sec.~\ref{sec:nogo} and from the returned set $S$ of Theorem~\ref{thm:leakage}: the three coincide only
by accident, and no statement here transfers to the other two (Sec.~\ref{sec:intro}).

\subsection{The support exponent and its resolution dependence}

Across $L=4$--$14$ ($8$ to $28$ qubits, half-filled \emph{periodic} rings at $U/t=8$; the three largest
sizes reached by matrix-free Lanczos/Haydock and validated against dense diagonalization at $L\le8$),
the determinant support needed to reconstruct $A(\omega)$ to relative $L_1<0.05$ at $\eta=0.15\,t$ and
$n_\ell=150$ grows exponentially,
\begin{widetext}
\begin{equation}
|\mathcal S|_{\eta=0.15} \;=\; 22,\;246,\;2180,\;19183,\;124407,\;824503
\;\;\Longrightarrow\;\; |\mathcal S|_\eta\sim e^{1.05\,L}\quad(L=4\text{--}14,\ R^{2}=0.998),
\label{eq:supportexp}
\end{equation}
\end{widetext}
more slowly than the $(N{+}1)$ sector itself,
$\dim = 24,\,300,\,3920,\,52920,\,731808,\,10306296\sim e^{1.30\,L}$. Restricted to $L=4$--$12$, the five sizes of the replication grid, the fit gives $1.082$ at
$R^{2}=0.998$. The fraction is the support divided by the sector dimension, exactly, so the fraction and
support exponents differ by the constant $b_{\dim}=1.291$ and move together, by $0.148$ over the decade
$\eta=0.05\to0.50\,t$ (Table~\ref{tab:fracgrid}): a seventh of the support exponent, three quarters of
the fraction exponent. We therefore report the support, the absolute resource a device has to populate,
and the fraction as the derived, resolution-dependent quantity that it is.

\subsection{The sector fraction is a property of the question, not of the model}

At $\eta=0.15\,t$ the fraction needed for a fixed spectral accuracy falls from $0.92$ at $L=4$ to
$0.08$ at $L=14$, but its exponent is set by $\eta$: over the decade of $\eta$ it moves by a factor
$2.2$, from $-0.1266$ to $-0.2743$ (Table~\ref{tab:fracgrid}), and even at fixed $\eta$ the local
logarithmic slopes between consecutive published points steepen by a factor $6.8$ across five
intervals, so the fraction is not a single exponential either. The reference is itself a
depth-$150$ Lanczos computation whose drift between depths $100$ and $200$ at $L=12$ is $0.93\%$ at
$\eta=0.15\,t$ and $3.0\%$ at $\eta=0.10\,t$, so the fraction is reported as a declared window,
$\eta\ge0.15\,t$ at $n_\ell=150$, with the full grid and the finite-depth caveats in
Sec.~\ref{sm:prices}.

\subsection{The third price: shots}

Populating a subspace by measurement is a coupon-collector problem on the Born distribution of the
time-evolved seed, so the budget exceeds the support it buys. Table~\ref{tab:shots} reports that budget
\emph{measured}, by multinomial sampling of the declared protocol, at $L=4$--$12$, in
two independent runs that agree where they overlap to $0.2\%$ and $3.2\%$; every budget is for
\emph{one} orbital seed and \emph{one} branch.

\begin{table*}[tp]\centering\small
\caption{\textbf{The third price, measured at every size but the largest.} $T$ is the total shot budget,
split evenly over the $19$ snapshots of the protocol, for a single orbital seed and a single branch; it
is the budget at which the finite-shot reconstruction itself first reaches relative $L_1<0.05$, and
$\pm$ is the sample standard deviation across seeds, not the standard error of the mean. Run~A is
\texttt{sampled\_honest.json} ($8$ seeds, $L\le8$); run~B is \texttt{honest\_sampling.json} ($4$ seeds,
$6$ at $L=10$), and the two are separate runs with different seeds and different budget grids, listed
separately at $L=6$ and $L=8$. At $L=4$ run~B
is censored by its own cap of $10^{8}$ shots per snapshot, the $24$-state sector being degenerate for
this purpose, and run~A is quoted there. The last column divides $T$ by the \emph{published} support;
dividing instead by the support each run actually populated gives $95$, $107$, $141$ and $124$ at
$L=6$--$12$. Only $L=14$ is a floor: the overhead held at its $L=12$ value, which is a lower bound and
not an estimate.}\label{tab:shots}
\setlength{\tabcolsep}{5pt}
\begin{tabular}{lccccl}
\toprule
$L$ & $\dim(N{+}1)$ & $|\mathcal S|_\eta$ & $T$ & $T/|\mathcal S|_\eta$ & status\\
\midrule
$4$  & $24$           & $22$       & $(7.8\pm2.0)\times10^{2}$   & $35$  & measured, run~A, $8$ seeds\\
$6$  & $300$          & $246$      & $(2.43\pm0.34)\times10^{4}$ & $99$  & measured, run~A, $8$ seeds\\
$6$  & $300$          & $246$      & $(2.44\pm0.23)\times10^{4}$ & $99$  & measured, run~B, $4$ seeds\\
$8$  & $3\,920$       & $2\,180$   & $(2.70\pm0.13)\times10^{5}$ & $124$ & measured, run~A, $8$ seeds\\
$8$  & $3\,920$       & $2\,180$   & $(2.61\pm0.08)\times10^{5}$ & $120$ & measured, run~B, $4$ seeds\\
$10$ & $52\,920$      & $19\,183$  & $(3.09\pm0.04)\times10^{6}$ & $161$ & measured, run~B, $6$ seeds\\
$12$ & $731\,808$     & $124\,407$ & $(1.91\pm0.01)\times10^{7}$ & $154$ & measured, run~B, $4$ seeds\\
$14$ & $10\,306\,296$ & $824\,503$ & $\gtrsim1.3\times10^{8}$    & $154$ & floor, not measured\\
\bottomrule
\end{tabular}
\end{table*}

The budget is two orders of magnitude larger than the support it buys, already at $L=6$, and it
grows exponentially, $7.8\times10^{2}$ to $1.9\times10^{7}$ over $L=4$--$12$, a log-linear fit over
$L=6$--$12$ giving $T\sim e^{1.12\,L}$ at $R^{2}=0.996$. At $L=8$ the sixteen-channel momentum-resolved
configuration of Fig.~\ref{fig:akwsampled} reached its own resolution and
accuracy with $8.3\times10^{7}$, and not with half as many; the $16\times2.7\times10^{5}\approx4\times10^{6}$
shots that sixteen channels at the local-seed budget would cost is a transfer and not a measurement,
because the momentum seeds have larger support. The $L=14$ floor of $1.3\times10^{8}$ shots for a
\emph{single} seed and branch (Table~\ref{tab:shots}) exceeds the largest shot budget drawn anywhere
in this work, the $8.3\times10^{7}$ shots of the sixteen-channel configuration of
Fig.~\ref{fig:akwsampled}. The overhead per determinant is a level and not a growth law---$35$,
$99$, $124$, $161$, $154$ against the published support across $L=4$--$12$, rising by a factor $3.5$
over the first three sizes ($35$ to $124$), then to $161$ at $L=10$, and not rising over the last
step, $161$ to $154$. Its exponent is not resolved by these data. Fitted on $L=6$--$12$ it is $0.08$ per site against the published support ($95\%$ interval $[-0.03,0.19]$) and $0.054$ ($[-0.066,0.174]$) against the support each run populated; both intervals contain zero, so \emph{these data do not establish that the shot budget grows with a larger exponent than the
support}, only that it is larger in absolute terms at every size measured.

\paragraph{What the published curve is.} The published curve is the $T\to\infty$ limit of the
declared protocol: the ranking that produced the published fractions is \texttt{argsort} applied to
the Born distribution of the time-evolved state accumulated over the $19$ snapshots of the
protocol---the distribution the device would itself measure with infinitely many shots---and not a
clairvoyant ordering by exact eigenvector amplitudes. With
finite shots the fraction needed rises, from $0.8200$ to $0.8571\pm0.0135$ at $L=6$ and from $0.5561$ to
$0.6275\pm0.0031$ at $L=8$ (Table~\ref{tab:finiteshot}), so the published curve is a strict lower bound
on the finite-shot cost. At fixed size the penalty is budgetary, not structural: at $L=8$ raising the budget by a
factor $12.7$ takes the ratio of the finite-shot error to the infinite-shot-ordering error from
$1.96\pm0.12$ to $1.06\pm0.01$. Finite-shot fractions are therefore reported only with their budget and
seed count attached.

\begin{figure*}[tp]\centering
\definecolor{scMag}{HTML}{D55E00}\definecolor{scFrac}{HTML}{0072B2}\definecolor{scInk}{HTML}{5B5B5B}
\definecolor{inkPrim}{HTML}{1B1B1B}\definecolor{scDim}{HTML}{8A8A8A}\definecolor{scHon}{HTML}{009E73}
\begin{tikzpicture}
\begin{groupplot}[group style={group size=3 by 1, horizontal sep=1.50cm},
  width=4.85cm, height=5.5cm, paperaxis,
  xmin=6.5, xmax=29.5, xtick={8,16,24}, xtick align=outside,
  title style={at={(0,1)},anchor=south west,font=\small\bfseries,yshift=1.5pt},
  label style={font=\small}, tick label style={font=\footnotesize},
  every node near coord/.append style={font=\scriptsize,inkPrim,anchor=south east,xshift=3pt,yshift=1pt}]
\nextgroupplot[title={(a)~magic grows},
  xlabel={qubits $\;2L$}, ylabel={fermionic magic $\mathcal{F}_1$}, ymin=0, ymax=30.5,
  ytick={0,10,20,30}, ymajorgrids, grid style={draw=black!10}]
\addplot[draw=none,fill=scMag,fill opacity=0.10] coordinates {(8,6.6074) (12,9.3851) (16,12.8138) (20,16.0183) (24,19.2721) (28,22.5050)} \closedcycle;
\addplot[scMag,line width=1.3pt,mark=*,mark size=2.5pt,mark options={fill=scMag,draw=white,line width=0.7pt},
  nodes near coords,point meta=explicit symbolic]
  coordinates {(8,6.6074)[] (12,9.3851)[] (16,12.8138)[] (20,16.0183)[] (24,19.2721)[] (28,22.5050)[22.5]};
\node[inkPrim,font=\scriptsize,anchor=north west,xshift=2pt,yshift=-3pt] at (axis cs:8,6.6074) {6.6};
\nextgroupplot[title={(b)~fraction falls},
  xlabel={qubits $\;2L$}, ylabel={$|\mathcal{S}|/\dim$}, ymin=0, ymax=1.22,
  ytick={0,0.5,1}, ymajorgrids, grid style={draw=black!10},
  legend style={font=\scriptsize,at={(0.98,0.98)},anchor=north east,draw=none,fill=white,
    fill opacity=0.65,text opacity=1,row sep=0.2pt,inner sep=1.5pt},legend cell align=left]
\addplot[scFrac,line width=1.1pt,densely dashed,mark=square,mark size=2.3pt,
  mark options={solid,draw=scFrac,line width=0.8pt},name path=Bpub]
  coordinates {(8,0.9167) (12,0.8200) (16,0.5561) (20,0.3625) (24,0.1700) (28,0.0800)};
\addlegendentry{$T\!\rightarrow\!\infty$}
\addplot[scHon,line width=1.3pt,mark=*,mark size=2.3pt,
  mark options={solid,fill=scHon,draw=white,line width=0.6pt},name path=Bhon,
  error bars/.cd, y dir=both, y explicit, error bar style={scHon,line width=0.6pt}, error mark options={scHon,mark size=1.6pt}]
  coordinates {(8,0.9167) +- (0,0.0000) (12,0.8527) +- (0,0.0096) (16,0.6246) +- (0,0.0033) (20,0.4141) +- (0,0.0012) (24,0.2104) +- (0,0.0004)};
\addlegendentry{finite $T$}
\addplot[scHon,fill opacity=0.20,forget plot] fill between[of=Bpub and Bhon,soft clip={domain=8:24}];
\nextgroupplot[title={(c)~absolute cost grows},
  xlabel={qubits $\;2L$}, ylabel={count}, ymode=log, ymin=1, ymax=5e7,
  ytick={1e1,1e3,1e5,1e7}, ymajorgrids, grid style={draw=black!10},
  legend style={font=\scriptsize,at={(0.985,0.02)},anchor=south east,draw=none,fill=white,
    fill opacity=0.0,text opacity=1,row sep=0.0pt,inner sep=1.0pt},legend cell align=left,
  legend image code/.code={\draw[#1,/tikz/mark=none] plot coordinates {(0cm,0cm) (0.38cm,0cm)};}]
\addplot[scHon,line width=1.3pt,mark=*,mark size=2.4pt,mark options={fill=scHon,draw=white,line width=0.6pt}]
  coordinates {(12,24377) (16,261426) (20,3088618) (24,19136905)};
\addlegendentry{$T$}
\addplot[scDim,line width=1.3pt,mark=triangle*,mark size=2.6pt,mark options={fill=scDim,draw=white,line width=0.6pt}]
  coordinates {(8,24) (12,300) (16,3920) (20,52920) (24,731808) (28,10306296)};
\addlegendentry{$\dim$}
\addplot[scFrac,line width=1.3pt,mark=square*,mark size=2.3pt,mark options={fill=scFrac,draw=white,line width=0.6pt}]
  coordinates {(8,22) (12,246) (16,2180) (20,19183) (24,124407) (28,824503)};
\addlegendentry{$|\mathcal S|$}
\end{groupplot}
\end{tikzpicture}
\input{figs/caption_fig12.tex}
\end{figure*}

\section{Physical validation: the two gaps scale in opposite directions}\label{sec:physics}

Across $L=6$--$12$ the exact full-sector references against which the sampled reconstructions are
scored reproduce the opposite finite-size scaling of the two collective scales of the half-filled
Hubbard ring. On the same lattice, the charge gap decreases monotonically to within $6.2\%$ of its
Lieb--Wu thermodynamic value, while the spin gap closes as $\approx2/L$, with an independent
effective-Heisenberg control (Table~\ref{tab:gaps}). The answer is known exactly and independently
of this work, so the scaling, reported across four sizes rather than as a picture at a single size,
is a check of the references and not a result (Sec.~\ref{sec:physics:known}).

\subsection{Charge converges to a finite value, spin closes as \texorpdfstring{$1/L$}{1/L}}\label{sec:gaps}

\begin{table*}[tp]\centering\small
\caption{\textbf{The two collective scales, across four system sizes.} Charge: the Mott gap
$\Delta=\mu^{+}-\mu^{-}$ from exact $(N{\pm}1)$ ground-state energies (half-filled Hubbard ring,
$U/t=8$; particle--hole symmetry $\mu^{+}+\mu^{-}=U$ verified to $10^{-8}$ at every size, ARPACK
checked against dense \texttt{eigh} on every sector of dimension $\le5000$). $\Delta_\infty=4.67952\,t$
is the Lieb--Wu thermodynamic value~\cite{liebwu1968}, evaluated here from its exact integral
representation to an absolute quadrature error of $2\times10^{-9}$ and cross-checked against a second
exact representation, the two agreeing to $10^{-13}$. Spin: the lowest pole of
$S^{zz}(q=\pi,\omega)$, a pole position of the tridiagonal recursion and therefore independent of the
broadening used to display it. Right-hand block: the same quantity in the pure Heisenberg ring at the
effective exchange $J=4t^{2}/U=0.5\,t$, a model with no charge degree of freedom at all. The residual
column is quoted to four decimals so that it equals the $\Delta/t$ column minus $\Delta_\infty/t=4.67952$;
at five it would not be.}\label{tab:gaps}
\begin{tabular}{lcccccc}
\toprule
 & \multicolumn{2}{c}{charge} & \multicolumn{2}{c}{spin (Hubbard)} & \multicolumn{2}{c}{spin (Heisenberg, $J=0.5\,t$)}\\
\cmidrule(lr){2-3}\cmidrule(lr){4-5}\cmidrule(lr){6-7}
$L$ & $\Delta/t$ & $(\Delta-\Delta_\infty)/t$ & $\Delta_{s}/t$ & $L\,\Delta_{s}/t$ & $\Delta_{\rm ST}/t$ & $L\,\Delta_{\rm ST}/t$\\
\midrule
$6$  & $5.35813$ & $0.6786$ & $0.34857$ & $2.091$ & $0.34237$ & $2.054$\\
$8$  & $5.18584$ & $0.5063$ & $0.23958$ & $1.917$ & $0.26134$ & $2.091$\\
$10$ & $5.02224$ & $0.3427$ & $0.19956$ & $1.996$ & $0.21162$ & $2.116$\\
$12$ & $4.96876$ & $0.2892$ & $0.16635$ & $1.996$ & $0.17792$ & $2.135$\\
\bottomrule
\end{tabular}
\end{table*}

\paragraph{Charge.} $\Delta(L)$ decreases monotonically and stays above the Lieb--Wu value
$4.67952\,t$ at every size, the residual falling from $0.679\,t$ at $L=6$ to $0.289\,t$ at $L=12$,
$6.2\%$ of the thermodynamic value; $\Delta(L=12)=4.96876\,t$ is the Hubbard-band separation of
Fig.~\ref{fig:lattice}, and every pole of $S(q,\omega)$ at $L=12$ carrying at least $1\%$ of the
largest weight at its momentum---$97.5\%$ or more of the weight at each $q$---lies above it
(Fig.~\ref{fig:sqw}); weaker poles are not resolved and are not bounded by this statement. No extrapolation exponent is fitted: a straight line in $1/L$ undershoots
Lieb--Wu by $0.12\,t$, which is evidence that the form is not $1/L$, not evidence about the limit.

\paragraph{Spin.} The lowest pole of $S^{zz}(q=\pi,\omega)$ gives $L\,\Delta_{s}=2.091$, $1.917$,
$1.996$, $1.996$: a $1/L$ closing with a coefficient of about $2$ (mean $2.000$, range
$1.917$--$2.091$), and a fit $L\,\Delta_{s}=d_0+d_1L$ has a slope consistent with zero
($d_1=-0.010\pm0.018$). A linear-in-$1/L$ intercept of $\Delta_s$ is $-0.020\pm0.020$, against
$4.561\pm0.044$ for the charge gap; but the pure Heisenberg ring, whose gap is known to vanish, gives
$0.0142\pm0.0014\,t$ at these sizes, so a vanishing intercept is not by itself evidence that a gap
closes, and the flatness of $L\,\Delta_s$ is what carries the statement. That Heisenberg ring at
$J=4t^{2}/U=0.5\,t$---a different Hamiltonian, Hilbert space and code path---reproduces the $1/L$ law,
$L\,\Delta_{\rm ST}=2.054$--$2.135$, differing from the Hubbard value by $6.5\%$ at $L=12$ with a
residual that changes sign: it confirms the scaling and the coefficient to within $10\%$ and does not
pin the coefficient (Sec.~\ref{sm:physics}).

The two channels therefore behave \emph{qualitatively oppositely} under a change of system size---one
converging to a finite thermodynamic value, the other extrapolating to zero as $1/L$---with a factor
$29.9$ between them at $L=12$, and with $1/L$ intercepts that separate. That
opposition, measured at four sizes with an independent control on the smaller scale, is the physical
check this section makes. It is a statement about pole positions, not about lineshapes, so it does
not depend on the broadening at which either response is plotted.

\subsection{Resolution: why the statement rests on pole positions}

A spin--charge statement cannot be read off a lineshape at a single size, for a reason of resolution. At $L=12$ the spin scale is $0.16635\,t$ against the $\eta=0.10\,t$ used for
Fig.~\ref{fig:spin}: the gap is $1.66\,\eta$, i.e.\ $0.83$ of the Lorentzian full width at half
maximum. At that resolution a spin response with a $0.17\,t$ gap and one with none are not
distinguishable by eye in a single-size plot. Terms that presuppose a resolved dense
spectrum---``continuum'', ``gapless'' and ``edge of the continuum''---would need that distinction,
and none is used for any channel of this work at this size.

What establishes the physics instead is the four-size scaling of Table~\ref{tab:gaps}, plotted in
Fig.~\ref{fig:gapscaling}, and a pole census (Table~\ref{tab:poles}): between $L=6$ and $L=12$ the
sector grows by a factor $2440$ while the largest per-momentum count of poles above $1\%$ of the channel
maximum grows from $3$ to $7$ in the addition branch, from $4$ to $9$ in the charge channel and from $2$
to $3$ in the spin channel, a discrete spectrum whose lines are broadened together rather than a dense
one forming. The census, the sum rule as measured, and the size convention of the
exact-diagonalization literature together with the dynamical-DMRG resolution
criterion~\cite{jeckelmann2002} this work does not meet at $L=12$, are in Sec.~\ref{sm:physics}.

\subsection{Relation to the known spin--charge physics}
\label{sec:physics:known}

That the spin and charge responses of the one-dimensional Hubbard model live on separate energy scales
has been known exactly since Lieb and Wu~\cite{liebwu1968}, is textbook~\cite{essler2005}, was seen in
angle-resolved photoemission thirty years ago~\cite{kim1996,damascelli2003}, is already legible in the
spectral weight of twelve-site clusters embedded by cluster perturbation theory~\cite{senechal2000},
and has been computed by dynamical DMRG since 2007~\cite{benthien2007} in all three of the
momentum-resolved channels reported here---in that reference's notation the charge factor
$N(q,\omega)$ at $L=60$ (at $q\approx\pi$ up to $L=120$), the spin factor $S(q,\omega)$ at $L=100$ and the one-particle function
$A(q,\omega)$ at $L=90$, which are respectively our $S(q,\omega)$, our $S^{zz}(q,\omega)$ and our
$A(k,\omega)$, for the extended Hubbard model of which ours is the $V=0$ line; the same literature treats the Hubbard chain at $L=48$ with $\eta=0.1\,t$ by
correction-vector DMRG, and computes the full Heisenberg-chain $S(k,\omega)$ at $L=64$, with the $k=\pi$ cut to $L=128$~\cite{nocera2016}. Here it serves as a check, against a known answer, of the exact full-sector references against which the sampled reconstructions are scored, with the des Cloizeaux--Pearson lower
boundary~\cite{descloizeaux1962} and the M\"uller upper bound~\cite{mueller1981} recovered as
thermodynamic-limit references rather than as findings. At $L=12$ the spin peaks follow that lower
boundary for $q\le\pi/2$ and lie above it toward $q=\pi$ (Fig.~\ref{fig:spin}).

\begin{figure*}[tp]\centering
\input{figs/fig_gapscaling_native.tex}
\input{figs/fig_gapscaling_caption.tex}
\end{figure*}

\begin{figure*}[tp]\centering\includegraphics[width=0.92\linewidth]{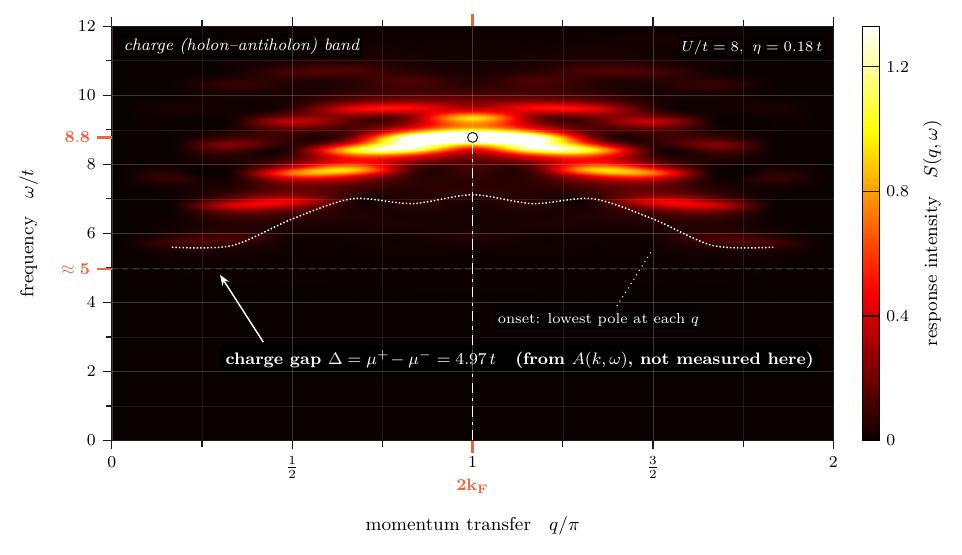}
\caption{\textbf{Dynamical structure factor: the density response.} $S(q,\omega)$ of the
half-filled periodic Hubbard ring ($U/t=8$, $L=12$, Lorentzian broadening $\eta=0.18\,t$)---the density--density
response, a number-conserving (neutral) excitation, computed by a Haydock continued fraction of
depth $n_\ell=200$ in the $N$-particle sector---a finite-depth quadrature of the spectral measure,
not the Lehmann sum (Sec.~\ref{sm:physics}). Color is the \emph{response intensity} $S(q,\omega)$ (units $1/t$): bright
means the system responds strongly at that $(q,\omega)$, dark means it does not. The weight forms the
two-particle charge (holon--antiholon) band, whose dotted curve is the onset, the lowest pole resolved at each $q$, and sets the
low-frequency scale of the band and traces its lower-boundary dispersion; the diffuse upper
boundary is left unmarked. At $L=12$ this band is six to nine discrete poles per momentum merged by
the broadening, not a dense spectrum (Table~\ref{tab:poles}). Two features are decisive. \emph{(i)}
The band is gapped on the scale of the charge gap: every resolved pole lies above
$\Delta=\mu^{+}-\mu^{-}=4.97\,t$ (exact $(N{\pm}1)$ ground-state energies), the lowest onset
reached here being $5.60\,t$ at $q/\pi=1/6$. The limit $q\to0$ is inaccessible to this observable,
$S(q{=}0,\omega{>}0)$ vanishing identically by density conservation, so $\Delta$ is drawn as a
horizontal reference level imported from the single-particle channel and nothing is claimed about
that limit. It is the \emph{same} Mott gap that
separates the Hubbard bands in the single-particle channel (Fig.~\ref{fig:lattice}): two independent
dynamical observables, one correlation gap. \emph{(ii)} The band peaks at $q=2k_F=\pi$ at
$\omega=8.78\,t$ (coral tick, at the measured peak), within $10\%$ of the bare Hubbard scale
$U=8\,t$ of a doubly-occupied site; the tick marks the peak, not $U$. The reflection
$S(q)=S(2\pi{-}q)$ holds to $\sim\!10^{-4}$ (Lanczos precision), and the analytic f-sum rule
$\int\omega\,S(q,\omega)\,d\omega\propto(1-\cos q)$ is recovered by the plotted (finite-window) grid to
$\sim\!0.3\%$. This is the collective observable the same bitstring-sampled primitive targets; as for
$A(k,\omega)$, faithfully sampling the \emph{full} response carries the sector-fraction cost of
Sec.~\ref{sec:prices}. This channel carries no shot budget: it is a full-sector reference
(Table~\ref{tab:provenance}).}\label{fig:sqw}\end{figure*}

\begin{figure*}[tp]\centering\includegraphics[width=0.92\linewidth]{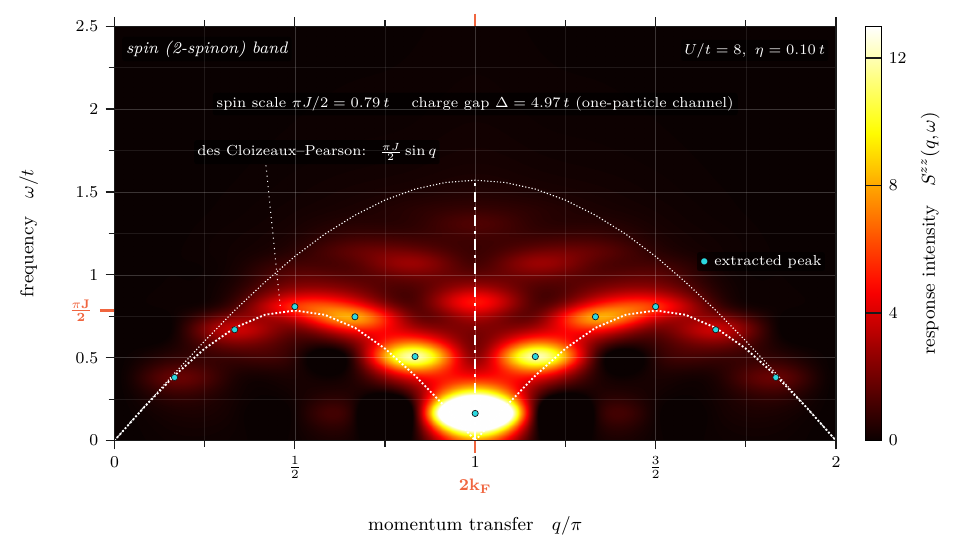}
\caption{\textbf{Spin structure factor: the low-frequency companion of the charge response.} The dynamical
spin structure factor $S^{zz}(q,\omega)$ of the same half-filled periodic Hubbard ring ($U/t=8$, $L=12$; here
$\eta=0.10\,t$, finer because the spin scale is small), computed by a Haydock continued fraction of
depth $n_\ell=200$ from a spin seed $S^z_q\lvert\psi_0\rangle$ in the $N$-particle sector; at this
$\eta$ a reference of that kind is not converged in general (the single-particle one moves by $3.0\%$
between $n_\ell=100$ and $200$, Sec.~\ref{sm:prices}), so the lineshape is
quoted with its depth and not as exact (Sec.~\ref{sm:physics}). \emph{Contrast Fig.~\ref{fig:sqw}}:
the charge scale stays at the Mott gap $\Delta\approx5\,t$ at every size measured, while the spin scale
is $0.16635\,t$ at $L=12$ and closes as ${\approx}2/L$ (Table~\ref{tab:gaps}). At this size the spin
response is one dominant pole plus at most two satellites per momentum (Table~\ref{tab:poles}), not a
dense band, and what the figure shows is a set of broadened lines; whether the band touches
$\omega=0$ in the thermodynamic limit is not established here and is not claimed. This is
spin--charge separation in the collective (two-particle) channel: charge and spin live in energetically
separate worlds, the spin set by the superexchange $J=4t^2/U=0.5\,t$, some $16\times$ smaller than $U$.
The spectral weight piles up on the des Cloizeaux--Pearson~\cite{descloizeaux1962} lower boundary $\tfrac{\pi J}{2}\lvert\sin q\rvert$
(dotted), bounded above by $\pi J\lvert\sin\tfrac q2\rvert$~\cite{mueller1981}; the numerically-extracted
peak dispersion (cyan markers: the maximum of each broadened line on the $0.006\,t$ frequency grid)
tracks this lower boundary to within about $3\%$ for $q\le\pi/2$ (and $q\ge3\pi/2$); toward $q=\pi$ it
lies above it, by $10\%$ at $q=2\pi/3$ and $29\%$ at $q=5\pi/6$, and at
$q=\pi$, where the boundary vanishes, it sits at the finite-size spin gap; the antiferromagnetic soft mode at $q=\pi$
carries the dominant weight. The reflection $S^{zz}(q)=S^{zz}(2\pi{-}q)$ holds to double-precision round-off and
$S^{zz}(q{\to}0)\to0$ (total-$S^z$ conservation). The same bitstring-sampled primitive that targets $A(k,\omega)$ and $S(q,\omega)$ targets this collective response as well. This channel carries no shot budget: it is a full-sector reference
(Table~\ref{tab:provenance}).}\label{fig:spin}\end{figure*}

\section{Hardware runs, classical reach, and scope}
\label{sec:scope}\label{sec:honest}

\subsection{The two hardware runs, and what ``exact by coverage'' means}
\label{sec:scope:hw}

Two device runs were made, and they are the only quantum hardware used in this work. The first, on the IBM~Heron processor \texttt{ibm\_fez}, sampled the $L=6$ periodic Hubbard ring at $U/t=4$
with $12$ qubits and $3.5\times10^{5}$ shots, and post-selection on particle number and $S_z$, with no
configuration recovery, retained all
$300$ determinants of the $(N{+}1)$ sector, so the reconstructed $A_{\rm hw}$ equals the exact $A$ \emph{by coverage}:
Rayleigh--Ritz on the whole sector is its diagonalization, and the vanishing error is an identity rather
than a measurement of device fidelity (Fig.~\ref{fig:heron}(a)). The circuits prepared a
computational-basis $(N{+}1)$ determinant rather than $\hat c^{\dagger}_p\lvert\Psi_0\rangle$ and applied
$R_{xx}R_{yy}$ hopping without Jordan--Wigner strings, so they are not the primitive of
Fig.~\ref{fig:circ}. The post-selection, moreover, read each bitstring in reversed qubit order, which
reflects the ring and exchanges the two spin species: under that reading none of the $3.5\times10^{5}$
shots of a noiseless simulation of the same seven circuits is retained, against all of them under the
correct reading, so every one of the $300$ retained determinants arose from a device error
(\texttt{src/hw\_bitorder\_check.py}). The run is therefore an execution record: the steps of the pipeline execute end to end on a
processor, but the run does not show that the device sampled the intended distribution, and it
establishes nothing about accuracy at complete or incomplete coverage. The second, stretched $\mathrm{N_2}$ at $24$ qubits on
\texttt{ibm\_marrakesh}, reaches $0.59\pm0.14$~mHa of the exact active-space energy after recovery,
below the $29.5$~mHa of the noiseless simulation of the same circuit (Fig.~\ref{fig:heron}(b)). The
record does not store the recovered subspace sizes on which that comparison turns, so it cannot say
whether device noise helps, although a study under a simple noise model reports the mechanism~\cite{vaquero2026}.
That study also finds that recovery from randomly generated configurations alone builds accurate CI
spaces, and no such control was run here. What is
controlled is recovery against no recovery on the same corrupted bitstrings, and there the gain grows
with the noise (Figs.~\ref{fig:noise} and~\ref{fig:noiserec}). Neither run is a spectral record in the sense of
Refs.~\cite{vilchez2025,granet2026akw,quantinuum_dsf2026}, whose spectra are determined by device
data; the full account, the job identifiers and the shot arithmetic are in
Sec.~\ref{sm:scope}.

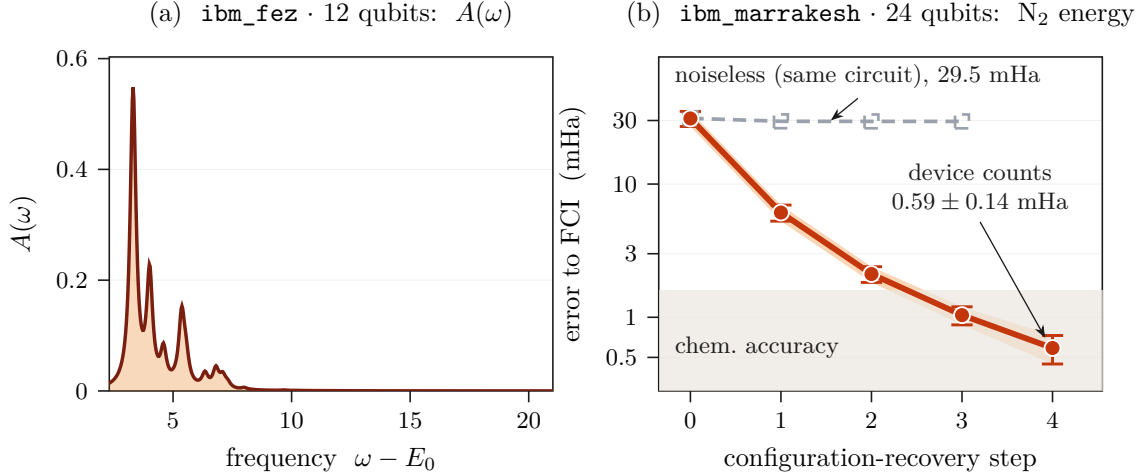
\begin{figure*}[tp]\centering
\leftskip=-34.6pt plus 1fil\relax
\definecolor{inkPrim}{HTML}{1B1B1B}\definecolor{inkMute}{HTML}{6E6E6E}
\definecolor{warmFill}{HTML}{E8770C}\definecolor{warmLine}{HTML}{7A1E10}
\definecolor{hwWarm}{HTML}{C4360C}\definecolor{simSlate}{HTML}{9AA1AC}\definecolor{band}{HTML}{ECE7DF}
\begin{tikzpicture}[font=\rmfamily]
\pgfplotsset{hpanel/.style={width=7.45cm,height=6.0cm,tick align=outside,tick pos=left,axis line style={inkPrim,line width=0.6pt},
  tick style={inkPrim,line width=0.6pt},tick label style={font=\footnotesize},label style={font=\small},
  title style={font=\small,yshift=1pt},ymajorgrids,major grid style={inkMute,opacity=0.10,line width=0.3pt}}}
\begin{axis}[name=axa,hpanel,
  title={(a)\ \ \texttt{ibm\_fez} $\cdot$ 12 qubits: \ $A(\omega)$},
  xlabel={frequency \ $\omega-E_0$}, ylabel={$A(\omega)$},
  xmin=2.31,xmax=21.02,ymin=0,ymax=0.603]
  \addplot[draw=none,fill=warmFill,fill opacity=0.28] table[x=w,y=A]{figs/heron_hot.dat} \closedcycle;
  \addplot[warmLine,line width=1.3pt] table[x=w,y=A]{figs/heron_hot.dat};
\end{axis}
\begin{axis}[name=axb,hpanel,at={(axa.south east)},anchor=south west,xshift=1.4cm,
  title={(b)\ \ \texttt{ibm\_marrakesh} $\cdot$ 24 qubits: \ $\mathrm{N_2}$ energy},
  xlabel={configuration-recovery step}, ylabel={error to FCI \ (mHa)},
  ymode=log,xmin=-0.35,xmax=4.55,ymin=0.28,ymax=90,
  xtick={0,1,2,3,4},ytick={0.5,1,3,10,30},yticklabels={$0.5$,$1$,$3$,$10$,$30$}]
  \fill[band,opacity=0.72] (axis cs:-0.35,0.28) rectangle (axis cs:4.55,1.6);
  \node[anchor=south west,font=\footnotesize,inkPrim] at (axis cs:-0.28,0.42) {chem.\ accuracy};
  \addplot[draw=none,name path=hwhi,forget plot] coordinates {(0,35.110) (1,6.958) (2,2.397) (3,1.199) (4,0.731)};
  \addplot[draw=none,name path=hwlo,forget plot] coordinates {(0,27.049) (1,5.267) (2,1.827) (3,0.877) (4,0.446)};
  \addplot[warmFill,fill opacity=0.30,forget plot] fill between[of=hwhi and hwlo];
  \addplot[simSlate,line width=1.4pt,dash pattern=on 5pt off 3pt,mark=square,mark size=2.6pt,
           mark options={draw=simSlate,fill=white,line width=1pt}] coordinates {(0,31.174) (1,29.600) (2,29.546) (3,29.546)};
  \addplot[hwWarm,line width=2.2pt,mark=*,mark size=3.0pt,mark options={fill=hwWarm,draw=white,line width=0.7pt},
     error bars/.cd,y dir=both,y explicit,error bar style={hwWarm,line width=1.1pt},
     error mark=-,error mark options={rotate=90,hwWarm,line width=1.1pt,mark size=4pt}]
     coordinates {(0,31.079)+-(0,4.030) (1,6.113)+-(0,0.845) (2,2.112)+-(0,0.285) (3,1.038)+-(0,0.161) (4,0.588)+-(0,0.142)};
  \node[anchor=south,font=\footnotesize,inkPrim] (nl) at (axis cs:1.85,44)
    {noiseless (same circuit), $29.5$~mHa};
  \draw[-{Stealth[length=1.7mm]},inkPrim,line width=0.5pt] (nl.south) -- (axis cs:1.55,30.0);
  \node[anchor=south,align=center,font=\footnotesize,inkPrim] (nh) at (axis cs:3.20,5.2)
    {device counts\\$0.59\pm0.14$~mHa};
  \draw[-{Stealth[length=1.7mm]},inkPrim,line width=0.5pt] (nh.south) -- (axis cs:3.92,0.70);
\end{axis}
\end{tikzpicture}
\caption{\textbf{The complete quantum-hardware record of this work: two runs on IBM~Heron processors.}
\textbf{(a)}~Electron-addition branch of the local spectral function, $A^{+}_{0\uparrow}(\omega)$, of the
$L=6$ periodic Hubbard ring at $U/t=4$,
$\eta=0.15\,t$ ($12$ qubits; the classical validation of Fig.~\ref{fig:method} is a different instance,
$U/t=8$), computed on the subspace post-selected from $350{,}000$ computational-basis bitstrings---$50{,}000$ from each of seven
circuits---measured on \texttt{ibm\_fez} (measurement and Pauli
twirling and dynamical decoupling; no readout-error mitigation and no configuration recovery). The retained shots
covered the complete $(N{+}1)$ sector ($|\mathcal S|\!=\!300/300$), so the reconstruction is exact by
coverage; they were post-selected in reversed bit order, so every retained shot is a device error
(Sec.~\ref{sec:scope:hw}). The panel is an execution record, and neither a fidelity measurement nor a
claim of advantage.
\textbf{(b)}~Ground-state energy error of stretched $\mathrm{N_2}$ ($R=2.0$~\AA, CAS(10e,12o)/cc-pVDZ,
$24$ qubits) versus self-consistent configuration-recovery step, on \texttt{ibm\_marrakesh} ($10{,}000$ shots,
CCSD-seeded local unitary cluster-Jastrow (LUCJ) circuit). Recovery on the device counts (solid) reaches
$0.59\pm0.14$~mHa of the exact FCI value---well below chemical accuracy---whereas the noiseless statevector
simulation of the \emph{same} circuit (dashed) is stranded at $29.5$~mHa (its recorded recovery history
ends at step~$3$, unchanged since step~$2$). That contrast is not read here as device noise \emph{helping}: the
mechanism such a reading needs is decided by the two recovered subspace sizes, and the record
stores neither (Sec.~\ref{sec:scope:hw}). The
solid curve is the mean over $8$ recovery seeds on the cached device counts; the shaded envelope and the
capped whiskers at every step both show the $\pm1$ population s.d.\ (normalized by $8$, not $7$) across those $8$ seeds (widest at
step~$0$, $31.1\pm4.0$~mHa, where it starts together with the noiseless curve, and tapering as recovery
converges to $0.59\pm0.14$~mHa; best seed $0.40$~mHa). No random-bitstring control was run, so the
panel cannot separate what the device counts contributed from what recovery alone supplies; the
controlled recovery-versus-no-recovery comparison, under a simulated bit-flip model, is
Fig.~\ref{fig:noiserec}. These two runs are the entirety of the quantum
hardware used here; every other result in this work is an exact classical simulation.}\label{fig:heron}\end{figure*}

\paragraph{The molecular suite.} On the seven Hamiltonians of the nineteen-molecule suite (Sec.~\ref{sec:scope:mol}, Table~\ref{tab:molscope}) whose sectors have dimension $20$ to $300$, the same primitive reproduces the orbital-projected $A(\omega)$ to a relative $L_1$ between $2.2\times10^{-6}$ and $1.0\times10^{-4}$ from $21\%$ to $50\%$ of the sector, a consistency check of the primitive outside the Hubbard family and not a scaling statement. The other twelve reconstructions are exact by construction, because their sampled subspaces, of at most twelve determinants, already contain the probe's Krylov space, so the suite cannot show that the reconstruction is efficient across chemistry.

\subsection{What the cost tracker implies for classical simulation}
\label{sec:scope:boomerang}

The cost tracker of Sec.~\ref{sec:nogo}, the bond dimension $\chi$, is by definition the cost
parameter of a matrix-product representation of the same state: wherever it says the determinant
support is small, a matrix-product state is small too. A sampler can therefore be needed only for a
target whose minimal bond dimension is large in every ordering while the sampled support stays
compact, a property of geometry and orbital connectivity that two-dimensional lattices and molecular
orbital graphs would test and that nothing here measures (Sec.~\ref{sm:scope}). Every target measured
here---Hubbard rings at $L\le14$, nineteen active spaces, and a largest minimal bond dimension of $42$
over the $74$ deposited exact ground states---lies on the other side: there is no computation
performed here that a classical method cannot perform, and we make no advantage claim at any size
reached. The literature shows the same pattern. The $7\,260$ trajectories of a $120$-qubit
Fermi--Hubbard dynamics experiment~\cite{qctrlfh2026} were reproduced by tensor contraction at bond
dimension $32$~\cite{ouyang2026}, the pattern in which tensor-network and Pauli-propagation methods
have overtaken earlier hardware demonstrations~\cite{kim2023,tindall2024,begusic2024}, with Pauli
propagation since extended to real-time correlators~\cite{kemper2026};
the energy of any single-layer unitary cluster-Jastrow circuit, the largest ($77$-qubit) experiment of
Ref.~\cite{robledomoreno2025} included, is now computed classically in polynomial
time~\cite{belagali2026}, and an extended matchgate simulator evaluates, to additive approximation,
the Born-rule probabilities of hardware-sampled cluster-Jastrow bitstrings at a cost exponential only in
the magnitudes of the controlled-phase angles~\cite{extraferm}; on the ground-state energy the strong classical selectors already win on
cost~\cite{reinholdt2025,zhangotten2026}, and even optimally ordered SQD subspaces need a number of
configurations that grows exponentially with lattice size for Heisenberg and Hubbard ground
states~\cite{gaberle2026}.

\paragraph{A dequantization test.} On an emulated, under-sampled head for $\mathrm{N_2}$ ($R=2.5$~\AA; $1\,000$ shots drawn from the exact FCI marginal and corrupted by a synthetic Heron-class noise model, ten paired seeds; no device data), a classical Epstein--Nesbet prior~\cite{cipsi1973} applied to the same histogram cuts the energy error by $36\%$, from $38.22$ to $24.38$~mHa (paired $t_9=12.37$, all ten seeds), while a generative flow network~\cite{bengio2021} fused onto that prior moves it a further $1.80$~mHa, a difference ten paired seeds do not resolve ($t_9=1.19$) and whose sign reverses between two five-seed halves; at $500$ shots the fused generator is worse than the classical control (Sec.~\ref{app:dequant}). On this instance the dominant determinants are supplied by the sampled histogram (emulated, so the test does not measure whether a device state would supply them) and completed by a classical prior rather than manufactured by a learned generator, and the largest improvement available to an under-sampled head costs no quantum resources at all.

\subsection{Scope}
\label{sec:scope:declared}

What the work establishes is four things, each with its measurement attached: one sampling primitive that targets four frequency-resolved channels, of which the single-particle ones are reconstructed here from ranked or sampled subspaces and the two neutral ones are full-sector references (Table~\ref{tab:provenance}); a
two-sided bound on the error of the reconstruction, whose leakage term is a posteriori, with its non-vacuity threshold
measured at five sizes up to a sector of dimension $10\,306\,296$ and its crossing calibrated to a true
relative error of ${\approx}3\times10^{-3}$ within a factor of three (Sec.~\ref{sec:frontier}); an
obstruction showing that no orbital-rotation invariant supplies an a priori predictor of that bound
useful on every state (Sec.~\ref{sec:nogo}); and three prices---support, resolution and shots---each measured here (Sec.~\ref{sec:prices}).

The momentum-resolved spectra are one-dimensional Hubbard rings at $L\le12$ with $\eta=0.10$--$0.18\,t$ (the support, certificate and shot measurements reach $L=14$), which is the
exact-diagonalization convention for this observable---Xu, Lu, Tohyama and Shao compute it at $L=12$
and $\eta=0.2\,t$ at half filling, and at $L=16$ at quarter filling, in the extended
Peierls--Hubbard chain~\cite{xu2024}---but the resolution criterion the dynamical-DMRG literature
applies at those sizes is not met here, so Sec.~\ref{sec:physics} rests its spin--charge statement on
pole positions at four sizes rather than enlarge $\eta$.

\subsection{Data availability statement}
\label{sec:scope:data}

The data that underlie the figures and the quoted numbers of this work are openly available in the
reproduction repository~\cite{p2repo} at
\url{https://github.com/nicolasbonilla/dynamical-spectral-functions-sqd}, archived on Zenodo under
\href{https://doi.org/10.5281/zenodo.22967220}{DOI 10.5281/zenodo.22967220},
the code under the MIT license and the manuscript, the deposited datasets and the reproduction
documentation under CC~BY~4.0. Every figure's plotted values are deposited---as a JSON table in
\texttt{data/} or as a plain-text table in \texttt{paper/figs/}---and
\texttt{docs/FIGURE\_PROVENANCE.md} names the file that is authoritative for each; the one exception
is Fig.~\ref{fig:circ}, a circuit schematic that plots no data. The classical data---the
exact-diagonalization spectral functions of the one-dimensional Hubbard model, the molecular suite and
the resource statistics---are additionally \emph{regenerable}: the computation scripts that produced
them are deposited in \texttt{src/}, and \texttt{docs/REPRODUCE.md} gives the reproduction commands
together with a measured table of what a clean clone does and does not reach. An automated guardian
(\texttt{python src/verify.py}) recomputes the underlying physics by exact diagonalization and checks
the deposited artifacts against it---$1280$ checks over $14\,563$ numeric assertions, exiting non-zero on
any discrepancy. It covers the deposited data files and the numbers printed in the figure sources, in
the tables and in the abstract; it is not a proof that every sentence of this manuscript is verified,
and four items are printed by name on every run, the known open defects; the ratios plotted in Fig.~\ref{fig:thm1iii-violation} are re-derived by it from the deposited certificate rows.

The quantum-hardware results are deposited as the arrays that enter the figures, together with the
backend name, the shot count and the IBM~Quantum job identifier for each of the two runs
(\texttt{data/heron\_spectral.json}, \texttt{data/hw\_lucj\_n2\_result.json}). \emph{The raw measured
bitstrings are not part of the deposit}: the counts can be obtained only from the IBM~Quantum
service, by retrieving the recorded job identifier (\texttt{service.job(...)}) in place of the
submission cell of the notebooks in \texttt{notebooks/}, which requires an account and a job whose
results are still retained, so the classical post-processing of the raw counts cannot be re-run from the deposited
files alone. The hardware runs themselves cannot be regenerated in any case---the jobs are closed, and
a device executed today has a different calibration, so new counts would be new data rather than a
reproduction. The dequantization control of Sec.~\ref{app:dequant} is regenerated seed by seed by \texttt{src/gflow\_dequant.py}, which needs \texttt{pyscf} and \texttt{torch} and writes \texttt{data/gflow.json}. The $L=14$ ground-state checkpoint ($94$~MB), on which the largest vacuity factor and the $L=14$ support count rest, is deposited in \texttt{ckpt/} with its SHA-256 checksums. Known gaps between the deposited code and the
deposited artifacts are enumerated in \texttt{docs/KNOWN\_DISCREPANCIES.md}; those still open are the
rasters embedded in Figs.~\ref{fig:lattice}, \ref{fig:sqw} and~\ref{fig:spin}, which are deposited as
rendered and cannot be regenerated from the deposit; the third panel of Fig.~\ref{fig:scaling} and
the captions of Figs.~\ref{fig:akwsampled} and~\ref{fig:gapscaling}, edited by hand after generation;
two computation scripts without deposited output; and the parts of the pipeline that require
\texttt{pyscf} or an IBM~Quantum account.

\section{Conclusion}
\label{sec:conclusion}

For a spectral function reconstructed by Rayleigh--Ritz on a sampled subspace, the weight the
subspace missed, together with the Hamiltonian coupling across its boundary, bounds the error. The
captured Born weight alone does not: an analytic counterexample and a counting argument rule out any
bound indexed on it (Sec.~\ref{sec:retraction}). The weight a subspace keeps is displaced in $\omega$, not
discarded. On the probe's own Born support, where the captured weight is exactly one, the relative
error at $\eta=0.18\,t$ is therefore still $0.43$ at $L=6$ and $0.35$ at $L=8$.
Theorem~\ref{thm:leakage} turns the answer into an a posteriori certificate for any orthogonal
projector, given the exact ground state. It bounds the $L_1$ error from below by the Born weight the
subspace failed to capture, and from above by the smaller of the trivial bound and a leakage branch,
built from that weight and the coupling that leaks across the subspace boundary at the resolution
asked for. That coupling is computed from the Ritz pairs and the retained part of the probe, without
the exact spectral function. What a Krylov-adequate subspace
does guarantee in addition is the classical moment exactness of Krylov projections, through order
$2K+1$ and no further, a reach stated exactly and measured on coordinate subspaces
(Theorem~\ref{thm:moments}).

Evaluated on the scanned subspaces at five system sizes, up to a sector of dimension
$10\,306\,296$, the certificate is vacuous at every operating fraction of the resource scan. There
its leakage branch exceeds the trivial bound by factors of $1.66$ to $8.61$, a verdict converged in
Krylov depth at $L=12$. The certificate becomes informative only well above those fractions, and
there it is calibrated. Over sixteen $(L,\eta)$ cells at $L=6$--$12$, it first improves on the
trivial bound at a true relative $L_1$ error between $1.2\times10^{-3}$ and $7.3\times10^{-3}$. That
spread is a drift, not scatter: it is ordered in both $L$ and $\eta$ and conservative in $L$. The
certificate locates the crossing without the exact spectral function, given the exact probe. A
two-constant rule in the missed weight $1-w_S$, fitted on some sizes and tested on the others,
localizes the same crossing more tightly on eleven of the fourteen splits. What distinguishes the
certificate from such a rule is that it is a proven inequality whose verdict still holds when
the captured weight $w_S$ is replaced by an assumed lower bound on it. The Born-ranked,
infinite-shot reconstruction of Fig.~\ref{fig:akwsampled}, at $85\%$ of its sector, is non-vacuous in
all sixteen of its momentum channels. Its relative bound there, $0.77$--$1.75$ against the trivial
$2$, is $350$--$900$ times its true error.

The leakage depends on the basis, and that dependence is measured. In the determinant basis in
which a device measures, the coordinate support of the Krylov space $\mathcal K(H,\phi)$ is every
determinant the reflection about the seed site allows, measured at $L=6$ to $14$. At full captured
weight, a subspace holding the whole probe therefore leaks unless it contains all of them. The
vacuity at the operating fractions, all below full captured weight, is measured and not deduced from
that fact. A momentum-resolved probe measured in plane-wave orbitals has its Krylov space inside one
momentum block instead, $1/L$ of the sector. On that block the certificate still needs $455$--$483$
of $490$ determinants at $L=8$, and the block, like the sector, grows exponentially with $L$
(Sec.~\ref{sec:frontier:obstruction}). No orbital-rotation invariant bounds the boundary leakage
\emph{usefully} on every state in advance: an invariant of the occupation spectrum cannot see a
boundary. On pairing-model fibers every such bound is vacuous, while the subspace-aware certificate
still separates a non-vacuous value from the trivial bound (Proposition~\ref{prop:leakinv}). The
one-body non-Gaussianity $\mathcal F_1$ is constant along a geminal orbit that contains a point of
classical cost $O(K)$ for $K$ geminals, at bond dimension $2$ and $2^{K}$ determinants in the
pairing basis (Theorem~\ref{thm:witness}). The orbit also contains points at which sampling is
plausibly hard, a hardness not established here. This is consistent with the separation drawn by
Tarabunga \emph{et al.}~\cite{tarabunga2026} between covariance-matrix quantities, two of which
lower-bound the non-Gaussian gate count, and full-state natural-orbital participation entropies,
which upper-bound the classical simulation cost in an orthonormal Gaussian basis.

The certificate does not certify the reconstructions at the operating fractions of the resource
scan, and the regime in which it would certify them costs more shots than the method saves. The method is priced in
support, resolution and shots (Sec.~\ref{sec:prices}). The determinant support that reconstructs
$A(\omega)$ to relative $L_1<0.05$ at $\eta=0.15\,t$ grows as $e^{1.05\,L}$ over $L=4$--$14$, and the
shot budget grows from $7.8\times10^{2}$ to $1.9\times10^{7}$ over $L=4$--$12$. Certification costs
more: moving from the operating fraction to the non-vacuous band multiplies the shots per
snapshot needed for the marginal configuration by $7.5$ to $10$ at $L=8$ and by $20$ to $35$ at
$L=12$. At the acceptance criterion this work uses elsewhere, relative $L_1<0.05$, the certificate
needs at least $0.98$ of the sector at $L=8$, a fraction at which one would diagonalize the whole
sector instead. On the scanned grid of fractions it does not reach that criterion at $L=6$, $10$ or $12$.

In the bound, the frequencies of both spectra are measured from the exact ground-state energy $E_0$
(hypothesis H1 of Theorem~\ref{thm:leakage}). In a shot-based pipeline $E_0$ is itself estimated on
a sampled subspace, and the bound then holds only up to an additive term the certificate does not
control. For the energy shift of $6.7\times10^{-4}\,t$ obtained from $60\%$ of the $N$-electron
sector at $L=8$, that term moves the bound by $+0.20\%$, while the measured error grows by a factor of
$16$ (Sec.~\ref{app:hyp}). The probe is likewise built from the exact ground state (hypothesis H4).
With an approximate ground state instead, the bound applies to that probe's own spectral function,
whose distance to the exact one nothing computed here controls; no certificate in this work has been
evaluated with an approximate probe. The full-sector references reproduce the opposite finite-size
scaling of the two gaps: the charge gap decreases to within $6.2\%$ of its Lieb--Wu value at
$L=12$, while the spin gap closes as
$\approx2/L$, with an independent effective-Heisenberg control (Table~\ref{tab:gaps}). All spectra
here are at exact-diagonalization sizes, the largest momentum-resolved one at $L=12$. By comparison,
dynamical DMRG computed the same three channels at $60$ to $120$ sites in 2007~\cite{benthien2007},
and correction-vector DMRG reaches $L=48$ in the Hubbard chain~\cite{nocera2016}. In one dimension
tensor networks are near-optimal~\cite{whitedmrg1992,schollwock2011}, and the largest minimal bond dimension over the
$74$ exact ground states deposited here is $42$. No quantum advantage is claimed at any size reached
(Sec.~\ref{sec:honest}).

Whether a certified regime exists at low sector fraction can now be settled by a well-posed
measurement, and the calibrated crossing is the instrument for making it. The answer turns on two
open questions. At full captured weight, the identity behind Theorem~\ref{thm:leakage} yields a certificate of second order in the boundary coupling.
The first question is whether it has a counterpart below full weight, where the operating fractions
lie, that can be evaluated without the exact spectral function. The second is whether some
single-particle basis gives $\mathcal K(H,\phi)$ a coordinate support that is a small fraction of
the sector. The plane-wave orbitals are one such basis for a momentum probe, at about $L(L-1)$
further two-qubit gates whose effect on the sampled distribution is unmeasured; whether the support
fills the block at $L\ge10$ is unmeasured as well. Toward both, this work provides a two-sided bound
whose boundary-leakage term is computable from the Ritz pairs and the retained probe, and a measured
map of where that bound is informative and at what cost. It also provides the coordinate support
that fixes, basis by basis, what zero leakage at full captured weight requires. For a sampled
pipeline validated at exact-diagonalization scale, where the exact ground state is in hand, the
captured weight is therefore not a convergence diagnostic for the spectral function built on the
subspace. The error of that spectral function is bounded, though not estimated, by the weight the
subspace missed together with the boundary coupling. The coupling is computable there from the same
Ritz pairs, the retained probe and one matrix--vector product per Ritz vector, and the missed weight
from the norm of the exact probe and its retained part.

\begin{acknowledgments}
We acknowledge the use of IBM~Quantum services for this work; the two device runs reported in
Sec.~\ref{sec:scope:hw} were executed on \texttt{ibm\_fez} and \texttt{ibm\_marrakesh}. The views
expressed are those of the author and do not reflect the official policy or position of IBM or the
IBM~Quantum team.

AI tools were used in the preparation of this manuscript; the author takes full responsibility for
its content.
\end{acknowledgments}

\begingroup
\appendix
\section{Proof of the leakage certificate and of its closed form}
\label{app:proof}

This appendix proves Theorem~\ref{thm:leakage} and the second equality of Eq.~\eqref{eq:lambda}.

\begin{proof}
Write $u=R(z)\phi$ and $u_P=R_P(z)\phi_P$, extended by zero outside $\operatorname{ran}P$.

\emph{(1) The split.} Since $u_P\in\operatorname{ran}P$ and $(z-H_P)u_P=\phi_P$ there,
$(z-H)u_P=zu_P-PHu_P-QHu_P=\phi_P-QHu_P$; subtracting this from $(z-H)u=\phi$ gives
\begin{equation}
u-u_P\;=\;R(z)\bigl[\,\phi_Q\;+\;QH\,R_P(z)\phi_P\,\bigr].
\label{eq:split}
\end{equation}
This is the Feshbach--L\"owdin partition of the resolvent~\cite{feshbach1958,lowdin1962}, written as an
exact identity rather than as an expansion: the first source term is the part of the probe the subspace
never held, the second is the part of the Hamiltonian that carries amplitude out of it. Only $P^{2}=P$,
$P^{\dagger}=P$, $H=H^{\dagger}$ and $\eta>0$ are used; no property of the determinant basis, of the
ranking, or of the sampler enters anywhere below.

\emph{(2) Pointwise.}
$A-A_P=(\eta/\pi)(\lVert u\rVert+\lVert u_P\rVert)(\lVert u\rVert-\lVert u_P\rVert)$ and
$\bigl|\lVert u\rVert-\lVert u_P\rVert\bigr|\le\lVert u-u_P\rVert$, so
$|A-A_P|\le(\eta/\pi)(\lVert u\rVert+\lVert u_P\rVert)\,\lVert u-u_P\rVert$ at every $\omega$.

\emph{(3) Integration.} By Cauchy--Schwarz in $\omega$,
\begin{widetext}
\begin{equation*}
\lVert A-A_P\rVert_1\;\le\;\frac{\eta}{\pi}
\Bigl[\int\bigl(\lVert u\rVert+\lVert u_P\rVert\bigr)^{2}d\omega\Bigr]^{1/2}
\Bigl[\int\lVert u-u_P\rVert^{2}d\omega\Bigr]^{1/2}.
\end{equation*}
\end{widetext}
The first bracket is at most $(\pi/\eta)^{1/2}(\lVert\phi\rVert+\lVert\phi_P\rVert)$, by Minkowski in
$L^{2}(d\omega)$ and the two sum rules. For the second, apply Minkowski to Eq.~\eqref{eq:split}: its
first term obeys the sum rule with seed $\phi_Q$, contributing exactly
$(\pi/\eta)^{1/2}\lVert\phi_Q\rVert$; its second obeys
$\lVert R(z)v\rVert\le\eta^{-1}\lVert v\rVert$, since $H$ is self-adjoint and $\operatorname{Im}z=\eta$,
and the definition~\eqref{eq:lambda} then contributes at most $(\pi/\eta)^{1/2}\Lambda_P(\eta)/\eta$.
Multiplying the three factors gives the leakage branch of~\eqref{eq:twobudget}.

\emph{(4) The other two branches.} $A$ and $A_P$ are nonnegative with known integrals, so
$\lVert A-A_P\rVert_1\le\int A+\int A_P=\lVert\phi\rVert^{2}(1+w_P)$ under no hypothesis at all; and
$\bigl|\int A-\int A_P\bigr|\le\lVert A-A_P\rVert_1$ gives the lower bound
$\lVert\phi\rVert^{2}(1-w_P)$.

\emph{(5) The closed form.} $R_P(z)\phi_P=\sum_m c_m(z-\tilde E_m)^{-1}\ket{\tilde m}$, so
\begin{equation*}
\bigl\lVert QHR_P(z)\phi_P\bigr\rVert^{2}\;=\;\sum_{m,m'}
\frac{c_m\bar c_{m'}\,\braket{g_{m'},g_m}}{(z-\tilde E_m)(\bar z-\tilde E_{m'})}\,.
\end{equation*}
With $x=\omega+E_0$ the remaining integral is
$\int dx\,[(x-\tilde E_m+i\eta)(x-\tilde E_{m'}-i\eta)]^{-1}$, which has one pole in each half
plane; closing above picks up $2\pi i/[(\tilde E_{m'}-\tilde E_m)+2i\eta]$, and multiplying by
$\eta/\pi$ gives Eq.~\eqref{eq:lambda}.
\end{proof}

\endgroup

\clearpage
\onecolumngrid
\begin{center}
{\large\bfseries Supplemental Material}\\[2pt]
{\normalsize Captured weight and boundary leakage bound the error of sample-based spectral functions}
\end{center}
\setcounter{section}{0}
\setcounter{subsection}{0}
\setcounter{figure}{0}
\setcounter{table}{0}
\setcounter{equation}{0}
\makeatletter\@removefromreset{equation}{section}\makeatother
\renewcommand{\thesection}{S\arabic{section}}
\renewcommand{\thesubsection}{\thesection.\arabic{subsection}}
\renewcommand{\thefigure}{S\arabic{figure}}
\renewcommand{\thetable}{S\arabic{table}}
\renewcommand{\theequation}{S\arabic{equation}}
\makeatletter\renewcommand{\p@subsection}{}\renewcommand{\p@subsubsection}{}\makeatother

\section{The primitive: provenance declarations in full}
\label{sm:method}

Section~\ref{sec:method} states the measurement primitive, fixes the one definition the rest of the paper turns on,
and then says --- channel by channel, figure by figure --- which objects here were reconstructed from a
finite number of measurements and which are full-sector references computed classically. That line is
routinely drawn nowhere in the text, and a single clause of an abstract is enough to draw it in the
wrong place. It is drawn here, explicitly, before any result is quoted.

\paragraph{Two consequences of Rayleigh--Ritz.} Two corollaries of Eq.~\eqref{eq:ritz} are
used later. A determinant acts on $A_{\mathcal S}$ through $H_{\mathcal S}$ even carrying no seed
amplitude: at sector fraction $0.85$, $30.2\%$ of the selected determinants at $L=6$ and $42.0\%$ at
$L=8$ lie outside $\operatorname{supp}(\phi)$, add exactly nothing to the captured weight, and still
move the Ritz spectrum. And on the smallest subspace of \emph{unit} captured weight, $\mathcal
S^{\star}=\operatorname{supp}(\phi)$, of size $(\tfrac12+\tfrac1L)\dim$ at $L=6$ and $0.571\dim$ at $L=8$, the relative $L_1$ error
of Eq.~\eqref{eq:ritz} does not vanish either; it is measured there at both resolutions, on both
boundary conditions and against its ceiling, in Sec.~\ref{sec:retraction}.

Table~\ref{tab:provenance} gives, for every dynamical object in this paper, the subspace it was built on
and the measurements that built it. Three readings the word \emph{sampled} invites are closed off by it.

\paragraph{The neutral channels are references, not reconstructions.} $S(q,\omega)$ and
$S^{zz}(q,\omega)$ are computed here by a continued fraction over the \emph{entire} $N$-particle sector
of $853\,776$ determinants at $L=12$, at recursion depth $n_\ell=200$: no subspace is selected and no
shot is drawn, and the generating scripts say so in their own headers. The primitive \emph{targets}
these two channels --- that is what Fig.~\ref{fig:circ} asserts, and what the seed algebra supports
--- and targeting a channel is not reconstructing it from measurements;
they were not, nor were the $A(k,\omega)$ of Fig.~\ref{fig:lattice} or the $\mathrm{N_2}$ series of
Fig.~\ref{fig:hero}. A reference at finite recursion depth is not the exact spectral function either, so
its depth is quoted with it, and we quote no error that lies within a factor $20$ of the drift of its
own reference.

\paragraph{What the Born ranking is, and what it is not.} The two objects it would be easiest to call
sampled without a budget --- Fig.~\ref{fig:akwsampled}(a,b) and the dashed curve of
Fig.~\ref{fig:scaling}(b) --- select $\mathcal S$ by sorting the accumulated Born weight
$\sum_{k}|\langle d\,|\,e^{-i\hat Ht_k}\phi\rangle|^{2}$ of the evolved state over the $K{+}1$ snapshots
and keeping the top fraction. That weight is the distribution the device draws from, not the
ground-state amplitudes and not the exact Lehmann residues, so the ranking is \emph{not} a clairvoyant
oracle: the repository contains one of those, in \texttt{amortized\_recovery.py}, and labels it as such
there. It is the infinite-shot limit $T\to\infty$ of the declared protocol, approached from above in
error by any finite budget --- at matched subspace size finite sampling is worse by $1.50\pm0.25$ at
$L=6$ and $1.96\pm0.12$ at $L=8$ (Table~\ref{tab:finiteshot}) --- and in that limit it selects nothing
the device could never return: the only determinants of exactly zero accumulated weight are the reflection-forbidden ones ($4$, $18$, $48$, $200$ at $L=6$--$12$), which rank last and never enter $\mathcal S$. What it is not is a
finite-shot experiment, and it must not be called one. The
replacement is measured and quoted with its budget: at $L=8$, $T=2.60\times10^{6}$ shots over $K{+}1=19$
snapshots reach sector fraction $0.835\pm0.0025$ at relative $L_1$ error $(4.81\pm0.27)\times10^{-3}$
over eight independent shot streams (Fig.~\ref{fig:akwsampled}(c)), and the measured ladder from
$2.4\times10^{4}$ to $1.9\times10^{7}$ shots for $L=6$ to $L=12$ is Fig.~\ref{fig:scaling}(c) and
Table~\ref{tab:shots}. The literal configuration of Fig.~\ref{fig:akwsampled}(a,b)---$\eta=0.18\,t$,
$K=16$, all sixteen channels---has also been run at finite shots, at two budgets a factor two apart: $5.2\times10^{6}$ per channel suffice
for panel~(b)'s accuracy, but on a \emph{larger} subspace, and at half that budget, at a fraction just under
the announced one ($0.839$ against $0.850$), the penalty is a measured $2.61$ and not the $2.05$
Table~\ref{tab:finiteshot} would transfer (numbers in the caption of Fig.~\ref{fig:akwsampled}).

\paragraph{What the a posteriori checks establish.} The branch sum rules $\int A_p^{\pm}\,d\omega$ are
identities of the ground state, not measurements of a reconstruction. Integrated on the grid the figures
use, at $L=6$ and $\eta=0.15\,t$, the exact and reconstructed addition branches carry $0.48659$ and
$0.48647$ --- agreement to $2.4\times10^{-4}$ relative, both falling $2.7\%$ short of the analytic
$1-\langle\hat n_p\rangle=0.5$ because Lorentzian tails of width $\eta$ leave the finite frequency
window. That is a genuine consistency check on a zeroth moment, and not an error bound. The certificate
is the two-sided bound of Theorem~\ref{thm:leakage}, and where it is non-vacuous is measured, not
asserted, in Sec.~\ref{sec:frontier} --- an essential property given that the classical hardness of a
specific instance is empirical rather than a theorem (Sec.~\ref{sec:honest}). What is guaranteed in advance is narrower: a subspace containing the order-$K$ Krylov space reproduces that branch's exact spectral moments through order $2K{+}1$ (Theorem~\ref{thm:moments}), and $T\ge\varepsilon^{-1}\log(1/(\varepsilon\delta))$ shots capture every configuration of Born probability at least $\varepsilon$ (Prop.~\ref{prop:shots}); the second does not imply the first (Sec.~\ref{sec:moments}).

\subsection{The two budgets}\label{sec:method:budgets}
The two halves of this paper are the two terms of one inequality. With $Q_{\mathcal S}$ the projector
complementary to $P_{\mathcal S}$ and
$w_{\mathcal S}=\lVert P_{\mathcal S}\phi\rVert^{2}/\lVert\phi\rVert^{2}$, the upper branch of
Eq.~\eqref{eq:twobudget} reads
\begin{widetext}
\begin{equation*}
\lVert A-A_{\mathcal S}\rVert_{1}\ \le\ \bigl(\lVert\phi\rVert+\lVert P_{\mathcal S}\phi\rVert\bigr)
\Bigl(\ \underbrace{\lVert Q_{\mathcal S}\phi\rVert}_{\text{missed weight}}
\ +\ \underbrace{\Lambda_{\mathcal S}(\eta)/\eta}_{\text{boundary leakage}}\ \Bigr).
\end{equation*}
\end{widetext}
The captured-weight term $\lVert Q_{\mathcal S}\phi\rVert=\lVert\phi\rVert\sqrt{1-w_{\mathcal S}}$ vanishes once the sampler has exhausted the Born support of the probe, and Table~\ref{tab:provenance} prices that exhaustion in measurements. The leakage term $\Lambda_{\mathcal S}(\eta)/\eta$ falls as $\mathcal S$ grows along the sampler's ranking (Sec.~\ref{sec:frontier}), but the captured weight does not control it, and it does not vanish with it. On the subspace that \emph{is} that support --- $\mathcal S^{\star}=\operatorname{supp}(\phi)$, of
size $(\tfrac12+\tfrac1L)\dim$ at $L=6$ and $0.571\dim$ at $L=8$, fixed by the probe and not by any ranking, so that
$w_{\mathcal S}=1$ exactly by construction, with no tie-breaking and no zero-amplitude padding --- the
weight term is identically zero and the relative error is still $0.43$ at $L=6$ and $0.35$ at $L=8$ at $\eta=0.18\,t$, rising to $1.15$ and $0.73$ at $\eta=0.05\,t$. $\Lambda_{\mathcal S}(\eta)$ is defined
by the resolvent at resolution $\eta$ and exists only inside the spectral half of this work: the
obstruction theorem of Sec.~\ref{sec:nogo} takes it as its subject and would have no subject without
Sec.~\ref{sec:leakage}. Conversely, Sec.~\ref{sec:leakage} leaves open precisely the question
Sec.~\ref{sec:nogo} answers --- whether the a posteriori certificate can be replaced by an a priori
predictor computed from the state. Neither half can be removed without leaving the other incomplete in
a way we can name.

\section{The leakage certificate and the weight-only bound: hypotheses and details}
\label{sm:leakage}

\paragraph{The hypotheses of Theorem~\ref{thm:leakage}, and its proof.} Four hypotheses carry the
theorem, each a property of the pipeline rather than of the inequality: (H1) $A$ and $A_P$ are referred
to the same $E_0$; (H2) the $L_1$ norm is taken over the whole line and normalized by
$\lVert\phi\rVert^{2}$; (H3) $A$ is the exact spectral function of the full sector; (H4) the probe is
built from the exact ground state. They are stated with the theorem in Sec.~\ref{sec:leakage}, and each
is checked against this work's own pipeline, with numbers and with the measured consequence of its
failure, in Sec.~\ref{app:hyp}. The proof of the theorem and of the closed form, Eq.~\eqref{eq:lambda},
is Appendix~\ref{app:proof} of the main text; its steps are numbered (1) to (5) there.

\paragraph{What object is certified.} If the resolvent inside $\operatorname{ran}P$ is itself evaluated
by a Krylov recursion of depth $n_\ell$, then $\Lambda$ computed on that Krylov space bounds the error
of the depth-$n_\ell$ reconstruction \emph{and of nothing else}. The convenient inequality
$\Lambda_{\mathcal K}\ge\Lambda_S$ that would transfer the statement to the dense $A_S$ is false:
$1169$ of $6000$ random triples violate it, at a minimum ratio of $0.5258$. Every table of
Sec.~\ref{sec:frontier} carries the depth of the object it certifies for that reason. Full
reorthogonalization is mandatory throughout: without it $P_{\mathcal K}$ is not a projector---the
orthogonality defect reaches $0.238$ at $L=8$, where the captured weight then reads $1.4868$, an
impossible value---and $\Lambda$ is inflated thirtyfold.

\paragraph{Relation to the classical error theory, and the scope of the novelty claim.} Nothing in the
proof of Appendix~\ref{app:proof} is new as analysis. The split,
Eq.~\eqref{eq:split}, is L\"owdin partitioning~\cite{lowdin1962}, the same algebra as the Feshbach projection of scattering
theory~\cite{feshbach1958}; and sixty years of Rayleigh--Ritz and Krylov error analysis for resolvents
and matrix functions~\cite{saad2011,golub2010} contains every rung of steps (2) and (3),
along with far sharper residual, Temple and Kato--Temple estimates for individual
eigenpairs~\cite{temple1928,weinstein1934,kato1949}, their extension to the noisy, ill-conditioned
eigenvalue problem of quantum subspace diagonalization~\cite{epperly2022}, and the
loss-of-orthogonality analysis of Paige~\cite{paige1976} that makes step (5) conditional on
reorthogonalization, and the polynomial-restart convergence theory of Beattie, Embree and Sorensen~\cite{beattie2005}. Step (3) in particular is the reduced-basis bound on a Galerkin projection onto an arbitrary subspace---the dual norm of the residual over an inf--sup constant~\cite{rozza2008,hesthaven2016}---with that constant equal to $\eta$.

Nor is ``a posteriori'' ours to claim, and we name the prior art rather than wait to be shown it.
Stopping criteria for rational matrix functions of Hermitian matrices, computed from the residual
alone, are Frommer and Simoncini's~\cite{frommersimoncini2008}; Jia and Lv's title is this section's
claim word for word---\emph{a posteriori error estimates of Krylov subspace approximations to matrix
functions}~\cite{jialv2015}---and Frommer and Schweitzer carry that line to Stieltjes
functions~\cite{frommerschweitzer2016}; Lanczos approximations to matrix functions carry a posteriori bounds from the same Cauchy-integral argument as step (5)~\cite{chen2022lanczos,xuchen2025block}, and stochastic Lanczos quadrature carries them on an approximated spectral measure itself, in Wasserstein and Kolmogorov--Smirnov distance, through Gauss-quadrature brackets that hold for each seeded (weighted) spectral measure and for the global density of states that averages them~\cite{chen2021slq}---bounds on a full-space Lanczos quadrature, not on a projection onto a sampled determinant subspace with boundary leakage as bounded here; and in April 2026, before this work, on \emph{quantum} Krylov,
Oliveira, Ziems and Glaser introduce ``two new metrics, the imaginary and unitary filters, that
successfully assess the reliability of the obtained solutions without any knowledge of the true
eigenspectrum''~\cite{oliveira2026}. That is the framing of this section, established earlier and by others, and the priority is theirs. On the selection side, that a subspace chosen for the energy misses configurations a response needs, and a selection criterion built for the response instead, are Reinholdt, Kjellgren and Kongsted's, for classical selected configuration interaction~\cite{reinholdt2026lr}.

What we did not find in that literature is the conjunction of three things, and it is the whole of the
novelty claimed here. First, a closed form for the frequency-integrated residual $\Lambda_P(\eta)$, assembled from the Ritz pairs, on a subspace of \emph{coordinates} carrying no Krylov structure, no invariance and no spectral-gap assumption---so that the estimate
attaches to the set of determinants a sampler returned, which is not the span of the iterates of any
recursion, and Eq.~\eqref{eq:split} is evaluated on that set directly. Second, that the bound is
two-sided on the same object: the lower branch is fixed by the captured weight and the upper branch adds what the weight does not control, and Sec.~\ref{sec:frontier} uses the distance between them rather than either alone.
Third, a \emph{measured} calibration of where that bound crosses the trivial one, over sixteen
$(L,\eta)$ cells whose true errors span five decades (Sec.~\ref{sec:frontier:calib}), in place of an
asymptotic statement about the crossing. Whether such a bound is \emph{useful} at the
fractions sampling reaches is a separate question; it is not settled by the proof, and
Sec.~\ref{sec:frontier} answers it---in the negative---by measurement.

\paragraph{The pooled sweep of Fig.~\ref{fig:thm1iii-violation}, and its limits.}
The pooled sweep behind Fig.~\ref{fig:thm1iii-violation} contributes $621$ subspaces from five Hamiltonian families and eight construction rules, written by two independent implementations; $15$
reproduce their target to double-precision round-off and define no ratio. On the $606$ that remain the
weight-only bound fails on $469$, and restricted to the regime this work operates in, $w_S\ge0.99$, it fails
on $363$ of $363$, the mildest by a factor $2.00$ (Sec.~\ref{app:counterex}). The compliant cases are
subspaces so poor that the relative error is already of order unity, where the bound asserts nothing.
In the configuration of Fig.~\ref{fig:akwsampled} itself---$L=8$, $\eta=0.18t$, the $3920$-dimensional
$(N{+}1)$ sector, momentum seeds $\hat c^{\dagger}_{k\uparrow}\ket{\Psi_0}$, both branches, eight
momenta, subspaces grown in order of accumulated Born weight---the median violation factor at sector
fractions $0.50$, $0.70$, $0.85$, $0.95$ is $18.5$, $23.6$, $26.5$, $29.9$, the $64$ rows spanning
$15.1$ to $57.3$ overall and $15.1$ to $33.7$ at the fraction $0.85$ of Fig.~\ref{fig:akwsampled}. Those medians
\emph{rise} monotonically, by a factor $1.6$, while the measured error over the same range falls by
three orders of magnitude, from the $10^{-1}$ level at the lowest fraction to $1.0\times10^{-4}$ at
$0.95$: the weight-only bound fails worst exactly where the reconstruction is best, which is the opposite
of what an error bound is for. That is a defect of the functional form and not a pathology at an edge of
parameter space, and it is therefore measurable: Fig.~\ref{fig:thm1iii-violation} plots it against
sector fraction and against $\eta$, beside the replacement bound on the same subspaces.

Three properties of that sweep are declared with it, because each of them limits how it may be read.
Its ranking uses $K=16$ slices, one protocol step short of the $K=18$ of the resource scan of
Sec.~\ref{sec:prices}. Its momentum seeds have support on $3848$ to $3920$ of the $3920$ determinants and captured weight below one at every fraction, so the support floor derived below for the local seed does not apply to it; the members of a translation orbit carry equal Born weight, so the cut can fall inside a tied orbit. How far a tie-breaking can move an error was measured on a local-seed subspace at $L=6$, fraction $0.85$, whose cut fell among zero-amplitude determinants: ten tie-breakings moved the measured error by ${\pm}20\%$, which is why the violation factors are quoted to three figures at most and why the rise
across fractions is read as a direction and not as a rate. (That $\pm20\%$ is the spread of the
\emph{measured error} over tie-breakings of a zero-amplitude cut. It is a different quantity from the
$6.6\times10^{-4}$ by which ties move the \emph{bound} in the sweep of Sec.~\ref{sec:frontier}, where
the cut falls below the support floor and only two determinants are tied.) And the error at fraction $0.85$ has two independent determinations that differ by a
factor $2.3$, so only its order of magnitude is quoted. None of this touches the conclusion: the
two-level counterexample is analytic, and the counting argument of the next paragraph but one uses no
ranking, no protocol and no numerics at all. The sweep is corroboration of a conclusion that does not
depend on it.

\paragraph{Why the weight-only bound fails, in full.}
The reconstruction is Rayleigh--Ritz on $H_S=P_SHP_S$ (Sec.~\ref{sec:method}), not the restriction of
$A$ to a subset of its own eigenpairs: the Ritz energies $\tilde E_m$ are not eigenvalues of $H$, so
weight the subspace \emph{keeps} is nonetheless moved in $\omega$. The obvious sketch elides exactly
this. Having correctly noted $\int A=\lVert\phi\rVert^{2}$ and $\int A_S=\lVert P_S\phi\rVert^{2}$, it
bounds the $L_1$ discrepancy by twice the missing weight---a step legitimate when $[P_S,H]=0$, where
$A_S$ is $A$ with the outside contributions deleted and
$\lVert A-A_S\rVert_1=\lVert\phi\rVert^{2}(1-w_S)$ identically, and controlling nothing otherwise. The
two-level example is the extreme: nothing is discarded there, and every unit of retained weight is
displaced by $g$. What the integrals do give, with no commutator hypothesis, is
$\lvert\int A-\int A_S\rvert\le\lVert A-A_S\rVert_{1}$, that is
$\lVert\phi\rVert^{2}(1-w_S)\le\lVert A-A_S\rVert_{1}$---a \emph{lower} bound, with $0$ violations across the $606$ subspaces of the pooled sweep. It survives intact as the
left-hand half of Theorem~\ref{thm:leakage}: the captured weight does say how wrong the reconstruction
must be, it never said how wrong it can be.

\paragraph{The normalization of the witness errors.} (Relative here and throughout this subsection means divided by $\lVert\phi\rVert^{2}$, the normalization
of Theorem~\ref{thm:leakage}; dividing instead by the in-window integral of the exact $A$, as the
repository's own relative $L_1$ does, moves them by about $1\%$ at $\eta=0.05\,t$ and by $2.6$ to
$3.3\%$ at $\eta=0.18\,t$, measured in Sec.~\ref{app:hyp}, and the two conventions are not interchanged
anywhere in this paper without saying so.)

\section{Moment exactness: the proof route, the conjecture and the shot bound}
\label{sm:moments}

That proof is three lines of Rayleigh--Ritz, and it is not the proof this statement is usually given.
The usual route argues that an order-$(K{+}1)$ quadrature ``matches the first $2K+1$ moments
exactly'', and $2K+1$ is exactly what Table~\ref{tab:moments} measures: the \emph{reach} is right,
the route is not. What does not hold here
is the identification of $A_{\mathcal S}$ with a Gaussian-quadrature rule built by a nonorthogonal
Lanczos/Krylov construction. No such object is built anywhere in this work: every reconstruction
diagonalizes $P_{\mathcal S}HP_{\mathcal S}$ on a coordinate subspace, so $d\mu_{\mathcal S}$ carries
$|\mathcal S|$ nodes, not $K+1$, and the orders that separate the two objects measure the distance between them,
not an error in the published reach. The conclusion is right, the route is not, and the correct route
is also the more general one (Sec.~\ref{app:moments:exact}).

Equation~\eqref{eq:momex} invites a corollary: ``consequently, at fixed broadening
$\eta>0$, the smoothed reconstruction converges geometrically,
$\lVert A-A_{\mathcal K_K}\rVert_{1}\le C(\eta)\,\rho(\eta)^{-K}$.'' It does not follow, and it is
not a corollary here. Moment exactness is a list of $2K+2$ linear identities, and the $L_1$
distance between two Lorentzian-broadened measures is not a finite linear functional of them;
moment exactness supplies neither of the two ingredients a conversion would need
(Sec.~\ref{app:moments:conj}), and neither $C(\eta)$ nor
$\rho(\eta)$ is derived or exhibited anywhere in this work or in the deposit. We therefore state the
claim as what it is.

One constraint on any answer is already available. The two-level family of Sec.~\ref{sec:retraction}
has $\mathcal K_0=\operatorname{span}\{\phi\}$, and by direct integration at $\eta=0.15$ its
$\lVert A-A_{\mathcal K_0}\rVert_{1}$ is $1.983$ at $g/\eta=2\times10^{2}$ and $1.998$ at
$g/\eta=2\times10^{3}$, rising to $2\lVert\phi\rVert^{2}$. Evaluating Conjecture~\ref{conj:geom} at
$K=0$ therefore forces $C(\eta)\ge2\lVert\phi\rVert^{2}$: the same family that forces the trivial
constant on the weight-only bound forces it here on the prefactor. Whatever content the conjecture has
must be carried by $\rho(\eta)^{-K}$ alone, and in particular it cannot be non-vacuous at the small $K$
at which a shot-limited subspace operates.

One figure of this work is read against that conjecture and not against a theorem. A geometric collapse $\text{rel-}L_1\sim\rho^{-K}$ with $\rho\approx1.7$, fitted to panel~(b) of Fig.~\ref{fig:method}, is easily reported as ``precisely the geometric convergence guaranteed by the
moment-exactness theorem''. It is guaranteed by no such thing---that is the corollary ruled out
above---and the fit does not carry the claim on its own either. The panel is one of the two places in this work where
the sampling is genuine and its budget is declared---multinomial draws at $4000$ shots per time step,
eight seeds, $\eta=0.15\,t$, so $T=4000(K{+}1)$ and $T=5.2\times10^{4}$ at the last
point---and the fit uses six of its thirteen measured points, inside a window selected after the fact
($4\le K\le9$, set as a literal in \texttt{spectral\_validation\_max.py}). Across that window the mean
subspace the sampler returns grows from $179$ to $236$ determinants of the $300$-dimensional sector, so
decay in the Krylov order is confounded with growth in $|\mathcal S|$; and the $K$ of the abscissa is
the number of measured time slices, not the $K_{\mathcal S}$ of the weight-only bound, which is $-1$ at
all thirteen points by Eq.~\eqref{eq:exclusion}. No rate is fitted to that window here, and the
caption claims none.
What the panel does show, over all thirteen points, with its eight-seed band and with $T$ attached, is
that the error falls by a factor $32$ as the measured Krylov order is climbed at fixed $\eta$, from
$0.59$ at $K=0$ to $0.0185$ at $K=12$, with a single non-monotone step at $K=3$ inside the discovery
phase; that is the statement its caption makes.

There is a sharper reason not to read Eq.~\eqref{eq:momex} as an error bound, and
Table~\ref{tab:moments} contains it. For a coordinate subspace $P_{\mathcal S}$ is diagonal, so
$\operatorname{span}(\mathcal S)\supseteq\mathcal K_K$ at any $K\ge0$ already forces
$\operatorname{supp}(\phi)\subseteq\mathcal S$, that is $w_{\mathcal S}=1$ exactly---as the fourth
column records for all three rows. Moment exactness here is strictly stronger than unit captured
weight, and unit captured weight is precisely the condition under which Sec.~\ref{sec:leakage} measures
the error to be $O(1)$. The two meet in the first row: the $K=0$ member of the ladder \emph{is} the
witness $\mathcal S^{\ast}=\operatorname{supp}(\phi)$ of Sec.~\ref{sec:leakage}, of size $(\tfrac12+\tfrac1L)D$ at $L=6$ and fixed by the probe rather than by any ranking. It is moment-exact through
order $1$, it holds the entire Born weight of the probe, and at $\eta=0.18\,t$ its relative $L_1$ error
is $0.28$ on the open chain of Table~\ref{tab:moments} at $L=6$, and $0.43$ at $L=6$ and $0.35$ at $L=8$ on the periodic rings on which Sec.~\ref{sec:leakage} reports it: a measured range of
$0.28$--$0.43$ at unit captured weight, whose lower end is the very system on which the method
validation of Sec.~\ref{sec:method} is performed. (All three are $\lVert A-A_{\mathcal S}\rVert_1$
divided by $\lVert\phi\rVert^{2}$, the normalization of Theorem~\ref{thm:leakage}; the size of the
conversion to the repository's own in-window relative $L_1$ is measured in
Sec.~\ref{app:hyp}, and the two conventions are never interchanged silently.) Exactness of $2K+2$ numbers is
what the structure certifies; the curve is certified, at fixed $\eta$ and on the subspace in hand, by
Theorem~\ref{thm:leakage}, whose own non-vacuity is measured rather than assumed
(Sec.~\ref{sec:frontier}).

What Eq.~\eqref{eq:shots} does \emph{not} say belongs in the same breath, and in the main text rather
than in an appendix. The parameter $\varepsilon$ thresholds a Born probability; it is not a spectral
accuracy, and no step of the argument connects it to $\lVert A-A_{\mathcal S}\rVert_{1}$. Arbitrarily
many configurations may each carry $p(x)<\varepsilon$ while jointly carrying most of the mixture; and,
decisively, even at $w_{\mathcal S}=1$ exactly---nothing left to capture at any $\varepsilon$---the
relative error is $O(1)$, as Sec.~\ref{sec:leakage} measures and the first row of
Table~\ref{tab:moments} exhibits. Capturing the high-probability configurations is a statement about
the sampler, not about the spectrum.

Two compositions of the bound follow easily from it, and neither is made here. ``A sampling-capture
bound polynomial in $|\mathcal S|$ and independent of the Hilbert-space dimension'' is a pair of
clauses that are literally true and jointly empty, because this same work
measures $|\mathcal S|\sim e^{1.05L}$ for these lattices, Eq.~\eqref{eq:supportexp}, so a budget
polynomial in $|\mathcal S|$ is exponential in $L$, and independence of the Hilbert-space dimension is
no saving when the quantity the budget is polynomial in grows exponentially itself. And the fall of the
required sector fraction with $L$ is not ``the favourable face of the
polynomial-shot bound'': it is not a face of that bound at all. Equation~\eqref{eq:shots} relates a
shot budget to a Born-probability threshold and contains no statement about $|\mathcal S|/D$, about
$L$, or about any trend in either.

What stands in their place is the narrower statement Eq.~\eqref{eq:shots} does support, and it is the stronger
one for being narrow: that many shots buy, at confidence $1-\delta$, every configuration above Born
weight $\varepsilon$, and nothing further. What the resulting subspace costs, and what it is worth, are
separate measurements, reported as three independent prices---support, resolution and shots---in
Sec.~\ref{sec:prices}.

\section{The non-vacuity threshold: protocol, tables and limits}
\label{sm:frontier}

Throughout this section we call $\mathrm{FR}^{\ast}$ the \emph{non-vacuity threshold of the certificate
of Theorem~\ref{thm:leakage}}: the sector fraction above which the leakage branch of the bound first
improves on the trivial bound $\lVert\phi\rVert^{2}(1+w_S)$. We do not call it a frontier of the
method, and it is not one. It is a property of the inequality chain that proves
Theorem~\ref{thm:leakage}, measured here for the first time, together with one structural fact about the determinant basis. Two measurements separate the two contributions. First, at
fixed system and fixed accuracy target the threshold sits far above the fraction the reconstruction
actually needs: at $L=8$, $\eta=0.18\,t$, a relative $L_1$ error of $0.05$---the acceptance criterion
used throughout this paper---is reached at $\mathrm{FR}\approx0.52$, while the certificate does not
reach $0.05$ until $\mathrm{FR}\approx0.98$, and the intervening shift of $\approx0.45$ in fraction is
the looseness of the proof. Second, at full captured weight $\Lambda_S=0$ requires $\operatorname{span}(S)\supseteq\mathcal K(H,\phi)$, whose coordinate support is every determinant the reflection about the seed site allows, at every size measured---a statement that the leakage is nonzero, not about its size, and not the source of the vacuity measured below, which is at $w_S<1$---and a Krylov-diagonal oracle ranking gives a \emph{worse} certificate than the Born ranking the device realizes at $L=6$, $\eta=0.18\,t$, at every fraction checked from $0.50$ on.
Section~\ref{sec:frontier:obstruction} measures both. What is a property of the method, and is new, is
the calibration of Sec.~\ref{sec:frontier:calib}: the crossing occurs at an essentially fixed level of
true error, so the crossing itself is worth computing even where the bound around it is not.

\paragraph{What is evaluated, in full.} We evaluate Theorem~\ref{thm:leakage} on the subspaces of Fig.~\ref{fig:scaling} at $L=6$, $8$, $10$
and $12$---sector dimensions $300$, $3920$, $52\,920$, $731\,808$---with the ranking protocol of
Sec.~\ref{sec:prices} ($K=18$, $\delta t=0.5$, $16$ substeps, $19$ snapshots), over fractions
$0.20$--$0.95$ and $\eta=0.10$--$0.25\,t$. Those subspaces are the $T\to\infty$ limit of that protocol
and not finite-shot draws, and every statement below inherits that limit: the marginal determinant at
the $L=12$ threshold carries Born probability $1$--$2\times10^{-9}$, so at the budgets of
Table~\ref{tab:shots} that subspace barely exists as a sampled object. Of the $204$ evaluations, $187$
are converged in Krylov depth, $6$ marginal and $11$ not converged, at a total cost of $2874$\,s with
a peak of $1.64$\,GiB. \emph{No row flagged not converged is quoted anywhere in this section}; the
tables mark the ones that are marginal.

To that grid we add twenty evaluations that are not on it. A grid point near a published fraction is
not the published subspace, and the difference is not cosmetic: this sweep reads $0.80$, $0.56$,
$0.40$ and $0.18$ against published fractions of $0.8200$, $0.5561$, $0.3625$ and $0.1700$. The rows
of Table~\ref{tab:published} are placed instead at the integer determinant count
$\lvert\mathcal S\rvert=\mathrm{round}(\mathrm{frac}\times\dim)$ that Fig.~\ref{fig:scaling}
reports---$246$, $2180$, $19\,183$, $124\,407$ and $824\,504$ determinants at $L=6$, $8$, $10$, $12$
and $14$---at the same four resolutions and under the same protocol, at depths $n_\ell=245$
(exhausting the subspace), $900$, $700$, $300$ and $250$, with orthogonality defect at most
$4.9\times10^{-15}$ and $\lvert w_{\mathcal K}-w_{\mathcal S}\rvert$ at most $6.4\times10^{-15}$. The
sixteen up to $L=12$ pass the depth test defined below; at $L=14$, where $\Lambda_{\rm out}$ was not
recorded, only its indicator half can be applied, and it passes at all four resolutions. All twenty
clear the reference-drift rule of Sec.~\ref{app:hyp}, the
weakest---$L=14$ at $\eta=0.10\,t$, where the reference itself drifts by $1.9\times10^{-3}$ on
halving its depth---by a factor $51$. They cost $129$\,s of core time up to $L=12$ and $2238$\,s of
wall time at $L=14$, and they add the size the sweep never reached.

Three conventions are fixed before any number is read. Full reorthogonalization is used throughout,
without which the projector is not a projector (Sec.~\ref{app:hyp} gives the measured defect). What
is certified is the object the code forms---$\Lambda$ on
$\mathcal K=\mathrm{Krylov}(H_S,\phi_S,n_\ell)$ bounds the error
of $A_{\mathcal K}$ at that depth and of nothing else, the convenient inequality
$\Lambda_{\mathcal K}\ge\Lambda_S$ that would extend it to the dense $A_S$ being false ($1169$ of
$6000$ random triples violate it, minimum ratio $0.5258$), and $A_S$ is never formed. And depth is
certified per point by the internal indicator $\Lambda_{\rm in}/\Lambda_{\rm out}$, the truncation
part of $\Lambda$ over its boundary part, rather than assumed from a fixed recursion
depth~\cite{haydock1972} or from a priori Krylov convergence bounds~\cite{epperly2022}: a point passes the depth test, and is counted converged, when
that indicator is below $0.05$ and $\Lambda_{\rm out}$ is constant to $10^{-3}$ relative over the last
three depths of its ladder, and is marginal when the indicator is below $0.30$ otherwise
(\texttt{fsum.py}); Sec.~\ref{app:hyp} bounds what the depth cap does to
$\mathrm{FR}^{\ast}$.

Four hypotheses carry the bound, each a statement about the pipeline rather than about the theorem;
they are stated in full with the theorem in Sec.~\ref{sec:leakage}, and all four are checked against
this pipeline in Sec.~\ref{app:hyp}, along with the ranking, tie and support declarations of the cut.
It is (H3), the exactness of the reference, that enforces the rule that no measured error is quoted
within a factor $20$ of its own reference drift.
One asymmetry runs through every table, and it is narrower than it looks: neither term of
Eq.~\eqref{eq:twobudget} is free of $\Psi_0$. The leakage needs the retained probe $\phi_S$, the
amplitudes of $\Psi_0$ on the determinants the probe operator maps into $S$, which the reconstruction
itself uses and which is exact here by (H4); the weight term needs in addition the full norm
$\lVert\phi\rVert^{2}$, and estimating $w_S$ against an already-truncated $\Psi_0$ is biased
optimistically, while $\Lambda_S$ evaluated with such a $\Psi_0$ belongs to another probe
(Sec.~\ref{sec:leakage}). What the tables report is
accordingly one object and not a family of them: a bound on
$\lVert A-A_{\mathcal K}\rVert_1/\lVert\phi\rVert^{2}$ for the reorthogonalized Krylov reconstruction
at a declared depth, against a reference whose drift is declared and enforced.

\paragraph{The two largest sizes.} The largest classically converged case carries the evidence. At $L=12$ and the published subspace
itself---$\lvert S\rvert=124\,407$ of $731\,808$ determinants---the captured weight is
$w_S=0.998622$, so Eq.~\eqref{eq:tau} asks for $\widehat\Lambda_S/\eta<\tau(w_S)=0.9625$ and the
measured value is $4.236$ at $\eta=0.18\,t$, putting the leakage branch at $8.544$ against a trivial
branch of $1.9986$. That is converged in depth: the depth indicator is
$\Lambda_{\rm in}/\Lambda_{\rm out}=2.8\times10^{-3}$ and between depths $250$ and $300$ $\Lambda_{\rm out}$ moves by $3.7\times10^{-7}$ relative and the bound by $1.4\times10^{-5}$, while at the neighboring grid point $\mathrm{FR}=0.18$ the same
quantity was checked twice more---$\Lambda_{\rm out}$ changing by $7.5\times10^{-7}$ between depths
$250$ and $300$, and an independently written re-implementation moving it by $8.6\times10^{-5}$
between depths $100$ and $300$. In all sixteen deposited depth series $\Lambda$ at every depth from $100$ on lies at or above its deepest value---it is not monotone in depth in general, and rises by up to $2.2\%$ from one depth to the next below depth $100$ at $L=12$---so the bound at any depth from $100$ on, including the depth $150$ at which Fig.~\ref{fig:scaling} is drawn, is no smaller, so the vacuity is not an artifact of Lanczos truncation in either direction.
Nor is it peculiar to these subspaces: over the $606$
live subspaces of the pooled sweep of Sec.~\ref{app:counterex}, Fig.~\ref{fig:thm1iii-violation}, the
leakage branch is the smaller of the two on only $61$. The reconstructions at the fractions this paper publishes---Fig.~\ref{fig:scaling} and the operating points of Secs.~\ref{sec:prices} and~\ref{sec:physics}---are therefore uncertified. Figure~\ref{fig:akwsampled}, drawn at $85\%$ of its sector, is the exception: there the certificate of the depth-$500$ Krylov reconstruction is non-vacuous in all sixteen channels, at a relative bound of $0.77$--$1.75$ against the trivial $2$, and loose, $350$--$900$ times the true error of the channel integrated over the whole frequency line (\texttt{cert\_akw.json}).

One row is quoted with a caveat attached. $L=14$ is the size at which Fig.~\ref{fig:scaling}
reports its headline fraction of $0.08$, and the certificate is evaluated there at $\lvert S\rvert=824\,504$ of $10\,306\,296$ determinants---one determinant above the
$\mathrm{round}(\mathrm{frac}\times\dim)=824\,503$ of the deposited grid---with $w_S=0.998388$ and a
leakage branch of $17.21$, $11.18$, $9.21$ and $6.50$ against a trivial branch of $1.9984$. Ground
state, seed and order there are the repository's stored checkpoint---residual $2.7\times10^{-9}$, the
order verified to be a permutation of the sector---and the $K=18$ ranking protocol, re-run at that
size, reproduces the stored selection as a set to $1.000000$ at
$\lvert S\rvert=824\,504$ and at the fractions $0.04$, $0.16$ and $0.32$, as it does at $L=10$. The caveat is the
depth indicator $\Lambda_{\rm in}/\Lambda_{\rm out}=0.0453$ at $\eta=0.10\,t$, the loosest value
quoted anywhere in this section, comparable to the $4.7\times10^{-2}$ of the last row of
Table~\ref{tab:certified}. It is not repaired, and the row is reported because it makes the
verdict of this section worse rather than better.

\begin{table*}[tp]\centering\small
\caption{\textbf{Where the certificate stops being vacuous, and what the reconstruction actually
needs, measured on the same sweep.} Protocol of Sec.~\ref{sec:prices}, full reorthogonalization, depth
certified per point. Column~3 is the fraction Fig.~\ref{fig:scaling} reports; column~4 the control,
where this sweep---independent code, the normalization of Theorem~\ref{thm:leakage}, the whole line,
a deep reference---first reaches relative $L_1<0.05$, the parenthesized value being the log-linear
interpolation inside the bracket and not a measurement at three digits; column~5 is
$\mathrm{FR}^{\ast}$; columns~6--7 are where the certified value itself first falls below $0.50$ and
$0.20$. Every bracket is delimited by two converged evaluations unless marked: $^{\ast}$ means the
upper end is marginal in depth, \emph{n.c.}\ that the bracket rests only on evaluations failing the
depth test and is therefore not quoted. Columns~5--7 are statements about $A_{\mathcal K}$ at the
declared depth, not about the dense $A_S$. Columns~4 and~5 must be read together: $\mathrm{FR}^{\ast}$
alone falls with size and would suggest the certificate improves as the problem grows, whereas what
the reconstruction needs falls with it and the gap widens---a gap we do not quote as a ratio, both
columns being interpolations inside brackets $0.02$--$0.10$ wide.}\label{tab:thresholds}
\setlength{\tabcolsep}{4pt}
\begin{tabular}{lcccccc}
\toprule
$L$ & $\dim(N{+}1)$ & published & rel-$L_1<0.05$ & $\mathrm{FR}^{\ast}$ & cert $<0.50$ & cert $<0.20$\\
\midrule
\multicolumn{7}{l}{\emph{$\eta=0.18\,t$, the resolution of the momentum-resolved channel}}\\
$6$  & $300$     & $0.82$ & $[0.70,0.80]$ & $[0.90,0.92]$ & $>0.97$ & $>0.97$\\
$8$  & $3\,920$  & $0.56$ & $[0.50,0.56]$ & $[0.75,0.80]$ & $[0.90,0.92]$ & $>0.95$\\
$10$ & $52\,920$ & $0.36$ & $[0.30,0.40]$ & $[0.60,0.70]$ & $[0.80,0.85]$ & $[0.85,0.90]$\\
$12$ & $731\,808$& $0.18$ & $[0.14,0.18]$ & $[0.40,0.45]$ & $[0.60,0.70]^{\ast}$ & \emph{n.c.}\\
\midrule
\multicolumn{7}{l}{\emph{$\eta=0.15\,t$, the resolution at which Fig.~\ref{fig:scaling} is drawn}}\\
$6$  & $300$     & $0.82$ & $[0.80,0.85]$ $(0.816)$ & $[0.92,0.95]$ & $>0.97$ & $>0.97$\\
$8$  & $3\,920$  & $0.56$ & $[0.50,0.56]$ $(0.553)$ & $[0.80,0.82]$ & $[0.92,0.95]$ & $>0.95$\\
$10$ & $52\,920$ & $0.36$ & $[0.30,0.40]$ $(0.356)$ & $[0.60,0.70]$ & $[0.80,0.85]$ & $[0.85,0.90]$\\
$12$ & $731\,808$& $0.17$--$0.18$ & $[0.14,0.18]$ $(0.168)$ & $[0.45,0.50]$ & $[0.60,0.70]^{\ast}$ & \emph{n.c.}\\
\bottomrule
\end{tabular}
\end{table*}

One column of Table~\ref{tab:thresholds} is a positive control and has to be read before the rest.
Applying this paper's own acceptance criterion to the same sweep---independent code, the theorem's
normalization, the whole line rather than a window, a deep reference rather than one recomputed at the
subspaces' own depth---returns $0.816$, $0.553$, $0.356$ and $0.168$ at $\eta=0.15\,t$ for
$L=6$--$12$, against the published $0.82$, $0.56$, $0.36$ and $0.18$: agreement to $7\%$ at worst, and
to under $2\%$ at every size once each value is compared with the fraction the scan itself reports
rather than with the two-figure entry of the replication grid ($0.5561$, $0.3625$, $0.1700$ at
$L=8,10,12$). None of the three ingredients is shared. Figure~\ref{fig:scaling} is thereby reproduced
independently.

The last column of Table~\ref{tab:published} is that control taken one step further, and it is the
sharper of the two: instead of re-deriving the fraction it evaluates the relative $L_1$ error
\emph{at} the published subspace, at the resolution and against the threshold that define it
($\eta=0.15\,t$, rel-$L_1<0.05$). The five values---$0.0477$, $0.0478$, $0.0484$, $0.0485$ and
$0.0463$---all sit just below the threshold and are flat to within $5\%$ of one another across a
sector dimension running from $300$ to $10\,306\,296$, $L=14$ included. The criterion that
\emph{defines} the published fractions is therefore reproduced at the subspaces themselves, at every
size the figure plots above $L=4$, again with nothing shared. What fails in this section is not the
reconstruction and not that curve; it is the word
\emph{certificate} attached to the fractions at which the curve is read.

\begin{table*}[t]\centering\small
\caption{\textbf{The threshold in the only currency a sampler has.} $T_{\rm marg}$ is the shots per snapshot at
which the marginal determinant of the cut appears once in expectation---the reciprocal of its Born
probability in the accumulated distribution of the $19$ snapshots, from the protocol's own cut
weights; covering all of $S$ costs a further factor $\log\lvert S\rvert$. \emph{It is a different
quantity from the measured budgets $T$ of Table~\ref{tab:shots}}, which are totals over the $19$
snapshots at which the finite-shot reconstruction reaches relative $L_1<0.05$; the two carry different
names for that reason, they are not interchangeable, and only ratios within this table are used.
Column~2 uses the grid point nearest the published fraction: $0.80$ against a published $0.82$ at
$L=6$, and $0.40$ against $0.36$ at $L=10$, where the published fraction is the worse of the two.
Factors are from unrounded budgets and need not reproduce from the two-figure entries. The marginal Born probability at
the $L=12$ threshold is $2.2\times10^{-9}$ at $\mathrm{FR}=0.40$, $1.2\times10^{-9}$ at
$0.45$.}\label{tab:shotfrontier}
\setlength{\tabcolsep}{6pt}
\begin{tabular}{lccc}
\toprule
$L$ & published fraction $\to T_{\rm marg}$ & $\mathrm{FR}^{\ast}\to T_{\rm marg}$ & factor\\
\midrule
$6$  & $0.80\to2.2\times10^{4}$ & $[0.90,0.92]\to3.1$--$3.6\times10^{4}$ & $1.4$--$1.6$\\
$8$  & $0.56\to2.1\times10^{5}$ & $[0.75,0.80]\to1.6$--$2.1\times10^{6}$ & $7.5$--$10$\\
$10$ & $0.40\to5.2\times10^{6}$ & $[0.60,0.70]\to3.4\times10^{7}$--$1.1\times10^{8}$ & $6.7$--$21$\\
$12$ & $0.18\to2.3\times10^{7}$ & $[0.40,0.45]\to4.7$--$8.1\times10^{8}$ & $20$--$35$\\
\bottomrule
\end{tabular}
\end{table*}

In shots the distance between those two columns is a budget and not a rhetorical gap
(Table~\ref{tab:shotfrontier}): reaching the non-vacuous band multiplies the shots per snapshot needed
for the marginal determinant by $7.5$--$10$ at $L=8$ and by $20$--$35$ at $L=12$, and only by
$1.4$--$1.6$ at $L=6$, where the sector is small. Driving the certified value down to $0.10$ costs more
again---a further factor $241$ at $L=8$ (fraction $0.95$) and $634$ at $L=10$ (fraction $0.90$), the
corresponding $L=12$ fraction resting only on evaluations that failed the depth test, so that no budget
is quoted for it. This section reports the measured regime of validity: four sizes, four resolutions, $204$ evaluations, every
bracket delimited by converged rows, and a shot budget attached to each.

\paragraph{The strongest certified statements.} Table~\ref{tab:certified} collects what the sweep does certify, taking at each size the strongest
statement whose bracket is converged---at $L=12$, that $60\%$ of a $731\,808$-dimensional sector gives
relative $L_1\le0.717$, tightening to $\le0.218$ at $70\%$ and $\eta=0.25\,t$. Each row is a statement
about $A_{\mathcal K}$ at the declared depth, and none of them uses the exact spectrum.

\begin{table*}[t]\centering\small
\caption{\textbf{The strongest certified statements available from converged data.} Each row is the
tightest certified value at that size whose bracket is delimited by evaluations passing the depth
test; rows that would be stronger but rest on evaluations failing it are not listed. ``True error'' is
the measured relative $L_1$ at the same point, in the normalization of Theorem~\ref{thm:leakage},
given only where it exceeds its own reference drift by more than a factor $20$ (Sec.~\ref{app:hyp},
H3); at $L=10$ it does not, and is withheld rather than rounded. At $L=6$ the recursion runs to
$n_\ell=\lvert S\rvert$, so the evaluation is exact and the depth indicator is not defined. The gap
between the certified value and the true error is three to four decades: the object is a certificate
and not an error estimate, and Sec.~\ref{sec:frontier:calib} states what it is good
for.}\label{tab:certified}
\setlength{\tabcolsep}{5pt}
\begin{tabular}{lccccc}
\toprule
$L$ & $\eta/t$ & fraction & certified rel-$L_1\le$ & true rel-$L_1$ & $\Lambda_{\rm in}/\Lambda_{\rm out}$\\
\midrule
$6$  & $0.18$ & $291/300=0.97$            & $0.502$ & $3.94\times10^{-4}$ & exact ($n_\ell=\lvert S\rvert$)\\
$8$  & $0.18$ & $3\,724/3\,920=0.95$      & $0.219$ & $4.33\times10^{-5}$ & $8\times10^{-4}$\\
$10$ & $0.18$ & $47\,628/52\,920=0.90$    & $0.136$ & withheld (H3)       & $1.6\times10^{-3}$\\
$12$ & $0.18$ & $439\,085/731\,808=0.60$  & $0.717$ & $1.89\times10^{-4}$ & $7.0\times10^{-3}$\\
$12$ & $0.15$ & $0.60$                    & $0.878$ & $2.66\times10^{-4}$ & $1.4\times10^{-2}$\\
$12$ & $0.25$ & $512\,266/731\,808=0.70$  & $0.218$ & $1.99\times10^{-5}$ & $4.7\times10^{-2}$\\
\bottomrule
\end{tabular}
\end{table*}

One restriction belongs with that table. At the acceptance criterion this paper uses everywhere else---relative $L_1<0.05$---the
certificate requires at least $0.98$ of the sector at $L=8$, measured at integer subspace size rather
than interpolated: the certified value is $0.084$ at $3802$ of $3920$ determinants, $0.060$ at $3830$
and $0.039$ at $3855$, so the crossing lies in $[0.977,0.983]$; on our grid it is not reached at all at
$L=6$, $L=10$ or $L=12$. At such fractions diagonalizing the whole sector is already what one would do
instead, and a certificate informative only there does not certify the regime in which sampling is of
interest. We state it without qualification: at the accuracy this paper accepts elsewhere,
Theorem~\ref{thm:leakage} does not certify the operating points of Fig.~\ref{fig:scaling}, and the shot
cost of the regime where it would exceeds the saving that figure reports.

What the table does certify at a fraction a sampler can plausibly reach is weaker, and it is not
empty. The trivial bound caps the relative error at $1+w_S\le2$ for \emph{any} subspace, a
reconstruction of pure noise included; at $L=12$ and $60\%$ of a $731\,808$-dimensional sector the
certified $0.717$ sits a factor $2.8$ below that cap, and at $70\%$ and $\eta=0.25\,t$ a factor $9.2$
below it. Those are exclusions of gross failure made from the subspace alone, at a dimension where the
exact spectrum is expensive and, at the sizes this method is aimed at, unavailable.

\paragraph{The calibration and its limits, in full.} The calibration is the reason the sweep is reported. Across all sixteen
$(L,\eta)$ cells the crossing of the trivial bound occurs at a true relative $L_1$ error of
$1.2\times10^{-3}$ to $7.3\times10^{-3}$, median $2.8\times10^{-3}$: a spread of $5.9$ over a table
whose true errors span five decades, and in every cell at least $47$ times the drift of our own
reference. $\widehat\Lambda_S(\eta)/\eta$ is therefore a calibrated a posteriori risk tracker---its
crossing of the trivial bound locates a true error of ${\approx}3\times10^{-3}$ to within a factor of
three over $L=6$--$12$ and $\eta=0.10$--$0.25\,t$, and it is computed from $H$, $S$, the retained probe $\phi_S$,
the probe norm $\lVert\phi\rVert$ and $\eta$, without $A$, given the exact probe (H4). That factor of three is a drift and not scatter, and the
deposited grid settles it in a few lines: the sixteen crossing errors are ordered without a single
exception, falling with $L$ at every resolution---$5.65$, $3.50$, $2.52$ and $1.60\times10^{-3}$ at
$\eta=0.18\,t$ for $L=6$--$12$, a factor $2.8$--$3.7$---and rising with $\eta$ at every size by
$1.6$--$2.1$, the product of the two being the observed $5.9$. The drift runs the conservative way in
the variable one would extrapolate: the larger the system, the smaller the true error at which the
certificate first says anything, so reading the crossing as ${\approx}3\times10^{-3}$ beyond the sizes
measured overstates the error rather than understating it. In $\eta$, which the user chooses and which
lies inside the measured range, it does not. The objection to anticipate is that this is
Fig.~\ref{fig:scaling} relabeled. It is not, and the difference is the content: that figure is built
knowing the exact spectrum and this quantity is computed without it. That the two track each other to
within a factor $5.9$ across five decades of error is not a redundancy---it is the calibration, and it
is what makes the quantity usable where the second curve cannot be drawn.

Four limits come with it. The crossing is one point per cell, so the tracker locates one error
level and does not estimate the error away from it: at a \emph{fixed} subspace the looseness of the
bound against the measured error is constant to within $\pm17\%$ ($L=6$) and $\pm25\%$ ($L=8$) over a
decade in $\eta$, but across subspaces it varies by a factor $51$, and at the threshold itself it is
$3\times10^{2}$ to $2.4\times10^{3}$---the object is a certificate, not an error estimate, and calling
it the latter would be refutable in one line. Neither of its two terms is free of the ground state (Sec.~\ref{sec:leakage}), and the comparison
needs more than the leakage: it is
$\widehat\Lambda_S/\eta$ against $\tau(w_S)$, and $\tau$ needs $w_S$, which a user does not have.
That dependence can be discharged conservatively, and doing so is what makes the decision itself a
posteriori. Equation~\eqref{eq:tau} is decreasing in $1-w_S$ for $w_S\ge0.066$, with $\tau(1)=1$, and positive for $w_S>0.296$, so any assumed floor on the captured weight above $0.296$ yields a usable threshold: at the $L=12$ published operating point
$w_S=0.998622$ and $\tau=0.9625$, and every subspace with $w_S\ge0.995$ has $\tau\ge0.9280$. A user
willing to assume only that the sampler captured $99.5\%$ of the probe's Born weight can therefore
run the comparison as $\widehat\Lambda_S(\eta)/\eta<0.928$ and use no other ground-state
information, at the cost of at most $7\%$ in the threshold. Where not even that floor can be
assumed, or where $w_S$ is estimated against an already-truncated $\Psi_0$---an estimate biased
optimistically---the test is not conservative and the bound is not a posteriori. All of this presupposes the exact
probe of (H4): with a $\Psi_0$ estimated on a sampled subspace, $\widehat\Lambda_S$ itself is evaluated
for another probe.
And sixteen cells over four sizes and four resolutions is the whole of the evidence: no exponent is
fitted to it and no extrapolation to $L=14$ is made, the local log--log slopes of
$\mathrm{FR}^{\ast}(L)$ already varying by a factor $4.1$ across the four sizes measured.

The fourth limit is the one a reader holding the deposited grid will test first, and it is a real
restriction on what the crossing is worth. A sampler already holds the captured Born weight, so the
question is whether a stopping rule indexed on $1-w_S$ alone locates the same error level, in which
case the crossing would be that weight relabeled. It is not: over the sixteen cells, $1-w_S$ at the
crossing ranges over a factor $45$ against the true error's $5.9$, and still over $9.4$ against $3.8$
once the extreme cell at each end is dropped. A rule with constants fitted to the grid, however, comes
close enough. Take $1-w_S<c\,\eta^{p}$, fit $c$ and $p$ on some of the four sizes
with its operating point matched to the certificate's, evaluate it on the sizes held out, and repeat
over all fourteen train/test splits: out of sample it localizes the crossing error about as tightly as
the parameter-free certificate does, and more tightly on eleven of the fourteen (\texttt{null\_ws7.py}).
Out of sample the certificate has no edge in sharpness over the captured weight as a tracker of the
error; its edge is provability (Sec.~\ref{sec:frontier:calib}). The difference lies elsewhere. The fitted rule fires at $1-w_S\sim10^{-5}$--$10^{-4}$, one to two orders of magnitude finer than the floor
$1-w_S\le5\times10^{-3}$ at which the threshold above still works, and it asks that precision of the one
quantity this section has already said a user does not have and, estimated against an already-truncated
$\Psi_0$, has optimistically; strip its second parameter and an $\eta$-blind $1-w_S<c$ is the looser of the two on thirteen of the fourteen splits. Theorem~\ref{thm:leakage}, for its part, is an inequality with a proof:
it holds at a size no grid has seen, and its decision survives replacing $w_S$ by a floor. The fitted
rule is a two-constant regularity with no guarantee outside its fitted range---the only place either
would be used.

Within those limits the procedure is operational and identical at every size we measured: assemble
$\widehat\Lambda_S(\eta)/\eta$ from the Ritz vectors already in hand and compare it with $\tau(w_S)$.
Above $\tau$ the certificate says nothing, which is the regime Fig.~\ref{fig:scaling} operates in; at
$\tau$ the true relative error is ${\approx}3\times10^{-3}$ within a factor of three; below $\tau$ it
is smaller than that, by an amount the bound states loosely and the calibration does not sharpen. That
is the whole content of the certificate, and it is a number a user can compute before knowing whether
the reconstruction is right.

\paragraph{Where the threshold comes from, in full.} Where does the ${\approx}0.45$ of sector fraction between what the reconstruction needs and what the
certificate certifies go? We measured the displacement rung by rung along the chain of inequalities
that proves Theorem~\ref{thm:leakage}, at $L=8$, and the answer removes the one repair that looked
available: the pointwise step from a scalar distance to a distance in the Hilbert space accounts for
$44\%$ of it, the operator-norm step $\lVert R(z)\rVert\le1/\eta$ applied to a source that depends on
$\omega$ for $38\%$, the Feshbach/Minkowski split for $11\%$, and the Cauchy--Schwarz step in $\omega$
for $4.6\%$ (medians over $15$ of $18$ cells; Sec.~\ref{app:hyp}). The working hypothesis that
Cauchy--Schwarz is what places the threshold is therefore refuted by its own measurement: its total
margin for tightening is a factor $1.05$ typically and $3.29$ at the worst point measured, which moves
the threshold by a fraction of one grid interval.

The second measurement concerns the subspace itself, and it is the one that turns part of the result into a statement about the problem. $\Lambda_S=0$ holds if and only if $\operatorname{span}(S)$ contains $\mathcal K(H,\phi_S)$, the Krylov space of the retained part of the probe (Sec.~\ref{sec:leakage}); for a subspace that holds the whole probe, $w_S=1$, that space is $\mathcal K(H,\phi)$ itself, so the coordinate support of $\mathcal K(H,\phi)$ is a floor under every coordinate subspace that holds the whole probe without leaking. Measured, it is $296$ of $300$ determinants at $L=6$, where $\dim\mathcal K=65$, and $3902$ of $3920$ at $L=8$, where $\dim\mathcal K=596$--$663$. (It is the \emph{coordinate support} of $\mathcal K$ that is measured, not $\mathcal K$ itself, whose dimension is three orders of magnitude smaller; containment, not spanning, is what $\Lambda_S=0$ requires.) The determinants missing are not an accident of threshold. The reflection $j\to-j$ about the seed site commutes with $H$, and $\phi$ is an eigenvector of it, so every vector of $\mathcal K$ vanishes exactly on the reflection-fixed determinants of the opposite parity: $4$, $18$, $48$, $200$ and $600$ at $L=6$, $8$, $10$, $12$ and $14$. Every other determinant is in the support. At $L=10$ the same
construction---a Lanczos recursion on the whole sector to depth $700$ with full reorthogonalization,
orthogonality defect $1.6\times10^{-15}$---leaves no row exactly zero in floating point, although $48$ of them vanish in exact arithmetic: round-off is not a support.

Beyond $L=10$ that basis cannot be stored, and we use a lower bound instead: every $H^{k}\phi$ lies in
$\mathcal K$, so the union over $k$ of the coordinate support of $H^{k}\phi/\lVert H^{k}\phi\rVert$ is
contained in the support of $\mathcal K$. A count of this kind is meaningless without its threshold,
so both are given. Projected onto the reflection sector of $\phi$ at every power (\texttt{z\_suppK\_sym.py}), so that round-off cannot populate the forbidden sector, the union over $k\le36$ is every symmetry-allowed determinant at every size: $296$, $3902$, $52\,872$, $731\,608$ and $10\,305\,696$ at $L=6$ to $14$, at an amplitude threshold of $10^{-12}$, and the same at $10^{-10}$ except at $L=14$, where $232$ genuine amplitudes lie between the two thresholds. Unprojected, the power iteration amplifies round-off in the forbidden sector until it crosses either threshold and reports the entire sector from $L=8$ on (\texttt{z\_suppK\_power.py}); that count is not a support.

At every size measured, therefore, a coordinate subspace that holds the whole probe leaks unless it contains every symmetry-allowed determinant---the sector minus its forbidden determinants contains $\mathcal K$ and does not leak---so at full captured weight no choice of determinants short of all of them makes $\Lambda_S$ vanish, and that is a property of the basis in which the device measures, not of the bound and not of the ranking. It is a statement that $\Lambda_S\neq0$, not about its size, and the vacuity at the published fractions is not deduced from it: that is the measurement of Table~\ref{tab:published}. Below full captured weight the space that decides is $\mathcal K(H,\phi_S)$, which changes with $S$, and nothing structural is claimed there. Nor is the ranking the weak link: replacing the Born ordering by a Krylov-diagonal oracle---determinants ordered by the diagonal of the projector onto $\mathcal K(H,\phi)$, an ordering no sampler could produce---gives a \emph{worse} certificate at $L=6$, $\eta=0.18\,t$, at every fraction checked from $0.50$ on, $12.97$ against $4.27$ at $0.85$ and $5.82$ against $0.911$ at $0.95$; only at $0.30$, where both are far above the trivial value, is the oracle marginally better, $21.76$ against $24.52$.

At fixed size the displacement is the looseness of the proof, and part of the looseness has a structural reason: the leakage branch is of first order in the boundary coupling, the error it bounds of second. For a subspace that holds the whole probe, eliminating $Q$ exactly replaces $H_P$ by $H_P+PHQ(z-QHQ)^{-1}QHP$, and the two sum rules then give $\lVert A-A_P\rVert_1\le\lVert\phi\rVert^{2}\lVert QHP\rVert^{2}/\eta^{2}$. Sharper, with $\phi_Q=0$ Eq.~\eqref{eq:split} and Cauchy--Schwarz in $\omega$ give $\lVert A-A_P\rVert_1\le\Lambda_P(\eta)\tilde\Lambda_P(\eta)/\eta^{2}$, where $\tilde\Lambda_P$ is Eq.~\eqref{eq:lambda} with the energy difference in its kernel reversed in sign, equal to $\Lambda_P$ for a real Hamiltonian and probe. The ratio of the leakage branch, $2\lVert\phi\rVert\Lambda_P/\eta$, to the error is therefore at least $2\eta\lVert\phi\rVert/\tilde\Lambda_P$, and grows without limit as the subspace approaches invariance; at full captured weight the product $\Lambda_P\tilde\Lambda_P/\eta^{2}$ is itself an a posteriori certificate of second order. Every published fraction lies below full captured weight, and whether a certificate of second order exists there is open; we have not constructed one. Section~\ref{sec:nogo} closes the picture
from the other side: nothing cheaper than the crossing of Sec.~\ref{sec:frontier:calib}---no invariant
of the state, computable before the subspace exists---can be put in its place.

One route out of the determinant basis has been measured, and it closes this section as an outlook and
not as a result, because the same measurement shows where it stops. The qualifier \emph{in the
determinant basis} above is load-bearing, and it is measurable. An orbital rotation is a Givens network
and the circuit of Sec.~\ref{sec:method} already contains them, so the measurement basis is a choice;
for the momentum-resolved probe of Fig.~\ref{fig:akwsampled} it is a forced one, since $H$ commutes with
lattice translation and $\phi_k$ is an eigenvector of it, so that $\mathcal K(H,\phi_k)$ lies inside a
single momentum block---rigorously, at every $L$---and that block is exactly $\dim(N{+}1)/L$
plane-wave determinants. At $L=6$ and $L=8$ the Krylov support fills nearly all of the block, and the certificate becomes non-vacuous only close to the full block: at $48$--$50$ of its $50$ determinants at $L=6$ and at $455$--$483$ of its $490$ at $L=8$, over the momenta of the addition branch (\texttt{z\_kblock.py}), a factor $5$--$7$ fewer than the site determinant basis requires; on the block itself $\Lambda_S$ vanishes identically and the bound is zero rather than vacuous. The gain, though, is $1/L$ against an exponential: as a fraction of
the fraction Fig.~\ref{fig:scaling} publishes the block is $0.20$ at $L=6$ and then $0.22$, $0.28$,
$0.49$ and $0.89$ at $L=14$, the two crossing just past $L=14$. It is also a laptop Lanczos---$60\,984$
determinants at $L=12$, $736\,164$ at $L=14$---so the regime it certifies is one that momentum-resolved
exact diagonalization has enumerated since the 1970s. Three things are unmeasured here and would have to
be measured before more is said: whether the support fills the block at $L\ge10$, where the threshold
then sits, and what the ${\approx}L(L-1)$ extra two-qubit gates do to the sampled distribution.

\section{One-body magic, the geminal witness and the determinant support: full analysis}
\label{sm:nogo}

\paragraph{The candidate measures.} The candidates have a sharp algebraic origin. Matchgate (free-fermion) circuits are classically
simulable~\cite{valiant2002,terhaldivincenzo2002,bravyi2005flo,jozsamiyake2008}, and what leaves that
manifold is fermionic non-Gaussianity: any pure non-Gaussian state injected into a matchgate circuit
promotes it to universality~\cite{hebenstreit2019}, the fermionic counterpart of qubit
magic~\cite{bravyikitaev2005,leone2022sre}. A ladder of measures quantifies it---Gaussian rank and
extent~\cite{cudbystrelchuk2023}, non-Gaussian gate counts~\cite{reardonsmith2024,diaskoenig2024},
including the count that makes a state efficiently learnable~\cite{melehera2025}, covariance-matrix
measures and monotones~\cite{sierant2026,tarabunga2026}, and convolution- and entropy-based measures built on the
violation of Wick's theorem~\cite{coffman2025}; distinct from these is the qubit non-stabilizerness
that Gaussian states themselves carry~\cite{collura2026}.

\paragraph{The collapse, in detail.} Equation~\eqref{eq:collapse} is exact on any number-conserving state and verified state by state to $10^{-13}$ across the dissociation of
$\mathrm{N_2}$ and the suite of Table~\ref{tab:molecules}, that is on the $38$ molecular ground states
of this work. For spin-restricted natural orbitals
($g_{a\uparrow}=g_{a\downarrow}=n_a/2$) it equals $2N_{\mathrm u}$, the quadratic
effectively-unpaired-electron index $N_{\mathrm u}=\sum_a n_a(2-n_a)$ of Takatsuka--Fueno--Yamaguchi and
of Staroverov--Davidson~\cite{takatsuka1978,staroverov2000}, and $\mathrm{tr}[\gamma(1-\gamma)]$ is one
quarter of the Gigena--Rossignoli one-body quadratic (linear) entropy, $S_2=4\,\mathrm{tr}[\gamma(1-\gamma)]$
on a number-conserving state~\cite{gigena2015}, a trace-form one-body entanglement entropy in the
sense of Gigena, Di~Tullio and Rossignoli~\cite{gigena2020}. Here $2N_{\mathrm u}$ is
the quadratic index, distinct from the linear $\sum_a\min(n_a,2-n_a)$ of
Head-Gordon~\cite{headgordon2003}. A cleanly singly-occupied spin orbital contributes $0$ to
$\mathcal F_1$ but $2$ to $2N_{\mathrm u}$, so on the five spin-polarized members of
Table~\ref{tab:molecules} ($\mathrm{O_2},\mathrm{NO},\mathrm{CN},\mathrm{OH},\mathrm{BeH}$) the two differ
and we report $\mathcal F_1=4\,\mathrm{tr}[\gamma(1-\gamma)]$ directly. Outside number conservation the
collapse fails: a paired Gaussian state hides fractional occupations in anomalous correlators and keeps
$\mathcal F_1=0$.

The algebra is not new. Equation~\eqref{eq:collapse} is an immediate number-conserving
specialization of the antiflatness of Ref.~\cite{sierant2026}; $N_{\mathrm u}$ has been the standard
diagnostic of unpaired density since 1978~\cite{takatsuka1978}; and the relation was noted
independently and contemporaneously by \'Swi\k{e}tek \emph{et al.}~\cite{swietek2026}, who state it
for fixed-particle-number states and write its ensemble-averaged form for random superpositions of
Gaussian states, their model of interacting-integrable eigenstates. What we add is the exact, state-by-state \emph{dictionary}: three
literatures---fermionic magic, effectively-unpaired electrons, one-body entanglement---name one quantity.
The verification is the $38$ molecular states just quoted; the further equality to $2N_{\mathrm u}$ is
checked on the closed-shell members, the five spin-polarized ones being excluded for the reason given
above. The $36$ lattice ground states of this work are number conserving, so
Eq.~\eqref{eq:collapse} applies to them by the same algebra, but they were not part of that
state-by-state check and we do not count them in it. The consequence we carry
forward is structural and equally exact: being a function of the occupation spectrum alone,
$\mathcal F_1$ inherits the orbital-rotation invariance of the eigenvalues of $\gamma$. It is a genuine,
cheap, exactly computed invariant of the state, and it is that invariance the two following subsections
use.

\paragraph{The witness and the fibers.} \noindent The proof is four short block-diagonal computations (Sec.~\ref{app:witness}); the step worth
flagging here is that at $\theta=\pi/4$ a single maximally-paired geminal is the two-determinant cat,
whose Majorana covariance \emph{vanishes identically}, so \emph{every} order of antiflatness saturates its
maximum at once. The values of $\mathcal F_1$, $\mathcal F_2$ and $|\mathcal S|$ are reproduced by
exact diagonalization for $K=1,\dots,4$, and so is $\chi=2$, the largest Schmidt rank over all cuts of
the pairing ordering (Sec.~\ref{app:witness}).

\begin{figure*}[tp]\centering\begin{tikzpicture}
\definecolor{cF1}{HTML}{D55E00}       %
\definecolor{cF2}{RGB}{123,63,160}    %
\definecolor{cS}{HTML}{0072B2}        %
\definecolor{cX}{RGB}{38,120,74}      %
\begin{groupplot}[
  group style={group size=2 by 1, horizontal sep=1.45cm},
  width=0.485\linewidth, height=6.1cm, paperaxis,
  title style={at={(0,1)},anchor=south west,font=\bfseries\small,yshift=1pt},
]
\nextgroupplot[
  ymode=log, xmin=0.6, xmax=4.5, ymin=0.8, ymax=48,
  xtick={1,2,3,4}, ytick={1,2,4,8,16,32}, yticklabels={1,2,4,8,16,32},
  xlabel={number of geminals $K$}, ylabel={resource value}, title={(a)},
  legend style={at={(0.5,0.985)},anchor=north,legend columns=2,font=\scriptsize,
                /tikz/every even column/.append style={column sep=6pt}},
  legend cell align=left,
]
\path[fill=cX,fill opacity=0.08] (axis cs:0.6,0.8) rectangle (axis cs:4.5,2);
\node[font=\scriptsize\itshape,align=center] at (axis cs:3.15,1.13)
  {classically tractable\\(bond dim.\ $\chi=2$)};
\addplot[cS,mark=*,mark size=1.9pt,line width=1pt] coordinates {(1,2)(2,4)(3,8)(4,16)};
\addlegendentry{support $|\mathcal{S}|=2^{K}$}
\addplot[cF1,mark=square*,mark size=1.7pt,line width=1pt] coordinates {(1,4)(2,8)(3,12)(4,16)};
\addlegendentry{$\mathcal{F}_1=2N_{\mathrm u}=4K$}
\addplot[cF2,mark=triangle*,mark size=2.3pt,line width=1pt,dashed] coordinates {(1,4)(2,8)(3,12)(4,16)};
\addlegendentry{$\mathcal{F}_2=4K$}
\node[font=\scriptsize,black,anchor=west] at (axis cs:2.45,3.05) {$\mathcal{F}_1\!=\!\mathcal{F}_2$ (here $\theta{=}\tfrac{\pi}{4}$)};
\addplot[cX,mark=diamond*,mark size=2.2pt,line width=1pt] coordinates {(1,2)(2,2)(3,2)(4,2)};
\addlegendentry{bond dim.\ $\chi=2$}
\nextgroupplot[
  xmin=0, xmax=0.5, ymin=0, ymax=14.5,
  xtick={0,0.1,0.2,0.3,0.4,0.5}, ytick={0,2,4,6,8,10,12,14},
  xlabel={geminal pairing angle $\theta/\pi$}, ylabel={magic (order 1 and 2)}, title={(b)},
  legend style={at={(0.97,0.97)},anchor=north east,font=\scriptsize},
  legend cell align=left,
]
\draw[black!30,dashed,line width=0.5pt] (axis cs:0.25,0) -- (axis cs:0.25,12.6);
\addplot[cX,line width=1pt,forget plot] coordinates {(0,2)(0.5,2)};
\node[font=\scriptsize,anchor=south,align=center] at (axis cs:0.25,0.4)
  {$\chi=2$ for $0<\theta<\pi/2$ (tractable $\forall\,\theta$)};
\addplot[cF1,mark=square*,mark size=1.7pt,line width=1pt]
  coordinates {(0.02,0.1885)(0.10,4.1459)(0.25,12)(0.40,4.1459)(0.50,0)};
\addlegendentry{$\mathcal{F}_1=2N_{\mathrm u}$}
\addplot[cF2,mark=triangle*,mark size=2.3pt,line width=1pt,dashed]
  coordinates {(0.02,0.3740)(0.10,6.8594)(0.25,12)(0.40,6.8594)(0.50,0)};
\addlegendentry{$\mathcal{F}_2$}
\node[font=\scriptsize,anchor=north west] at (axis cs:0.012,14.2) {maximal pairing};
\end{groupplot}
\end{tikzpicture}
\caption{\textbf{The geminal orbit witness (Theorem~\ref{thm:witness}).} Antisymmetrized products of
strongly orthogonal geminals in their localized (pairing) basis, where $\chi=2$ and the state is a
classically tractable tensor network. \textbf{(a)}~As $K$ grows, $\mathcal F_1=2N_{\mathrm u}=4K$,
$\mathcal F_2=4K$ and $|\mathcal S|=2^{K}$ all diverge while the bond dimension is $2$; the $\chi$
markers are the largest Schmidt rank over all cuts of the pairing ordering, computed by the deposited
diagonalization (Theorem~\ref{thm:witness}); at the central cut alone it is $2$, $1$, $2$, $1$ for
$K=1,\dots,4$, because that cut falls between two geminals when $K$ is even.
\textbf{(b)}~At $K=3$, sweeping the pairing angle drives both $\mathcal F_1$ and $\mathcal F_2$ to their
maximum at $\theta=\pi/4$---where the covariance vanishes and all orders saturate---with
$\mathcal F_2>\mathcal F_1$ off the peak, and $\chi=2$ at every interior angle, dropping to $1$ at
$\theta=0$ and $\theta=\pi/2$, where the state is a single determinant (the line drawn at $\chi=2$ is
the analytic interior value). Because $\mathcal F_k$ is
orbital-rotation invariant the same values persist on the whole orbit, which contains the $O(K)$-cost
pairing point and rotated points at which Born sampling is plausibly hard (not established here).}\label{fig:witness}\end{figure*}
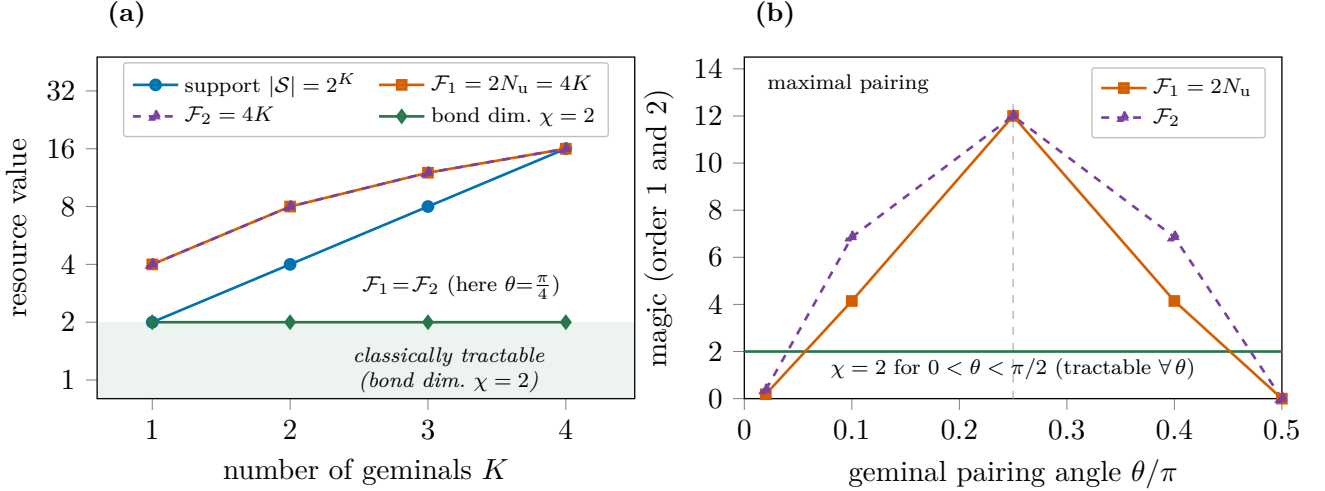

\paragraph{The orbit is the argument.} The witness is not the trivial ``a product state is easy''
(Fig.~\ref{fig:witness}). Because $\mathcal F_1$ and $\mathcal F_2$ are invariant under number-conserving
orbital rotations~\cite{sierant2026}, the whole orbit $\{U\ket{\Psi_K}\}$ carries the \emph{same}
$(\mathcal F_1,\mathcal F_2)=(4K,4K)$ while the cost in the fixed computational basis---Born-sampling
time, through the bond dimension---is $O(K)$ at the localized representative and, at generic rotated
points, not established here. Two conclusions
follow, each in the only direction the construction supports: a Gaussian-invariant fermionic magic is not
a sufficient certificate of hardness and does not lower-bound the bond dimension, nor the cost beyond linear order in $K$ (there $\chi=2$ while
$\mathcal F_k=4K\to\infty$); and $\mathcal F_k$ determines the cost only if the rotated points are in
fact easy, being constant on an orbit that contains an $O(K)$-cost point and points at which sampling is
plausibly hard. In a generic rotated basis the Born sampling of $\ket{\Psi_K}$
is plausibly classically hard---the free-fermionic-linear-optics-with-magic-input regime of the Fermion
Sampling advantage scheme~\cite{fermionsampling2022}---while $\mathcal F_k$ is unchanged: any hardness
that appears belongs to the sampler's fixed frame. We are silent on whether magic is \emph{necessary} for
hardness, and we anchor the easy end to an external result rather than our own, since for this paired
class the additive-error estimation of amplitudes, overlaps and number correlators under free-fermionic
dynamics is classically efficient~\cite{oh2026}.

\paragraph{The same orbit reaches the leakage.} The orbit argument was written for $|\mathcal S|$; it
extends to the second budget of Eq.~\eqref{eq:twobudget}. An invariant upper
bound on $\Lambda_S(\eta)/\eta$ \emph{does} exist---it is exhibited inside the proposition below, and we
tested it on $1600$ hostile cases (GOE/GUE, near-degenerate, banded, $\lVert H\rVert\sim300$) with zero
violations and with its first link attained, minimum ratio $1.0000000000$. The correct statement is
therefore one of non-vacuity, not of non-existence: the same logical form the $|\mathcal S|$ no-go always
had, where the sector dimension is likewise an invariant upper bound that says nothing.

Table~\ref{tab:fiber} is the exhibit. The two endpoints of each fiber are the \emph{same} rotation angle
applied within one geminal block and across two blocks, which is why the captured weight agrees bit for
bit. An exhaustive search over all $276$ coordinate subspaces of that size at $K=2$ settles both ends.
At the aligned endpoint the global minimum of the certificate over the $276$ is $0.856$ and it is
attained by the protocol's own subspace: the subspace the sampler would return is the best available
one. At the rotated endpoint the minimum over the same $276$ is exactly the trivial bound, attained
only by subspaces of vanishing captured weight: \emph{no} determinant subspace of that size certifies
anything there.

\begin{table*}[tp]\centering\small
\caption{\textbf{Fibers of fixed invariants with moving leakage}, at $s=\pi/16$, $\delta=0.01$,
$\eta=0.10$, $|\mathcal S|=2^{K-1}$, $w_S=0.961940$, $\tau(w_S)=0.795$. Along each row the spectrum,
every Lehmann weight, $\lVert\phi\rVert^2$, $|\mathcal S|$ and all $\mathcal F_k$ are identical between
the two endpoints ($w_S$ to $2.2\times10^{-16}$), so the first term of Eq.~\eqref{eq:twobudget} is fixed
and only the second moves. Quantities are relative ($\widehat\Lambda_S=\Lambda_S/\lVert\phi\rVert$) and
``certificate'' is $\min\{(1+\sqrt{w_S})(\sqrt{1-w_S}+\widehat\Lambda_S/\eta),\,1+w_S\}$, the relative
form of Theorem~\ref{thm:leakage}; ``al.''\ and ``rot.''\ label the aligned and rotated endpoints of the
fiber. Bold marks a non-vacuous certificate, below the trivial bound $1+w_S$. A continuous angle sweep
is in Sec.~\ref{app:witness}.}
\label{tab:fiber}
\setlength{\tabcolsep}{4pt}
\begin{tabular}{@{}lccccc@{}}
\toprule
$K$ ($\dim(N{+}1)$) & $\mathcal F_1{=}\mathcal F_2{=}\mathcal F_3$
 & $\widehat\Lambda/\eta$ al. & $\widehat\Lambda/\eta$ rot.
 & cert.\ al. & cert.\ rot. \\
\midrule
$2$ \ ($24$)   & $8$  & $0.237$ & $2.326$ & $\mathbf{0.856}$ & $1.962$ (trivial) \\
$3$ \ ($300$)  & $12$ & $0.323$ & $2.336$ & $\mathbf{1.026}$ & $1.962$ (trivial) \\
$4$ \ ($3920$) & $16$ & $0.391$ & $2.346$ & $\mathbf{1.161}$ & $1.962$ (trivial) \\
\bottomrule
\end{tabular}
\end{table*}

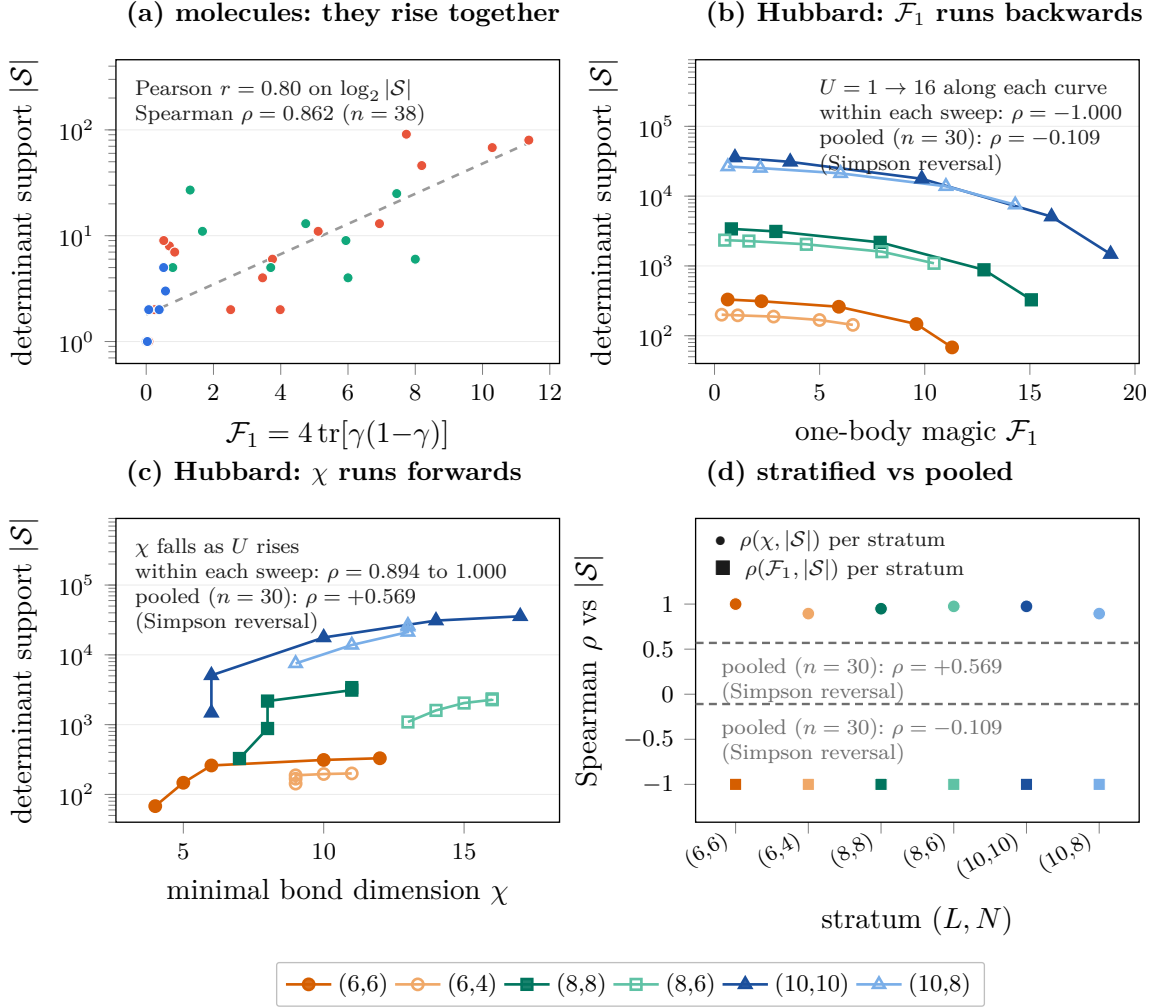
\begin{figure*}[tp]\centering
\definecolor{dcCov}{HTML}{E8543A}\definecolor{dcMrf}{HTML}{10A87B}\definecolor{dcIon}{HTML}{2E6FE0}
\definecolor{dcTrend}{HTML}{9A9A9A}\definecolor{dcPool}{HTML}{6E6E6E}\definecolor{inkPrim}{HTML}{1B1B1B}
\definecolor{scAh}{HTML}{D55E00}\definecolor{scAd}{HTML}{F0A868}\definecolor{scBh}{HTML}{00755E}\definecolor{scBd}{HTML}{5FC2A5}\definecolor{scCh}{HTML}{1B4F9C}\definecolor{scCd}{HTML}{79AEE8}
\begin{tikzpicture}
\begin{groupplot}[group style={group size=2 by 2, horizontal sep=1.8cm, vertical sep=2.05cm},
  width=7.45cm, height=5.6cm, paperaxis,
  title style={at={(0,1)},anchor=south west,font=\small\bfseries,yshift=2pt},
  label style={font=\normalsize}, tick label style={font=\footnotesize},
  ymajorgrids, grid style={draw=black!8}]
\nextgroupplot[title={(a)~molecules: they rise together},
  xlabel={$\mathcal F_1=4\,\mathrm{tr}[\gamma(1{-}\gamma)]$}, ylabel={determinant support $|\mathcal S|$},
  ymode=log, xmin=-0.9, xmax=12.29, ymin=0.62, ymax=460, enlarge x limits=false]
\addplot[dcTrend,line width=1pt,dashed,domain=0.000:11.377,samples=60,forget plot] {2^(0.4736*x+0.8502)};
\addplot[only marks,mark=*,mark size=1.9pt,draw=white,line width=0.4pt,fill=dcCov] coordinates {(0.0380,1) (2.5114,2) (0.0008,1) (3.4603,4) (0.0676,1) (3.9849,2) (0.6882,8) (8.1924,46) (0.8474,7) (11.3765,80) (0.0392,1) (3.7542,6) (0.0068,1) (6.9358,13) (0.0803,1) (10.2890,68) (0.2515,2) (5.1131,11) (0.5214,9) (7.7351,91)};
\addplot[only marks,mark=*,mark size=1.9pt,draw=white,line width=0.4pt,fill=dcMrf] coordinates {(4.7428,13) (8.0009,6) (0.0127,1) (3.7026,5) (1.6728,11) (6.0001,4) (1.3066,27) (5.9378,9) (0.7900,5) (7.4487,25)};
\addplot[only marks,mark=*,mark size=1.9pt,draw=white,line width=0.4pt,fill=dcIon] coordinates {(0.0021,1) (0.0750,2) (0.0000,1) (0.0007,1) (0.5181,5) (0.5731,3) (0.0353,1) (0.3864,2)};
\node[inkPrim,font=\scriptsize,anchor=north west,align=left] at (rel axis cs:0.02,0.97) {Pearson $r=0.80$ on $\log_2|\mathcal S|$\\Spearman $\rho=0.862$ ($n=38$)};
\nextgroupplot[title={(b)~Hubbard: $\mathcal F_1$ runs backwards},
  xlabel={one-body magic $\mathcal F_1$}, ylabel={determinant support $|\mathcal S|$},
  ymode=log, xmin=-0.9, xmax=20.2, ymin=40, ymax=900000, enlarge x limits=false]
\addplot[scAh,mark=*,mark size=2.1pt,line width=1.0pt,mark options={fill=scAh,draw=scAh}] coordinates {(0.6258,330) (2.2316,312) (5.9132,260) (9.6047,147) (11.2997,68)};
\addplot[scAd,mark=o,mark size=2.1pt,line width=1.0pt,mark options={fill=scAd,draw=scAd}] coordinates {(0.3396,200) (1.1047,196) (2.8100,188) (4.9945,168) (6.5809,143)};
\addplot[scBh,mark=square*,mark size=2.0pt,line width=1.0pt,mark options={fill=scBh,draw=scBh}] coordinates {(0.8020,3395) (2.9184,3120) (7.8807,2176) (12.8163,880) (15.0697,327)};
\addplot[scBd,mark=square,mark size=2.0pt,line width=1.0pt,mark options={fill=scBd,draw=scBd}] coordinates {(0.4855,2342) (1.6348,2272) (4.3683,2041) (7.9600,1600) (10.4394,1091)};
\addplot[scCh,mark=triangle*,mark size=2.6pt,line width=1.0pt,mark options={fill=scCh,draw=scCh}] coordinates {(0.9724,35678) (3.6002,30975) (9.8568,17732) (16.0306,5084) (18.8404,1476)};
\addplot[scCd,mark=triangle,mark size=2.6pt,line width=1.0pt,mark options={fill=scCd,draw=scCd}] coordinates {(0.6329,26503) (2.1803,25235) (5.9969,21203) (11.0065,13832) (14.3023,7527)};
\node[inkPrim,font=\scriptsize,anchor=north east,align=left] at (rel axis cs:0.985,0.97) {$U=1\to16$ along each curve\\within each sweep: $\rho=-1.000$\\pooled ($n=30$): $\rho=-0.109$\\(Simpson reversal)};
\nextgroupplot[title={(c)~Hubbard: $\chi$ runs forwards},
  xlabel={minimal bond dimension $\chi$},
  ylabel={determinant support $|\mathcal S|$},
  ymode=log, xmin=2.6, xmax=18.4, ymin=40, ymax=900000, enlarge x limits=false,
  legend style={at={(1.16,-0.46)},anchor=north,legend columns=6,draw=black!25,font=\footnotesize},
  legend cell align=left]
\addplot[scAh,mark=*,mark size=2.1pt,line width=1.0pt,mark options={fill=scAh,draw=scAh}] coordinates {(12,330) (10,312) (6,260) (5,147) (4,68)};
\addlegendentry{$(6{,}6)$}
\addplot[scAd,mark=o,mark size=2.1pt,line width=1.0pt,mark options={fill=scAd,draw=scAd}] coordinates {(11,200) (10,196) (9,188) (9,168) (9,143)};
\addlegendentry{$(6{,}4)$}
\addplot[scBh,mark=square*,mark size=2.0pt,line width=1.0pt,mark options={fill=scBh,draw=scBh}] coordinates {(11,3395) (11,3120) (8,2176) (8,880) (7,327)};
\addlegendentry{$(8{,}8)$}
\addplot[scBd,mark=square,mark size=2.0pt,line width=1.0pt,mark options={fill=scBd,draw=scBd}] coordinates {(16,2342) (16,2272) (15,2041) (14,1600) (13,1091)};
\addlegendentry{$(8{,}6)$}
\addplot[scCh,mark=triangle*,mark size=2.6pt,line width=1.0pt,mark options={fill=scCh,draw=scCh}] coordinates {(17,35678) (14,30975) (10,17732) (6,5084) (6,1476)};
\addlegendentry{$(10{,}10)$}
\addplot[scCd,mark=triangle,mark size=2.6pt,line width=1.0pt,mark options={fill=scCd,draw=scCd}] coordinates {(13,26503) (13,25235) (13,21203) (11,13832) (9,7527)};
\addlegendentry{$(10{,}8)$}
\node[inkPrim,font=\scriptsize,anchor=north west,align=left] at (rel axis cs:0.02,0.97) {$\chi$ falls as $U$ rises\\within each sweep: $\rho=0.894$ to $1.000$\\pooled ($n=30$): $\rho=+0.569$\\(Simpson reversal)};
\nextgroupplot[title={(d)~stratified vs pooled},
  xlabel={stratum $(L,N)$}, ylabel={Spearman $\rho$ vs $|\mathcal S|$},
  xmin=0.45, xmax=6.55, ymin=-1.42, ymax=1.95, ytick={-1,-0.5,0,0.5,1},
  xtick={1,2,3,4,5,6}, xticklabels={$(6{,}6)$,$(6{,}4)$,$(8{,}8)$,$(8{,}6)$,$(10{,}10)$,$(10{,}8)$},
  x tick label style={font=\scriptsize,rotate=30,anchor=east}, ymajorgrids=false]
\addplot[dcPool,densely dashed,line width=0.9pt,forget plot] coordinates {(0.45,0.5692) (6.55,0.5692)};
\addplot[dcPool,densely dashed,line width=0.9pt,forget plot] coordinates {(0.45,-0.1088) (6.55,-0.1088)};
\addplot[only marks,mark=*,mark size=2.4pt,scAh,mark options={fill=scAh,draw=white,line width=0.4pt},forget plot] coordinates {(1,1.0000)};
\addplot[only marks,mark=square*,mark size=2.3pt,scAh,mark options={fill=scAh,draw=white,line width=0.4pt},forget plot] coordinates {(1,-1.0000)};
\addplot[only marks,mark=*,mark size=2.4pt,scAd,mark options={fill=scAd,draw=white,line width=0.4pt},forget plot] coordinates {(2,0.8944)};
\addplot[only marks,mark=square*,mark size=2.3pt,scAd,mark options={fill=scAd,draw=white,line width=0.4pt},forget plot] coordinates {(2,-1.0000)};
\addplot[only marks,mark=*,mark size=2.4pt,scBh,mark options={fill=scBh,draw=white,line width=0.4pt},forget plot] coordinates {(3,0.9487)};
\addplot[only marks,mark=square*,mark size=2.3pt,scBh,mark options={fill=scBh,draw=white,line width=0.4pt},forget plot] coordinates {(3,-1.0000)};
\addplot[only marks,mark=*,mark size=2.4pt,scBd,mark options={fill=scBd,draw=white,line width=0.4pt},forget plot] coordinates {(4,0.9747)};
\addplot[only marks,mark=square*,mark size=2.3pt,scBd,mark options={fill=scBd,draw=white,line width=0.4pt},forget plot] coordinates {(4,-1.0000)};
\addplot[only marks,mark=*,mark size=2.4pt,scCh,mark options={fill=scCh,draw=white,line width=0.4pt},forget plot] coordinates {(5,0.9747)};
\addplot[only marks,mark=square*,mark size=2.3pt,scCh,mark options={fill=scCh,draw=white,line width=0.4pt},forget plot] coordinates {(5,-1.0000)};
\addplot[only marks,mark=*,mark size=2.4pt,scCd,mark options={fill=scCd,draw=white,line width=0.4pt},forget plot] coordinates {(6,0.8944)};
\addplot[only marks,mark=square*,mark size=2.3pt,scCd,mark options={fill=scCd,draw=white,line width=0.4pt},forget plot] coordinates {(6,-1.0000)};
\node[inkPrim,font=\scriptsize,anchor=north west] at (rel axis cs:0.02,0.995) {$\bullet$\; $\rho(\chi,|\mathcal S|)$ per stratum};
\node[inkPrim,font=\scriptsize,anchor=north west] at (rel axis cs:0.02,0.90) {$\blacksquare$\; $\rho(\mathcal F_1,|\mathcal S|)$ per stratum};
\node[dcPool,font=\scriptsize,anchor=north west,align=left] at (rel axis cs:0.035,0.578) {pooled ($n=30$): $\rho=+0.569$\\(Simpson reversal)};
\node[dcPool,font=\scriptsize,anchor=north west,align=left] at (rel axis cs:0.035,0.377) {pooled ($n=30$): $\rho=-0.109$\\(Simpson reversal)};
\end{groupplot}
\end{tikzpicture}
\caption{\label{fig:decoupling}\textbf{Inside every stratum one-body magic $\mathcal F_1$ orders the
static determinant support $|\mathcal S|$ perfectly and with the sign inverted; what fails is the
\emph{pooled} coefficient, not $\mathcal F_1$.}
\textbf{(a)}~Nineteen molecular active spaces at equilibrium and at dissociation ($38$ points):
in the static-correlation regime $\mathcal F_1$ and $|\mathcal S|$ rise together
(Pearson $r=0.80$ on $\log_2|\mathcal S|$; Spearman $\rho=0.862$). The grey dashed line is the least-squares fit $\log_2|\mathcal S|=0.474\,\mathcal F_1+0.850$ over those same points, a guide to that correlation and not a model of the cost. Color: covalent bond-breaking (red, $20$ points), multireference/open-shell (green, $10$) and ionic/weakly correlated (blue, $8$), the groups of Table~\ref{tab:molecules}.
\textbf{(b)}~The Hubbard interaction sweep is not one cloud of $30$ points but six strata
$(L,N)$ of five on-site repulsions $U=1,2,4,8,16$; one curve per stratum. Inside every one of
the six, raising $U$ raises $\mathcal F_1$ and lowers $|\mathcal S|$, and does so with a perfect
rank reversal: $\rho(\mathcal F_1,|\mathcal S|)=-1.000$ in all six. Pooling across strata
collapses the coefficient to $\rho=-0.109$; that collapse is a Simpson reversal, not an absence
of signal.
\textbf{(c)}~The same six strata against the minimal bond dimension: $\rho(\chi,|\mathcal S|)$
runs from $+0.894$ to $+1.000$ within strata, with a mean of $0.948$, and does not change sign between
them. Its pooled coefficient, like $\mathcal F_1$'s, is a Simpson reversal and is not used: pooling six
deterministic five-point sweeps destroys within-stratum structure whichever way the result falls, and
it fell in our favour on one of the two.
\textbf{(d)}~The six within-stratum coefficients against the two pooled ones. Stratified, the
$\mathcal F_1$ evidence is one-sided and exact: Kendall $\tau=-1.000$ with $0$ concordant of
$60$ within-stratum pairs and no ties, one-sided exact permutation $p=3.35\times10^{-13}$ over
the $(5!)^6=2.99\times10^{12}$ within-stratum relabellings. The pooled reading, $\rho=-0.11$ ($n=30$, $p\approx0.56$), is not used.
Two limits of this evidence, which the stratified test does not remove. First, the five points
of a stratum are a deterministic monotone sweep in $U$, so the effective number of independent
observations is $6$, not $30$; the exact $p$ counts relabellings, not independent samples.
Second, this is a statement about rank direction only: inside a fixed sweep $\mathcal F_1$
orders the five costs \emph{perfectly}, so the claim is not that $\mathcal F_1$ is uninformative
about $|\mathcal S|$ but that its sign is not stable---inverted here, positive over the
molecular set of panel~(a)---which is exactly what an orbital invariant compared against a
basis-dependent cost must do.}
\end{figure*}

\begin{figure*}[tp]\centering
\definecolor{rmMol}{HTML}{0072B2}\definecolor{rmHub}{HTML}{E69F00}
\definecolor{rmChn}{HTML}{009E73}\definecolor{rmLad}{HTML}{D55E00}\definecolor{rmInk}{HTML}{1B1B1B}
\definecolor{rmGrey}{HTML}{6E6E6E}
\begin{tikzpicture}
\begin{groupplot}[group style={group size=2 by 2, horizontal sep=1.8cm, vertical sep=2.05cm},
  width=7.45cm, height=5.6cm, paperaxis,
  title style={at={(0,1)},anchor=south west,font=\small\bfseries,yshift=2pt},
  label style={font=\normalsize}, tick label style={font=\footnotesize},
  legend style={font=\scriptsize,draw=black!25,fill=white,fill opacity=0.92,text opacity=1},
  legend cell align=left]
\nextgroupplot[title={(a)~cost tracks $\chi$}, xmode=log, ymode=log,
  xlabel={minimal bond dimension $\chi$}, ylabel={determinant support $|\mathcal S|$},
  xmin=0.8,xmax=70,ymin=0.50,ymax=900000, ymajorgrids,xmajorgrids, grid style={draw=black!8},
  legend to name=fig2legend, legend columns=3, legend style={draw=black!25,/tikz/every even column/.append style={column sep=11pt}}]
\addplot[only marks,mark=*,mark size=1.9pt,rmMol,mark options={fill=rmMol,draw=white,line width=0.3pt}] coordinates {(2,2) (4,2) (4,4) (2,2) (2,5) (2,3) (14,8) (30,46) (7,7) (32,80) (11,13) (2,6) (8,6) (13,13) (42,68) (4,2) (8,11) (8,9) (36,91) (2,2) (4,5) (4,11) (2,4) (17,27) (7,9) (11,5) (7,25)};
\addlegendentry{molecules, $|\mathcal S|>1$ ($\times27$)}
\addplot[only marks,mark=o,mark size=2.2pt,rmMol,line width=0.8pt] coordinates {(1,1) (1,1) (1,1) (2,1) (1,1) (1,1) (1,1) (1,1) (4,1) (1,1) (1,1)};
\addlegendentry{molecules at the floor $|\mathcal S|=1$ ($\times11$)}
\addplot[only marks,mark=square*,mark size=1.9pt,rmHub,mark options={fill=rmHub,draw=white,line width=0.3pt}] coordinates {(12,330) (10,312) (6,260) (5,147) (4,68) (11,200) (10,196) (9,188) (9,168) (9,143) (11,3395) (11,3120) (8,2176) (8,880) (7,327) (16,2342) (16,2272) (15,2041) (14,1600) (13,1091) (17,35678) (14,30975) (10,17732) (6,5084) (6,1476) (13,26503) (13,25235) (13,21203) (11,13832) (9,7527)};
\addlegendentry{Hubbard sweep ($\times30$)}
\addplot[only marks,mark=triangle*,mark size=2.2pt,rmChn,mark options={fill=rmChn,draw=white,line width=0.3pt}] coordinates {(5,147) (8,880) (6,5084)};
\addlegendentry{chain ($\times3$)}
\addplot[only marks,mark=diamond*,mark size=2.2pt,rmLad,mark options={fill=rmLad,draw=white,line width=0.3pt}] coordinates {(13,224) (38,1479) (38,11375)};
\addlegendentry{2-leg ladder ($\times3$)}
\addplot[rmInk,densely dashed,line width=0.9pt,mark=none,forget plot] coordinates {(2,2) (2,4) (2,8) (2,16) (2,32)};
\node[rmInk,font=\scriptsize,anchor=west] at (axis cs:2.15,120000) {witness: $\chi{=}2,\;|\mathcal S|{=}2^{K}$};
\node[rmInk,font=\scriptsize,anchor=south east] at (rel axis cs:0.98,0.02) {pooled $\rho=+0.722$};
\nextgroupplot[title={(b)~$\mathcal F_1$ does not}, ymode=log,
  xlabel={one-body magic $\mathcal F_1$}, ylabel={determinant support $|\mathcal S|$},
  xmin=-1.2,xmax=23,ymin=0.50,ymax=900000, ymajorgrids, grid style={draw=black!8}]
\addplot[only marks,mark=*,mark size=1.9pt,rmMol,mark options={fill=rmMol,draw=white,line width=0.3pt}] coordinates {(2.5114,2) (0.0750,2) (3.4603,4) (3.9849,2) (0.5181,5) (0.5731,3) (0.6882,8) (8.1924,46) (0.8474,7) (11.3765,80) (4.7428,13) (8.0009,6) (3.7542,6) (6.9358,13) (10.2890,68) (0.2515,2) (5.1131,11) (0.5214,9) (7.7351,91) (0.3864,2) (3.7026,5) (1.6728,11) (6.0001,4) (1.3066,27) (5.9378,9) (0.7900,5) (7.4487,25)};
\addplot[only marks,mark=o,mark size=2.2pt,rmMol,line width=0.8pt] coordinates {(0.0380,1) (0.0021,1) (0.0008,1) (0.0676,1) (0.0000,1) (0.0007,1) (0.0392,1) (0.0068,1) (0.0803,1) (0.0353,1) (0.0127,1)};
\addplot[only marks,mark=square*,mark size=1.9pt,rmHub,mark options={fill=rmHub,draw=white,line width=0.3pt}] coordinates {(0.6258,330) (2.2316,312) (5.9132,260) (9.6047,147) (11.2997,68) (0.3396,200) (1.1047,196) (2.8100,188) (4.9945,168) (6.5809,143) (0.8020,3395) (2.9184,3120) (7.8807,2176) (12.8163,880) (15.0697,327) (0.4855,2342) (1.6348,2272) (4.3683,2041) (7.9600,1600) (10.4394,1091) (0.9724,35678) (3.6002,30975) (9.8568,17732) (16.0306,5084) (18.8404,1476) (0.6329,26503) (2.1803,25235) (5.9969,21203) (11.0065,13832) (14.3023,7527)};
\addplot[only marks,mark=triangle*,mark size=2.2pt,rmChn,mark options={fill=rmChn,draw=white,line width=0.3pt}] coordinates {(9.6047,147) (12.8163,880) (16.0306,5084)};
\addplot[only marks,mark=diamond*,mark size=2.2pt,rmLad,mark options={fill=rmLad,draw=white,line width=0.3pt}] coordinates {(9.6948,224) (12.9107,1479) (16.0889,11375)};
\node[rmInk,font=\scriptsize,anchor=north west,align=left] at (rel axis cs:0.02,0.97) {pooled $\rho=+0.599$ over the same 74 rows\\Hubbard branch: $U\!\uparrow\;\Rightarrow\;\mathcal F_1\!\uparrow,\;|\mathcal S|\!\downarrow$};
\nextgroupplot[title={(c)~coefficients; intervals for $n{=}38$ only},
  xlabel={Spearman $\rho$ against $|\mathcal S|$}, ylabel={},
  xmin=-0.03,xmax=1.03,ymin=-1.55,ymax=5.45, xtick={0,0.25,0.5,0.75,1},
  ytick={1,2,3,4}, yticklabels={{$\mathcal F_1$, $n{=}74$},{$\chi$, $n{=}74$},{$\mathcal F_1$, $n{=}38$},{$\chi$, $n{=}38$}},
  y tick label style={font=\scriptsize}, ymajorgrids=false, xmajorgrids, grid style={draw=black!8}]
\addplot[rmGrey,densely dashed,line width=0.9pt,forget plot] coordinates {(0.9863,0.45) (0.9863,4.45)};
\addplot[rmMol,line width=1.0pt,forget plot] coordinates {(0.8051,4) (0.9589,4)};
\addplot[rmMol,line width=1.0pt,forget plot] coordinates {(0.8051,3.87) (0.8051,4.13)};
\addplot[rmMol,line width=1.0pt,forget plot] coordinates {(0.9589,3.87) (0.9589,4.13)};
\addplot[rmLad,line width=1.0pt,forget plot] coordinates {(0.7091,3) (0.9294,3)};
\addplot[rmLad,line width=1.0pt,forget plot] coordinates {(0.7091,2.87) (0.7091,3.13)};
\addplot[rmLad,line width=1.0pt,forget plot] coordinates {(0.9294,2.87) (0.9294,3.13)};
\addplot[only marks,mark=*,mark size=2.4pt,rmMol,mark options={fill=rmMol,draw=white,line width=0.4pt},forget plot] coordinates {(0.9011,4)};
\addplot[only marks,mark=*,mark size=2.4pt,rmLad,mark options={fill=rmLad,draw=white,line width=0.4pt},forget plot] coordinates {(0.8618,3)};
\addplot[only marks,mark=*,mark size=2.4pt,rmMol,mark options={fill=rmMol,draw=white,line width=0.4pt},forget plot] coordinates {(0.7224,2)};
\addplot[only marks,mark=*,mark size=2.4pt,rmLad,mark options={fill=rmLad,draw=white,line width=0.4pt},forget plot] coordinates {(0.5992,1)};
\node[rmGrey,font=\scriptsize,anchor=east] at (axis cs:0.9863,5.05) {tie ceiling $0.986$\;};
\node[rmInk,font=\scriptsize,anchor=west] at (axis cs:-0.01,-0.40) {$n{=}38$: $\Delta\rho=0.039$ ($t=0.94$, $p=0.35$)};
\node[rmInk,font=\scriptsize,anchor=west] at (axis cs:-0.01,-1.15) {$n{=}74$: $\Delta\rho=0.123$ ($t=1.50$, $p=0.14$)};
\nextgroupplot[title={(d)~the 38 molecules in rank space},
  xlabel={midrank of $|\mathcal S|$}, ylabel={midrank of the predictor},
  xmin=0,xmax=40,ymin=0,ymax=59, xmajorgrids, ymajorgrids, grid style={draw=black!8},
  ytick={0,10,20,30,40},
  legend style={at={(0.98,0.985)},anchor=north east,font=\scriptsize}]
\addplot[rmGrey,densely dashed,line width=0.8pt,forget plot] coordinates {(1,1) (38,38)};
\addplot[rmGrey,line width=0.7pt,forget plot] coordinates {(6.0,25.0) (6.0,41.5)};
\addplot[only marks,mark=*,mark size=1.9pt,rmMol,mark options={fill=rmMol,draw=white,line width=0.3pt}] coordinates {(6.0,5.0) (14.0,13.5) (6.0,5.0) (14.0,20.5) (6.0,5.0) (18.5,20.5) (6.0,13.5) (14.0,13.5) (6.0,5.0) (6.0,5.0) (21.0,13.5) (17.0,13.5) (26.0,33.0) (35.0,35.0) (25.0,25.0) (37.0,36.0) (31.5,30.5) (23.5,13.5) (6.0,5.0) (23.5,28.0) (6.0,5.0) (31.5,32.0) (6.0,20.5) (36.0,38.0) (14.0,20.5) (29.5,28.0) (27.5,28.0) (38.0,37.0) (6.0,5.0) (14.0,13.5) (6.0,5.0) (21.0,20.5) (29.5,20.5) (18.5,13.5) (34.0,34.0) (27.5,25.0) (21.0,30.5) (33.0,25.0)};
\addlegendentry{$\chi$, $\rho=0.901$}
\addplot[only marks,mark=square*,mark size=1.9pt,rmLad,mark options={fill=rmLad,draw=white,line width=0.3pt}] coordinates {(6.0,8.0) (14.0,23.0) (6.0,4.0) (14.0,11.0) (6.0,3.0) (18.5,24.0) (6.0,10.0) (14.0,27.0) (6.0,1.0) (6.0,2.0) (21.0,15.0) (17.0,17.0) (26.0,18.0) (35.0,36.0) (25.0,20.0) (37.0,38.0) (31.5,28.0) (23.5,35.0) (6.0,9.0) (23.5,26.0) (6.0,5.0) (31.5,32.0) (6.0,12.0) (36.0,37.0) (14.0,13.0) (29.5,29.0) (27.5,16.0) (38.0,34.0) (6.0,7.0) (14.0,14.0) (6.0,6.0) (21.0,25.0) (29.5,22.0) (18.5,31.0) (34.0,21.0) (27.5,30.0) (21.0,19.0) (33.0,33.0)};
\addlegendentry{$\mathcal F_1$, $\rho=0.862$}
\node[rmInk,font=\scriptsize,anchor=north west,align=left] at (rel axis cs:0.02,0.985) {11 of 38 tied at $|\mathcal S|=1$\\midrank $6.0$: $|\rho|\le 0.986$};
\node[rmGrey,font=\scriptsize,anchor=west] at (axis cs:7.1,33.0) {tied block};
\end{groupplot}
\node[anchor=north,yshift=-1.25cm] at (group c1r2.south east) {\ref*{fig2legend}};
\end{tikzpicture}
\caption{\label{fig:master}\textbf{The resource thesis over the $74$ deposited rows
($71$ distinct exact ground states).} Nineteen molecules at equilibrium and dissociation, a
$30$-point Hubbard interaction sweep, and a matched chain/two-leg-ladder pair; all by exact
diagonalization.
\textbf{(a)}~The determinant support tracks the minimal bond dimension in every family (pooled
Spearman $\rho=+0.722$ over all $74$ rows, $+0.746$ over the $71$ distinct states---see
\emph{Provenance} below), with the geminal witness $\chi{=}2$, $|\mathcal S|{=}2^{K}$ the
deliberately loose lower bound. Open circles are the $11$ of $38$ molecular points pinned at the
floor $|\mathcal S|=1$---$29\%$ of the molecular set---where the support cannot resolve anything
and from which both molecular coefficients draw much of their leverage.
\textbf{(b)}~One-body magic does not ($\rho=+0.599$ over the same rows, $+0.592$ over the $71$
distinct states): the Hubbard branch runs
the wrong way ($U\!\uparrow$ raises $\mathcal F_1$ and lowers $|\mathcal S|$) and the
chain/ladder pair sits at $\mathcal F_1$ values agreeing to within $0.9\%$ while $|\mathcal S|$
differs by $1.52$, $1.68$ and $2.24\times$.
\textbf{(c)}~The two coefficients with the intervals the deposited data actually support. Bars
are $95\%$ cluster bootstrap intervals over the $19$ molecules ($B=20\,000$, cluster $=$
molecule): $\chi$ gives $0.901$ $[0.805,0.959]$, $\mathcal F_1$ gives $0.862$ $[0.709,0.929]$.
\emph{The molecular comparison does not resolve the two resources}: $\Delta\rho=0.039$,
Williams $t=0.94$ on $35$ df, $p=0.35$ (Steiger $z=0.94$), bootstrap $95\%$ CI of the difference
$[-0.070,+0.189]$; neither does the pooled comparison ($\Delta\rho=0.123$, $t=1.50$, $p=0.14$).
The $n=74$ entries are point estimates---only the difference was resampled. The dashed line is
the ceiling $|\rho|\le0.986$ imposed by the ties in $|\mathcal S|$, against which $\chi$ reaches
$91\%$ and $\mathcal F_1$ $87\%$ of what is attainable. What discriminates between the two
candidate resources is therefore not this figure but the stratified Hubbard sweep of
Fig.~\ref{fig:decoupling}(b,d).
\textbf{(d)}~The same $38$ molecular points in the rank space in which a Spearman coefficient is
actually computed, against the identity line. The $11$ floor points share midrank $6.0$---the
vertical block at the left---which is the origin of the ceiling quoted in~(c).
\emph{Provenance and caveats.} The three ``chain'' rows are the same states as the $U=8$
half-filled entries of the Hubbard sweep (identical $|\mathcal S|$ and $\chi$, identical
$\mathcal F_1$ to the precision at which they are stored), re-plotted as the ladder's matched
partner; the $74$ rows are therefore $71$ distinct states and three of them enter the pooled
coefficients twice. Recomputed over the $71$ distinct states the pooled coefficients are
$\rho(\chi,|\mathcal S|)=+0.746$ and $\rho(\mathcal F_1,|\mathcal S|)=+0.592$. The $38$
molecular points are $19$ molecules $\times$ $2$ geometries, i.e.\ dependent pairs (effective
$n\approx31.3$), which is why every molecular interval here is clustered by molecule.}
\end{figure*}
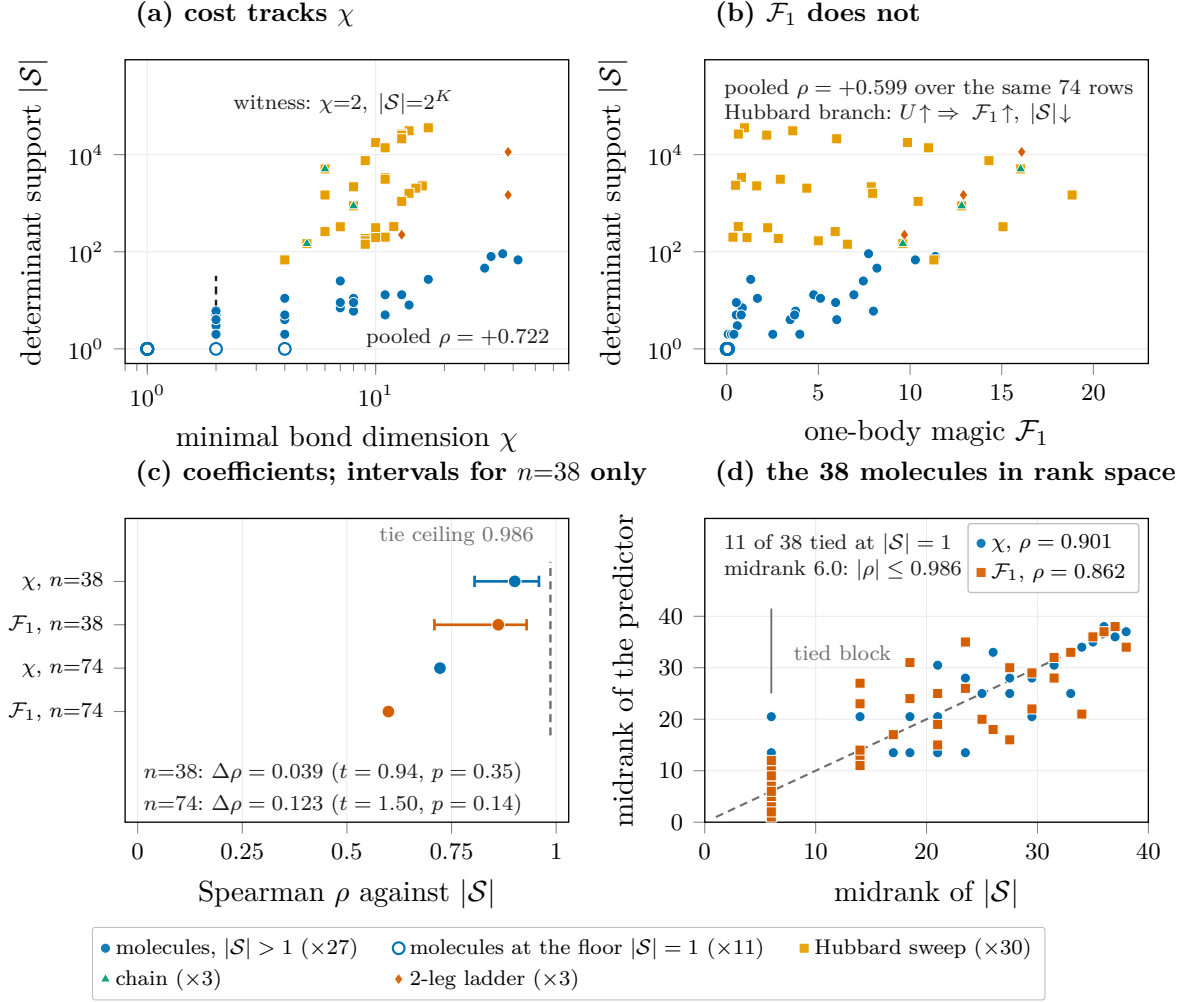

\begin{figure*}[tp]\centering\definecolor{ldChain}{HTML}{0072B2}\definecolor{ldLad}{HTML}{D55E00}\definecolor{ldInk}{HTML}{1B1B1B}
\begin{tikzpicture}
\begin{groupplot}[group style={group size=2 by 1, horizontal sep=1.75cm},
  width=7.15cm, height=6.0cm, paperaxis,
  xmin=10.5, xmax=21.5, xtick={12,16,20}, xtick align=outside,
  title style={at={(0,1)},anchor=south west,font=\normalsize\bfseries,yshift=2pt},
  label style={font=\normalsize}, tick label style={font=\normalsize},
  legend style={font=\footnotesize,draw=black!25,fill=white,fill opacity=0.9,text opacity=1},legend cell align=left]
\nextgroupplot[title={(a)~minimal bond dimension $\chi$},
  xlabel={qubits $\;2L$}, ylabel={$\chi$ (min.\ MPS bond dim.)}, ymin=0, ymax=62,
  ytick={0,10,20,30,40}, ymajorgrids, grid style={draw=black!10},
  legend pos=north west]
\addplot[ldLad,line width=1.5pt,mark=*,mark size=3pt,mark options={fill=ldLad,draw=white,line width=0.8pt}]
  coordinates {(12,13)(16,38)(20,38)}; \addlegendentry{2-leg ladder}
\addplot[ldChain,line width=1.5pt,mark=square*,mark size=2.8pt,mark options={fill=ldChain,draw=white,line width=0.8pt}]
  coordinates {(12,5)(16,8)(20,6)}; \addlegendentry{chain}
\node[ldInk,font=\scriptsize,anchor=north east] at (axis cs:21.0,55.0) {$\mathcal F_1$ matched};
\node[ldInk,font=\scriptsize,anchor=west] at (axis cs:16.55,27.0) {$\times4.8$};
\nextgroupplot[title={(b)~response-targeted $(N{\pm}1)$ fraction},
  xlabel={qubits $\;2L$}, ylabel={$|\mathcal S|/D$ for rel-$L_1<0.05$}, ymin=0, ymax=1.0,
  xmin=10.5, xmax=17.5, xtick={12,16}, ytick={0,0.25,0.5,0.75,1}, ymajorgrids, grid style={draw=black!10},
  legend pos=south west]
\addplot[ldLad,line width=1.5pt,mark=*,mark size=3pt,mark options={fill=ldLad,draw=white,line width=0.8pt}]
  coordinates {(12,0.807)(16,0.626)};
\addplot[ldChain,line width=1.5pt,mark=square*,mark size=2.8pt,mark options={fill=ldChain,draw=white,line width=0.8pt}]
  coordinates {(12,0.717)(16,0.426)};
\node[ldInk,font=\scriptsize,anchor=north east] at (axis cs:17.3,0.97) {same $\mathcal F_1$};
\end{groupplot}
\end{tikzpicture}
\caption{\textbf{The cost is tracked by entanglement, not one-body magic---a controlled chain-versus-ladder
contrast.} A two-leg Hubbard ladder versus a chain of the \emph{same} site count (half filling,
$U/t=8$; exact diagonalization). At matched fermionic magic $\mathcal F_1$ (open chain vs.\ ladder: $9.60/9.69$,
$12.82/12.91$, $16.03/16.09$ at $12/16/20$ qubits---$\mathcal F_1$ is a geometry-blind one-body invariant)
the ladder's perpendicular cut has an area-law boundary of two rather than one. \textbf{(a)}~Its minimal
matrix-product-state bond dimension $\chi$ is far larger; the chain-to-ladder ratio grows over
$12\to20$ qubits ($\times2.6,\times4.8,\times6.3$), the ladder $\chi$ saturating at $38$ while the chain
$\chi$ stays small and does \emph{not} grow ($5,8,6$ at $12/16/20$ qubits---non-monotone, with no systematic
increase over this range, consistent with the at most logarithmic entanglement growth of the spin
sector of the one-dimensional chain)---the reported $\chi$ are threshold-truncated Schmidt ranks. \textbf{(b)}~The response-targeted $(N{\pm}1)$ fraction needed for spectral
relative-$L_1<0.05$ is correspondingly higher for the ladder (these union $(N{\pm}1)$ response fractions
differ from the single-sector $(N{+}1)$ $A(\omega)$ fractions of Fig.~\ref{fig:scaling}, a different
observable). Same one-body magic, higher cost: the
determinant support is lower-bounded in the raw count and empirically tracked by the basis-dependent
bond dimension $\chi$ (Sec.~\ref{sec:nogo:chi}), not by $\mathcal F_1$. A ladder of this width
($\chi\!\approx\!38$) remains classically near-exact for dynamical DMRG; it is a resource contrast and
no claim of advantage is attached to it.}\label{fig:ladder}\end{figure*}
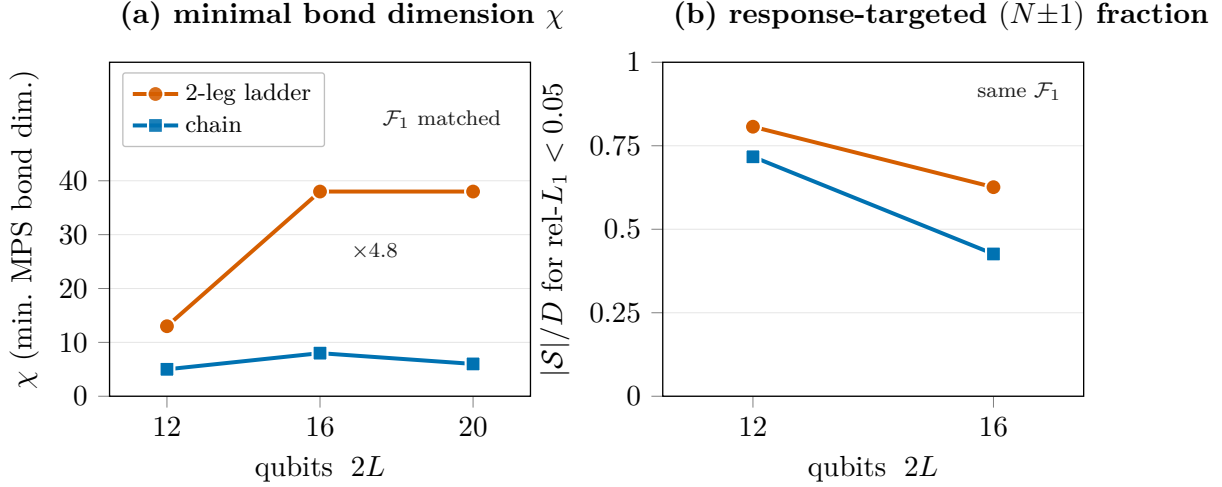

\paragraph{What this exhibit is not.} Four limits are declared with it, and they are why we call it a
fiber and not a sweep of the orbit. \emph{(i)}~The evidence is a matched pair at $s=\pi/16$ for $K=2,3,4$.
At the maximal-pairing angle $s=\pi/4$ the Born distribution of $\phi$ is exactly flat over $2^{K}$
determinants, so ``the top $|\mathcal S|$ of the Born distribution'' does not define a subspace at all
($6$, $70$ and $12870$ equally valid outputs at $K=2,3,4$); any contrast quoted at that angle is a
tie-breaking artifact and not a property of the fiber, and the sweep of Sec.~\ref{app:witness} stops
at $s=0.654$ for that reason. \emph{(ii)}~Nothing here grows with the system, and two separate quantities shrink. The ratio of the
two leakages falls, $9.80\to7.23\to6.01$ from $K=2$ to $K=4$, and it dies as the perturbation is turned
up, from $28.7$ at $\delta=0.0025$ to $1.03$ at $\delta=0.3$. More to the point---because it is the
quantity the proposition actually rests on---so does the margin of the certificate over the trivial
bound: $1.962/0.856=2.29$, $1.962/1.026=1.91$, $1.962/1.161=1.69$, a fall of $26\%$ across the only
three sizes available. Where that fall would reach unity we do not say: a linear extrapolation in $K$,
a log-linear one and the last step alone place it between $K=6$ and $K=8$, the decrements are
themselves decelerating, and three points do not support a number---which is exactly the objection this
paper makes to its own fraction headline in Sec.~\ref{sec:prices}. A statement of non-vacuity needs one
fiber and not a magnitude, and the fiber at $K=4$ is still non-vacuous by $69\%$; but ``the effect grows
with $K$'' is false here. \emph{(iii)}~The small-leakage
endpoint is realized only by the aligned basis, which is the stabilizer of the state
($|\braket{\Psi(0)|\Psi(U_{\rm in})}|=1.000000000000$); at the fiber angle $s=\pi/16$, every generator
that actually moves the state gives $\widehat\Lambda/\eta\in[2.34,5.61]$. That range is a statement
about generators at one angle and not about the angle sweep of Sec.~\ref{app:witness}, whose rotated
endpoint descends to $1.61$ as $s\to0$ because the rotation itself is being switched off. The rhetorical parallel with the cost witness---where
the orbit contains an $O(K)$-cost point and points of plausibly exponential cost---is therefore lost; Sec.~\ref{sec:nogo:asym} states
that asymmetry. \emph{(iv)}~The scale is a toy: $2K\le8$ orbitals, sector
dimension $\le3920$, no instance at the $L=12$ Hubbard scale of Sec.~\ref{sec:frontier}.

What the exhibit delivers survives all four limits and is measured: on fibers where every invariant is
fixed to $10^{-14}$ and the captured Born weight to $2.2\times10^{-16}$, the certificate of
Theorem~\ref{thm:leakage} separates a non-vacuous value from the trivial bound in every row of
Table~\ref{tab:fiber}, and a quantity constant on those fibers cannot make that distinction. The a priori
predictor asked for at the head of this section does not exist in the invariant class; what stands in its
place is the subspace-aware certificate, whose non-vacuous region is mapped, not asserted, in
Sec.~\ref{sec:frontier}.

\subsection{The obstruction is bilateral, and its sign is not stable}
\label{sec:nogo:bilateral}

The same symmetry closes the cost side, in both directions (Fig.~\ref{fig:decoupling}). The cost
statistic of this section is the \emph{static} determinant support $|\mathcal S|_{\varepsilon_w}$: the
number of fixed-basis Slater determinants of the \emph{ground state} carrying all but $\varepsilon_w$ of
its weight. Three tolerances called $0.05$ or $0.01$ appear in this paper and they are three different
objects, so we subscript all of them. $\varepsilon_w$ cuts the support: $\varepsilon_w=10^{-2}$ for the
molecular suite (\texttt{n19\_suite.py}, $99\%$ of the weight) and $\varepsilon_w=2.5\times10^{-3}$ for
the lattice families (\texttt{cost\_vs\_entanglement.py}, $99.75\%$, written there as
$1-\varepsilon^{2}$ with $\varepsilon=0.05$). $\varepsilon_\chi=0.05$ cuts the Schmidt spectrum for
$\chi$, by discarded tail $\le\varepsilon_\chi^{2}$, in both families. Neither is the third $0.05$ of
this work, which is a threshold on a relative $L_1$ spectral error at a stated resolution and defines
the \emph{dynamical} support $|\mathcal S|_\eta$ of Sec.~\ref{sec:prices}; and neither is the $S$ of
Theorem~\ref{thm:leakage}, which is whatever the sampler returned. This subsection is about
$|\mathcal S|_{\varepsilon_w}$ and about nothing else. Where the \emph{raw} support---determinants of
nonzero amplitude---is meant we say so, because the two do not obey the same bounds. Unlike $\mathcal F_1$, $|\mathcal S|$ is basis dependent: a
single Slater determinant has $|\mathcal S|=1$ in its own natural orbitals and exponentially many terms in
a generic basis. An orbital rotation is a free operation that leaves the occupation spectrum, hence
$\mathcal F_1$, exactly fixed while changing $|\mathcal S|$; a quantity invariant under operations that
change $|\mathcal S|$ at fixed value cannot determine it.

Both directions are realized and neither is a limiting case. A free-fermion ground state at $U=0$---one
Slater determinant, $\mathcal F_1=0$, the minimum---has site-basis support growing exponentially,
$|\mathcal S|_{\varepsilon_w}=35,\,336,\,3496,\,37361$ for $L=4,6,8,10$ on the open chain at the lattice
threshold $\varepsilon_w=2.5\times10^{-3}$, the discarded weight being at most $\varepsilon_w$
(\texttt{free\_fermion\_support.py}; at $L=4$ the four smallest determinant weights equal
$\varepsilon_w$ exactly, and a strict cut gives $36$): minimal invariant, divergent cost. Along a
Hubbard chain at fixed filling, raising $U$ \emph{raises} $\mathcal F_1$ while \emph{lowering}
$|\mathcal S|$ as the ground state localizes: rising invariant, falling cost. Across the nineteen
molecules of Table~\ref{tab:molecules} the two rise \emph{together} (Pearson $r=0.80$ on
$\log_2|\mathcal S|$), both merely tracking static correlation. The relationship is two-sided and its sign
is not stable---the signature of comparing an orbital invariant to a basis variant.

We state what this is and what it is not. It is not a refutation of anyone's claim:
Ref.~\cite{sierant2026} proposes $\mathcal F_k$ as a measure of fermionic magic (non-Gaussianity) for
matchgate circuits---faithful, Gaussian invariant and (sub)additive---which is what it is, and no one
proposed it as a predictor of the determinant support. It is our own
no-go, about our own cost, and a corollary consistent with---not a competitor to---the two families of
Tarabunga \emph{et al.}~\cite{tarabunga2026}, who in July 2026, before this work, showed that
occupation entropies depend only on the covariance matrix, two members of which (orders $1$ and $2$) are
strong monotones that lower-bound the non-Gaussian gate count,
while natural-orbital participation entropies depend on the full state and \emph{upper-bound the
classical simulation cost} in an orthonormal Gaussian basis. Both of their families are Gaussian
invariant; the basis dependence at issue here enters only because $|\mathcal S|$ is counted in a fixed
determinant basis rather than in the state's own natural orbitals. That is the quantity on the useful
side of the division, and the priority is theirs. On the invariant side the same acknowledgement is owed
to Leone and Bittel~\cite{leonebittel2026}, who in July 2026, before this work, proved that the
fermionic entropy $M_f=n\,(1-\|\Gamma\|_2^2/2n)$ on $n$ modes---an affine function of the squared
Frobenius norm of the correlation matrix---is a strong pure-state Gaussian monotone (its convex roof one for mixed states)---on $n=2n_{\rm o}$ modes
$M_f$ is $\mathcal F_1$ itself and Tarabunga \emph{et al.}'s quadratic occupation entropy, whose strong
monotonicity the revised version of Ref.~\cite{tarabunga2026} proves independently---and that its
fermionic purity $\|\Gamma\|_2^2/2n$ has an unbiased estimator using $O(\varepsilon^{-2})$ two-copy
measurements.
What our construction adds to that division is the constructive instance it does not
contain: a family (Theorem~\ref{thm:witness}) on which every order of the invariant is maximal while the
cost is $O(K)$, and fibers (Proposition~\ref{prop:leakinv}) on which the invariants are frozen to
$10^{-14}$ while the certificate of Theorem~\ref{thm:leakage} still separates.

\subsection{Two blindnesses, and why we do not sell them as one symmetry}
\label{sec:nogo:asym}

The two halves of this section are not one statement wearing two hats. On the cost side the invariance
is a genuine physical group: the Gaussian orbital
rotations $\mathrm U(n_{\rm o})$ are free operations, they act on the whole problem, and along one orbit
the cost is $O(K)$ at one point and, at others, plausibly exponential but not established here, at fixed invariant---a symmetry argument in
the ordinary sense. On the precision side it is weaker and of a different kind: $\Lambda_S(\eta)$ reads
the off-diagonal block $P_SHQ_S$ in the working basis, and an invariant is blind to that
block---informational blindness, not a gauge. The orbit there does not sweep an extensive range; as
declared above, the small-leakage endpoint is realized only by the stabilizer of the state, and the
contrast shrinks with $K$ rather than growing.

What joins the two is therefore not a symmetry at all, and nothing stronger is claimed for it than what
it is: the ``$+$'' of Eq.~\eqref{eq:twobudget}, one inequality whose two summands are the two halves of
this paper, one controlled by the captured weight and one not, each carrying its own obstruction. What makes that join a
measurement rather than an analogy is that both summands have been evaluated on the same object,
$\mathcal S^{\ast}=\operatorname{supp}(\phi)$, where the first summand is identically zero while
the second still leaves a relative $L_1$ error of $0.43$ at $L=6$ and $0.35$ at $L=8$ on the periodic
rings at $\eta=0.18\,t$, and $0.28$ on the open chain (Secs.~\ref{sec:leakage} and~\ref{sec:moments}).

\subsection{What does track the cost: \texorpdfstring{$\chi$}{chi}, and the statistics the right way up}
\label{sec:nogo:chi}
\global\count\footins=1000\relax

The computable lower bounds on $|\mathcal S|$ are themselves basis dependent. The orbital-Schmidt rank
across a bipartition---the tensor-network bond dimension $\chi$~\cite{schollwock2011}---is a rigorous
floor on the \emph{raw} support,\footnote{For any state and any bipartition into $(\mathrm L,\mathrm R)$,
$\chi\le|\mathcal S|$ with $|\mathcal S|$ the number of determinants of nonzero amplitude: each
determinant factorizes across the cut, so the matricization $A_{D_{\mathrm L},D_{\mathrm R}}=c_D$ has at
most $|\mathcal S|$ nonzero entries and rank at most $|\mathcal S|$. It is elementary, which is why we
record it as an observation, and it is proved for the raw support only: once amplitudes are thresholded
the inequality need not survive---in the molecular suite $|\mathcal S|$ is cut at $99\%$ of the weight
while $\chi$ is cut at a discarded tail of $\varepsilon_\chi^{2}=2.5\times10^{-3}$, i.e.\ at $99.75\%$,
so the two are not truncated at the same place---and the measured $|\mathcal S|_{\varepsilon_w}$ falls
below $\chi$ in $7$ of the $38$ molecular points. It is also one-sided and can be exponentially loose: Theorem~\ref{thm:witness} has
$\chi=2$ at $|\mathcal S|_{\mathrm{pair}}=2^{K}$.} and, empirically, a tracker of the measured
$|\mathcal S|_{\varepsilon_w}$. It is not the winner of a contest: $\chi$ is measured rather than
predicted, and it has no purchase on the leakage at all. On the fibers of Table~\ref{tab:fiber} it is
\emph{identical} at the two endpoints---$2$ and $2$ at $K=3$, $1$ and $1$ at $K=4$, at equal raw
support---while $\widehat\Lambda/\eta$ moves by a factor $6$--$7$; in a control at $K=3$ it jumps from
$2$ to $16$ with $\widehat\Lambda/\eta=2.336$ and relative $L_1=0.238$ unchanged to three figures; and
across a subfamily of $58$ points on which $\chi$ and the raw support are \emph{exactly} equal,
$\widehat\Lambda/\eta$ still varies by $2.8$ and the measured error by $8.3$. The only tracker of the
leakage is the boundary coupling $\lVert P_SHQ_S\rVert$ of Sec.~\ref{sec:leakage}. What follows is a
statement about $|\mathcal S|_{\varepsilon_w}$ and about nothing else.

\global\count\footins=500\relax %
\paragraph{The Hubbard sweep, stratified.} The $30$ lattice points are not $30$ independent observations.
They are six $(L,N)$ strata---$(6,6),(6,4),(8,8),(8,6),(10,10),(10,8)$---each swept over the same five
values $U=1,2,4,8,16$. Within every stratum the rank correlation of $\mathcal F_1$ with
$|\mathcal S|_{\varepsilon_w}$ is $-1.0000$, exactly, in all six; the corresponding
$\rho(\chi,|\mathcal S|_{\varepsilon_w})$ are $1.000$, $0.894$, $0.949$, $0.975$, $0.975$, $0.894$, mean
$0.948$. The stratified concordance statistics $(n_C-n_D)/(n_C+n_D)$ are $-1.000$ for $\mathcal F_1$,
with $0$ concordant pairs out of $60$ and no ties, where the statistic is Kendall's $\tau$, and $+1.000$
for $\chi$, with $50$ concordant, $0$ discordant and $10$ tied, where it is the Goodman--Kruskal
$\gamma$, which drops the tied pairs; Kendall's $\tau_b$, which charges them, is below $1$.
Because each $\mathcal F_1$ stratum is a perfect reversal of five points, the exact one-sided
stratified permutation $p$-value is $(1/120)^{6}=3.35\times10^{-13}$, against $1.93\times10^{-10}$ for
$\chi$, which has ties. Both take a product across strata and therefore assume the six independent.
They are not six independent experiments: they are three lengths crossed with two fillings, on one code
path and one $U$ grid. The per-stratum coefficients, Fisher's more conservative combination and the tie
structure are in Sec.~\ref{app:stats}.

Two pooled coefficients are available over these same $30$ points: $\rho=0.57$ for $\chi$, which
invites the reading that the bond dimension tracks the cost, and $\rho=-0.11$ for
$\mathcal F_1$, which invites the reading that it carries no evidence. Neither reading is used
here, and for one reason. The $30$
points are six deterministic five-point sweeps; the effective $n$ is $\approx6$, not $30$; and pooling
strata with different offsets destroys within-stratum structure. Here it destroys it in our favour on
one coefficient and against us on the other, which is precisely why neither may be kept. What is used
instead is the within-stratum evidence above, which is stronger on both sides and symmetric between
them:
$\mathcal F_1$ anti-orders the support perfectly in all six strata, and $\chi$ orders it at a mean
$\rho$ of $0.948$.

That the evidence is stronger does not make it what we were about to call it. A rank correlation of
exactly $-1$ in all six strata is the strongest evidence a rank statistic can give of a
\emph{deterministic monotone} relation between $\mathcal F_1$ and the support inside a stratum---with
the wrong sign. In this family $\mathcal F_1$ is not decoupled from the cost; it is an excellent rank
predictor of it whose \emph{sign} reverses from family to family: $+1.000$ in the chain and ladder
sweeps below, $+0.86$ across the molecules, $-1.000$ within every Hubbard stratum. That is the failure
mode the no-go of Secs.~\ref{sec:nogo:witness} and~\ref{sec:nogo:bilateral} predicts for an orbital
invariant compared against a basis variant, and ``unstable in sign'', not ``decoupled'', is the phrase
the measurement supports. The physical content is one mechanism observed six times---raising $U$
localizes the ground state, so $|\mathcal S|_{\varepsilon_w}$ falls and $\chi$ falls while
$\mathcal F_1$ rises---and a $p$-value of $3.35\times10^{-13}$ is the probability of six perfect reversals under independent random orderings within each stratum---a null that six strata sharing one mechanism do not satisfy---and not evidence that it generalizes beyond the one-dimensional
Hubbard family in which it was measured.

\paragraph{The molecular set does not discriminate.} The $n=38$ molecular
points are $19$ molecules at two geometries---dependent pairs, not $38$ draws---and $76\%$ of them lie
in some tie in $|\mathcal S|_{\varepsilon_w}$, $11$ pinned at the floor $|\mathcal S|=1$. The
comparison the molecular set is asked to settle is not resolvable by any of three procedures: $\rho=0.90$ for $\chi$ against
$0.86$ for $\mathcal F_1$ differ by $0.039$, a molecule-level cluster bootstrap puts a $95\%$ interval
of $[-0.070,+0.189]$ on that difference, and Williams's and Steiger's tests for dependent correlations
give $p=0.35$. The pooled pair over all $74$ states, $0.72$ against $0.60$, is likewise unresolved. The molecular set therefore does not separate the two candidate resources (Sec.~\ref{app:stats}), and the whole weight of the separation rests on the stratified
Hubbard sweep above.

\paragraph{The two families where $\mathcal F_1$ wins, written out.} In the matched chain and two-leg
ladder families of Fig.~\ref{fig:ladder} ($n=3$ each), $\rho(\mathcal F_1,|\mathcal S|)=+1.000$ exactly,
against $\rho(\chi,|\mathcal S|)=0.50$ and $0.87$. We report it and then dismantle it. The parameter swept
there is $L$, not $U$: for the open chains $\mathcal F_1/L=1.6008,\,1.6020,\,1.6031$ at $L=6,8,10$, that
is $\mathcal F_1=1.60\,L$ to three figures, so $\rho=+1$ restates ``both quantities grow with system
size'' and nothing else, and conditioning on $L$ leaves zero degrees of freedom. The three chain rows are
moreover literal duplicates of three of the $30$ Hubbard rows---the $U=8$, half-filled entries at
$L=6,8,10$, identical to the last digit in $\mathcal F_1$, $|\mathcal S|$ and $\chi$---so they are not
independent evidence. And with $n=3$ the smallest attainable permutation $p$-value is $1/6$: less evidence
than three coin flips. What the chain/ladder pair does establish is the controlled contrast it was built
for, which needs no correlation coefficient to read: at $\mathcal F_1$ equal to within $1\%$, the ladder's
support exceeds the chain's by up to ${\sim}2.2\times$, while $\chi$ moves from $5$ to $13$ at $L=6$ and
from $8$ to $38$ at $L=8$.

\paragraph{What stands, and where the tracker fails too.} Figure~\ref{fig:master} consolidates the
picture over all $74$ exact ground states: the determinant support is bounded below by the bond
dimension in the raw count, and the bond dimension is the better empirical tracker of the thresholded
one. It is not a monotone tracker, and the exception is in the same file as the rule. Within every
Hubbard stratum $\chi$ orders the support at $\rho=0.89$--$1.00$. Along the chain family, where the
swept parameter is $L$ and not $U$, it does not: $\chi=5,8,6$ at $L=6,8,10$ while
$|\mathcal S|_{\varepsilon_w}=147,\,880,\,5084$, so $\chi$ \emph{falls} across the last step while the
cost grows by $5.8$, and $\rho(\chi,|\mathcal S|_{\varepsilon_w})$ there is $0.50$; in the ladder family
$\chi$ is flat at $38$ over that same step while the cost grows by $7.7$. That non-monotonicity of the
chain $\chi$ is stated here explicitly, and it is a limit on $\chi$ itself. What
$\mathcal F_1$ does instead is fail in two \emph{additional} and visibly distinct ways---the Hubbard
branch runs the wrong way at a stratified $\tau=-1.000$, and the chain/ladder points sit at identical
$\mathcal F_1$ with $|\mathcal S|_{\varepsilon_w}$ differing by up to ${\sim}2.2\times$. Neither statement supplies the
predictor this section set out to find, $\chi$ being itself basis dependent, measured rather than
predicted, and blind to the leakage. What $\mathcal F_1$ \emph{is} remains intact and useful: a faithful,
efficiently measurable detector of static (multireference) correlation, generalizing to fermionic
non-Gaussianity, and to the nineteen molecules of Table~\ref{tab:molecules}, the qubit non-stabilizerness
peak of molecular bonding~\cite{sarkis2025} and its reported proportionality to Hartree--Fock weight,
itself reported to break down beyond the Coulson--Fischer point~\cite{seibert2026}; that a single-particle
invariant can be maximal while bond entanglement stays $O(1)$ is the fermionic counterpart of the
separation between magic and entanglement reported for qubit matrix-product
states~\cite{frautarabunga2024}. The cost is a
property of the sampler's frame, read after the fact by $\chi$ and priced in Sec.~\ref{sec:prices}; the
error is a property of the subspace's boundary, read after the fact by $\Lambda_S(\eta)$ and mapped in
Sec.~\ref{sec:frontier}. Per-species $|\mathcal S|_{\varepsilon_w}$, $D$ and $\chi$ are in
Sec.~\ref{app:repro}; the full stratified analysis---the six per-stratum coefficients, the exact
one-sided permutation values, the tie structure, the cluster bootstrap and the dependent-correlation
tests---is in Sec.~\ref{app:stats}.

\section{The three prices: the resolution grid and finite shots}
\label{sm:prices}

\paragraph{The support exponent and its fit range.} Restricting the fit of Eq.~\eqref{eq:supportexp} to
$L=4$--$12$, the five sizes Table~\ref{tab:fracgrid} covers, gives $1.082$ at $R^{2}=0.998$; we quote
the fit range with the exponent throughout, because it matters at the second decimal. The same caution
applies to the object fitted: $1.082$ is the fit to the published points, while the $1.090$ in the
bottom block of Table~\ref{tab:fracgrid} is the fit to that table's own coarser replication grid at the
same resolution. They are two fits to two datasets and the difference is theirs, not a rounding.

\paragraph{The two exponents are one exponent.} It would be convenient to report that the support
exponent, unlike the fraction exponent, is resolution-independent. It is not, and the arithmetic is
immediate: the fraction \emph{is} the support divided by the sector dimension---exactly, to the last
digit, at all six sizes---and the sector dimension does not depend on $\eta$. The two exponents
therefore differ by a constant,
$b_{|\mathcal S|}(\eta)-b_{\rm frac}(\eta)=b_{\dim}=1.291$ on $L=4$--$12$, and they move together:
across the decade $\eta=0.05\to0.50\,t$ both shift by $0.148$ in absolute value
(Table~\ref{tab:fracgrid}). There is one resolution-dependent exponent in this analysis, and quoting
the fraction and the support separately reports it twice.

What differs is the base it is measured against, and that is not cosmetic. On the fraction exponent
$0.148$ is three quarters of the quantity itself: $e^{-0.13\,L}$ becomes $e^{-0.27\,L}$, and the
headline changes qualitatively with a choice the model does not make. On the support exponent the same
$0.148$ carries it from $1.165$ to $1.017$---a seventh of its own value, $13\%$ measured against the
larger end and $15\%$ against the smaller---and ``exponential in $L$ with an exponent near one''
survives at every resolution in the decade. We therefore report the support, and we report it for two reasons, neither of which is
stability: it is the absolute resource a device has to populate, and it is not divided by a sector
dimension that grows for a reason---the combinatorics of the Hilbert space---that has nothing to do
with the method. The fraction is reported as the derived, resolution-dependent quantity that it is,
with its full grid attached.

\begin{table*}[tp]\centering\small
\caption{\textbf{Sector fraction required for relative $L_1<0.05$, versus system size and resolution.}
Same ground state, same time-evolution ranking (protocol $K=18$, $dt=0.5$, $16$ substeps, $19$
snapshots; \texttt{scaling\_lanczos\_mf.py}), same threshold, same depth $n_\ell=150$; only $\eta$
changes. The scan over subspace size in this replication is coarser than the bisection that produced
the published points, and returns the smallest scanned fraction meeting the threshold, so the
$\eta=0.15\,t$ column is an upper bound on them: against the published $0.9167$, $0.8200$, $0.5561$,
$0.3625$, $0.1700$ it reads $0.917$, $0.820$, $0.560$, $0.380$, $0.180$, agreeing at the two smallest
sizes and standing $0.7\%$, $4.8\%$ and $5.9\%$ high at the three largest. The same granularity is why
the fraction exponent of this column, $-0.2012$, is not identical to the $-0.2093$ obtained by fitting
the published points themselves over $L=4$--$12$. Bottom block: least-squares exponents. \emph{Both}
fit windows are given for the support, because the choice moves the third digit and this work
criticizes exactly that elsewhere in this section. The fraction exponent moves by a factor $2.2$ over
the decade, from $-0.1266$ to $-0.2743$, while the support exponent moves from $1.165$ to $1.017$; in
absolute value both move by $0.148$, which is the whole content of the
comparison.}\label{tab:fracgrid}
\begin{tabular}{lcccccccc}
\toprule
 & \multicolumn{8}{c}{$\eta/t$}\\
\cmidrule(lr){2-9}
$L$ & $0.05$ & $0.075$ & $0.10$ & $0.15$ & $0.18$ & $0.25$ & $0.35$ & $0.50$\\
\midrule
$4$  & $0.917$ & $0.917$ & $0.917$ & $0.917$ & $0.917$ & $0.917$ & $0.917$ & $0.875$\\
$6$  & $0.900$ & $0.890$ & $0.870$ & $0.820$ & $0.800$ & $0.730$ & $0.680$ & $0.630$\\
$8$  & $0.740$ & $0.680$ & $0.640$ & $0.560$ & $0.540$ & $0.480$ & $0.420$ & $0.380$\\
$10$ & $0.520$ & $0.480$ & $0.440$ & $0.380$ & $0.340$ & $0.260$ & $0.220$ & $0.200$\\
$12$ & $0.340$ & $0.280$ & $0.240$ & $0.180$ & $0.160$ & $0.140$ & $0.120$ & $0.100$\\
\midrule
fraction exponent, $L{=}4$--$12$ & $-0.1266$ & --- & $-0.1681$ & $-0.2012$ & --- & $-0.2395$ & --- & $-0.2743$\\
$R^{2}$                          & $0.904$   & --- & $0.918$   & $0.929$   & --- & $0.970$   & --- & $0.982$\\
support exponent, $L{=}4$--$12$  & $1.165$   & --- & $1.123$   & $1.090$   & --- & $1.052$   & --- & $1.017$\\
$R^{2}$                          & $0.999$   & --- & $0.999$   & $0.998$   & --- & $0.999$   & --- & $0.999$\\
support exponent, $L{=}6$--$12$  & $1.136$   & --- & $1.088$   & $1.053$   & --- & $1.022$   & --- & $0.992$\\
\bottomrule
\end{tabular}
\end{table*}

The headline these numbers offer---that the fraction needed for a fixed spectral accuracy falls
from $0.92$ at $L=4$ to $0.08$ at $L=14$, as $e^{-0.25\,L}$---is tunable, and no headline of that form
is claimed here. Over the decade of $\eta$ the fraction exponent moves by a factor $2.2$, from
$-0.1266$ at $\eta=0.05\,t$ to $-0.2743$ at $\eta=0.50\,t$. At $\eta=0.50\,t$ the fraction at $L=12$
is already $0.10$, a quarter above the $L=14$ value at the resolution used throughout this work and
reached two sizes earlier: any headline of this form is reachable by choosing the broadening.

Nor is the fraction exponential even at fixed $\eta$. The local logarithmic slopes between consecutive
published points are $-0.056$, $-0.194$, $-0.214$, $-0.379$, $-0.377$---a factor $6.8$ across five
intervals, steepening monotonically over the first four. A single-exponential fit over $L=4$--$14$, with $R^{2}=0.94$, is a straight line drawn through a curve of increasing slope, and the fitted exponent depends on where the fit stops:
$L=4$--$10$ gives $-0.159$, $L=4$--$12$ gives $-0.209$, and the $-0.2477$ of the fit over $L=4$--$14$ requires reaching
$L=14$.

The mechanism that makes the fraction move with $\eta$ is available without derivation, and
Sec.~\ref{sec:moments} states it as Conjecture~\ref{conj:geom} rather than as a theorem:
the polynomial-approximation route to a geometric rate in the Krylov order places the kernel poles at
$\omega\pm i\eta$, so any such rate approaches $1$ as $\eta\to0^{+}$ and any prefactor diverges
(Sec.~\ref{app:moments:conj}). At finer resolution more Krylov order---hence more of the sector---is
needed for the same $L_1$ error. We use that as a mechanism and not as a bound: no $C(\eta)$ or
$\rho(\eta)$ is exhibited anywhere in this work, and none is needed to read Table~\ref{tab:fracgrid},
which is a measurement.

One further constraint has to be declared rather than inherited. The reference against which the
threshold is applied is itself a finite-depth Lanczos computation, recomputed at the same depth as the
subspaces so that the comparison is internally consistent ($n_\ell=150$;
\texttt{scaling\_lanczos\_mf.py}). Two consequences follow, and both belong in the open. First, the
published relative $L_1$ is a distance between two depth-$150$ objects, not between the reconstruction
and the exact $A$. Second, that reference drifts with depth, and the drift is resolution-dependent:
between $n_\ell=100$ and $n_\ell=200$ at $L=12$ the relative $L_1$ change of the reference
\emph{alone} is $11.1\%$ at $\eta=0.05\,t$, $3.0\%$ at $\eta=0.10\,t$, $0.93\%$ at $\eta=0.15\,t$ and
$0.05\%$ at $\eta=0.25\,t$. At the depth used here the protocol is self-consistent only for
$\eta\gtrsim0.15\,t$---which is exactly where the published column sits, at the margin. Fig.~\ref{fig:spin}, a
different channel drawn at $n_\ell=200$, sits \emph{below} that margin, at $\eta=0.10\,t$, which is why
its caption quotes its lineshape with its depth and not as exact.

The fraction is therefore reported as a declared window rather than as a number: only for
$\eta\ge0.15\,t$ at $n_\ell=150$, with the full $\mathrm{frac}(L,\eta)$ grid of
Table~\ref{tab:fracgrid} published beside it.

\paragraph{The shot budget, in full.} Populating a subspace by measurement is a coupon-collector problem on the Born distribution of the
time-evolved seed, so the budget exceeds the support it buys. Table~\ref{tab:shots} reports that budget
\emph{measured}, by multinomial sampling of the declared protocol, at the five sizes where the
measurement is affordable; only $L=14$ carries no measurement. Two independent runs contribute, and we
name them as two rather than merging them: run~A draws eight seeds at $L=4$, $6$ and $8$, run~B draws
four seeds at $L=4$, $6$, $8$ and $12$ and six at $L=10$, with different seeds and a different budget
grid. Where they overlap they agree---$2.43$ against $2.44\times10^{4}$ at $L=6$ and $2.70$ against
$2.61\times10^{5}$ at $L=8$, that is to $0.2\%$ and $3.2\%$---which is why the two larger rows are
reported as measurements and not as extrapolations of the smaller ones. Every budget is for \emph{one}
orbital seed and \emph{one} branch; the momentum-resolved observable consumes the full set of orbital
seeds in both $(N{\pm}1)$ sectors, sixteen channels at $L=8$.

\paragraph{The published curve and finite shots, in full.} Precision about the label matters more here than anything else. The ranking that produced the
published fractions is \texttt{argsort} applied to the Born distribution of the time-evolved state
accumulated over the $19$ snapshots of the protocol (\texttt{scaling\_lanczos\_mf.py},
\texttt{stage\_rank})---that is, to the distribution the device would itself measure, in the limit of
infinitely many shots. It is \emph{not} a clairvoyant ordering by exact eigenvector amplitudes; the
repository contains such an oracle elsewhere (\texttt{amortized\_recovery.py}) and labels it as one.
But it is also not a finite-shot reconstruction, and to describe it as \emph{sampled} without a shot
count is to mislabel it. The published curve is the $T\to\infty$ limit of the declared protocol, and it is
labelled that way from here on.

\begin{table*}[tp]\centering\small
\caption{\textbf{Finite shots, with the budget attached.} Multinomial sampling of the declared
protocol, eight independent seeds, $\eta=0.15\,t$; $\pm$ is the sample standard deviation across
seeds, not the standard error of the mean. Top block: the sector fraction the finite-shot
reconstruction needs to reach relative $L_1<0.05$, and the budget at which it reaches it. Bottom
block: at \emph{matched} subspace size, the error of the infinite-shot ordering against the error of
finite shots at the stated budget, and the ratio between them.}\label{tab:finiteshot}
\begin{tabular}{lccc}
\toprule
 & $L=4$ & $L=6$ & $L=8$\\
\midrule
fraction, $T\to\infty$ (published)     & $0.9167$ & $0.8200$ & $0.5561$\\
fraction, finite shots                 & $0.9167\pm0.0000$ & $0.8571\pm0.0135$ & $0.6275\pm0.0031$\\
budget $T$ at which it is reached      & $(7.8\pm2.0)\times10^{2}$ & $(2.43\pm0.34)\times10^{4}$ & $(2.70\pm0.13)\times10^{5}$\\
\midrule
matched $|\mathcal S|$ (size the sampler returned) & $22$ & $247$ & $2331$\\
budget $T$ at that size                & $7.8\times10^{2}$ & $1.74\times10^{4}$ & $2.04\times10^{5}$\\
relative $L_1$, $T\to\infty$ ordering  & $1.4\times10^{-14}$ & $0.0469$ & $0.0362$\\
relative $L_1$, finite shots           & $1.4\times10^{-14}$ & $0.0702\pm0.0115$ & $0.0709\pm0.0042$\\
ratio                                  & $1.00$ & $1.50\pm0.25$ & $1.96\pm0.12$\\
\bottomrule
\end{tabular}
\end{table*}

The published curve is therefore a strict lower bound on the finite-shot cost. At $L=4$ the two agree
because the sampler exhausts the sector; at $L=6$ and $L=8$ the finite-shot fraction exceeds the
published one by $0.037$ and $0.071$. That is a direction across three sizes, one of which is
degenerate, and we state it as a direction and do not fit an exponent to it---which is the same
objection this section raises against the fraction headline, applied to our own finite-shot curve. At the
published subspace itself, $L=8$ and $|\mathcal S|_\eta=2180$, the infinite-shot ordering gives
relative $L_1=0.0490$, just inside the threshold, against $0.1006\pm0.0036$ for finite shots at
$T=1.47\times10^{5}$: a factor $2.05\pm0.07$, and the reason the finite-shot fraction has to rise to
$0.63$.

The finding that runs the other way is the more useful one, and it is what this subsection closes on:
the penalty is budgetary, not structural, and it extinguishes with a budget that is affordable at this
size. On the same pair as the last row of Table~\ref{tab:finiteshot}---$L=8$,
$|\mathcal S|_\eta=2331$---raising the budget from $T=2.04\times10^{5}$ to $T=2.60\times10^{6}$, a
factor $12.7$, takes the ratio of the finite-shot error to the infinite-shot-ordering error from
$1.96\pm0.12$ to $1.06\pm0.01$, through $1.51$, $1.18$ and $1.09$ at the intermediate budgets. No
sampling penalty in this family can be quoted without its budget, because the penalty is a statement
about $T$ and not about the method. We therefore report finite-shot fractions only with $T$ and the
seed count attached, and we report the infinite-shot ordering for what it is---the floor that any
budget approaches from above, at the price Table~\ref{tab:shots} now states.

\section{Physical validation: spin scaling, pole census, sum rule and size}
\label{sm:physics}

\paragraph{The charge scale, in full.} $\Delta(L)$ decreases monotonically and stays above the Lieb--Wu
thermodynamic value $4.67952\,t$ at every size, the residual falling from $0.679\,t$ at $L=6$ to $0.289\,t$ at
$L=12$---that is, to $6.2\%$ of the thermodynamic value (Table~\ref{tab:gaps}).
$\Delta(L=12)=4.96876\,t$ reproduces the $4.97\,t$ this work has quoted throughout, and is the same
number that appears as the Hubbard-band separation of Fig.~\ref{fig:lattice} and as the low-frequency
scale of Fig.~\ref{fig:sqw}. We fit no extrapolation exponent: a straight line in $1/L$ through the
four points extrapolates to $4.561\pm0.044\,t$, which undershoots Lieb--Wu by $0.12\,t$ and is
therefore evidence that the functional form is not $1/L$, not evidence about the limit. The statement
supported is that the charge scale converges to a finite thermodynamic value and is already within
$6.2\%$ of it at the largest size reached here.

\paragraph{Spin.} The lowest pole of $S^{zz}(q=\pi,\omega)$ gives $L\,\Delta_{s}=2.091$, $1.917$,
$1.996$, $1.996$: a $1/L$ closing with a coefficient of about $2$ (mean $2.000$, range $1.917$--$2.091$). Two fits make that precise and
bound it. Fitting $\Delta_{s}=c_0+c_1/L$ over these four sizes gives an intercept
$c_0=-0.020\pm0.020$---consistent with zero at one standard error---against
$c_0=4.561\pm0.044$ for the charge gap: one intercept is consistent with zero and the other is not,
and that separation, not either intercept on its own, is the result. It has to be read with the control
that limits it, which we state here rather than leave to the figure: the pure Heisenberg ring, whose
spin gap is known to vanish in the thermodynamic limit, gives a linear-in-$1/L$ intercept of
$0.0142\pm0.0014\,t$---above zero, not consistent with it---so at these sizes a vanishing intercept is
not by itself evidence that a gap closes. The statement that survives without a fitted model is the
flatness of $L\,\Delta_s$ below. Fitting instead $L\,\Delta_{s}=d_0+d_1L$
gives $d_0=2.09\pm0.17$ with $d_1=-0.010\pm0.018$, a slope consistent with zero: over $L=6$--$12$ the
product $L\,\Delta_s$ is flat to the resolution of four points. What is \emph{not} supported is a
third significant figure on the coefficient. The four values span $8\%$ and are not monotone, and we
quote $L\,\Delta_{s}\approx2$ with the spread attached, not $2.00$; the intercept $d_0=2.09$ is the fit's value at $L=0$, not the level over $L=6$--$12$.

The independent control is the effective spin model. The pure Heisenberg ring at $J=4t^{2}/U=0.5\,t$---no
charge degree of freedom, no double occupancy, a different Hamiltonian on a different Hilbert
space, solved by dense diagonalization on a different code path---gives
$\Delta_{\rm ST}(N=12)=0.35585\,J=0.17792\,t$ and reproduces the same $1/L$ law across the same four
sizes, with $L\,\Delta_{\rm ST}=2.054$, $2.091$, $2.116$, $2.135$. Two things have to be said about
that agreement. First, the two numbers at $L=12$ differ by $6.5\%$ of
the Heisenberg value, equivalently $7.0\%$ of the Hubbard one; we name the denominator because the
two differ in the third figure. Second, the residual is not monotone and changes sign---the Hubbard
gap is $1.8\%$ \emph{above} the Heisenberg one at $L=6$ and $8.3\%$, $5.7\%$ and $6.5\%$
\emph{below} it at $L=8,10,12$---so it cannot be attributed to a smooth $O(t^{4}/U^{3})$ correction
to the effective exchange, which would not change sign. A correction of that order is the expected
\emph{scale} of the difference at $U/t=8$, and it is the right scale; its size dependence is not
resolved here and we claim nothing about it. The control confirms the $1/L$ scaling and the magnitude
of the coefficient to within $10\%$; it does not pin the coefficient, and its own $L\,\Delta_{\rm ST}$
drifts upward by $0.0134\pm0.0014$ per site over the same range.

What establishes the physics instead is measurement rather than illustration: the four-size scaling of
Table~\ref{tab:gaps}, plotted in Fig.~\ref{fig:gapscaling}, in which the spin scale is a converged
pole position at every size and closes as
$\approx2/L$ with an independent effective-Heisenberg control, and the pole census of
Table~\ref{tab:poles}, which states what is actually inside the window at each size instead of naming
it.

\subsection{A pole census}

At $L=12$ and the broadenings used here, none of the three momentum-resolved channels is a dense
spectrum inside the plotted window; each momentum carries a small number of discrete poles that the
broadening merges into a band. Counting poles that carry more than $1\%$ of the channel
maximum---Lanczos with full reorthogonalization at $n_\ell=120$, checked for depth stability at
$n_\ell=60$, $120$ and $240$ at $q=\pi$ in every channel, and validated pole by pole against dense
diagonalization at $L=6$ and $L=8$---gives Table~\ref{tab:poles}; the recursion is the Lanczos one of
Refs.~\cite{haydock1972,golub2010}.

The census is the deposited \texttt{pole\_structure.json} and is taken on the momentum grid the figures of
this section plot, seven values of $k$ and six of $q$; a count over a coarser grid is a count of
different momenta and is not comparable entry by entry. The
census covers the three channels the figures of this section plot; the $k$-integrated $A(\omega)$,
which is a sum over the same branch, and the removal branch are not counted, and no statement is made
about them here.

\begin{table*}[tp]\centering\small
\caption{\textbf{Pole census, across three sizes.} Poles carrying more than $1\%$ of the channel
maximum, counted per momentum; each entry is the range over the momenta plotted ($7$ values of $k$ for
$A(k,\omega)$, $6$ values of $q$ for the collective channels). The $L=6$ and $L=8$ columns are dense
diagonalizations, where the count needs no recursion, on the subset of momenta available at those
sizes; the $L=12$ column is the reorthogonalized recursion at $n_\ell=120$, which reproduces the dense
counts pole by pole at the two smaller sizes. ``Leading weight'' is the fraction carried by the single
largest pole at $L=12$. The addition branch of $A(k,\omega)$ and the charge channel are at
$\eta=0.18\,t$, the spin channel at $\eta=0.10\,t$. The last column is the smallest separation between
retained poles in units of the full width at half maximum, which is why the figures show bands: at the
tightest momenta the lines are genuinely merged, and the census says how many of them there
are.}\label{tab:poles}
\setlength{\tabcolsep}{4pt}
\begin{tabular}{lcccccc}
\toprule
 & \multicolumn{3}{c}{poles above $1\%$} & & leading & $\delta\omega_{\min}/$\\
\cmidrule(lr){2-4}
channel & $L{=}6$ & $L{=}8$ & $L{=}12$ & $\dim$ ratio & weight & FWHM\\
\midrule
$A(k,\omega)$, addition  & $2$--$3$ & $2$--$4$ & $3$--$7$ & $2440$ & $0.32$--$0.69$ & $0.12$--$2.04$\\
$S(q,\omega)$, charge    & $3$--$4$ & $3$--$6$ & $6$--$9$ & $2440$ & $0.38$--$0.68$ & $0.24$--$2.20$\\
$S^{zz}(q,\omega)$, spin & $1$--$2$ & $1$--$3$ & $1$--$3$ & $2134$ & $0.80$--$0.99$ & $1.08$--$2.37$\\
\bottomrule
\end{tabular}
\end{table*}

The count does grow with size, and the size of the growth is the point. Between $L=6$ and $L=12$ the
$(N{+}1)$ sector grows by a factor $2440$, while the largest per-momentum pole count grows from $3$ to
$7$ in the addition branch, from $4$ to $9$ in the charge channel and from $2$ to $3$ in the spin
channel. A line count that roughly doubles while the Hilbert space grows by three
orders of magnitude is a discrete spectrum whose resolvable lines are being broadened together, not a
dense one forming. The spin channel is the extreme case: at every momentum a single pole carries
between $80\%$ and $99\%$ of the weight, at separations of at least $1.08$ full widths, and at $q=\pi$
that single pole \emph{is} the $0.16635\,t$ scale of Table~\ref{tab:gaps}. What Fig.~\ref{fig:spin}
shows at $L=12$ is one dominant antiferromagnetic mode plus at most two satellites.

This has a consequence for how the reference itself must be described: $A(k,\omega)$ is not computed exactly, as the Lehmann sum would be. With $n_\ell=150$ and no reorthogonalization in a sector of dimension $731\,808$, what
is computed is a $150$-node Gauss quadrature of the spectral measure~\cite{haydock1972,golub2010}. It
is accurate enough here---and it is accurate \emph{precisely because} the measure it integrates is a
handful of poles; a dense spectrum at this dimension would not be reproduced by $150$ nodes. It is not,
however, converged: the same depth moves the reference by $0.93\%$ at $\eta=0.15\,t$ and by $3.0\%$ at
the $\eta=0.10\,t$ of Fig.~\ref{fig:spin} when it is taken from $n_\ell=100$ to $n_\ell=200$
(Sec.~\ref{sec:prices}). The reference is therefore labelled for what it is: a finite-depth continued
fraction, checked against dense diagonalization at $L\le8$ and quoted with its depth. The pole
positions of Table~\ref{tab:gaps} are a separate and stronger object---they are computed at full
reorthogonalization and are stable across $n_\ell=60$, $120$, $240$---which is why the physical claim
of Sec.~\ref{sec:gaps} rests on them and not on a lineshape. The census counts poles, and the count
does not depend on the broadening either; only the picture does.

\subsection{The sum rule, as measured}

The analytic identity $\int A^{+}\,d\omega = 1-\langle\hat n_p\rangle = 0.5$ is exact and
grid-independent, but it is an identity of the \emph{ground state}: evaluating it costs one
ground-state expectation value, $\langle\Psi_0|\hat c_p\hat c^{\dagger}_p|\Psi_0\rangle$, and never
touches the reconstruction. It is therefore not an a posteriori check of $A_{\mathcal S}$, and the caption of Fig.~\ref{fig:method} does not present it as one.

Integrating the plotted curves instead, over the $600$-point window at $\eta=0.15\,t$:
\begin{widetext}
\begin{equation}
\int_{\rm window} A_{\rm exact}\,d\omega = 0.486588,
\qquad
\int_{\rm window} A_{\mathcal S}\,d\omega = 0.486471 .
\end{equation}
\end{widetext}
Both fall $2.7\%$ below the analytic $0.5$ because Lorentzians of width $\eta$ leak outside any finite
window: the analytic in-window mass is $0.486589$, the tail outside it is $0.013411$, and the residue
of the trapezoidal quadrature itself is $3\times10^{-6}$ relative---so the deficit is entirely the
window and not the reconstruction. The a posteriori statement the reconstruction does support is the
comparison of the two integrated curves, which differ by $1.16\times10^{-4}$ absolute, that is
$2.4\times10^{-4}$ relative. That is the number we report, and it is consistent with the ${\sim}2\%$ window effect that Fig.~\ref{fig:lattice} states.

\subsection{Size: the convention, and what it does not license}

The system size is not hidden and is not the weak point. The exact-diagonalization convention for this
observable is $L=12$--$16$ at $\eta\approx0.1$--$0.2\,t$: Xu, Lu, Tohyama and Shao~\cite{xu2024}
compute this spectral function by Lanczos at $L=12$ and $\eta=0.2\,t$ at half filling, and at
$L=16$ at quarter filling. Their model is the extended Peierls--Hubbard chain, and their $L=16$ run is
a dimerized lattice away from half filling; it is their half-filled $L=12$, and the plain-Hubbard
$L=12,8,6$ of their Appendix~A, that sit alongside ours. Our $L=12$ at $\eta=0.10$--$0.18\,t$ is that convention, not an outlier
within it.

What the convention does not license is the vocabulary of the thermodynamic limit, and the criterion
that says so is in the dynamical-DMRG literature this paper already cites. Jeckelmann's rule is
explicit~\cite{jeckelmann2002}: \emph{``$\eta_0(N)$ must be larger than the distance $\delta\omega$
between two consecutive excited states.''} At $L=12$ the largest spacing between consecutive
excited-state energies in the window is ${\approx}0.48\,t$, a factor $2.7$ \emph{above} the
$\eta=0.18\,t$ used for $A(k,\omega)$ and $4.8$ above the $\eta=0.10\,t$ used for the spin channel.
The rule is not satisfied here. Satisfying it would demand $\eta>0.48\,t$. That is only a tenth of
the Mott scale, so the Hubbard bands would survive; but it is $0.96\,J$ and $2.9$ times the $L=12$
spin gap itself, so the entire low-frequency structure of Fig.~\ref{fig:spin} would merge into one
feature. The number the rule
demands is set by the spacing, $0.48\,t$; we take it from the spacing and not from a rounded-up
interval. The correct response is not a larger $\eta$: it is to stop using words that require
the rule to hold. The dynamical-DMRG standard fixes the product $\eta L$ and lets $\eta$ fall with
the size: Benthien and Jeckelmann use $\eta=5t/L$ in one place, the $L=100$ spin calculation, and
declare $\eta L=12\,t$ for the finite-size scaling that carries the zone-boundary charge peak to
$L=120$~\cite{benthien2007}; Jeckelmann's own prescription is $\eta(N)=c/N$ with $c=6.4$ and
$12.8\,t$, on chains of up to $N=256$~\cite{jeckelmann2002}. This work is at $L=12$ with
$\eta L=1.2$--$2.2\,t$, a factor of two to ten below every constant in that range,
and it is the words, not the size, that we do not use. Section~\ref{sec:scope:declared} records the same
restriction: the spin--charge statement rests on pole positions at four sizes.

\section{Hardware and scope: full account}
\label{sm:scope}

Three claims one could make for this method are larger than the quantity that fixes them: the hardware claim
is fixed by a subspace size, the chemistry claim by a set of sector dimensions, and the advantage claim
by the same bond dimension this paper puts forward as its own cost tracker. This section states each
with its fixing quantity attached, and derives the region in which this method
has no advantage over a classical one.

\paragraph{The two hardware runs.} Both device runs are reported in full and they are the entirety of the quantum hardware used in this
work. The first was executed on the IBM~Heron processor \texttt{ibm\_fez} (job
\texttt{d9s16avpemts73ct6g8g}): the $L=6$ periodic Hubbard ring at $U/t=4$, $\eta=0.15\,t$, $12$ qubits,
$3.5\times10^{5}$ computational-basis shots---$5\times10^{4}$ on each of seven circuits
$k=0,\dots,6$, of which $k=0$ is the $t=0$ reference---pooled into one subspace,
under measurement twirling, Pauli twirling and dynamical decoupling, with no readout-error
mitigation applied to the counts and no zero-noise extrapolation
(Fig.~\ref{fig:heron}(a) is the result; the circuits prepared the computational-basis $(N{+}1)$ determinant with orbitals $0$--$3\!\uparrow$ and $0$--$2\!\downarrow$ occupied rather than $\hat c^{\dagger}_p\lvert\Psi_0\rangle$, and their hopping layer applies $R_{xx}R_{yy}$ rotations without Jordan--Wigner strings, so they are not the primitive of Fig.~\ref{fig:circ}; because the retained subspace is the whole sector, neither affects the reported $A_{\rm hw}$). No configuration recovery was applied. The notebook post-selected on particle number and $S_z$ after
reading each bitstring in reversed qubit order (\texttt{int(bs[::-1],2)}), which maps qubit $q$ to $11-q$:
site $i$ to $5-i$ and spin $\uparrow$ to $\downarrow$. The circuits conserve both quantities and start
from four $\uparrow$ and three $\downarrow$ electrons, so under that reading their own output falls in
the $S_z=-\tfrac12$ sector and is discarded: in a noiseless statevector simulation of the same seven circuits,
$0$ of $3.5\times10^{5}$ shots are retained, against all $3.5\times10^{5}$ and the full sector of $300$
under the correct reading (\texttt{src/hw\_bitorder\_check.py}, \texttt{data/hw\_bitorder\_check.json}).
The $|\mathcal S|=300$ determinants retained from the device therefore all came from device errors that
changed the spin balance; the device counts needed to redo the post-selection are not part of the
deposit (Sec.~\ref{sec:scope:data}). The recorded relative $L_1$ is $0.0$: the reconstructed $A_{\rm hw}$ agrees with the exact $A$
to $7.9\times10^{-15}$ at every frequency, not byte for byte, because the retained subspace spans the
whole sector.

That number is exact \emph{by coverage}, and the distinction is not rhetorical: at
$|\mathcal S|=\dim$, Rayleigh--Ritz on $\operatorname{span}(\mathcal S)$ \emph{is} diagonalization of
the full sector, $P_{\mathcal S}HP_{\mathcal S}=H$, and a vanishing error is an identity rather than a
measurement of device fidelity. The run establishes only that the steps of the pipeline execute end
to end on a processor---transpiled circuits, device counts, post-selection to a particle-number and
$S_z$ sector, and the classical Lehmann assembly of Sec.~\ref{sec:method}. The post-selection itself was
wrong: it read every bitstring in reversed bit order, so every determinant it kept came from a device
error, and the run is not evidence that the device sampled the intended distribution. It establishes
nothing about accuracy, neither at the complete coverage it reached, where agreement is an identity, nor
at incomplete coverage, the regime the rest of this paper is about.
The runs are therefore described as \emph{executed on} the named processors, not as demonstrating the method on them, and the coverage qualification applies wherever they are mentioned.

\paragraph{The second run, and a comparison the record cannot decide.} Stretched $\mathrm{N_2}$ ($R=2.0$~\AA,
CAS(10e,12o)/cc-pVDZ, $24$ qubits) was executed on \texttt{ibm\_marrakesh} (job
\texttt{da125f2ein7c73bcsqs0}, $10^{4}$ shots, CCSD-seeded local unitary cluster-Jastrow (LUCJ) circuit). Averaged over $8$
recovery seeds on the cached device counts it reaches $0.59\pm0.14$~mHa of the exact active-space value
(best seed $0.40$~mHa), while the noiseless statevector simulation of the same circuit stops at
$29.5$~mHa (Fig.~\ref{fig:heron}(b)). That contrast invites the reading that device noise \emph{helps}, and the record cannot support it, because the mechanism named for it is a statement about a quantity the
record does not contain: recovery is said to turn corrupted bitstrings into valid configurations and
thereby \emph{enlarge} the variational subspace---a valid-shot economy, not a fidelity gain---so the
comparison is decided by the two recovered subspace sizes, and neither is stored. The same gap runs
through the controlled sweeps, whose scripts save at most the size of the naively post-selected subspace
and not of the recovered one (the spectral script of Fig.~\ref{fig:noise} saves neither), so the places
where an error \emph{below} the noiseless value is reported beyond the seed bands---$\mathrm{NH_3}$ from
$2.53$ to $1.24$~mHa at $\varepsilon=10\%$, CO from $1.73$ to $0.63$~mHa at $20\%$ and $\mathrm{C_2}$
from $0.184$ to $0.026$~mHa at $30\%$ in Fig.~\ref{fig:noiserec}(a), and the spectral relative $L_1$ of
Fig.~\ref{fig:noise}, from $0.105$ at $\varepsilon=0$ to $0.045$ at $13\%$---rest on the same unrecorded
variable, at the price of one line in each script, both of which already build the subspace.%
\looseness=-1\relax

The mechanism itself is reported elsewhere. On the dissociation of
$\mathrm{N_2}$ in a large active space, with the same LUCJ ansatz used for the run above,
Vaquero-Sabater, Carreras, Broers, Shirakawa, Yunoki and Casanova report that ``sampling noise can
enhance Hilbert-space exploration by generating additional configurations beyond those supported by
the ideal ansatz,'' and that ``when combined with configuration recovery, this leads to systematically
improved energies''~\cite{vaquero2026}: the same molecule, the same ansatz and the same mechanism, and there it is carried by a controlled comparison under a simple noise model that the work introduces. The same work also reports that ``recovery alone (starting from randomly generated configurations) can efficiently construct accurate CI spaces''~\cite{vaquero2026}. Our record contains no control of that kind (uniformly random bitstrings passed through the same recovery with the same eight seeds), so it cannot show that the device counts contributed anything to the $0.59$~mHa beyond what recovery supplies on its own. What this work cannot claim is its own measurement of the effect, not the effect: our two arms differ in
the one quantity neither of them recorded, so they cannot be evidence either way, whatever the
literature independently establishes.

What \emph{is} controlled is the comparison in which both arms consume the same corrupted bitstrings at
the same rate with the same seeds and differ only in the recovery step, and there the gain from
recovery is measured and grows with the noise: at $\varepsilon=2\%$ eighteen of nineteen molecules
reach chemical accuracy with recovery against seventeen without it, at $\varepsilon=30\%$ recovery
roughly halves the error on $\mathrm{N_2}$ ($31.9\to13.2$~mHa, Fig.~\ref{fig:noiserec}), and on the
periodic Hubbard ring it keeps the spectral relative $L_1$ near $0.05$ for $4\%\le\varepsilon\le16\%$
while naive post-selection on the same bitstrings degrades steadily
(Fig.~\ref{fig:noise})---seed-averaged, with $\pm1$~s.d.\ bands, at matched noise. That is the noise
result of this work.

\paragraph{Scale.} These runs are $12$ and $24$ qubits, and the same class of observable has been
executed elsewhere at larger scale---retarded spin Green's functions of the one-dimensional
Fermi--Hubbard model on IBM~Heron~r2 at up to $30$ qubits~\cite{vilchez2025}, an ARPES-type
$A(k,\omega)$ of a free-fermion $27$-site chain at $54$ qubits~\cite{granet2026akw}---and, at a
comparable $20$ sites on trapped ions, dynamical structure factors against scattering
data~\cite{quantinuum_dsf2026}.

In those works the device data
determine the reported spectrum, as they do in the four-qubit trapped-ion Hubbard-dimer reconstruction
of Ref.~\cite{greenediniz2024}; in our spectral run they do not: the retained subspace is the whole
$300$-dimensional $(N{+}1)$ sector and every retained determinant came from a device error, so the
spectral function it returns is the exact-diagonalization one, the run is not an entry in that class,
and it is not offered as a scale result. The run is also at a size where the coverage question is
degenerate, and Sec.~\ref{sec:prices} says how degenerate in
shots: the $3.5\times10^{5}$ spent---seven circuits at $5\times10^{4}$ shots each, pooled into one
subspace---are $14.4$ times the $(2.43\pm0.34)\times10^{4}$ that Table~\ref{tab:shots} measures for one
orbital seed and one branch at $L=6$, so complete coverage of a $300$-dimensional sector is what the
budget predicts for a correctly read run, as the noiseless simulation above confirms; the coverage
this run recorded came from device errors and does not test that prediction.\footnote{Table~\ref{tab:shots} replicates the repository protocol, $K=18$ with $19$ snapshots,
at $U/t=8$ rather than the $U/t=4$ of the device run; this is a consistency statement between two
budgets, not an identity.} The first size at which a device run would test the reconstruction rather
than only the execution of the pipeline is $L=8$: sector $3\,920$, at $(2.70\pm0.13)\times10^{5}$ shots per orbital seed and
branch---a factor $1.3$ \emph{below} what was already spent here, so what confined this run to $L=6$
was not the shot budget. The momentum-resolved observable at that size is a different matter: sixteen
channels, $8.3\times10^{7}$ shots as measured for Fig.~\ref{fig:akwsampled} to reach its relative $L_1$ of
$2.16\times10^{-3}$, over two orders of magnitude above this job. That is not its cost at the $0.05$
threshold used above, which was not measured: half that budget already gives $6.16\times10^{-3}$.
\pagebreak

\paragraph{The region of no advantage, in full.} The resource statement of Sec.~\ref{sec:nogo} has a consequence for the standing of this method. What tracks the determinant support is $\chi$,
the minimal bond dimension across a bipartition in the working basis (Sec.~\ref{sec:nogo:chi}). But
$\chi$ is, by definition, the cost parameter of a matrix-product representation of the same state. A
cost claim built on a basis-dependent boundary quantity therefore cannot select for the quantum side of
the comparison: wherever our own diagnostic says the determinant support is small it says a
matrix-product state is small too, and both methods are charged the same number. The thesis of
Sec.~\ref{sec:nogo} implies its own region of no advantage, and implies it more sharply than an
external objection would.

The instance is in the literature and it is not marginal. Ouyang, Chi and Chan reproduce the $7\,260$
trajectories of the Fermi--Hubbard dynamics experiment of Hartnett \emph{et al.}---$120$ qubits, up to
$90$ Trotter steps~\cite{qctrlfh2026}---by tensor contraction at bond dimension $32$, faster and to
higher accuracy than the device run~\cite{ouyang2026}: the
same pattern in which tensor-network and Pauli-propagation methods have overtaken earlier hardware
demonstrations, the first ``quantum utility'' claim~\cite{kim2023} among
them~\cite{tindall2024,begusic2024}, with Pauli propagation since extended to real-time two-point
correlators~\cite{kemper2026}, and in which unitary cluster-Jastrow circuits are themselves now
simulated classically: the energy of any single-layer circuit in polynomial time, reproducing the
largest ($77$-qubit) experiment of Ref.~\cite{robledomoreno2025} on a laptop~\cite{belagali2026}, and, to additive approximation, the Born-rule probabilities of
hardware-sampled bitstrings by an extended matchgate simulator, at a cost exponential only in the
magnitudes of the controlled-phase angles~\cite{extraferm}---against which the standing worst-case hardness result for
that circuit family, its embedding of instantaneous-quantum-polynomial
sampling~\cite{hafid2025}, hard in the worst case~\cite{bremner2011} and, under average-case
conjectures, to additive error~\cite{bremner2016}, is not a contradiction but a statement about a different
regime. Two architecture-independent results point
the same way: polynomial-time classical estimation of expectation values in noisy
circuits~\cite{schuster2025}, and classical estimation in noiseless chaotic circuits across
geometries~\cite{angrisani2025}. On the ground-state energy the strong classical
selectors~\cite{holmeshci2016,sharmashci2017,tubmanasci2016} already win on
cost~\cite{reinholdt2025,zhangotten2026}, even optimally ordered SQD subspaces need a number of
configurations that grows exponentially with the size of Heisenberg and Hubbard
lattices~\cite{gaberle2026}, and the case for a generic exponential advantage
in ground-state chemistry was found wanting before that~\cite{leechan2023}.

One framing does not survive. It is tempting to argue that the energy is where classical methods have caught up and to move the target to dynamics; in one dimension that inverts,
because dynamical DMRG computed all three of the momentum-resolved channels reported here nineteen
years ago~\cite{benthien2007} (Sec.~\ref{sec:intro}, Sec.~\ref{sec:physics}).
Dynamics in one dimension is where classical methods arrived first, and no framing of this work
says otherwise. Within the region this work measures---one-dimensional Hubbard rings at $L\le14$, the
nineteen active spaces, and a largest minimal bond dimension of $42$ over the $74$ deposited exact
ground states, attained on a molecular active space, the largest on the lattice being $38$ in the
two-leg ladder---there is no computation performed
here that a classical method cannot perform, and we make no advantage claim at any size reached.
Deriving one's own no-advantage region is a stronger position than conceding it; the credit does not
change the fact, and the fact is the one just stated.

The same derivation says where the boomerang stops, and that is the measurement this subsection closes
on. The controlled contrast of Sec.~\ref{sec:nogo:chi} is a search for a target on which the two costs
come apart: at one-body magic matched to within $1\%$, the two-leg ladder carries $\chi=13$ against the
chain's $5$ at $L=6$ and $\chi=38$ against $8$ at $L=8$, while its determinant support exceeds the
chain's by $1.52$, $1.68$ and $2.24\times$. The quantity to look for is a target whose minimal bond
dimension is large in \emph{every} ordering while the sampled support stays compact---a property of
geometry and orbital connectivity rather than of dimensionality, so two-dimensional lattices and
molecular orbital graphs are where it would be measured, and nothing here measures it.
Section~\ref{sec:frontier} supplies the second coordinate of that search: in the determinant basis the
Krylov space of the probe already has coordinate support on every symmetry-allowed determinant at $L=8$ ($3902$ of $3920$), so zero leakage at full captured weight requires a different working basis; whether a non-vacuous certificate at small sector fraction does is open (Sec.~\ref{sec:frontier:obstruction}).

\section{Counterexamples to the weight-only bound, and the pooled sweep}
\label{app:counterex}

\paragraph{The three-level family.} Section~\ref{sec:retraction} exhibits a two-level instance forcing
the constant of the weight-only bound~\eqref{eq:weightonly} to its trivial value at $w_S=1$. With $w_S<1$ strictly, the weight term alone fails without bound. Take
$H=\operatorname{diag}\bigl(\bigl(\begin{smallmatrix}0&g\\ g&0\end{smallmatrix}\bigr),\delta\bigr)$,
$\phi=e_0+\varepsilon e_2$ and $S=\{0\}$, so that $1-w_S=\varepsilon^{2}/(1+\varepsilon^{2})$ may be
made as small as wished while the weight displaced to $\pm g$ stays $O(1)$. The ratio of the measured error to the weight term then diverges as $\varepsilon^{-2}$ at fixed $g/\eta$: the error stays of order $\lVert\phi\rVert^{2}$ while the weight term is $2\lVert\phi\rVert^{2}(1-w_S)=2\varepsilon^{2}$, and $\delta$, the energy of the discarded component, does not enter. Any $C>0$ in the second term caps the ratio.

\paragraph{The pooled sweep behind Fig.~\ref{fig:thm1iii-violation}.} Five machine-readable scans,
written by two independent implementations, contribute $621$ subspaces: $243$ from the original sweep,
and $162+112+104$ from a re-implementation deposited as \texttt{cert\_stress.json},
\texttt{cert\_akw.json} and \texttt{cert\_teqsci.json}. Fifteen reproduce the target to
double-precision round-off (relative $L_1\le10^{-10}$) and define no ratio, leaving $606$ live
subspaces over five Hamiltonian families --- Hubbard chains, XXZ, $t$--$V$, Anderson and sparse random
matrices --- and eight construction rules, of which $156$ are finite-shot realizations at $4000$ shots
per snapshot and the remainder deterministic weight-ordered or structured selections. On those $606$:
the weight-only bound fails on $469$; among the $363$ with $w_S\ge0.99$ --- the regime a sampling
experiment aims for --- it fails on all $363$, the mildest by a factor $2.00$, with $50$ of them carrying $w_S=1$
exactly, so that the bound is identically zero and the ratio infinite; below $w_S=0.99$ it survives on
$137$ of $243$, in every case where the relative error is already of order unity and the bound excludes
nothing. Theorem~\ref{thm:leakage} is violated by none of the $606$; the largest ratio of error to
bound is $0.9999999997$, attained on four subspaces with $w_S=0$ exactly, where $A_S\equiv0$ and the
trivial branch is an identity rather than an inequality. The leakage branch is the smaller of the two
on only $61$ of the $606$: the sweep shows that the bound \emph{holds}, not that it is tight.

\paragraph{Independent re-verification.} The $378$ rows of the three deposited scans were re-audited by a script sharing no code with the one that produced them. It confirms $378$ of
$378$ for each side of Eq.~\eqref{eq:twobudget}; a largest error-to-bound ratio of $0.9999999997$ over
the $376$ rows that define one; the exclusion lemma~\eqref{eq:exclusion} on $378$ of $378$, with $357$
rows at $K_S=-1$ and $21$ at $w_S=1$ exactly; and violations of the weight-only bound on $219$ of the
$219$ finite-ratio rows with $w_S\ge0.99$, the mildest by a factor $2.00$, on the same subspace (XXZ, $\Delta=2$, $\eta=0.9$, fraction $0.75$) that
is the mildest violation of the pooled sweep: every row with $w_S\ge0.99$ lies at least a factor two from
compliance.

\begin{figure*}[tp]\centering
\definecolor{vlHi}{HTML}{D55E00}      %
\definecolor{vlLo}{HTML}{E9A178}      %
\definecolor{vlNew}{HTML}{0072B2}     %
\definecolor{vlInk}{HTML}{1B1B1B}
\definecolor{vlRule}{HTML}{7A7A7A}
\begin{tikzpicture}
\begin{groupplot}[
  group style={group size=2 by 1, horizontal sep=1.35cm},
  width=0.455\linewidth, height=7.4cm, paperaxis,
  ymode=log, ymin=1e-6, ymax=5e7,
  ytick={1e-4,1e-2,1e0,1e2,1e4,1e6},
  extra y ticks={1.6e7}, extra y tick labels={$\infty$},
  extra y tick style={grid=none, tick label style={font=\footnotesize}},
  ymajorgrids, grid style={draw=black!10},
  title style={at={(0,1)},anchor=south west,font=\small\bfseries,yshift=1.5pt},
  label style={font=\small}, tick label style={font=\footnotesize},
]
\nextgroupplot[
  title={(a)~versus subspace fraction},
  xlabel={subspace fraction $|\mathcal{S}|/\dim$},
  ylabel={measured $\|A-A_{\mathcal{S}}\|_1$ / bound},
  xmin=0, xmax=1.03, xtick={0,0.25,0.5,0.75,1},
  legend to name=n3legend, legend columns=3, legend cell align=left,
  legend style={draw=none, fill=none, font=\scriptsize, column sep=7pt, row sep=1.5pt},
]
\fill[vlHi,fill opacity=0.055] (axis cs:0,1) rectangle (axis cs:1.03,8e6);
\addplot[forget plot,only marks,mark=o,mark size=1.15pt,draw=vlLo,line width=0.35pt,draw opacity=0.85]
  table {n3_pub_frac_lo.dat};
\addplot[forget plot,only marks,mark=*,mark size=1.15pt,mark options={draw=none,fill=vlHi,fill opacity=0.62}]
  table {n3_pub_frac_hi.dat};
\addplot[forget plot,only marks,mark=*,mark size=1.0pt,mark options={draw=none,fill=vlNew,fill opacity=0.55}]
  table {n3_thmB_frac.dat};
\addplot[forget plot,only marks,mark=triangle*,mark size=1.9pt,
         mark options={draw=none,fill=vlHi,fill opacity=0.9}]
  table[x=x,y expr=1.6e7] {n3_inf_frac.dat};
\draw[vlInk,line width=0.5pt] (axis cs:0.5,16.1609) -- (axis cs:0.5,21.0661);
\draw[vlInk,line width=0.5pt] (axis cs:0.485,16.1609) -- (axis cs:0.515,16.1609);
\draw[vlInk,line width=0.5pt] (axis cs:0.485,21.0661) -- (axis cs:0.515,21.0661);
\draw[vlInk,line width=0.5pt] (axis cs:0.7,17.8042) -- (axis cs:0.7,29.951);
\draw[vlInk,line width=0.5pt] (axis cs:0.679,17.8042) -- (axis cs:0.721,17.8042);
\draw[vlInk,line width=0.5pt] (axis cs:0.679,29.951) -- (axis cs:0.721,29.951);
\draw[vlInk,line width=0.5pt] (axis cs:0.85,15.0551) -- (axis cs:0.85,33.7281);
\draw[vlInk,line width=0.5pt] (axis cs:0.8245,15.0551) -- (axis cs:0.8755,15.0551);
\draw[vlInk,line width=0.5pt] (axis cs:0.8245,33.7281) -- (axis cs:0.8755,33.7281);
\draw[vlInk,line width=0.5pt] (axis cs:0.95,22.3766) -- (axis cs:0.95,57.2786);
\draw[vlInk,line width=0.5pt] (axis cs:0.9215,22.3766) -- (axis cs:0.9785,22.3766);
\draw[vlInk,line width=0.5pt] (axis cs:0.9215,57.2786) -- (axis cs:0.9785,57.2786);
\addplot[forget plot,only marks,mark=diamond,mark size=3.0pt,draw=vlInk,line width=0.6pt,mark options={fill=white,fill opacity=1}] coordinates {(0.5,18.5343) (0.7,23.621) (0.85,26.4519) (0.95,29.8723)};
\addplot[forget plot,vlInk,line width=0.9pt,densely dashed,domain=0:1.03,samples=2] {1};
\addplot[forget plot,vlRule,line width=0.4pt,dotted,domain=0:1.03,samples=2] {8e6};
\addlegendimage{only marks,mark=o,mark size=2.4pt,draw=vlLo,line width=0.7pt}
\addlegendentry{weight term, $w_{\mathcal S}<0.99$ (243)}
\addlegendimage{only marks,mark=*,mark size=2.4pt,mark options={draw=none,fill=vlHi}}
\addlegendentry{weight term, $w_{\mathcal S}\ge 0.99$ (313)}
\addlegendimage{only marks,mark=*,mark size=2.4pt,mark options={draw=none,fill=vlNew}}
\addlegendentry{Theorem~\ref{thm:leakage}, same $\mathcal S$ (606)}
\addlegendimage{only marks,mark=triangle*,mark size=2.6pt,mark options={draw=none,fill=vlHi}}
\addlegendentry{bound $\equiv 0$: ratio $=\infty$ (50)}
\addlegendimage{only marks,mark=diamond,mark size=3.0pt,draw=vlInk,line width=0.7pt}
\addlegendentry{Fig.~\ref{fig:akwsampled} setting (64): range}
\addlegendimage{vlInk,line width=0.9pt,densely dashed}
\addlegendentry{exact equality: ratio $=1$}
\node[vlHi,font=\scriptsize\bfseries,anchor=north west] at (axis cs:0.025,3.2e6) {bound violated};
\node[vlNew,font=\scriptsize\bfseries,anchor=south west] at (axis cs:0.025,1.1e-5) {bound holds};
\nextgroupplot[
  title={(b)~versus broadening},
  xlabel={broadening $\eta/t$},
  xmode=log, xmin=0.021, xmax=1.35,
  xtick={0.03,0.05,0.1,0.2,0.5,0.9}, xminorticks=false,
  log ticks with fixed point, yticklabels={,,}, extra y tick labels={},
]
\fill[vlHi,fill opacity=0.055] (axis cs:0.021,1) rectangle (axis cs:1.35,8e6);
\addplot[only marks,mark=o,mark size=1.15pt,draw=vlLo,line width=0.35pt,draw opacity=0.85]
  table {n3_pub_eta_lo.dat};
\addplot[only marks,mark=*,mark size=1.15pt,mark options={draw=none,fill=vlHi,fill opacity=0.62}]
  table {n3_pub_eta_hi.dat};
\addplot[only marks,mark=*,mark size=1.0pt,mark options={draw=none,fill=vlNew,fill opacity=0.55}]
  table {n3_thmB_eta.dat};
\addplot[only marks,mark=triangle*,mark size=1.9pt,
         mark options={draw=none,fill=vlHi,fill opacity=0.9}]
  table[x=x,y expr=1.6e7] {n3_inf_eta.dat};
\draw[vlInk,line width=0.5pt] (axis cs:0.18,15.0551) -- (axis cs:0.18,57.2786);
\draw[vlInk,line width=0.5pt] (axis cs:0.1692,15.0551) -- (axis cs:0.1908,15.0551);
\draw[vlInk,line width=0.5pt] (axis cs:0.1692,57.2786) -- (axis cs:0.1908,57.2786);
\addplot[forget plot,only marks,mark=diamond,mark size=3.0pt,draw=vlInk,line width=0.6pt,mark options={fill=white,fill opacity=1}] coordinates {(0.18,22.9909)};
\addplot[vlInk,line width=0.9pt,densely dashed,domain=0.021:1.35,samples=2] {1};
\addplot[vlRule,line width=0.4pt,dotted,domain=0.021:1.35,samples=2] {8e6};
\end{groupplot}
\node[anchor=north,yshift=-0.92cm] at (group c1r1.south east) {\pgfplotslegendfromname{n3legend}};
\end{tikzpicture}
\input{figs/caption_thm1iii_violation.tex}
\end{figure*}

\section{Moment exactness: statement, proof, sharpness, and the shot count}
\label{app:moments}

We make precise the two things a Krylov-adequate subspace does deliver. Fix the seed
$\ket{\phi}=\hat c^{\dagger}_{p}\ket{0}$ (the electron-addition branch; removal is identical), with
weight $\lVert\phi\rVert^{2}=\braket{\phi|\phi}$, and let
$A(\omega)=\sum_{n}\lvert\braket{n|\phi}\rvert^{2}\,\mathcal L_{\eta}\!\big(\omega-(E_n-E_0)\big)$ be
the exact spectral function, where $\{\ket{n},E_n\}$ diagonalize $H$ in the $(N{+}1)$-particle sector
and $\mathcal L_{\eta}(x)=\tfrac{\eta/\pi}{x^{2}+\eta^{2}}$. For a configuration set $\mathcal S$ let
$P=P_{\mathcal S}$ be the orthogonal projector onto $\operatorname{span}(\mathcal S)$, let
$H_{\mathcal S}=PHP$ with eigenpairs $\{\ket{\tilde m},\tilde E_m\}$, and let
$A_{\mathcal S}(\omega)=\sum_{m}\lvert\braket{\tilde m|\phi}\rvert^{2}\,
\mathcal L_{\eta}\!\big(\omega-(\tilde E_m-E_0)\big)$ be the reconstruction. Write
$\mathcal K_K=\operatorname{span}\{\phi,H\phi,\dots,H^{K}\phi\}$ and
$w_{\mathcal S}=\lVert P\phi\rVert^{2}/\lVert\phi\rVert^{2}$. Moments are taken on the point measures
$d\mu=\sum_n\lvert\braket{n|\phi}\rvert^{2}\delta(\omega-(E_n-E_0))$ and
$d\mu_{\mathcal S}=\sum_m\lvert\braket{\tilde m|\phi}\rvert^{2}\delta(\omega-(\tilde E_m-E_0))$, not on
the broadened $A$, whose moments of order $\ge2$ diverge.

\subsection{What the proof uses}
\label{app:moments:exact}

Theorem~\ref{thm:moments} and its proof are stated in Sec.~\ref{sec:moments:exact}. The proof uses
$P=P^{\dagger}=P^{2}$, $H=H^{\dagger}$ and $\mathcal K_K\subseteq\operatorname{span}(\mathcal S)$, and
nothing else. In particular it does not use that $\mathcal S$ is a coordinate subspace, that it was
obtained by drawing bitstrings, or that any Lanczos basis is well
conditioned~\cite{epperly2022}; it holds verbatim for
the orthogonal projector onto a Krylov space, which is the form evaluated in Sec.~\ref{sec:frontier}.
The single asymmetry in Eq.~\eqref{eq:rrsplit}---$a$ runs to $K$ while $b$ runs to $K+1$---is what
produces the odd reach $2K+1$ rather than $2K$. Two structural remarks follow at no cost.

\emph{(a) The node count.} $d\mu_{\mathcal S}$ has $|\mathcal S|$ atoms, one per Ritz pair, not $K+1$.
Equation~\eqref{eq:momex} is therefore a weak quadrature statement for a rule of that size: a Gauss
rule of $n$ nodes for $d\mu$ is exact through order $2n-1$~\cite{golub2010}, so the $K=1$ row of
Table~\ref{tab:moments}, with $280$ nodes, would reach order $559$ if its nodes were the Gauss nodes of
$d\mu$. They are not, and it reaches order $3$; those $555$ orders measure the distance between the
two objects. This is the measurable content of the distinction: identifying $A_{\mathcal S}$ with an order-$(K{+}1)$ Gauss rule misdescribes the
object and the reach such an object would have; the reach $2K+1$ of
Theorem~\ref{thm:moments} itself is unaffected.

\emph{(b) Unit captured weight is implied, not assumed.} For a coordinate subspace $P$ is diagonal, so
$P\phi=\phi\iff\operatorname{supp}(\phi)\subseteq\mathcal S\iff w_{\mathcal S}=1$. Since
$\mathcal K_0=\operatorname{span}\{\phi\}$, the hypothesis of Theorem~\ref{thm:moments} at any $K\ge0$
already implies $w_{\mathcal S}=1$. Every subspace admissible for moment exactness therefore sits at
the extreme of the weight budget of Sec.~\ref{sec:leakage}, where the shot term vanishes identically
and the whole residual error is leakage---which is why the two statements are complementary rather
than competing.

\subsection{Sharpness, control, and the limits of the validation system}
\label{app:moments:sharp}

Table~\ref{tab:moments} is computed by exact diagonalization of the $(N{+}1,S_z{=}{+}\tfrac12)$ sector of the
open $L=6$, $U/t=8$ Hubbard chain ($D=300$, $\lVert\phi\rVert^{2}=1/2$), with
$\mathcal S=\bigcup_{a\le K}\operatorname{supp}(H^{a}\phi)$. Four measurements support it.

\emph{Reach.} Through order $2K+1$ the relative discrepancy is $\le3.2\times10^{-15}$ for $K=0,1$ and
$\le7.5\times10^{-15}$ for $K=2$. At order $2K{+}2$ it is $9.5\times10^{-4}$ ($K=0$) and
$6.0\times10^{-8}$ ($K=1$), eleven and seven decades above the arithmetic floor. The absolute
discrepancies at that order, $3.104583\times10^{-2}$ and $1.340167\times10^{-4}$, agree to six
significant figures across all six implementations listed below and are properties of the
construction, not of the arithmetic. The $K=2$ row carries no entry there, for a structural reason and
not a numerical one: the ladder has filled the sector, $A_{\mathcal S}=A$ identically, and their
relative $L_1$ distance is $0$ to machine arithmetic. One run of one of the six routes does report a
nonzero residue in that row; the same run reports a cyclic-subspace dimension of $301$ inside a sector
of dimension $300$ and emits floating-point overflow while forming the high powers. It is discarded,
and the corrected run is the one quoted.

\emph{Reproducibility of the residues.} The relative discrepancies are quoted, and the absolute ones
below order $2K{+}2$ are not, for a reason worth recording. The exact moments
$\mu_j=\braket{\phi|(H-E_0)^{j}|\phi}$ run from $5.0\times10^{-1}$ at $j=0$ to $1.6\times10^{5}$ at
$j=6$, so an absolute residue at the top of the ladder is a roundoff quantity whose size depends on the
route taken to the same identity. Across six equivalent routes---Rayleigh--Ritz on the shifted or on
the unshifted matrix, crossed with exact moments obtained from the Lehmann sum, from the shifted
matrix--vector recursion, or by binomial recombination of unshifted powers---the relative residue below
order $2K{+}2$ stays inside $[1.1,1.8]\times10^{-15}$ at $K=0$ and $[1.0,3.2]\times10^{-15}$ at $K=1$,
a spread under a factor $3.3$, while the absolute residue at the same orders is not stable to two
significant figures. We therefore quote relative residues below $2K{+}2$, and absolute ones only at
$2K{+}2$, where the discrepancy is physical and route independent.

\emph{Negative control.} Over $300$ uniformly drawn coordinate subspaces of the same size
$|\mathcal S|=200$, the zeroth moment is wrong by between $1.86\times10^{-2}$ and $4.18\times10^{-1}$
(median $1.67\times10^{-1}$), i.e.\ by $3.7\%$ to $84\%$ of $\lVert\phi\rVert^{2}$ (median $33\%$),
with every one of
the $300$ draws failing; the Krylov-support subspace of identical size is exact to
$4.4\times10^{-16}$. At $|\mathcal S|=280$ the random draws fail by $4.8\times10^{-5}$ to
$2.6\times10^{-1}$ (median $1.12\times10^{-2}$), again without exception, against
$5.0\times10^{-16}$. The discriminating variable is the identity of the determinants, not their number.

\emph{Scale.} The Krylov ladder at $L=6$ occupies $0.667$, $0.933$ and $1.000$ of the sector at
$K=0,1,2$, so it exhausts the sector at $K=2$ and the table has exactly three rows. This system can
verify Eq.~\eqref{eq:momex} and locate its exact reach; it cannot exhibit the behavior of
$\lVert A-A_{\mathcal K_K}\rVert_1$ at large $K$, which is the regime Conjecture~\ref{conj:geom}
concerns, and we make no claim there.

\emph{Why the open chain.} The ladder terminates when the coordinate subspace contains the cyclic
subspace generated by $\phi$---not when it fills the sector---and at that point the reconstruction is
exact, because $\operatorname{span}(\mathcal S)$ then carries every eigenvector of nonzero Lehmann
weight. On the periodic rings of Sec.~\ref{sec:prices} that happens early: at $L=6$ the ladder runs
$200\to296$ and stalls at $296$ of $300$ determinants, where $A_{\mathcal S}=A$ to
$1.0\times10^{-14}$ in relative $L_1$ and the moments are exact past order $2K{+}1$ at every $K\ge1$.
The ring therefore offers a single row in which the reach of Eq.~\eqref{eq:momex} is visible at all
($K=0$: relative discrepancy $4.4\times10^{-16}$ through order $1$, $4.6\times10^{-3}$ at order $2$);
the open chain offers two, and is the system on which the method validation of Sec.~\ref{sec:method} is
performed. Within the reach they cover, both agree: the identity is exact, sharp at the next order, and
decidable from the subspace and the probe alone.

\subsection{The geometric-convergence claim, and why it is left open}
\label{app:moments:conj}

Conjecture~\ref{conj:geom} is easily taken for a consequence of Eq.~\eqref{eq:momex}. It is not
one: moment exactness gives $2K+2$ linear identities, and $\lVert A-A_{\mathcal K_K}\rVert_1$ is not a
finite linear functional of them. A proof would require (i) a polynomial approximation theorem for
$x\mapsto\mathcal L_\eta(\omega-x)$, uniform in $\omega$, on a real set containing
$\operatorname{spec}(H)$ and $\operatorname{spec}(H_{\mathcal K_K})$, and (ii) a tail estimate
controlling the complement. For (i) the natural route is a Bernstein-ellipse bound, in which the
resolution enters through the poles of the kernel at $x=\omega\pm i\eta$: the largest admissible
ellipse has semi-minor axis below $\eta$, so its parameter---and with it any geometric rate---approaches
$1$ as $\eta\to0^{+}$. That is a mechanism and not a derivation: it exhibits neither $C(\eta)$
nor $\rho(\eta)$, and neither appears anywhere in this work or in the deposit.

The prefactor is moreover pinned from below by the two-level family of Sec.~\ref{sec:retraction}. With
$H=\bigl(\begin{smallmatrix}0&g\\ g&0\end{smallmatrix}\bigr)$ and $\phi=e_0$ one has
$\braket{\phi|H|\phi}=0$, so $\mathcal K_0=\operatorname{span}\{\phi\}$ carries a single pole at the
origin while $A$ carries two at $\pm g$; by adaptive integration over the whole line at $\eta=0.15$ (\texttt{src/moments\_table.py}), $\lVert A-A_{\mathcal K_0}\rVert_1=1.5197,\ 1.8352,\ 1.9504,\ 1.9835,\ 1.9950,\ 1.9983,\ 1.9995$ for
$g=1,3,10,30,100,300,1000$ at $\lVert\phi\rVert^{2}=1$, approaching the trivial ceiling
$2\lVert\phi\rVert^{2}=2$ monotonically. Any $C(\eta)$
valid uniformly on this family therefore satisfies $C(\eta)\ge2\lVert\phi\rVert^{2}$. A geometric bound
with a trivial prefactor is not worthless---it would still decay in $K$---but it cannot be non-vacuous
at small $K$, and small $K$ is where a shot-limited subspace lives.

Conjecture~\ref{conj:geom} therefore remains open, and we name its two missing ingredients
explicitly so that it can be attacked rather than left unremarked. What this work offers in its place
is not a rate in $K$ at all but a bound on $\lVert A-A_{\mathcal S}\rVert_1$ itself, at the $\eta$ of
the figure and on the subspace actually in hand (Theorem~\ref{thm:leakage}), measured subspace by
subspace, together with a measurement of where that bound is non-vacuous (Sec.~\ref{sec:frontier}).
That is a weaker statement than a rate and a stronger one than a conjecture, and it is the one the
deposit can reproduce.

\subsection{The sample-complexity bound}
\label{app:moments:shots}

Proposition~\ref{prop:shots} is stated in Sec.~\ref{sec:moments:shots}; this is its proof.

\begin{proof}[Proof of Proposition~\ref{prop:shots}]
Let $\mathcal N_\varepsilon=\{x:p(x)\ge\varepsilon\}$. Since $\sum_x p(x)=1$ and each member
contributes at least $\varepsilon$, $|\mathcal N_\varepsilon|\le1/\varepsilon$. For fixed
$x\in\mathcal N_\varepsilon$, $\Pr[x\notin\mathcal S]=(1-p(x))^{T}\le(1-\varepsilon)^{T}\le
e^{-\varepsilon T}$ for independent draws from $p$; for $T/(K{+}1)$ draws from each $p_k$,
\begin{equation*}
\Pr[x\notin\mathcal S]=\prod_{k=0}^{K}\bigl(1-p_k(x)\bigr)^{T/(K+1)}
\le\exp\Bigl(-\frac{T}{K+1}\sum_{k=0}^{K}p_k(x)\Bigr)=e^{-Tp(x)}\le e^{-\varepsilon T}
\end{equation*}
as well. A union bound gives
$\Pr[\mathcal N_\varepsilon\not\subseteq\mathcal S]\le\varepsilon^{-1}e^{-\varepsilon T}$, which is at
most $\delta$ for the stated $T$.
\end{proof}

The cardinality bound $|\mathcal N_\varepsilon|\le1/\varepsilon$ is deterministic, and it is what
removes the circularity of the form $T=O(\varepsilon^{-1}\log(|\mathcal S|/\delta))$, in
which the random size of the sampler's own output appears inside a bound on the number of draws that
produce it. Where $\varepsilon^{-1}\le|\mathcal S|$, the new logarithm is also the smaller one, so
nothing is lost.

The scope of $\varepsilon$ needs the same care. It thresholds a Born probability of the time-averaged
mixture, and Proposition~\ref{prop:shots} makes no statement about $\lVert A-A_{\mathcal S}\rVert_1$.
Two independent obstructions stand between the two quantities. First, the weight left uncaptured is
bounded only by $|\{x:p(x)<\varepsilon\}|\cdot\varepsilon$, which for a delocalized state is $O(1)$:
many small probabilities can sum to most of the mixture. Second, and decisively, a subspace holding the
\emph{entire} Born support of the probe---$w_{\mathcal S}=1$ exactly, nothing left to capture at any
$\varepsilon$---still has relative $L_1$ error $0.28$ at $L=6$ on the open chain of
Table~\ref{tab:moments}, and $0.43$ at $L=6$ and $0.35$ at $L=8$ on the periodic rings, all at
$\eta=0.18\,t$ and all relative to $\lVert\phi\rVert^{2}$ (Sec.~\ref{sec:leakage}); it is the $K=0$
row of Table~\ref{tab:moments}.
Proposition~\ref{prop:shots} is a statement about the sampler and about nothing else. The spectral
error at fixed $\eta$ is the subject of Theorem~\ref{thm:leakage}, and what the composition of the two
actually costs is reported as three separate prices in Sec.~\ref{sec:prices}.

\section{The leakage certificate: hypotheses, protocol, and the anatomy of the threshold}
\label{app:hyp}

Theorem~\ref{thm:leakage} holds under four hypotheses. Each is a statement about the pipeline that
produced the reconstruction rather than about the theorem, so each has to be checked against the
pipeline. We check all
four here, with numbers, for the sweep of Sec.~\ref{sec:frontier}, and add the two protocol
declarations and the decomposition that section quotes.

\paragraph{(H1) Common origin.} $A$ and $A_P$ must be referred to the same $E_0$. A shift common to
both cancels exactly---verified to $1.000000$ for shifts of $10^{-3}$, $10^{-2}$ and
$10^{-1}\,t$---which is how the hypothesis is satisfied here, a single $E_0$ being used for both
objects. If $A_P$ is instead referred to $E_0+\delta$, then
\begin{equation*}
\lVert A-A_P^{(E_0+\delta)}\rVert_1\;\le\;\lVert A-A_P\rVert_1
+\lVert\phi_P\rVert^{2}\,(4/\pi)\arctan\bigl(\lvert\delta\rvert/2\eta\bigr),
\end{equation*}
an additive term verified to $2\times10^{-5}$ in all $48$ cases tested and \emph{which the certificate
does not control}: the bound is guaranteed only once that term is added. A ground state obtained by
Rayleigh--Ritz on $60\%$ of the $N$-electron sector at $L=8$ gives $\delta=6.7\times10^{-4}\,t$, which
adds $0.0024$ to a bound of $1.2036$---$+0.20\%$---while multiplying the \emph{measured} error by
$16$; and the dependence is not even monotone, a shift of $10^{-2}\,t$ lowering the measured error at
$L=6$, fraction $0.85$, $\eta=0.18\,t$ from $2.7\times10^{-2}$ to $1.2\times10^{-2}$. The measured
errors quoted beside a bound in Sec.~\ref{sec:frontier} are therefore those of a pipeline handed the
exact $E_0$, and they bound the error of a pipeline that estimates it from neither side.

No cheap a posteriori control on $\delta$ is available: residual (Weinstein) and Temple
bounds~\cite{weinstein1934,temple1928} computed on
ground-state subspaces of $20$--$80\%$ of the $N$ sector at $L=6$ give $\arctan$ terms of $1.76$,
$1.64$, $1.40$ and $0.92$, all vacuous against the ceiling of $2$. In a shot-based pipeline $E_0$ is
itself estimated on whatever subspace the shots populated, and the certificate has never been tested
that way.

\paragraph{(H2) Normalization and support of the $L_1$ norm.} The bound is on
$\int_{\mathbb R}\lvert A-A_P\rvert\,d\omega$ normalized by
$\lVert\phi\rVert^{2}=\int_{\mathbb R}A\,d\omega$. A relative error normalized instead by $\int_W A$
over a finite window $W$ and evaluated by finite quadrature---the convention of the repository, in
which Sec.~\ref{sec:prices} reports---is a different quantity, differing by the product of
$\int_W\lvert A-A_P\rvert/\int_{\mathbb R}\lvert A-A_P\rvert$ and $\lVert\phi\rVert^{2}/\int_W A$.
Both factors are measured here: the second is $1.0100$, $1.0309$ and $1.0373$ at $\eta=0.05$, $0.15$
and $0.18\,t$ at $L=6$ and $1.0110$, $1.0278$ and $1.0335$ at $L=8$, while the first lies between
$1.0040$ and $1.0115$ across the $204$ configurations of this study. Every number in
Sec.~\ref{sec:frontier} is in the normalization of Theorem~\ref{thm:leakage}, and the two conventions
are nowhere interchanged.

\paragraph{(H3) Exact reference.} $A$ denotes the exact spectral function of the full sector of the
probe. If the reference is itself a Lanczos reconstruction of finite depth $n_\ell$, the theorem does
\emph{not} bound $\lVert A^{(n_\ell)}-A_P^{(n_\ell)}\rVert_1$, and a measured error below the
reference's own drift measures the reference rather than the reconstruction. That drift runs from
$2.6\times10^{-5}$ ($L=12$, $\eta=0.10\,t$) to $2.8\times10^{-7}$ ($\eta=0.18\,t$), and at $L=10$ the
reorthogonalization gap, $1.71\times10^{-5}$, is of the same order as the depth drift. Every measured error in this work is therefore labeled with the depth of its reference and with
that reference's own drift, and we
quote no measured error anywhere in Sec.~\ref{sec:frontier} within a factor of $20$ of its own
reference drift: twelve rows at $L\ge10$ fall under that rule and are left blank rather than rounded,
the weakest surviving cell---$L=12$, $\eta=0.10\,t$---standing at a factor $47$. One gap is declared
and not bounded. The reconstructions this paper publishes are produced by bare Lanczos at depth $150$,
while the certificate is evaluated on the reorthogonalized Krylov space; the difference between the
two has been measured at a single point at $L=8$---relative $L_1$ of $1.2454\times10^{-3}$ against
$1.2464\times10^{-3}$---and never bounded, and at $L=12$ it has never been measured at all.

\paragraph{(H4) Exact probe.} $A$ and $A_P$ must be built from one probe, and it must be the one whose
spectral function is wanted: $\phi=\hat c^{\dagger}_p\ket{\Psi_0}$, or its removal, density or spin
counterpart, with the exact $N$-electron ground state $\Psi_0$. That holds by construction in every
certificate evaluated in this work: $A$ and $A_P$ are formed from the same $\phi$, and on the lattices
that $\phi$ is built from a ground state computed on the full $N$-electron sector---at $L=14$ the stored
checkpoint of Sec.~\ref{sm:frontier}, residual $2.7\times10^{-9}$---while the Rayleigh--Ritz ground
state of the (H1) test above supplies $E_0$ alone and leaves $\phi$ exact. A pipeline whose $\Psi_0$ is
itself estimated on a sampled subspace forms another vector $\tilde\phi$; Theorem~\ref{thm:leakage} then
bounds $\lVert A[\tilde\phi]-A_P[\tilde\phi]\rVert_1$, and nothing the pipeline computes bounds the
distance from $A[\tilde\phi]$ to $A[\phi]$. The certificate has never been tested with an approximate
probe.

\paragraph{Ranking, ties and support.} Two properties of the cut are declared with the protocol. Only the reflection-forbidden determinants ($4$, $18$, $48$, $200$ at $L=6$--$12$) carry exactly zero accumulated Born weight, and they rank last, so the cut is never filled with configurations the
device would never sample in the infinite-shot limit; ties at the cut involve two determinants (four
at $L=6$, at fractions $0.80$ and $0.85$) and move the bound by at most $6.6\times10^{-4}$ relative.
Above the support floor $\lvert\operatorname{supp}(\phi)\rvert/D$ of Sec.~\ref{sec:retraction}, on the other hand, the
ranking selects among determinants outside $\operatorname{supp}(\phi)$: at fraction $0.85$ that is
$30.2\%$ of $S$ at $L=6$, $42.0\%$ at $L=8$ and ${\approx}39.6\%$ at $L=10$. Those contribute nothing
to $w_S$, and they do change $H_S$. The depth uncertainty quoted in Sec.~\ref{sec:frontier} is the
full recomputation of $\mathrm{FR}^{\ast}$ with $L=12$ capped at $n_\ell=300$: it moves by $0.0024$,
$0.0021$, $0.0012$ and $0.0000$ at $\eta=0.10$, $0.15$, $0.18$ and $0.25\,t$.

\paragraph{Reorthogonalization and the certified object.} Full reorthogonalization is used at every
point of the sweep, without which the projector is not a projector: the orthogonality defect is at
most $4.7\times10^{-15}$, and $P_{\mathcal K}\phi=\phi_S$ to $1.1\times10^{-14}$. At $L=12$ the
certificate is evaluated at $n_\ell=300$--$500$ on up to $658\,627$ determinants, the dense $A_S$
never being formed---which is why what the bound controls is $A_{\mathcal K}$ at the declared depth
and not $A_S$ (Sec.~\ref{sec:frontier:hyp}).

\paragraph{Anatomy of the displacement.} The percentages quoted in
Sec.~\ref{sec:frontier:obstruction} are medians over $15$ of the $18$ $(\eta,\text{target})$ cells at
$L=8$; three cells are censored from above, the certificate never reaching the target inside the grid,
and are excluded rather than imputed. Pooling $L=6$ with $L=8$ shifts them to $52\%$, $31\%$, $11\%$
and $4.1\%$, which is why the $L=8$ figures are the ones quoted. One caveat travels with them, and it
is why they are reported as a diagnosis and not as a quantity: of the rungs of the inequality chain
only the last is computable by a user from $(H,S,\phi,\eta)$, so the decomposition is of an object the
user never sees.

\section{The geminal witness: proof and angle sweep}
\label{app:witness}

\begin{proof}[Proof of Theorem~\ref{thm:witness}]
Each geminal occupies its own four spin-orbitals, so \eqref{eq:apsg} is a product across the $K$ blocks
and $\gamma$ is block diagonal with spatial natural occupations
$(n_{b_i},n_{a_i})=(2\cos^2\theta,2\sin^2\theta)$; hence $N_{\mathrm u}=8K\cos^2\theta\sin^2\theta=2K$ at
$\theta=\pi/4$, and $\mathcal F_1=4K$ by Eq.~\eqref{eq:collapse}. The Majorana covariance is likewise
block diagonal, and at $\theta=\pi/4$ a single maximally-paired geminal is the two-determinant cat, whose
covariance vanishes identically ($\Gamma=0$): every order saturates its maximum,
$\mathcal F_k=2n_{\rm o}=4$ per geminal for all $k$, giving $\mathcal F_2=4K$. This equality of orders is
a knife-edge of maximal pairing; away from it the orders separate ($\theta=\pi/6$ gives
$\mathcal F_1=3$, $\mathcal F_2=3.75$ per pair, both still large on the same family). In the pairing basis
each geminal is a two-configuration superposition, so the product carries $2^{K}$ determinants. Finally,
as a product of two-orbital factors the state is a matrix-product state whose every bond crosses at most
one geminal, so $\chi=2$ in that ordering, and contraction and Born sampling then cost $O(K)$. The
values of $\mathcal F_1$, $\mathcal F_2$ and $|\mathcal S|_{\mathrm{pair}}$ are reproduced by exact
diagonalization for $K=1,\dots,4$. So is the value $\chi=2$ derived here, the largest Schmidt rank
over all cuts of the pairing ordering, for $K=1,\dots,4$ and at every interior angle of the sweep
($\chi=1$ at $\theta=\pi/2$); the Schmidt rank at the central spin-orbital cut alone is $2$, $1$, $2$,
$1$ for $K=1,\dots,4$, because that cut falls between two geminals when $K$ is even.
\end{proof}

Table~\ref{tab:fibersweep} continues the fiber exhibit of Table~\ref{tab:fiber} into a continuous sweep of
the rotation angle at $K=3$. It stops at $s=0.654$: at the maximal-pairing angle $s=\pi/4$ the Born
distribution of $\phi$ is exactly flat over $2^{K}$ determinants, so the top-$|\mathcal S|$ rule does not
select a subspace and no contrast measured there is well defined. Across every row the captured weight is
identical at the two endpoints to $2.2\times10^{-16}$, so the entire difference sits in the leakage term,
and the certificate at the rotated endpoint is the trivial bound $1+w_S$ in every row.

\begin{table*}[tp]\centering\small
\caption{\textbf{Angle sweep of the fiber exhibit}, $K=3$, $\delta=0.01$, $\eta=0.10$,
$|\mathcal S|=4$. Certificate columns as in Table~\ref{tab:fiber}; $w_S$ is identical at the two
endpoints, and bold marks a certificate below the trivial bound $1+w_S$.}
\label{tab:fibersweep}
\setlength{\tabcolsep}{4pt}
\begin{tabular}{@{}lccccc@{}}
\toprule
$s$ & $w_S$ (both endpoints)
 & $\widehat\Lambda/\eta$ al. & $\widehat\Lambda/\eta$ rot.
 & cert.\ al. & cert.\ rot. \\
\midrule
$0.131$ & $0.982963$ & $0.328$ & $1.610$ & $\mathbf{0.913}$ & $1.983$ (trivial) \\
$0.262$ & $0.933013$ & $0.317$ & $3.010$ & $\mathbf{1.131}$ & $1.933$ (trivial) \\
$0.393$ & $0.853553$ & $0.299$ & $4.124$ & $\mathbf{1.312}$ & $1.854$ (trivial) \\
$0.524$ & $0.750000$ & $0.277$ & $4.849$ & $\mathbf{1.450}$ & $1.750$ (trivial) \\
$0.654$ & $0.629410$ & $0.252$ & $5.139$ & $\mathbf{1.543}$ & $1.629$ (trivial) \\
\bottomrule
\end{tabular}
\end{table*}

\section{One-body magic cannot bound the determinant support}
\label{app:theorem}
\begin{proposition}[Coefficient-insensitivity obstruction]
\label{prop:obstruction}
Let $\mu$ be any functional of a fixed finite set of reduced density matrices, continuous in the
state. Then $\mu$ cannot usefully upper-bound the \emph{raw} determinant support (the unthresholded
Slater rank): for every $\epsilon>0$ there are states with $|\mu-\mu(\ket{D_0})|<\epsilon$ (for
$\mathcal F_k$, $\mu<\epsilon$) whose raw support is every determinant of the particle-number sector, so
any bound $|\mathcal S|_{\rm raw}\le f(\mu)$ valid on the sector takes the trivial value, the sector
dimension, at values of $\mu$ arbitrarily close to $\mu(\ket{D_0})$. The count is lower- but not
upper-semicontinuous; the CP tensor rank fails upper semicontinuity in the same way and, unlike
the count, is not even lower-semicontinuous: its sets of rank at most $r$ are not
closed~\cite{desilvalim2008} (that reference names the two semicontinuities the other way round,
calling matrix rank upper-semicontinuous). Nor can the invariant $\mathcal F_1$ usefully lower-bound it:
for $N$ electrons on $2n_{\rm o}$ modes it implies $|\mathcal S|_{\rm raw}\ge N/(N-\mathcal F_1/4)$,
which never exceeds $2n_{\rm o}/N$, and at half filling ($N=n_{\rm o}$), where that is $2$, any
increasing $f$ with $|\mathcal S|_{\rm raw}\ge f(\mathcal F_1)$ satisfies $f\le2$.
\end{proposition}
\begin{proof}
Rank is a coefficient-insensitive count. Take
$\ket{\Psi_\epsilon}=\ket{D_0}+\epsilon\sum_{k=1}^{K}\ket{D_k}$ with the $D_k$ distinct determinants of
the sector of $\ket{D_0}$. As $\epsilon\to0$ every reduced density matrix converges to that of
$\ket{D_0}$, so any continuous $\mu\to\mu(\ket{D_0})$ (for $\mathcal F_k$, to $0$), while the raw
support stays $K+1$ for every $K$ up to the sector dimension minus one. For the second part, every
natural orbital lies in the span of the at most $N|\mathcal S|_{\rm raw}$ orbitals occupied in the
support, so $\mathrm{rank}\,\gamma\le N|\mathcal S|_{\rm raw}$; with $\mathrm{tr}\,\gamma=N$ and
$\mathrm{tr}\,\gamma^{2}=N-\mathcal F_1/4$, the Cauchy--Schwarz inequality
$(\mathrm{tr}\,\gamma)^{2}\le\mathrm{rank}\,\gamma\,\mathrm{tr}\,\gamma^{2}$ gives the bound, and
$\mathrm{rank}\,\gamma\le2n_{\rm o}$ caps it at $2n_{\rm o}/N$. At half filling with $n_{\rm o}\ge2$ the two-determinant cat
$(\ket{1\cdots10\cdots0}+\ket{0\cdots01\cdots1})/\sqrt2$, whose determinants differ in $N\ge2$ orbitals,
has $\gamma=\tfrac12\openone$, hence $\mathcal F_1=2n_{\rm o}$ (maximal) and Slater rank $2$, so
$f(2n_{\rm o})\le2$; for $n_{\rm o}=1$ every state is a single particle on two modes, of raw support at
most $2$.
\end{proof}
\noindent The obstruction is about the unthresholded rank; the operative, coefficient-\emph{sensitive}
support $|\mathcal S|_\varepsilon$ of Sec.~\ref{sec:nogo} collapses to $1$ along the same
$\epsilon\to0$ path, so continuity alone does not obstruct it. What obstructs $|\mathcal S|_\varepsilon$
itself is the orbital-rotation symmetry---no invariant functional survives it---together with the two
equal-weight witnesses below, whose amplitudes are $\Theta(1)$ and therefore threshold-robust.
Two explicit witnesses rule out the natural repairs. The two-determinant cat of
Prop.~\ref{prop:obstruction} is maximal in $\mathcal F_1$ yet of Slater rank $2$, so at half filling no
$\mathcal F_1$-monotone bounds $|\mathcal S|$ from below beyond $2$; a paired (BCS) Gaussian state has
Gaussian rank $1$ yet fractional occupations, so no occupation functional bounds it from above. The
computable lower bounds on $|\mathcal S|$ that are not capped in this way are the basis-dependent
orbital-Schmidt (matrix-flattening) ranks~\cite{schollwock2011}---the tensor-network bond
dimension---to which any Gaussian-invariant magic measure is blind. That lower bound is exact:

\begin{proposition}[The bond dimension lower-bounds the support]
\label{prop:chibound}
For any state and any bipartition of the spin-orbitals into $(\mathrm L,\mathrm R)$, the Schmidt rank
$\chi$ across the cut (the minimal matrix-product-state bond dimension there) satisfies
\begin{equation}
\chi \;\le\; |\mathcal S|,
\end{equation}
where $|\mathcal S|$ is the number of fixed-basis Slater determinants of nonzero amplitude. Hence the
minimal bond dimension is a rigorous, basis-dependent lower bound on the determinant support. The bound is
one-sided: the geminal witness of Thm.~\ref{thm:witness} has $\chi=2$ yet $|\mathcal S|_{\mathrm{pair}}=2^{K}$,
so $|\mathcal S|$ can exceed $\chi$ exponentially already in the pairing basis, which for that family is
a natural-orbital basis, while the orbital-invariant one-body magic $\mathcal F_1$ cannot usefully
upper-bound $|\mathcal S|$ and at half filling cannot lower-bound it beyond $2$
(Prop.~\ref{prop:obstruction}).
\end{proposition}
\begin{proof}
Write $\ket{\Psi}=\sum_{D\in\mathcal S}c_D\ket{D}$ over its support. Each determinant factorizes across the
cut, $\ket{D}=\ket{D_{\mathrm L}}\otimes\ket{D_{\mathrm R}}$, so the matricization
$A_{D_{\mathrm L},D_{\mathrm R}}=c_D$ has at most $|\mathcal S|$ nonzero entries; its rank---the Schmidt
rank $\chi$---is therefore at most $|\mathcal S|$.
\end{proof}
\noindent Since a matrix-product-state simulation costs $\mathrm{poly}(\chi)$, Prop.~\ref{prop:chibound}
says that \emph{in the sampler's own basis a tensor network is never less compact than the sampled
subspace}: this is why we make no advantage claim in one dimension, where $\chi$ is small, and why
any advantage must live in the higher-entanglement regime where $\chi$ itself---and with it the classical
tensor-network cost---grows (the chain-to-ladder jump of Fig.~\ref{fig:ladder}). Empirically $\chi$ is also the
operative correlate of the \emph{measured} cost: across the $74$ exact ground states of
Fig.~\ref{fig:master} the thresholded support $|\mathcal S|_\varepsilon$ tracks $\chi$ (pooled Spearman
$\rho=0.72$; $\rho=0.90$ within the $n=38$ molecular suite, where however one-body magic scores a
comparable $\rho=0.86$---the two resources are separated not there but by the Hubbard sweep, where
$\mathcal F_1$ anti-orders the support in all six strata while $\chi$ orders it at a mean $\rho$ of
$0.948$; the pooled coefficients for that sweep, $\rho=0.57$ and $-0.11$, are Simpson reversals of the
six strata and are not the evidence here, Sec.~\ref{sec:nogo:chi}). Two caveats belong with that number. A rank correlation is not a statement of tightness,
and per the geminal witness the bound can be exponentially loose. And the proposition is proved for the
raw support: once amplitudes are thresholded the inequality need not survive---here $|\mathcal S|$ is cut
at $99\%$ of the weight while $\chi$ is cut at $\varepsilon=0.05$, a tighter tolerance---and
$|\mathcal S|_\varepsilon$ falls below $\chi$ in $7$ of the $38$ molecular points. What stands is that $\chi$ is a rigorous floor on
the raw support and an empirical correlate of the measured one, whereas $\mathcal F_1$ orders the
support with a sign that is not stable from family to family (Sec.~\ref{sec:nogo:chi}).

\section{The resource statistics, stratified}
\label{app:stats}

Section~\ref{sec:nogo:chi} reports the correlation evidence within the design that produced it. This
appendix gives the design, the per-stratum coefficients and the three tests, so that the decision not
to use the two pooled coefficients can be checked rather than taken.

\paragraph{The lattice design.} The $30$ Hubbard points are six $(L,N)$ strata swept over the same five
on-site repulsions $U=1,2,4,8,16$; they are not $30$ independent observations, and the effective number
of independent sweeps is six. Table~\ref{tab:strata} gives each stratum. The rank correlation of
$\mathcal F_1$ with the static support $|\mathcal S|_{\varepsilon_w}$ is $-1$ in all six, exactly; that
of $\chi$ runs from $+0.894$ to $+1.000$ with mean $+0.948$, its departures from unity being ties in
$\chi$ and not inversions.

\begin{table*}[tp]\centering\small
\caption{\textbf{The six Hubbard strata.} Each is a five-point sweep in $U$ at fixed $(L,N)$. The
one-sided exact permutation probability of a perfect reversal within one stratum of five distinct
values is $1/120=0.0083$; ties in $\chi$ raise it. ``distinct $\chi$'' counts how many of the five
$\chi$ values differ, which is what caps $\rho(\chi,|\mathcal S|)$ below unity in five of the
six.}\label{tab:strata}
\begin{tabular}{cccccc}
\toprule
$L$ & $N$ & $\rho(\mathcal F_1,|\mathcal S|)$ & $\rho(\chi,|\mathcal S|)$ & $\chi$ over $U=1,2,4,8,16$ & distinct $\chi$\\
\midrule
$6$  & $6$  & $-1.000$ & $+1.000$ & $12,10,6,5,4$   & $5$\\
$6$  & $4$  & $-1.000$ & $+0.894$ & $11,10,9,9,9$   & $3$\\
$8$  & $8$  & $-1.000$ & $+0.949$ & $11,11,8,8,7$   & $3$\\
$8$  & $6$  & $-1.000$ & $+0.975$ & $16,16,15,14,13$& $4$\\
$10$ & $10$ & $-1.000$ & $+0.975$ & $17,14,10,6,6$  & $4$\\
$10$ & $8$  & $-1.000$ & $+0.894$ & $13,13,13,11,9$ & $3$\\
\bottomrule
\end{tabular}
\end{table*}

\paragraph{Three tests, and what each assumes.} The stratified concordance statistics
$(n_C-n_D)/(n_C+n_D)$ are $-1.000$ for $\mathcal F_1$, with $0$ concordant and $60$ discordant
within-stratum pairs and no ties, so that the statistic is Kendall's $\tau$, and $+1.000$ for $\chi$,
with $50$ concordant, $0$ discordant and $10$ tied, where it is the Goodman--Kruskal $\gamma$, which
drops the tied pairs; Kendall's $\tau_b$, which charges them, is below $1$. Because each
$\mathcal F_1$ stratum is a perfect reversal of five distinct points, the exact one-sided stratified
permutation probability is the product of the six per-stratum chances,
$(1/120)^{6}=3.35\times10^{-13}$, over a relabelling space of $2.99\times10^{12}$; the same test
applied to $\chi$, which has ties, gives $1.93\times10^{-10}$. Fisher's combination of the six
per-stratum probabilities is conservative here and returns $6.56\times10^{-8}$ and
$1.14\times10^{-5}$ ($\chi^{2}=57.4$ and $44.7$ on $12$ degrees of freedom); we quote the permutation
values and record Fisher's beside them. All are \emph{one-sided}---a perfect reversal is one of the
$120$ orderings, not two---and all take a product across strata, so all assume the six strata
independent. They are not six independent experiments: they are three lengths crossed with two
fillings, on one code path and one $U$ grid. What the numbers establish is that the direction is
reproduced six times out of six, not that it generalizes beyond this family.

\paragraph{Why neither pooled coefficient is used.} Pooled over the same $30$ points
the coefficients are $\rho(\chi,|\mathcal S|)=+0.569$ and
$\rho(\mathcal F_1,|\mathcal S|)=-0.109$, the second with a nominal two-sided $p=0.57$ and the first
with $p=0.0010$. Pooling six deterministic five-point sweeps with different offsets destroys the
within-stratum structure that carries the signal; it is a Simpson reversal in both cases. It happens to
destroy the signal in our favour on $\chi$ and against us on $\mathcal F_1$, and that is exactly why
neither may be kept: a rule that is applied to one coefficient and not to the other is not a rule.

\paragraph{The molecular set.} The $n=38$ molecular points are $19$ molecules at two geometries, i.e.\
dependent pairs. The intraclass correlation of $\log_2|\mathcal S|$ within a molecule is $0.214$, a
design effect of $1.214$ and a design-corrected effective $n$ of $31.3$. Ties are pervasive: $29$ of
the $38$ ($76\%$) lie in some tied group of $|\mathcal S|_{\varepsilon_w}$ and $11$ are pinned at the
floor $|\mathcal S|=1$, which caps the attainable Spearman coefficient at $\rho_{\max}=0.9863$;
against that ceiling $\chi$ reaches $91.4\%$ and $\mathcal F_1$ $87.4\%$. Under a molecule-level
cluster permutation ($10^{6}$ resamples) both coefficients are beyond the resolution floor of the
test, $p\le1.0\times10^{-6}$. The comparison the molecular set is asked to settle is not resolvable,
and three
procedures agree that it is not: $\rho=0.9011$ for $\chi$ against $0.8618$ for $\mathcal F_1$ differ by
$0.0393$; a molecule-level cluster bootstrap ($2\times10^{4}$ resamples) puts a $95\%$ interval on the
difference of $[-0.070,+0.189]$, which contains zero; Williams's test for dependent correlations gives
$t=0.94$ on $35$ degrees of freedom, $p=0.35$, and Steiger's $z=0.94$, $p=0.35$. Both parametric tests
assume bivariate normality and untied ranks while $76\%$ of these points sit in a tie, which is why the
bootstrap is quoted first; at the cluster-level effective $n$ of $19$, Williams gives $p=0.53$. Pooled
over all $74$ rows the pair is $0.722$ against $0.599$ and is likewise unresolved. Recomputed over the
$71$ \emph{distinct} states---three chain rows duplicate three Hubbard rows---it is $0.746$ against
$0.592$.

\paragraph{One artefact of the deposit, recorded because it changes a digit.} The pooled
$\rho(\mathcal F_1,|\mathcal S|)$ stored in the deposit is $0.5991984$ and recomputing it from the
per-species values gives $0.5992028$. The difference is not a bug in either: it is a tie in
$\mathcal F_1$ created when two rows that are the same state are stored at different precisions. The
quantity is quoted to three figures throughout for that reason.

\section{Reproducibility data for the molecular suite}
\label{app:repro}

\begin{table*}[tp]\centering\small
\caption{\textbf{Reproducibility data for the nineteen-molecule suite} (cc-pVDZ, valence bond-breaking
active spaces). $R_{\rm eq},R_{\rm diss}$ in \AA{}; for the polyatomics ($\mathrm{BeH_2},\mathrm{H_2O},
\mathrm{NH_3},\mathrm{H_4},\mathrm{H_6}$) the representative varied bond length is shown and full
Cartesians are in \texttt{src/n19\_suite.py} of the deposit. $D$: dimension of the $N$-electron (neutral, ground-state) sector,
in which the tabulated $|\mathcal S|$ lives; $|\mathcal S|$: ground-state determinant support ($\ge99\%$
FCI weight); $\chi$: minimal MPS bond dimension across the central spin-orbital cut. The $|\mathcal S|/D$
of Table~\ref{tab:molecules} is a distinct ratio, taken in the $(N{\pm}1)$ spectral sector.}
\label{tab:repro}
\begin{tabular}{llccccccc}\hline\hline
Molecule & CAS$(n_e,n_{\rm o})$ & $R_{\rm eq}$ & $R_{\rm diss}$ & $D$ & $|\mathcal S|^{\rm eq}$ & $|\mathcal S|^{\rm diss}$ & $\chi^{\rm eq}$ & $\chi^{\rm diss}$ \\ \hline
F$_2$ & (1+1,2) & 1.412 & 3.00 & 4 & 1 & 2 & 2 & 2 \\
N$_2$ & (3+3,6) & 1.098 & 2.40 & 400 & 7 & 80 & 7 & 32 \\
H$_2$O & (2+2,4) & 0.958 & 2.30 & 36 & 1 & 13 & 1 & 13 \\
HF & (1+1,2) & 0.917 & 2.30 & 4 & 1 & 4 & 1 & 4 \\
NH$_3$ & (3+3,6) & 1.012 & 2.33 & 400 & 1 & 68 & 4 & 42 \\
CO & (3+3,6) & 1.128 & 2.40 & 400 & 8 & 46 & 14 & 30 \\
H$_6$ & (3+3,6) & 0.900 & 2.00 & 400 & 9 & 91 & 8 & 36 \\
H$_4$ & (2+2,4) & 0.900 & 2.00 & 36 & 2 & 11 & 4 & 8 \\
H$_2$ & (1+1,2) & 0.741 & 2.00 & 4 & 1 & 2 & 1 & 2 \\
BeH$_2$ & (2+2,4) & 1.334 & 2.90 & 36 & 1 & 6 & 1 & 8 \\
C$_2$ & (4+4,6) & 1.243 & 2.60 & 225 & 13 & 6 & 11 & 2 \\
NO & (4+3,6) & 1.151 & 2.50 & 300 & 5 & 25 & 11 & 7 \\
O$_2$ & (5+3,6) & 1.208 & 2.60 & 120 & 11 & 4 & 4 & 2 \\
CN & (4+3,6) & 1.172 & 2.50 & 300 & 27 & 9 & 17 & 7 \\
OH & (3+2,4) & 0.970 & 2.40 & 24 & 1 & 5 & 1 & 4 \\
BH & (2+2,4) & 1.232 & 2.60 & 36 & 5 & 3 & 2 & 2 \\
BeH & (2+1,3) & 1.343 & 2.70 & 9 & 1 & 2 & 1 & 2 \\
LiH & (1+1,2) & 1.596 & 3.20 & 4 & 1 & 2 & 1 & 4 \\
LiF & (1+1,2) & 1.564 & 3.20 & 4 & 1 & 1 & 1 & 1 \\
\hline\hline\end{tabular}
\end{table*}

Table~\ref{tab:repro} lists, for every molecule of Table~\ref{tab:molecules}, the equilibrium and
dissociation bond lengths, the active space, the $N$-electron (neutral, ground-state) sector dimension $D$, the ground-state
determinant support $|\mathcal S|$ (the number of determinants carrying $\ge99\%$ of the FCI weight), and
the minimal matrix-product-state bond dimension $\chi$ at both geometries. These are the quantities that
enter the $\mathcal F_1$--$|\mathcal S|$ and $\chi$--$|\mathcal S|$ correlations of
Sec.~\ref{sec:nogo}, pooled over the $n=38$ equilibrium-and-dissociation points (Spearman
$\chi$--$|\mathcal S|$ $\rho=0.90$, $\mathcal F_1$--$|\mathcal S|$ $\rho=0.86$; Pearson
$\mathcal F_1$--$\log_2|\mathcal S|$ $r=0.80$). The $|\mathcal S|/D$ reported in
Table~\ref{tab:molecules} is a \emph{distinct} quantity---the fraction of the $(N{+}1)$ sector used in
the spectral \emph{reconstruction}, not the ground-state support tabulated here. Every active-space FCI
energy is converged to better than $10^{-12}$~Ha (the largest residual across all $38$ points is
$1.1\times10^{-13}$~Ha); the full Cartesian geometries and the active-space definitions (frozen core,
$n_e$, $n_{\rm o}$) are given in \texttt{src/n19\_suite.py} of the deposit, and the absolute ground-state
energies are not deposited, only these FCI residuals (\texttt{data/n19\_suite.json}).

\subsection{The molecular suite, and the size of its sectors}
\label{sec:scope:mol}
\begin{figure*}[tp]\centering\includegraphics[width=\linewidth]{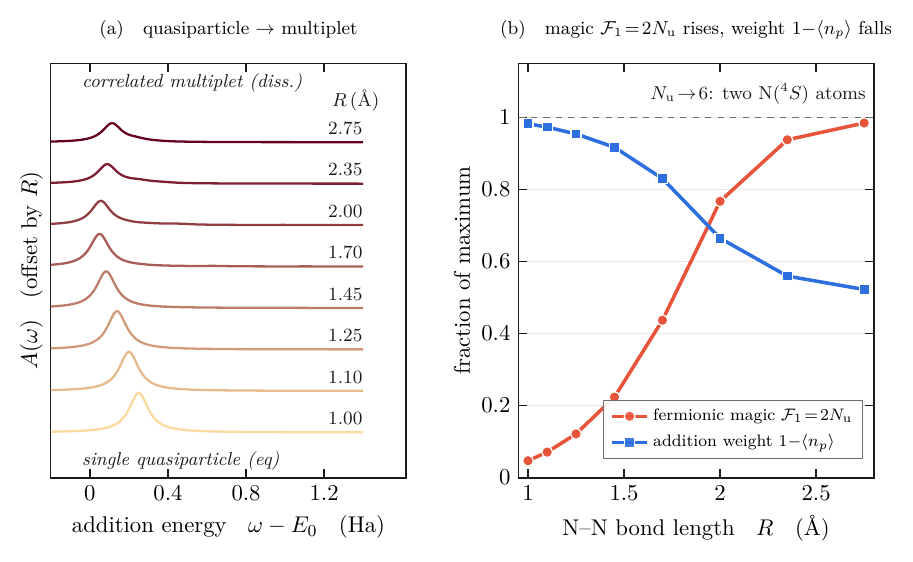}
\caption{\textbf{Magic and multireference character in $\mathrm{N_2}$ dissociation.} (a)~The exact
addition spectral function $A(\omega)$ (orbital-projected, CAS(6,6), cc-pVDZ) along the triple-bond
stretch, offset vertically, one row per $R$: the sharp coherent quasiparticle peak near equilibrium
($R=1.0$~\AA, slightly compressed from the $R_{\rm eq}=1.098$~\AA\ of Table~\ref{tab:repro}) broadens
into a correlated multiplet on dissociation as its spectral weight collapses. (b)~Two physically
independent quantities along the same coordinate. The fermionic magic obeys the \emph{exact} identity
$\mathcal F_1=2N_{\mathrm u}$ [Eq.~\eqref{eq:collapse}, verified to $<10^{-13}$], so a single curve carries both: it rises
with the orbital-rotation-invariant multireference count---the effective number of unpaired electrons
$N_{\mathrm u}=\sum_i n_i(2-n_i)$ from the natural-orbital occupations $n_i$, approaching its atomic limit of $6$
(the two triplet $\mathrm{N}(^4S)$ atoms) as the bond fully breaks. The integrated addition weight $\int A^{+}d\omega=1-\langle\hat n_p\rangle$
(from panel a) falls from $0.98$ to $0.52$ over the same stretch; the two cross as the bond
breaks. The active-space CASCI energies are converged to $1.1\times10^{-13}$~Ha, which is a
convergence check and not a validation against an independent method
(Sec.~\ref{sec:scope:mol}).}\label{fig:hero}\end{figure*}

\begin{figure*}[tp]\centering\includegraphics[width=0.7\linewidth]{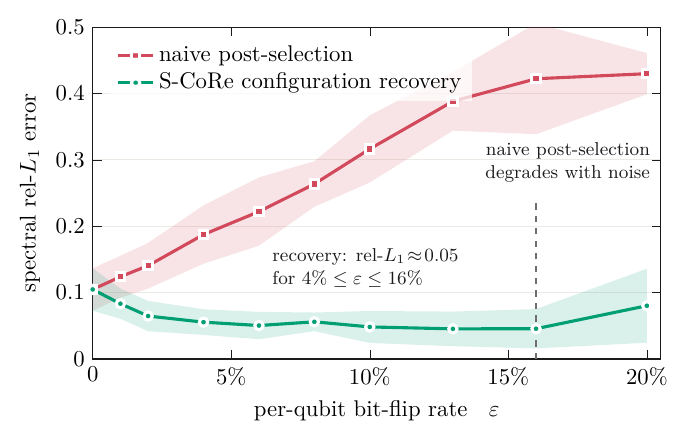}
\caption{\textbf{Noise robustness (synthetic per-qubit bit-flip channel, $L=6$).} Spectral
relative-$L_1$ error vs the per-qubit bit-flip rate $\varepsilon$, applied to bitstrings Born-sampled from
the exact distribution (this is a controlled classical simulation, not a device run). Naive
post-selection degrades steadily, whereas S-CoRe self-consistent configuration recovery, applied to the
same corrupted bitstrings, keeps the error near $0.05$ for $4\%\le\varepsilon\le16\%$ (periodic Hubbard
ring, $U/t=4$; seed-averaged over $6$ seeds at $6000$ shots each; bands $=\pm1$ s.d.\ across seeds). At
$\varepsilon=0$ the two arms coincide at $0.105$, about twice the acceptance threshold $0.05$ of the
resource scan, and the recovered error at finite $\varepsilon$ falls below that noiseless value; that
second comparison is not used here (Sec.~\ref{sec:scope:hw}): it is decided by the size of the recovered
subspace, which this sweep does not record.}\label{fig:noise}\end{figure*}

\begin{figure*}[tp]\centering\definecolor{nrH6}{HTML}{6A4C93}
\definecolor{nrNH3}{HTML}{E8880C}
\definecolor{nrN2}{HTML}{D7263D}
\definecolor{nrCO}{HTML}{12876F}
\definecolor{nrC2}{HTML}{2E6F95}
\definecolor{nrOK}{HTML}{0E7C7B}\definecolor{nrBad}{HTML}{5E4B8B}\definecolor{nrNv}{HTML}{D1495B}\definecolor{nrChem}{HTML}{2E7D32}
\begin{tikzpicture}
\begin{groupplot}[group style={group size=2 by 1, group name=nrg, horizontal sep=1.45cm}, height=9.5cm, paperaxis,
  title style={at={(0,1)},anchor=south west,font=\bfseries,yshift=1pt}]
\nextgroupplot[width=0.505\linewidth, ymode=log, xmin=-1, xmax=31, ymin=3e-3, ymax=80,
  xlabel={per-qubit bit-flip rate $\varepsilon$ (\%)}, ylabel={energy error vs FCI (mHa)}, title={(a)},
  legend style={at={(0.5,0.03)},anchor=south,font=\normalsize,draw=black!20,fill=none},legend columns=3,legend cell align=left]
\draw[nrChem,dashed,line width=0.6pt] (axis cs:-1,1.6) -- (axis cs:31,1.6);
\addplot[name path=H6lo,draw=none,forget plot] coordinates {(0,3.2914) (2,3.0066) (5,3.6550) (10,3.7659) (15,5.3712) (20,8.9720) (30,15.6267)};
\addplot[name path=H6hi,draw=none,forget plot] coordinates {(0,4.5041) (2,4.1385) (5,4.4412) (10,6.5199) (15,9.5780) (20,16.7961) (30,29.8587)};
\addplot[nrH6,opacity=0.15,forget plot] fill between[of=H6lo and H6hi];
\addplot[name path=NH3lo,draw=none,forget plot] coordinates {(0,2.1436) (2,1.4495) (5,1.0457) (10,0.7919) (15,1.3563) (20,1.8046) (30,5.2913)};
\addplot[name path=NH3hi,draw=none,forget plot] coordinates {(0,2.9078) (2,1.9853) (5,1.6281) (10,1.6847) (15,2.0580) (20,2.8561) (30,20.1868)};
\addplot[nrNH3,opacity=0.15,forget plot] fill between[of=NH3lo and NH3hi];
\addplot[name path=N2lo,draw=none,forget plot] coordinates {(0,0.1680) (2,0.0796) (5,0.3169) (10,0.0030) (15,0.4045) (20,1.0050) (30,8.4636)};
\addplot[name path=N2hi,draw=none,forget plot] coordinates {(0,0.6006) (2,0.6093) (5,0.6846) (10,1.3158) (15,2.3324) (20,5.5026) (30,17.8734)};
\addplot[nrN2,opacity=0.15,forget plot] fill between[of=N2lo and N2hi];
\addplot[name path=COlo,draw=none,forget plot] coordinates {(0,1.5510) (2,1.5669) (5,1.0491) (10,0.7344) (15,0.5781) (20,0.4698) (30,0.3702)};
\addplot[name path=COhi,draw=none,forget plot] coordinates {(0,1.9155) (2,2.0418) (5,1.4264) (10,1.3026) (15,1.0358) (20,0.7942) (30,1.0249)};
\addplot[nrCO,opacity=0.15,forget plot] fill between[of=COlo and COhi];
\addplot[name path=C2lo,draw=none,forget plot] coordinates {(0,0.1703) (2,0.1351) (5,0.0827) (10,0.0677) (15,0.0479) (20,0.0385) (30,0.0089)};
\addplot[name path=C2hi,draw=none,forget plot] coordinates {(0,0.1967) (2,0.1807) (5,0.1240) (10,0.0817) (15,0.0730) (20,0.0627) (30,0.0433)};
\addplot[nrC2,opacity=0.15,forget plot] fill between[of=C2lo and C2hi];
\addplot[nrH6,line width=1.4pt,mark=*,mark size=1.3pt,mark options={fill=nrH6,draw=white,line width=0.4pt}] coordinates {(0,3.8978) (2,3.5726) (5,4.0481) (10,5.1429) (15,7.4746) (20,12.8840) (30,22.7427)}; \addlegendentry{H$_6$}
\addplot[nrH6,line width=0.9pt,densely dashed,forget plot] coordinates {(0,3.8978) (2,4.6251) (5,5.6347) (10,7.8697) (15,12.9901) (20,23.8560) (30,50.3837)};
\addplot[nrNH3,line width=1.4pt,mark=*,mark size=1.3pt,mark options={fill=nrNH3,draw=white,line width=0.4pt}] coordinates {(0,2.5257) (2,1.7174) (5,1.3369) (10,1.2383) (15,1.7072) (20,2.3304) (30,12.7391)}; \addlegendentry{NH$_3$}
\addplot[nrNH3,line width=0.9pt,densely dashed,forget plot] coordinates {(0,2.5257) (2,2.9622) (5,2.7617) (10,2.9227) (15,3.7294) (20,5.5505) (30,31.5152)};
\addplot[nrN2,line width=1.4pt,mark=*,mark size=1.3pt,mark options={fill=nrN2,draw=white,line width=0.4pt}] coordinates {(0,0.3843) (2,0.3445) (5,0.5007) (10,0.6589) (15,1.3685) (20,3.2538) (30,13.1685)}; \addlegendentry{N$_2$}
\addplot[nrN2,line width=0.9pt,densely dashed,forget plot] coordinates {(0,0.3843) (2,0.4335) (5,0.6929) (10,1.1118) (15,2.8006) (20,6.3706) (30,31.8749)};
\addplot[nrCO,line width=1.4pt,mark=*,mark size=1.3pt,mark options={fill=nrCO,draw=white,line width=0.4pt}] coordinates {(0,1.7332) (2,1.8043) (5,1.2377) (10,1.0185) (15,0.8069) (20,0.6320) (30,0.6975)}; \addlegendentry{CO}
\addplot[nrCO,line width=0.9pt,densely dashed,forget plot] coordinates {(0,1.7332) (2,1.9851) (5,1.9371) (10,1.4901) (15,1.8518) (20,2.5550) (30,2.8339)};
\addplot[nrC2,line width=1.4pt,mark=*,mark size=1.3pt,mark options={fill=nrC2,draw=white,line width=0.4pt}] coordinates {(0,0.1835) (2,0.1579) (5,0.1034) (10,0.0747) (15,0.0604) (20,0.0506) (30,0.0261)}; \addlegendentry{C$_2$}
\addplot[nrC2,line width=0.9pt,densely dashed,forget plot] coordinates {(0,0.1835) (2,0.1635) (5,0.1143) (10,0.0811) (15,0.0746) (20,0.0647) (30,0.0657)};
\nextgroupplot[width=0.425\linewidth, xmode=log, xmin=2.2e-3, xmax=14, enlarge y limits=0.03,
  ytick={0,1,2,3,4,5,6,7,8,9,10,11,12,13,14,15,16,17,18}, yticklabels={BeH$_2$,H$_2$,LiH,HF,F$_2$,H$_2$O,H$_4$,OH,LiF,BH,N$_2$,BeH,C$_2$,NO,O$_2$,CN,CO,NH$_3$,H$_6$}, y tick label style={font=\normalsize},
  xtick={0.01,0.1,1,10}, xticklabels={$0.01$,$0.1$,$1$,$10$}, xminorticks=false,
  xlabel={energy error at $\varepsilon=2\%$ (mHa)}, title={(b)},
  legend style={at={(0.03,0.03)},anchor=south west,font=\normalsize,draw=black!20},legend cell align=left]
\fill[nrChem,fill opacity=0.06] (axis cs:2.2e-3,-0.6) rectangle (axis cs:1.6,18.6);
\draw[nrChem,dashed,line width=0.7pt] (axis cs:1.6,-0.6) -- (axis cs:1.6,18.6);
\node[nrChem,font=\tiny,anchor=south,rotate=90] at (axis cs:1.6,9.5) {};
\draw[black!32,line width=0.8pt] (axis cs:0.0022,0) -- (axis cs:0.0022,0); \draw[black!32,line width=0.8pt] (axis cs:0.0022,1) -- (axis cs:0.0022,1); \draw[black!32,line width=0.8pt] (axis cs:0.0022,2) -- (axis cs:0.0022,2); \draw[black!32,line width=0.8pt] (axis cs:0.0022,3) -- (axis cs:0.0022,3); \draw[black!32,line width=0.8pt] (axis cs:0.0022,4) -- (axis cs:0.0022,4); \draw[black!32,line width=0.8pt] (axis cs:0.0327,5) -- (axis cs:0.0022,5); \draw[black!32,line width=0.8pt] (axis cs:0.3178,6) -- (axis cs:0.0022,6); \draw[black!32,line width=0.8pt] (axis cs:0.0194,7) -- (axis cs:0.0194,7); \draw[black!32,line width=0.8pt] (axis cs:0.0254,8) -- (axis cs:0.0254,8); \draw[black!32,line width=0.8pt] (axis cs:0.0845,9) -- (axis cs:0.0605,9); \draw[black!32,line width=0.8pt] (axis cs:0.0688,10) -- (axis cs:0.0636,10); \draw[black!32,line width=0.8pt] (axis cs:0.1367,11) -- (axis cs:0.1367,11); \draw[black!32,line width=0.8pt] (axis cs:0.1711,12) -- (axis cs:0.1615,12); \draw[black!32,line width=0.8pt] (axis cs:0.3275,13) -- (axis cs:0.1942,13); \draw[black!32,line width=0.8pt] (axis cs:0.4696,14) -- (axis cs:0.3363,14); \draw[black!32,line width=0.8pt] (axis cs:0.7815,15) -- (axis cs:0.7261,15); \draw[black!32,line width=0.8pt] (axis cs:1.3923,16) -- (axis cs:1.2450,16); \draw[black!32,line width=0.8pt] (axis cs:2.6491,17) -- (axis cs:1.5063,17); \draw[black!32,line width=0.8pt] (axis cs:4.3814,18) -- (axis cs:3.5606,18);
\addplot[only marks,mark=o,mark size=1.9pt,nrNv,line width=0.8pt] coordinates {(0.0022,0) (0.0022,1) (0.0022,2) (0.0022,3) (0.0022,4) (0.0327,5) (0.3178,6) (0.0194,7) (0.0254,8) (0.0845,9) (0.0688,10) (0.1367,11) (0.1711,12) (0.3275,13) (0.4696,14) (0.7815,15) (1.3923,16) (2.6491,17) (4.3814,18)}; \addlegendentry{naive}
\addplot[only marks,mark=*,mark size=2.1pt,nrOK] coordinates {(0.0022,0) (0.0022,1) (0.0022,2) (0.0022,3) (0.0022,4) (0.0022,5) (0.0022,6) (0.0194,7) (0.0254,8) (0.0605,9) (0.0636,10) (0.1367,11) (0.1615,12) (0.1942,13) (0.3363,14) (0.7261,15) (1.2450,16) (1.5063,17) (3.5606,18)}; \addlegendentry{S-CoRe}
\node[black,font=\normalsize,anchor=north east] at (axis cs:1.4720000000000002,18.3) {$1.6$ mHa};
\end{groupplot}
\node[nrChem,font=\footnotesize,anchor=south east,inner sep=1pt]
      at (nrg c1r1.north east) {chemical accuracy ($1.6$\,mHa)};
\end{tikzpicture}
\caption{\textbf{SQD ground-state energy under a controlled bit-flip model across the molecular suite.}
\textbf{(a)}~Energy error vs the per-qubit bit-flip rate $\varepsilon$ for five representative molecules
($\mathrm{H_6}$, $\mathrm{NH_3}$, $\mathrm{N_2}$, CO, $\mathrm{C_2}$; seed-averaged over $8$ seeds at $4000$ shots each, bands $=\pm1$ s.d.\ across seeds, running to the bottom of the panel where the mean minus one s.d.\ falls below it, as for $\mathrm{N_2}$ at $\varepsilon=10\%$; log scale). Solid $=$ S-CoRe recovery, dashed $=$ naive
post-selection. The advantage of recovery over naive post-selection, both arms consuming the same corrupted
bitstrings at the same rate with the same seeds, grows with noise. For $\mathrm{NH_3}$, CO and
$\mathrm{C_2}$ the error at finite $\varepsilon$ also falls below its own $\varepsilon=0$ value, beyond
the seed bands ($\mathrm{NH_3}$ $2.53\to1.24$~mHa at $\varepsilon=10\%$);
that second comparison is not used here (Sec.~\ref{sec:scope:hw}): it is decided by the size of
the recovered subspace, which this sweep does not record. \textbf{(b)}~The full nineteen-molecule suite at a realistic
$\varepsilon=2\%$ (log scale, sorted; errors that are exact to round-off, in both arms for
$\mathrm{F_2}$, HF, LiH, $\mathrm{H_2}$ and $\mathrm{BeH_2}$, where the open circle is hidden behind the
filled one, and after recovery for $\mathrm{H_4}$ and $\mathrm{H_2O}$, are drawn at the left edge of the
axis): for each molecule the naive post-selection error (open circle) is
linked to the S-CoRe value (filled), so the leftward pull is the recovery gain; the shaded band is
chemical accuracy ($<1.6$~mHa). Eighteen of nineteen land inside it; $\mathrm{H_6}$, the most
multireference case, is the residual challenge. Panels (a) and (b) are independent seed ensembles
(panel (a): $8$ seeds; panel (b): $5$ seeds, $4000$ shots), so a given molecule's $\varepsilon=2\%$ value
can differ by up to a factor of a few between them (for $\mathrm{N_2}$, $0.34$ vs $0.06$~mHa; both far below chemical accuracy).}\label{fig:noiserec}\end{figure*}
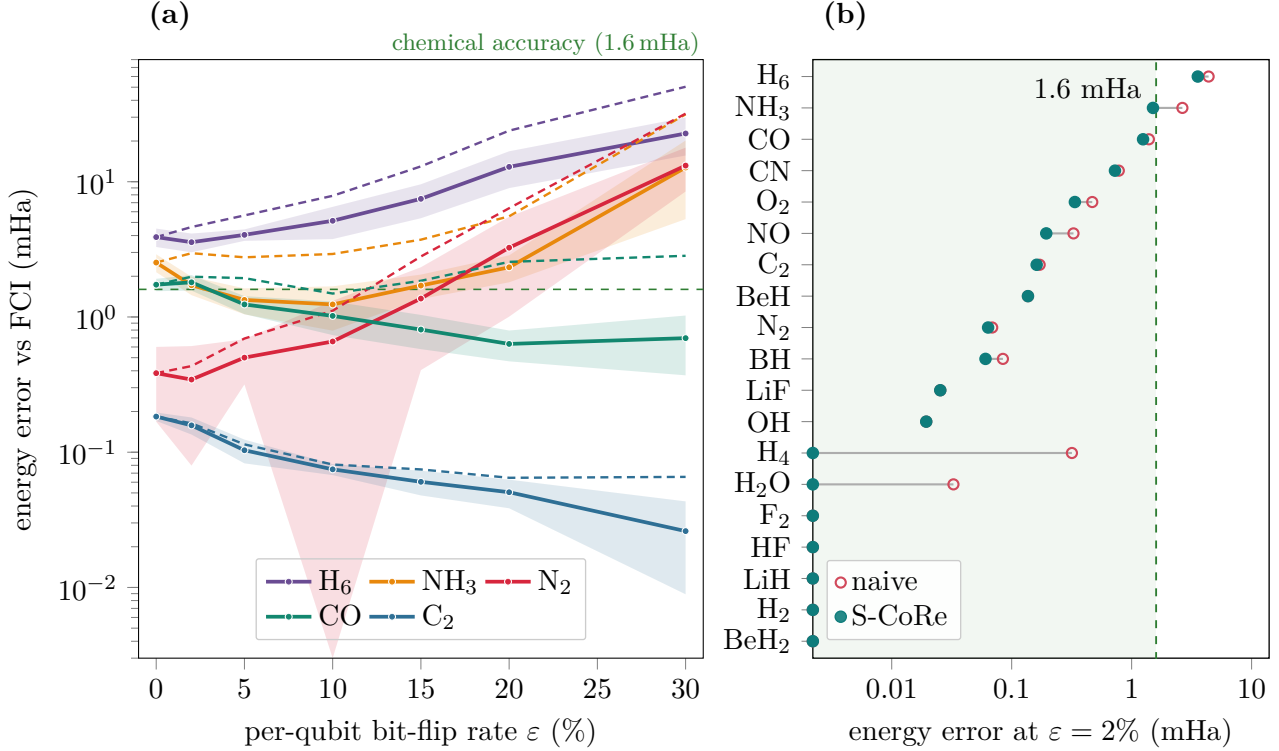

The nineteen-molecule suite reconstructs the orbital-projected addition spectral function $A(\omega)$
at dissociation in a valence bond-breaking active space, against a subspace grown by multinomial
sampling of the time-evolved probe, seeded with the probe's largest-amplitude determinant, which is
inserted rather than drawn: $8\,000$ draws at each of $27$ evolution times,
$2.16\times10^{5}$ per molecule, at $\eta=0.10$~Ha, from a single fixed random seed and with no spread
reported across seeds. What decides the reach of that test is not the shot budget but the sector it is
poured into, and Table~\ref{tab:molscope} gives it.

\begin{table*}[tp]\centering\small
\caption{\textbf{What the nineteen-molecule reconstruction tests.} The suite grouped by whether the
reported error is an approximation or an identity. $D$ is the dimension of the $(N{+}1)$ sector the
reconstruction runs in---all nineteen species take the addition branch---and $|\mathcal S|$ the
absolute subspace size at the end of the protocol. Top row: five sectors of
dimension $2$, one of $3$ and one of $6$, where the reported error is identically zero and where in
three cases the sampled subspace is the entire two-dimensional sector. Middle row: five further cases
whose subspace contains the probe's Krylov space $\mathcal K(H,\phi)$, so what is reported is floating-point round-off (containing the probe's support alone would not suffice, Sec.~\ref{sec:retraction}). Bottom row: the seven cases in which a truncation costs something. The largest sector in the
suite is $300$, the dimension of the $L=6$ Hubbard sector of Sec.~\ref{sec:scope:hw}, against
$731\,808$ at $L=12$ on the lattice side of this work.}\label{tab:molscope}
\begin{tabular}{lcccc}
\toprule
 & count & $D$ & $|\mathcal S|$ & rel-$L_1$\\
\midrule
sector of dimension $\le6$       & $7$ & $2$--$6$    & $1$--$2$   & $0$ exactly\\
subspace contains $\mathcal K(H,\phi)$ & $5$ & $24$--$90$  & $6$--$12$  & $\le1.4\times10^{-14}$\\
genuine truncation               & $7$ & $20$--$300$ & $5$--$149$ & $2.2\times10^{-6}$--$1.0\times10^{-4}$\\
\bottomrule
\end{tabular}
\end{table*}

\begin{table*}[tp]\centering\small
\caption{\textbf{The nineteen-molecule fermionic-magic suite}, computed in a cc-pVDZ
active space; the active-space CASCI energies are converged to $1.1\times10^{-13}$~Ha, which is a
convergence check and not a validation against an independent method (Sec.~\ref{sec:scope:mol}). The fermionic
AntiFlatness $\mathcal F_1$ ($k{=}1$ Majorana antiflatness; $\mathcal F_1{=}0$ iff Gaussian/mean-field, maximum
$\mathcal F_1^{\max}{=}2n_{\rm o}$) is listed near equilibrium and at dissociation, in a valence
bond-breaking active space CAS$(n_e,n_{\rm o})$. Fermionic magic switches on for covalent
bond-breaking, is already large at equilibrium for inherently multireference species
($\mathrm{C_2}$, $\mathrm{O_2}$), and stays near-Gaussian for ionic and weakly-correlated
bonds---so $\mathcal F_1$ marks the multireference, strongly-correlated regime rather than bond length
\emph{per se}. The last two columns validate the \emph{method}: at the dissociation geometry the
reconstruction of the (orbital-projected) spectral function $A(\omega)$ from a subspace drawn in
$2.16\times10^{5}$ shots matches the exact Lehmann result to relative-$L_1$ error, on a subspace that
is a fraction $|\mathcal{S}|/D$ of the full $(N{+}1)$ sector. In twelve of the nineteen the subspace contains the probe's Krylov space and the agreement is exact by construction; see
Table~\ref{tab:molscope}. The listed $\mathcal F_1=4\,\mathrm{tr}[\gamma(1-\gamma)]$
are directly computed spin-orbital values (for the five open-shell species they are \emph{not} $2N_{\rm u}$;
see Sec.~\ref{sec:nogo}). Per-molecule bond lengths, active spaces, FCI residuals, and the
absolute $|\mathcal S|$, sector dimension $D$ and minimal bond dimension $\chi$ are listed in
Sec.~\ref{app:repro}; the absolute FCI energies are not deposited. $^{\dagger}$: near-degenerate ground state.}
\label{tab:molecules}
\begin{tabular}{llcccccc}\hline\hline
Molecule & Bonding & CAS$(n_e,n_{\rm o})$ & $\mathcal F_1^{\rm eq}$ & $\mathcal F_1^{\rm diss}$ & $\mathcal F_1^{\rm diss}/\mathcal F_1^{\max}$ & rel-$L_1$ & $|\mathcal{S}|/D$ \\ \hline
\multicolumn{8}{l}{\emph{Covalent bond-breaking}}\\
F$_2$ & covalent single & (1+1,2) & 0.068 & 3.985 & 1.00 & $<\!10^{-13}$ & 0.50 \\
N$_2$ & covalent triple & (3+3,6) & 0.847 & 11.377 & 0.95 & $2\times10^{-5}$ & 0.21 \\
H$_2$O & multicentre & (2+2,4) & 0.007 & 6.936 & 0.87 & $<\!10^{-13}$ & 0.50 \\
HF & polar single & (1+1,2) & 0.001 & 3.460 & 0.87 & $<\!10^{-13}$ & 1.00 \\
NH$_3$ & multicentre & (3+3,6) & 0.080 & 10.289 & 0.86 & $1\times10^{-5}$ & 0.47 \\
CO & polar triple & (3+3,6) & 0.688 & 8.192 & 0.68 & $1\times10^{-4}$ & 0.41 \\
H$_6$ & H-chain & (3+3,6) & 0.521 & 7.735 & 0.64 & $2\times10^{-5}$ & 0.50 \\
H$_4$ & H-chain & (2+2,4) & 0.252 & 5.113 & 0.64 & $<\!10^{-13}$ & 0.50 \\
H$_2$ & covalent single & (1+1,2) & 0.038 & 2.511 & 0.63 & $<\!10^{-13}$ & 0.50 \\
BeH$_2$ & multicentre & (2+2,4) & 0.039 & 3.754 & 0.47 & $<\!10^{-13}$ & 0.25 \\
\multicolumn{8}{l}{\emph{Multireference / open-shell}}\\
C$_2$$^{\dagger}$ & covalent double & (4+4,6) & 4.743 & 8.001 & 0.67 & $<\!10^{-13}$ & 0.10 \\
NO & polar (radical) & (4+3,6) & 0.790 & 7.449 & 0.62 & $2\times10^{-6}$ & 0.25 \\
O$_2$ & covalent double (triplet) & (5+3,6) & 1.673 & 6.000 & 0.50 & $1\times10^{-5}$ & 0.25 \\
CN & polar (radical) & (4+3,6) & 1.307 & 5.938 & 0.49 & $1\times10^{-5}$ & 0.24 \\
OH$^{\dagger}$ & polar (radical) & (3+2,4) & 0.013 & 3.703 & 0.46 & $<\!10^{-13}$ & 0.33 \\
\multicolumn{8}{l}{\emph{Ionic / weakly correlated}}\\
BH & metal hydride & (2+2,4) & 0.518 & 0.573 & 0.07 & $<\!10^{-13}$ & 0.25 \\
BeH & metal hydride (radical) & (2+1,3) & 0.035 & 0.386 & 0.06 & $<\!10^{-13}$ & 0.67 \\
LiH & metal hydride & (1+1,2) & 0.002 & 0.075 & 0.02 & $<\!10^{-13}$ & 1.00 \\
LiF & ionic & (1+1,2) & 0.000 & 0.001 & 0.00 & $<\!10^{-13}$ & 1.00 \\
\hline\hline\end{tabular}
\end{table*}

Twelve of the nineteen reconstructions are exact by construction rather than by approximation, and the
dimensions say why: seven sectors have dimension six or less, and in five more the six to twelve sampled determinants already contain $\mathcal K(H,\phi)$. The suite therefore cannot show that the reconstruction is efficient across chemistry: on twelve of the nineteen there is nothing for efficiency to mean. Its energy check is an active-space CASCI energy against itself, and the residual the appendix records is $1.1\times10^{-13}$~Ha across all
$38$ points---a convergence check on a deterministic diagonalization, worth keeping in the
reproducibility appendix at its true magnitude, but not a validation against an independent method.

What the suite does test is the bottom row of Table~\ref{tab:molscope}, and it is a result. On seven
Hamiltonians---$\mathrm{O_2}$, CN, NO, CO, $\mathrm{N_2}$, $\mathrm{NH_3}$ and $\mathrm{H_6}$, in
sectors of dimension $20$ to $300$---the same primitive reproduces the orbital-projected $A(\omega)$ to
a relative $L_1$ between $2.2\times10^{-6}$ and $1.0\times10^{-4}$ from $21\%$ to $50\%$ of the sector.
Those seven are independent chemical Hamiltonians rather than one molecule repeated, and their
agreement is a consistency check of the primitive outside the Hubbard family. It is not a scaling
statement and is not on a common axis with Sec.~\ref{sec:prices}, since the Hamiltonians, the
broadening and the error level all differ. Per-species dimensions, subspace sizes and bond dimensions
are in Sec.~\ref{app:repro}, whose column $D$ is the $N$-electron dimension; the $(N{+}1)$ dimension of each species is $\binom{n_{\rm o}}{n_\uparrow+1}\binom{n_{\rm o}}{n_\downarrow}$ ($300$ against $400$ for $\mathrm{N_2}$).

\section{The dequantization test}
\label{sec:scope:dequant}\label{app:dequant}
The one place where a learned model could remove the quantum sampling from this pipeline is the
generation of the dominant determinants, and we tested it on stretched $\mathrm{N_2}$ (cc-pVDZ,
$R=2.5$~\AA, CAS(10e,12o), $792$ strings per spin, orbitals in the $D_{\infty h}$ symmetry-adapted
gauge), where the exact energy is available to score against. A device head is emulated by drawing
$1\,000$ configurations from the exact FCI marginal and corrupting them with a Heron-class noise
model---$5\%$ depolarizing replacement and asymmetric readout at the median rates of the
FakeTorino calibration, $0.0229$ for $1\to0$ and $0.0200$ for $0\to1$---after which
$24\%$ of the shots carry the wrong particle number. From that one histogram a subspace of
$D=120$ strings per spin is built four ways and scored on the exact energy, over ten seeds paired
across the four. A single-pass analogue of the occupation-based configuration recovery of
Ref.~\cite{robledomoreno2025} gives $38.22\pm2.62$~mHa. A cheap classical Epstein--Nesbet perturbative
prior~\cite{cipsi1973}, which completes the subspace from one application of $H$ to the Hartree--Fock
determinant---no additional shots, no additional quantum processing---gives $24.38\pm4.02$~mHa. A generative
flow network~\cite{bengio2021} trained on the recovered histogram fused with that prior gives
$22.57\pm1.60$~mHa, and trained on the histogram alone $25.55\pm1.84$~mHa (means and standard deviations over the
ten seeds).

The classical step is resolved beyond any reading of the seeds: it lowers the error by
$13.85$~mHa, a $36\%$ reduction, in all ten seeds, with paired $t_9=12.37$; at $500$
shots it is larger still, $54\%$ with $t_9=36.2$. The generative step is not. On top
of the classical prior the fused generator lowers the error by a further $1.80\pm4.78$~mHa,
in six of the ten seeds, with paired $t_9=1.19$ ($p=0.26$), and without the prior it is
not distinguishable from the classical control ($1.17$~mHa, $\lvert t_9\rvert=0.92$). At
$500$ shots both generative arms are \emph{worse} than the classical control: the fused one by
$5.31$~mHa ($\lvert t_9\rvert=3.27$, $p=0.010$, worse in nine of the ten seeds),
the one without the prior by $15.33$~mHa ($\lvert t_9\rvert=10.3$).

The two halves of the ten seeds are the two five-seed runs of the companion
study~\cite{review}, and they show why five seeds do not decide this comparison. Over seeds
$0$--$4$ the fused generator is ahead of the classical control by $5.23\pm3.19$~mHa in
every seed, with $t_4=3.67$, beyond the two-sided critical value $2.776$; over seeds $5$--$9$
it is behind, $-1.62\pm3.46$~mHa with $t_4=-1.05$. The companion reports
$+5.23\pm3.19$ and $-1.83$~mHa for the same two runs; the first is reproduced here to the
fourth decimal and the second to $0.21$~mHa. A five-seed $t$ test can therefore clear its critical
value and fail to replicate, and a distribution-free test cannot help, since with five seeds the
exact two-sided sign test and the signed-rank test both floor at $p=2/2^{5}=0.0625$; with ten the
floor is $0.002$, which the classical step reaches. \texttt{src/gflow\_dequant.py} writes every seed
to \texttt{data/gflow\_runs/}, and \texttt{data/gflow.json} holds the paired statistics quoted here.

What the test establishes is the statement we keep. On this instance the dominant determinants are
supplied by the sampled histogram (here an emulated one, drawn from the exact ground state and then
noise-corrupted, so the test does not measure whether a device state would supply them) and completed
by a classical prior, not manufactured by a learned generator: the largest improvement available to an under-sampled head costs no quantum resources at
all and should be the default post-processing step, since a demonstration that omits it is measured
against a baseline that is too weak, and a learned generator adds nothing on top of it that ten
seeds resolve.

The class of learned constructions that select or generate the determinants
themselves is by now several---neural quantum states built on selected configurations,
autoregressive and dual-channel samplers for selected configuration interaction, and neural-network
classifiers~\cite{nqssc2026,arnnsci2026,hinqs2026,zeni2026}---and the defensible role for a learned
model is the one that consumes the device histogram, as physics-informed generative recovery of
hardware samples does~\cite{pigensqd}.

A learned model can also schedule device work, as generative models for
measurement grouping do~\cite{vargashernandez2024,vargashernandez2025}, alongside the non-learned
error-mitigation and validation protocols that already sit at that
interface~\cite{qesem2025,algorithmiq2026} and learners that process device output coherently in
quantum memory~\cite{huang2022}; this is the boundary the companion study maps from the
ground-state side~\cite{review}.

\end{document}